\documentclass[11pt]{article}
\pdfoutput=1
\usepackage{jheppub}

 \usepackage{tabu}
\usepackage[vcentermath]{youngtab}
\usepackage[usenames,dvipsnames,table]{xcolor}
\usepackage{graphicx,amsmath,amssymb,amsthm,multirow,array,bm,bbm,esint}
\usepackage[mathscr]{eucal}
\usepackage[bbgreekl]{mathbbol}
\usepackage{epsf,amsfonts}
\usepackage{slashed}
\usepackage[numbers,sort&compress]{natbib}
\usepackage{minibox}
\usepackage{array,tikz-cd}
\usepackage{slashed}
\usepackage{rotfloat}
\usepackage{afterpage}
\usepackage{ccaption}
\usepackage{arydshln}
\usepackage[T1]{fontenc}

  \captionnamefont{\bfseries}
  \captiontitlefont{\small\sffamily}
  \captiondelim{: }
  \hangcaption

\newtheorem{theorem}{Theorem}

\usepackage{rotating}
\usepackage{pdflscape}
\usepackage{array}
\newcolumntype{H}{>{\setbox0=\hbox\bgroup}c<{\egroup}@{}}

 \def\shadeB{\cellcolor{blue!5}}
\def\shadeR{\cellcolor{red!5}}

\definecolor{rust}{rgb}{0.8,0.2,0.2}

\def\FH#1{{\color{cyan}{ #1}}}

\newcommand{\Ad}{\text{ad}}

\newcommand{\Dut}{\SF{\mathcal{D}}}
\newcommand{\z}{\mathbf{z}}
\newcommand{\zsf}{\SF{\z}}
\newcommand{\Adj}{\mathfrak{w}_{\,\smallT}}
\newcommand{\uA}{\breve{\mu}}
\newcommand{\uB}{\breve{\nu}}
\newcommand{\uC}{\breve{\rho}}
\newcommand{\uD}{\breve{\lambda}}
\newcommand{\Dwv}{\SF{\mathfrak{D}}}
\newcommand{\Cwv}{\SF{\mathscr{C}}}
\newcommand{\Cwvt}{\widetilde{\mathscr{C}}}

\newcommand{\gpsi}{\SF{\mathfrak{g}}^{(\psi)}}
\newcommand{\gpsib}{\SF{\mathfrak{g}}^{(\psib)}}

 \def\shadeB{\cellcolor{blue!5}}
\def\shadeR{\cellcolor{red!5}}

\definecolor{rust}{rgb}{0.8,0.2,0.2}

\def\FH#1{{\color{cyan}{ #1}}}

\newcommand{\prn}[1]{\left ( #1 \right )}
\newcommand{\brk}[1]{\left [ #1 \right ]}
\newcommand{\bigbr}[1]{\Bigl\{ #1 \Bigr\} }
\newcommand{\dbrk}[1]{(#1)_\Kref}

\newcommand{\half}{\frac{1}{2}}
\newcommand{\quarter}{\frac{1}{4}}

\newcommand{\Tr}[1]{\hbox{Tr}\left(#1\right)}

\newcommand{\vev}[1]{\langle #1 \rangle}

\newcommand{\rhoi}{\hat{\rho}_{\text{initial}}}
\newcommand{\rhoT}{\hat{\rho}_{_T}}

\newcommand{\QSK}{\mathcal{Q}_{_{SK}}}
\newcommand{\QSKb}{\overline{\mathcal{Q}}_{_{SK}}}
\newcommand{\QKMS}{\mathcal{Q}_{_{KMS}}}
\newcommand{\QKMSb}{\overline{\mathcal{Q}}_{_{KMS}}}

\newcommand{\dSK}{\mathbb{d}_{_\text{SK}}}
\newcommand{\dSKb}{\bar{\mathbb{d}}_{_\text{SK}}}

\newcommand{\Qzero}{\mathcal{Q}^0_{_{KMS}}}
\newcommand{\Qbeta}{\mathscr{L}_{_{KMS}}}

\newcommand{\IKMS}{\SF{\cal I}^{\text{\tiny{KMS}}}}
\newcommand{\LKMS}{\SF{\cal L}^{\text{\tiny{KMS}}}}
\newcommand{\IKMSb}{\SF{\overline{\cal I}}{}^{^{\text{\tiny{KMS}}}}}
\newcommand{\IKMSzero}{\SF{\cal I}^{\text{\tiny{KMS}}}_0\,}
\newcommand{\iKMS}{{\iota}^{\text{\tiny{KMS}}}}
\newcommand{\iKMSb}{{\overline{\iota}^{\text{\tiny{KMS}}}}}
\newcommand{\iKMSzero}{{\iota}^{\text{\tiny{KMS}}}_0\,}
\newcommand{\lKMS}{{\ell}^{\text{\tiny{KMS}}}}

\newcommand{\Q}{\mathcal{Q}}
\newcommand{\Qb}{\overline{\mathcal{Q}}}

\newcommand{\QC}{\mathbb{d}_{_{\text{\tiny C}}}}
\newcommand{\QCb}{\overline{\mathbb{d}}_{_{\text{\tiny C}}}}
\newcommand{\QW}{\mathbb{d}_{_{\text{\tiny W}}}}

\newcommand{\QWb}{\overline{\mathbb{d}}_{_{\text{\tiny W}}}}

\newcommand{\DSK}{\mathcal{D}_{_\text{SK}}}
\newcommand{\DSKb}{\overline{\mathcal{D}}_{_\text{SK}}}

\newcommand{\Op}[1]{\mathbb{#1}}
\newcommand{\OpH}[1]{\widehat{\mathbb{#1}}}

\newcommand{\SKR}[1]{\mathbb{#1}_{\skR}}
\newcommand{\SKL}[1]{\mathbb{#1}_{\skL}}

\newcommand{\SKAv}[1]{\mathbb{#1}_{{av}}}

\newcommand{\SKDif}[1]{\mathbb{#1}_{{dif}}}

\newcommand{\SKG}[1]{\mathbb{#1}_{_G}}
\newcommand{\SKGb}[1]{\mathbb{#1}_{_{\overline{G}}}}

\newcommand{\gh}[1]{\text{gh}(#1)}

\newcommand{\gradcomm}[2]{ \brk{ #1, #2 }_{\scriptscriptstyle \pm} }
\newcommand{\comm}[2]{ \brk{ #1, #2 }}

\newcommand{\thb}{{\bar{\theta}} }
\newcommand{\thetab}{\bar{\theta} }
\newcommand{\SF}[1]{\mathring{#1}}
\newcommand{\Ibar}{\overline{\mathcal{I}}}

\newcommand{\LamS}{\SF{\Lambda}}
\newcommand{\As}{\SF{\mathscr{A}}}
\newcommand{\Ascr}{\mathscr{A}}
\newcommand{\Ath}{\mathscr{A}_{\theta}}
\newcommand{\Athb}{\mathscr{A}_{\thetab}}
\newcommand{\Fs}{\SF{\mathscr{F}}}

\newcommand{\Bs}{\SF{\mathscr{B}}}

\newcommand{\SLref}{\SF{\Lambda}_\Kref}

\newcommand{\psib}{\overline{\psi}}

\newcommand{\phiT}{\phi_{_{\,\smallT}}}
\newcommand{\phibT}{\overline{\phi}_{_{\,\smallT}}}
\newcommand{\phizT}{\phi^0_{_{\,\smallT}}}
\newcommand{\GT}{G_{_{\,\smallT}}}
\newcommand{\GbT}{\overline{G}_{_{\,\smallT}}}
\newcommand{\GhT}{G_{_{\,\smallT}}}
\newcommand{\GhbT}{\overline{G}_{_{\,\smallT}}}
\newcommand{\BT}{B_{_{\,\smallT}}}
\newcommand{\etaT}{\eta_{_{\,\smallT}}}
\newcommand{\etabT}{\overline{\eta}_{_{\,\smallT}}}

\newcommand{\etab}{\overline{\eta}}

\newcommand{\tx}{\tilde{X}}
\newcommand{\xpsi}{X_\psi}
\newcommand{\xpsib}{X_{\psib}}

\definecolor{arsenic}{rgb}{0.28, 0.28, 0.6}

\def\source#1{{\textcolor{arsenic}{ #1}}}

\def\sBdel{\source{{\sf B}_\Delta}}
\def\sAt#1{\source{\mathcal{F}_{#1}}}

\newcommand{\Kref}{{\bm \beta}}
\newcommand{\Lref}{\Lambda_\Kref}

\newcommand{\gref}{{\sf g}}
\newcommand{\Aref}{{\sf A}}
\newcommand{\shref}{{\sf h}}

\newcommand{\Lagref}{\mathscr L}

\newcommand{\etaref}{\bm \eta}

\newcommand{\Gref}{{\sf \mathbf G}}
\newcommand{\Nref}{{\sf \mathbf N}}
\makeatletter
  \newcommand\Ttiny{\@setfontsize\Ttiny{1pt}{2}}
\makeatother

\newcommand{\TEMref}{{\sf\mathbf T}}

\newcommand{\Cref}{{\sf C}}

\newcommand{\Dref}{\mathbbm d}

\newcommand{\lieD}{\pounds}

\newcommand{\acc}{{\mathfrak a}}

\newcommand{\Lag}{{\mathcal L}}

\newcommand{\K}{\mathrm{K}}

\newcommand{\skR}{\text{\tiny R}}
\newcommand{\skL}{\text{\tiny L}}

\newcommand{\smallT}{{\sf \!{\scriptscriptstyle{T}}}}

\newcommand{\UT}{U(1)_{\scriptstyle{\sf T}}}

\newcommand{\PS}{{\rm H}_S}
\newcommand{\PV}{{\rm H}_V}
\newcommand{\PF}{{\rm H}_F}
\newcommand{\LS}{{\overline{\rm H}}_S}
\newcommand{\GV}{{\overline{\rm H}}_V}

\newcommand{\LT}{{\rm L}_{\,\smallT}}

\title{Effective Action for Relativistic Hydrodynamics: Fluctuations, Dissipation, and Entropy Inflow}

\author[a]{Felix M. Haehl}
\author[b]{\!, R.\ Loganayagam}
\author[c]{\!, Mukund Rangamani}
\affiliation[\,a]{Department of Physics and Astronomy, University of British Columbia,\\
6224 Agricultural Road, Vancouver, B.C.\ V6T 1Z1, Canada.}
\affiliation[\,b]{International Centre for Theoretical Sciences (ICTS-TIFR), \\
Shivakote, Hesaraghatta Hobli, Bengaluru 560089, India.}
\affiliation[\,c]{
Center for Quantum Mathematics and Physics (QMAP)  \\
Department of Physics, University of California, Davis, CA 95616 USA.}
\emailAdd{f.m.haehl@gmail.com}
\emailAdd{nayagam@gmail.com}
\emailAdd{mukund@physics.ucdavis.edu}

\abstract{We present a detailed  and self-contained analysis of the universal Schwinger-Keldysh effective field theory which describes macroscopic thermal fluctuations of a relativistic field theory, elaborating on our earlier construction \cite{Haehl:2015uoc}. We write an effective action for appropriate hydrodynamic Goldstone modes and fluctuation fields, and discuss the symmetries to be imposed. The constraints imposed by fluctuation-dissipation theorem are manifest in our formalism. Consequently, the action reproduces hydrodynamic constitutive relations consistent with the local second law at all orders in the derivative expansion, and captures the essential elements of the eightfold classification of hydrodynamic transport of \cite{Haehl:2015pja}. We demonstrate how to recover the hydrodynamic entropy and give predictions for the non-Gaussian hydrodynamic fluctuations.
 
The basic ingredients of our construction involve (i) doubling of degrees of freedom \emph{a la} Schwinger-Keldysh, 
(ii) an emergent gauge $\UT$ symmetry associated with entropy which is encapsulated in a Noether current \emph{a la} Wald, and (iii) a BRST/topological supersymmetry imposing the fluctuation-dissipation theorem \emph{a la} Parisi-Sourlas. The overarching mathematical framework for our construction is provided by the  balanced equivariant cohomology of thermal translations, which captures the basic constraints arising from the Schwinger-Keldysh doubling, and the thermal Kubo-Martin-Schwinger relations. All these features are conveniently implemented in a covariant superspace formalism. An added benefit is that the second law can be understood as being due to entropy inflow from the Grassmann-odd directions of superspace.}

\begin{document}
\maketitle


\newpage 

\newpage 
\part{Introduction \& Background }
\section{Introduction}
\label{sec:intro}

The dynamics of quantum field theories out-of-equilibrium encompasses many interesting physical phenomena which are readily observable in nature. Often one is interested in the macroscopic behaviour of the system after transient effects have settled down. In a wide variety of examples we know empirically  that the collective dynamics of the low energy degrees of freedom leads to new dramatic effects including effective non-unitarity, entropy production and dissipation. One would like to have a theoretical framework to address these issues and isolate potential universal characteristics that are insensitive to the specific microscopic details. In the case of equilibrium (typically near ground state) dynamics, the Wilsonian paradigm makes clear that one has to isolate the relevant macroscopic degrees of freedom and ascertain the generic dynamics for them subject to various symmetry considerations. An analogous framework for non-equilibrium dynamics is the goal one would like to aspire to.

Whilst the question for generic out-of-equilibrium dynamics remains as yet unclear, in recent years progress has been made on understanding the situation in a near-equilibrium regime where hydrodynamic effective field theories operate. Inspired by the structure of the Schwinger-Keldysh functional integral  there have been several works \cite{Haehl:2015foa,Haehl:2015pja,Crossley:2015evo,Haehl:2015uoc,Haehl:2016pec,Haehl:2016uah,Glorioso:2016gsa,Jensen:2017kzi,Gao:2017bqf,Glorioso:2017fpd} dedicated to constructing a framework to capture near thermal effective field theories in the hydrodynamic regime.\footnote{ An  earlier attempt to construct dissipative hydrodynamic effective actions was made in \cite{Kovtun:2014hpa} which took its inspiration from the Martin-Siggia-Rose (MSR) construction \cite{Martin:1973zz}.  We will give a more complete discussion of some of the earlier attempts later in the introduction, \S\ref{sec:related}.}  Our aim is to elaborate on these constructions and set out a comprehensive framework for studying such effective field theories. We build on recent work and provide further details underlying the construction of topological sigma models capturing dissipative hydrodynamics as described in \cite{Haehl:2015uoc}. 

The central problem in this regard is to decide on  the basic symmetry principles/symmetry breaking pattern behind the effective field theory in the hydrodynamic regime. 
The core challenge for these symmetry principles is to automatically explain both the emergence of a  macroscopic arrow of time and the existence of a local entropy current, along with the non-conservation of the latter. These requirements necessitate a novel form of effective field theory very different from existing paradigms. In our previous work \cite{Haehl:2015foa,Haehl:2015uoc}, we posited a 
three-fold symmetry structure whose interplay successfully reproduces the hydrodynamic effective theory. Our proposal consists of: 
\begin{enumerate}
\item A set of `twisted' super-symmetries   emerging from the Schwinger Keldysh doubling.
\item An emergent thermal or `entropic' gauge symmetry (denoted as $U(1)_T$) emerging from the near thermal structure. Its gauge current is the entropy current.
\item A particular superspace component of $U(1)_T$ field strength acts as an order parameter for CPT breaking. Its expectation value then leads to the emergence of arrow of time.
\end{enumerate}
In this work, we will substantially add to the explicit computations which support the above conjecture. In particular, we will show that a non-trivial statement required for the self-consistency of our proposal does hold: \emph{for the entropy current to be a gauge current, the apparent non-conservation of entropy in fluid dynamics should lift to an appropriate conservation statement within our framework}. We will see that this is indeed true; the physical hydrodynamic entropy is only a part of larger conserved super-current. Our companion paper \cite{Haehl:2018uqv} summarizes the salient features of our construction, especially the fact that \emph{entropy production can be understood as a superspace inflow mechanism}.

In the rest of this introduction, we will summarize various features of our proposal and the resulting framework. While the entire discussion is framed in terms of hydrodynamics, the 
reader should note that holography (more specifically fluid-gravity correspondence \cite{Bhattacharyya:2008jc,Hubeny:2011hd}) requires that these statements also be true in  gravity. If the above set of symmetries are indeed the correct framework for fluid dynamics, it follows that the same symmetries should also underlie black hole physics including the emergence of arrow of time
via a superspace field strength as well as superspace inflow of entropy. The hydrodynamic computations in this work  when combined with fluid-gravity correspondence, force upon us this somewhat radical conclusion.

\subsection{Preview of the general framework}
\label{sec:preview}

The basic philosophy behind our construction, first detailed in \cite{Haehl:2015foa,Haehl:2015uoc}, can be understood as follows (see \cite{Haehl:2016pec,Haehl:2016uah} for detailed reviews). Hydrodynamics is supposed to capture the causal, non-linear response of a physical system perturbed away from equilibrium, in a long-wavelength, low-energy regime. In the microscopic presentation of the theory such response functions are  computed using the Schwinger-Keldysh formalism \cite{Schwinger:1960qe,Keldysh:1964ud} which involves a complex time contour (or a doubled contour with forward (R) and backward (L) evolution) in the functional integral. A key feature is that the response functions are the first non-trivial observables; we insert a sequence of  identical sources to disturb the system followed by a mis-aligned source to facilitate a response/measurement, see Fig.~\ref{fig:hydrosk} for an illustration.

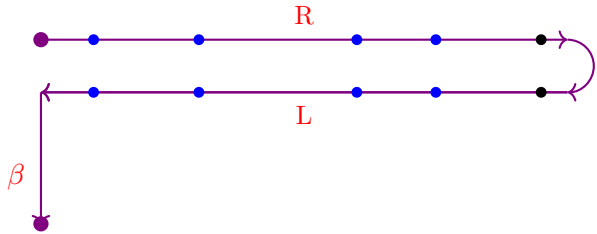
\begin{figure}[t!]
\centering
\begin{tikzpicture}[scale=0.7]
\draw[thick,color=violet,->] (-5,1)  -- (5,1);
\draw[thick,color=violet,->] (5,0) -- (-5,0);
\draw[thick,color=violet,->] (5,0) -- (-5,0);
\draw[thick,color=violet,->] (-5,0) -- (-5,-2.5);
\draw[thick,color=violet,fill=violet] (-5,1) circle (0.75ex);
\draw[thick,color=violet,fill=violet] (-5,-2.5) circle (0.75ex);
\draw[thick,color=violet,->] (5,1) arc (90:-90:0.5);
\draw[thick, color=red]
{ (0,1.1) node [above] {\small{R}}
(0,-0.8) node [above] {\small{L}}
(-5.1,-1.6) node [left] {$\beta$}
 };
\draw[thick,color=blue,fill=blue] (-4,1) circle (0.5ex);
\draw[thick,color=blue,fill=blue] (-4,0) circle (0.5ex);
\draw[thick,color=blue,fill=blue] (-2,1) circle (0.5ex);
\draw[thick,color=blue,fill=blue] (-2,0) circle (0.5ex);
\draw[thick,color=blue,fill=blue] (1,1) circle (0.5ex);
\draw[thick,color=blue,fill=blue] (1,0) circle (0.5ex);
\draw[thick,color=blue,fill=blue] (2.5,1) circle (0.5ex);
\draw[thick,color=blue,fill=blue] (2.5,0) circle (0.5ex);
\draw[thick,color=black,fill=black] (4.5,1) circle (0.5ex);
\draw[thick,color=black,fill=black] (4.5,0) circle (0.5ex);
\end{tikzpicture}
\caption{An illustration of the hydrodynamic observables in the Schwinger-Keldysh functional integral phrased in the Keldysh average-difference basis. Response functions are correlators of a sequence of difference operators (the external disturbances, denoted as blue dots) followed by an average operator (the measurement, denoted by the black dots) in the future.}
\label{fig:hydrosk}
\end{figure}

The canonical way to view the Schwinger-Keldysh path integral is in terms of an action on an extended Hilbert space $\mathcal{H}_\skR \otimes \mathcal{H}_\skL^*$ corresponding to the top and bottom legs of the contour, respectively.  One correspondingly doubles the operator algebra of the theory for we can independently insert operators on either leg of the contour. The key insight of the recent discussions was to interpret some of the well known identities for correlation functions which are, for example, well reviewed in \cite{Chou:1984es,Weldon:2005nr}, as consequences of  topological BRST symmetries inherent in the Schwinger-Keldysh construction (see \cite{Geracie:2017uku} for a Hilbert space perspective). 

The primary contention is that a Schwinger-Keldysh effective theory is characterized by a doubling of fields and a degeneration to a topological theory when the  sources of two copies are appropriately aligned  \cite{Haehl:2015foa,Crossley:2015evo}.  As explained in  \cite{Haehl:2016pec} the BRST symmetries may be understood as a consequence of the redundancy built into the Schwinger-Keldysh doubling. This line of reasoning led us to argue that the natural language for Schwinger-Keldysh is to work with a quartet of operators corresponding to the conventional doubled operators $\Op{O}_\skL, \Op{O}_\skR$, and their Grassmann odd counterparts $\SKG{O}, \SKGb{O}$ obtained by the action of the BRST charges.  One can succinctly capture this by working in a superspace with two Grassmann odd scalar directions $\{\theta,\thb \}$ which take care of the doubling and the topological limit aspects of a Schwinger-Keldysh construction. Altogether the most natural way to view the Schwinger-Keldysh construction is to work with a super-operator algebra on which the BRST charges dubbed $\QSK$ and $\QSKb$ act as super-derivations $\QSK \sim \partial_{\thb}$ and  $\QSKb \sim \partial_\theta$. The topological invariance is then just the Grassmann translation invariance. A detailed account of how this translational invariance can be used to give superspace rules for computing the microscopic field theory correlation functions can be found in \cite{Haehl:2016pec,Geracie:2017uku}.

It is important to  emphasize that the Schwinger-Keldysh BRST symmetries essentially capture the constraints imposed by unitary evolution of the microscopic degrees of freedom. Furthermore, by virtue of being topological one expects these symmetries to be robust under the renormalization group flow. Consequentially, they provide useful guideposts on how to organize the low energy dynamics.

\subsection{$\UT$ gauge invariance and thermal equivariance}
\label{sec:related}

There is additional structure for near-thermal dynamics owing to the Gibbs structure of the density matrix. Thermal field theory in its Euclidean avatar implements thermality by imposing periodicity in imaginary time. If we wanted to study real-time dynamics of thermal  systems, this Euclidean description should then be analytically continued into a Schwinger-Keldysh path integral.  The periodicity in the Euclidean description then translates to a set of non-local conditions called  Kubo-Martin-Schwinger (KMS) conditions \cite{Kubo:1957mj,Martin:1959jp} satisfied by the real time correlators. We expect that these conditions should deform appropriately to describe real time correlators in a generic fluid dynamical state. 
A fundamental question in any effective theory of thermal systems is how to implement a sensible deformation of these  (temporal) non-local conditions in a local effective theory.

Recently an answer to this question has emerged, mainly in the context of treating fluid dynamics as an effective theory \cite{Haehl:2014zda,Haehl:2015pja}. It was conjectured the correct local principle to enforce is to demand an emergent `entropic' gauge symmetry dubbed $\UT$, which in real time provides the  correct analytic continuation of the Euclidean periodicity. Its gauge current is  the entropy current (this statement can be thought of as a generalization of Wald's idea that equilibrium entropy is a Noether charge \cite{Wald:1993nt}). This emergent  KMS gauge symmetry can be understood in terms of the topological structure of the Schwinger-Keldysh construction. In particular, it was argued in \cite{Haehl:2016pec} that the additional constraints coming from KMS invariance lead to a quartet of operations that act on the Schwinger-Keldysh super-operator algebra. Two of these are Grassmann odd, thermal counterparts of the BRST charges,  $\QKMS, \QKMSb$ and the two others are Grassmann even generators $\Qzero, \Qbeta$.\footnote{ More precisely, $\{\QKMS,\QKMSb,\Qzero\}$ are interior contractions in the language of extended equivariant cohomology, while $\Qbeta$ is a Lie derivation. A closely related superalgebra was posited in \cite{Crossley:2015evo} where the authors convolve the KMS symmetry with a discrete {\sf CPT} action. This latter algebra can be viewed in the statistical (high temperature) limit as a restriction of the equivariant algebra with the symmetry left  ungauged.}

One can intuitively understand the thermal generators in the following fashion. For a thermal system one can view real time dynamics as occurring on a background spacetime that admits a fibration by a thermal circle. Recall that we are used to analyzing equilibrium dynamics in the Euclidean framework as a statistical field theory in a  geometry which is a thermal circle fibration over a spatial background. We argue that this perspective continues to be useful in the dynamical context. Given a notion of local temperature and a local choice of inertial frame measuring it (as is usual in hydrodynamics), our contention is that the background spacetime geometry on which the quantum dynamics occurs, should be viewed as living on a thermal bundle over a Lorentzian base. The local fibres being given by the thermal vector $\Kref^\mu(t,x)$, whose norm gives a measure of the local temperature; it picks out the inertial frame for local equilibrium.

The generator $\Qbeta$ implements  translations along the thermal vector $\Kref^\mu$. Owing to the KMS conditions, we can equivalently say that $\Qbeta$ implements gauge transformations around the thermal circle. An operator $\Op{O}(t,x)$ when acted upon by this generator gets Lie dragged along the local fibre, viz.,  $\Op{O} \mapsto \Op{O} + \lieD_\Kref \Op{O}$. In general this is a non-local, discrete operation since one compares an operator with its thermal counterpart. The latter is separated by an imaginary amount set by the local temperature. We will eventually postulate a continuum version, but before doing so let us intuit the rationale for the other KMS charges.

Owing to the underlying Schwinger-Keldysh BRST symmetries inherent in real time dynamics, it follows that $\Qbeta$ cannot act in isolation. Given that the symmetries act on the operator superalgebra, it also follows that there ought to be quartet of KMS operations that are interlinked by the Schwinger-Keldysh BRST charges. Said differently, it does not suffice for there to be a single KMS operation $\Qbeta$ since the superspace structure demands that it uplift to an appropriate superspace operation. The explicit action of these charges on the super-operator algebra can then be constructed directly. We will review the resulting SK-KMS algebra generated by the  six charges $\{\QSK,\QSKb, \QKMS, \QKMSb, \Qzero,\Qbeta\}$ below.  As noted in the references cited above and explained in \cite{Haehl:2016uah} this algebra exemplifies an extended equivariant cohomology algebra.

The complete structure of this algebra and its implications for generic thermal systems have not yet been fully understood. However, for the analysis of low energy dynamics in near-thermal situations, as in the hydrodynamic context, we can make some useful simplifications. Insofar as the low energy hydrodynamic regime is concerned, one can effectively work in the high temperature limit, where the local thermal circle becomes infinitesimal. This has the salubrious effect of allowing us to both make the thermal translations local, and pass into the continuum limit.  We then view the low energy theory as being equivariant with respect to the thermal gauge symmetry translating operators around the thermal circle, which is nothing but the $\UT$ KMS-gauge symmetry.

The basic framework for viewing such thermal equivariant cohomology algebras was outlined in
\cite{Haehl:2016uah}.  Our aim there was to explain the general structures and explicate the origins of the thermal $\UT$ gauge symmetry. We also argued that this framework could be used to understand the simplest Schwinger-Keldysh effective theory in the thermal regime. The system under question is the worldline description of a thermal particle, often called Brownian particle.  The resultant macroscopic
dynamics  is the one given by the Langevin equation. We kept the discussion there simple by focusing on the particle motion in one-dimension, which amounts to studying a worldline sigma model with a one dimensional target space constrained by the thermal equivariant cohomology algebra. The natural generalization is to extend the discussion to Brownian branes \cite{Haehl:2015foa} which can be viewed as reparameterization invariant worldvolume sigma models. The target space is the spacetime in which these branes lie embedded. The particular theory arising out of our considerations is a natural generalization of twisted supersymmetric quantum mechanics studied by Witten in the context of Morse theory 
 \cite{Witten:1982im}. Amongst the Brownian branes, the space-filling one, captures, upon imposition of target space diffeomorphisms, the hydrodynamic effective field theory \cite{Haehl:2015uoc}.

The object of the current analysis  is to examine deeply the underlying mathematical structure of the thermal gauge theory and build up the necessary machinery to construct the sigma models of interest. Whilst we view the current discussion as a necessary elaboration of \cite{Haehl:2015uoc}, it is the first step we need to take to check whether these sigma models with their attendant KMS gauge symmetry match against expectations from  thermal field theory. In particular, we would want to show based on the formalism we are about to explicate that the eightfold classification of hydrodynamic transport described in \cite{Haehl:2015pja} is indeed comprehensive and can be recovered from an effective action.
A related  objective  is to develop necessary mathematical machinery for near-equilibrium effective theory which will capture well the thermal correlations.

In this work, we will focus on the topological or aligned limit, where one needs to write down a theory with doubled topological invariance along with the above mentioned $\UT$ gauge invariance.
This naturally brings us into the remit of equivariant cohomology algebras and one effectively desires a superspace construction that encodes the relevant constraints. We will work in the aforementioned superspace with two Grassmann odd directions and $\{\partial_\theta,\partial_{\thetab} \}$ providing the necessary topological charges. Thus, our problem involves studying $\UT$ gauge invariance associated with thermal translations in the context of a superspace. This interplay between thermal translations and Schwinger-Keldysh superspace results in a rich \emph{thermal supergeometry}
which forms the central subject of this work.

\subsection{A brief history of hydrodynamic effective field theories}
\label{sec:related}

To put our construction in perspective, we give a brief history of hydrodynamic effective actions. For the case of ideal fluids, efforts to construct an action principle date back several decades with works by Taub \cite{Taub:1954zz} and Carter \cite{Carter:1973fk,Carter:1987qr}. 

Recent interest in understanding effective desciption of fluids was rejuvenated in an interesting paper \cite{Nickel:2010pr}, where the authors proposed a useful strategy for identifying the low energy dynamical degrees of freedom in terms of Goldstone modes for broken symmetries.  At the same time working with the local fluid element variables (the Lagrangian description),  \cite{Dubovsky:2011sj} gave a general framework to describe non-dissipative fluid dynamics, which was employed to understand anomalous transport for Abelian flavour anomalies in $1+1$ dimensions in 
\cite{Dubovsky:2011sk}. Building on these works, \cite{Bhattacharya:2012zx} demonstrated how these could be used to understand non-linear dissipative fluids, while \cite{Saremi:2011ab,Haehl:2013kra,Geracie:2014iva} explored the specific features of Hall viscosity in $2+1$ dimensions.

In our first attempt to understand the generality of the formalism, we constructed an action principle for anomalous transport in general in \cite{Haehl:2013hoa}. A curious feature of this construction was that despite the anomalous transport being non-dissipative or adiabatic, one nevertheless had to resort to a Schwinger-Keldysh type doubling of degrees of freedom to construct an effective action. 

While efforts were being expended to understand hydrodynamic effective actions, progress was being made on understanding the constraints on transport from viewing hydrodynamics as a long-wavelength effective field theory, constrained by the requirement that a local form of the second law is upheld on-shell in every fluid configuration (i.e., there exists an entropy current with locally non-negative divergence). This problem was initially explored in \cite{Romatschke:2009kr} and a complete solution to neutral fluids at second order was finally obtained by Sayantani Bhattacharyya in \cite{Bhattacharyya:2012nq}. In the latter work it was shown that there are non-trivial identities that transport coefficients need to satisfy in order to satisfy the second law. Inspired by these developments, \cite{Banerjee:2012iz,Jensen:2012jh} developed the equilibrium partition function formalism from which all the constraints on transport can be obtained. Application of this formalism to understand anomalous transport was explored in  \cite{Jensen:2012kj,Jensen:2013kka,Jensen:2013rga}. This analysis was further refined by Sayantani Bhattacharyya in \cite{Bhattacharyya:2013lha,Bhattacharyya:2014bha} who went on to prove a remarkable  theorem: apart from constraints coming from equilibrium, and the positivity requirements on lowest order dissipative terms (viz., transport coefficients such as shear viscosity, bulk viscosity, conductivity, etc., are non-negative definite), there are no constraints on higher order dissipative terms.

Our first attempt to synthesize these results into a coherent picture culminated in the \emph{eightfold classification of hydrodynamic transport} as described in \cite{Haehl:2014zda,Haehl:2015pja}. This work involved two distinct lines of development: firstly we took the axioms of fluid dynamics at face value, and constructed an explicit  parametrization of independent classes of transport consistent with the second law at all orders in the derivative expansion. This classification, which was named the Eightfold Way, indicated that apart from the obvious dissipative transport, there are 7 independent adiabatic (i.e., non-dissipative) classes. Any transport coefficient not in these classes is forbidden from appearing and belongs to the set of hydrostatic forbidden terms, which can already be inferred from the equilibrium analysis. In a parallel development we demonstrated that an action principle involving two essential ingredients, (a) Schwinger-Keldysh like doubling, and (b) an emergent thermal $\UT$ symmetry, could capture all of the 7 adiabatic classes of transport. The former requirement was on the one hand familiar from the issues encountered in the construction of anomalous transport effective actions, but was on the other surprising since we were dealing with a conservative system. The $\UT$ symmetry, however, was imperative given the doubling, to forbid terms that would be in tension with microscopic unitarity. Equivalently, this symmetry  is necessary to  implement the constraints arising from the KMS condition in equilibrium and ensures that all the Schwinger-Keldysh influence functionals are consistent. 

The structure of this action of adiabatic transport, dubbed the Class $\LT$ action,  is reminiscent of the MSR \cite{Martin:1973zz} construction. Roughly, the hydrodynamic terms arise from an action of the form $T^{ab} \, \tilde{g}_{ab}$, where $T^{ab}$ is the energy-momentum tensor, and $\tilde{g}_{ab}$ is the difference metric in the Schwinger-Keldysh doubled construction. It was independently argued by \cite{Kovtun:2014hpa} that such an MSR like  construction should be the right framework  for dissipative fluid dynamics. Other attempts to construct actions for dissipative hydrodynamics include 
\cite{Grozdanov:2013dba,Endlich:2012vt,Hayata:2015lga,Floerchinger:2016gtl}.

Seeking to understand the origins of the Schwinger-Keldysh doubling and the $\UT$ symmetry led us to unearthing the general framework of thermal equivariance, which has been explained in the sequence of papers \cite{Haehl:2015foa,Haehl:2015uoc,Haehl:2016pec,Haehl:2016uah}. The key ingredients of this construction, viz., the BRST symmetry were also independently argued for by \cite{Crossley:2015evo}. Their construction was further explored in \cite{Glorioso:2016gsa,Gao:2017bqf,Glorioso:2017fpd}. There are some overarching similarities, and some key differences between the two approaches. These have been spelled out in some detail in \cite{Haehl:2017zac} so we will refrain from providing further commentary here. The key similarity is that other approaches seems to result in the same Lagrangian with the same final symmetries in the hydrodynamic limit.

The key difference is that these works do not have an emergent $\UT$ gauge symmetry, and their implementation of KMS invariance as a discrete ${\mathbb Z}_2$ symmetry differs from our viewpoint. Further, much of their work is also done within a non-covariant framework and amplitude expansions, thus obscuring  relations to previous literature as well as our work.  As explained in \cite{Haehl:2017zac}, their $\mathbb{Z}_2$ symmetry is however consistent with our proposal. Finally, \cite{Jensen:2017kzi} provides a superspace description of this alternate construction and sketch some features of hydrodynamic actions involving  Lagrangians with mutually non-local terms. 
 A related approach with a focus on the path integral derivation of hydrodynamics was spelled out in \cite{Hongo:2016mqm,Hongo:2018nzb}.  While we have not tried to reassemble every non-covariant/non-local expressions occurring in these works into a form that would facilitate more direct comparison to our local covariant expressions, we anticipate 
agreement in the regime where the fluid dynamics is correctly reproduced.
 
\paragraph*{Note added:} While this paper was nearing completion, \cite{Jensen:2018hhx} appeared on the arXiv, which has some overlap with this work (especially elements of \S\ref{sec:topsigma}).

\subsection{Outline of the paper}
\label{sec:outline}

In \S\ref{sec:thermaleq} we briefly review the equivariant superspace construction that encodes the Schwinger-Keldysh constraints in near thermal states. 

The second part of this paper contains details on the application to hydrodynamic effective field theories. In \S\ref{sec:hydrodof} we review the sigma model degrees of freedom describing fluid dynamics, and the symmetries to be imposed on an effective action. We give an abstract discussion of the structure of the resulting hydrodynamic effective action in \S\ref{sec:topsigma}. This leads to the realization that the second law can be naturally understood as due to an entropy inflow from the Grassmann odd directions in superspace. To make the abstract discussion more accessible and provide a concrete framework for calculations, we develop the ``MMO limit'' (after Mallick-Moshe-Orland \cite{Mallick:2010su}) of our formalism in \S\ref{sec:mmo}. This limit provides a truncation of our formalism to keep track of only the fields relevant to extract certain information of interest, such as the currents and the entropy production. We provide several examples at low orders in the derivative expansion in \S\ref{sec:mmoEx}. In \S\ref{sec:SigmaModel} we demonstrate how the complete eightfold classification of hydrodynamic transport can be reproduced from our effective action. We discuss some open questions in \S\ref{sec:discussion}. 

The third and fourth parts of this paper, comprise of  a number of appendices. Here we give further details on the formalism and explicit expressions for the superspace representations of various relevant fields. The reader interested in the general formalism of thermal supergeometry is invited to consult these.  In Part III, the  Appendices \ref{sec:skkmsalg}-\ref{sec:formal} captures the details of the $\UT$ gauge symmetry, in particular general aspects of the representation theory. In Part IV, the Appendices \ref{sec:PositionMult} and \ref{sec:MetCurvMult} describe the matter multiplets entering into the construction, while Appendices \ref{sec:gaugesdiff} and \ref{sec:wvconnect} delve into details about various gauge fixing conditions, and construction of appropriate covariant objects which enter the hydrodynamic sigma models. Finally, Appendix~\ref{sec:MMOfurther} collects various expressions that enter into the explicit computations of \S\ref{sec:mmo}and \S\ref{sec:mmoEx}.  This split allows us to keep the main line of development in the paper free of technical complications (to the extent possible).

\section{Review of thermal equivariance}
\label{sec:thermaleq}

Given a unitary QFT in a thermal state, real time correlation functions are computed using the Schwinger-Keldysh functional integral contour; Fig.~\ref{fig:hydrosk}. Some key features of this construction are a doubling of degrees of freedom, and an attendant set of constraints induced by the redundancy thus introduced. Furthermore, in thermal equilibrium, correlation functions are required to satisfy suitable analyticity properties encoded in the KMS condition. The latter owes its origin to the special nature of the thermal density matrix, which involves Hamiltonian evolution in Euclidean time. We review some of the salient facts relating to these features, and their encoding in terms of a thermal equivariance algebra. A detailed exposition of these ideas can be found in our earlier papers \cite{Haehl:2016pec,Haehl:2016uah}.

\subsection{The SK-KMS  superalgebra}
\label{sec:skkms}

The Schwinger-Keldysh generating functional starts with an initial state $\rhoi$ (w.l.o.g.\ a mixed state) of a QFT, and implements a source-deformed forward (R)/backward (L) evolution so as to be agnostic of the future state of the system, viz., 
\begin{equation}\label{eq:SKdef}
 \mathcal{Z}_{SK}[J_\skR,J_\skL] \equiv \text{Tr}\bigbr{ U[J_\skR]\ \rhoi\  (U[J_\skL])^\dag } \,.
\end{equation} 
A key feature of this construction is the localization of  the generating functional on to the initial state when sources are aligned: 
\begin{equation}
\mathcal{Z}_{SK}[J_\skR=J_\skL\equiv J] = \text{Tr}\bigbr{  \rhoi  }. 
\label{eq:eqsource}
\end{equation} 
We take this statement as encoding unitarity in the microscopic formulation, and it implies that correlation functions vanish if their future-most insertion is a `difference operator' of the form $\SKR{O} - \SKL{O}$ (we refer to this as the {\it largest time equation}). In particular, correlators of only difference operators generically vanish. We argued that this would naturally be enforced by a topological BRST symmetry and introduced a pair of supercharges $\{\QSK, \QSKb\}$. The two topological charges are CPT conjugates of each other,\footnote{ We note here that \cite{Crossley:2015evo} also argue for a BRST symmetry, with the difference that they require only one BRST charge to ensure the correct Schwinger-Keldysh functional integral localization upon source alignment. See \cite{Haehl:2017zac} for further comments.} and ensure that the largest time equation is upheld. We refer the reader to  \cite{Geracie:2017uku} where the  Hilbert description of this construction is discussed in detail.

When the initial state is furthermore thermal, $\rhoi = \rhoT = e^{-\beta\, \OpH{H}}$, the correlators can be obtained  (in equilibrium) by analytically continuing Euclidean thermal correlators. Euclidean thermal periodicity translates then into a set of non-local KMS conditions \cite{Kubo:1957mj,Martin:1959jp,Haag:1967sg}. The schematic way to understand these conditions is by conjugating any operator through the density matrix, noting that its Gibbsian nature results in shifting the temporal argument of the operator 
$t \mapsto t-i\,\beta$.\footnote{ The direction of the shift and conjugation is dictated by convergence, so all motion is into the lower-half complex time plane.}

These conditions encode the important fluctuation-dissipation relations; indeed for two-point functions a simple rearrangement of the primitive KMS condition results in the familiar relation between the commutator and the anti-commutator up to a statistical factor:\footnote{ For higher-point functions, a more elaborate analysis involving more switchbacks in time is necessary \cite{Haehl:2017eob}.}
\begin{equation}
\begin{split}
\Tr{\OpH{A}(t_{_A}) \, \OpH{B}(t_{_B}) \, \rhoT} &= \Tr{\OpH{B}(t_{_B} - i\,\beta) \, \OpH{A}(t_{_A})\,\rhoT}  \\
 \;\; \Longrightarrow \;\;
\vev{\{\OpH{A} , \OpH{B} \} } &= - \coth\left(\frac{1}{2}\,\beta\, \omega_{_B}\right) \, \vev{ [\OpH{A} , \OpH{B} ] } 
\end{split}
\label{eq:kms}
\end{equation}	
where the second line is written in momentum space.
 It is traditional in much of the literature to view the KMS condition as a ${\mathbb Z}_2$ involution,  owing to the fact that the first line of \eqref{eq:kms} involves a swap of operator order after conjugation; see \cite{Sieberer:2015hba} for a nice discussion. Indeed \cite{Crossley:2015evo} implement the KMS condition as a   ${\mathbb Z}_2$  symmetry, and argue for an emergent second topological charge to encode the KMS condition. 

 However, inspired by our previous studies of  transport in relativistic fluids that is consistent with the  second law, we argued for a different approach involving an 
emergent KMS-$\UT$ gauge invariance \cite{Haehl:2014zda,Haehl:2015pja,Haehl:2015foa}. This KMS symmetry acts on the fields by thermal translations, as it must. The actual implementation of the KMS gauge invariance works through the formalism of extended equivariant cohomological algebras. While the two approaches appear to be superficially different, at the end of the day, we will end up with a very similar superalgebra constraining low energy dynamics, modulo the following key distinctions. 
Rather than repeat the technical exposition which can be found  in \cite{Haehl:2016uah} we shall give a brief qualitative picture of our construction, pausing to note some key differences in the two approaches:
\begin{itemize}
\item In our construction, there is an emergent $\UT$ symmetry which acts via thermal diffeomorphisms. This symmetry is gauged, and its conserved charge is the net free energy. We do not however describe the topological gauge dynamics of this field; see \S\ref{sec:caveat}. 
\item If we freeze out the gauge modes, and instead treat them as non-trivial background gauge fields which enforce {\sf CPT} breaking, we leave behind an algebra that agrees with \cite{Crossley:2015evo} in the high temperature limit.\footnote{ We believe that their construction is also best understood in this limit. The high temperature limit of the superalgebra is well known in the statistical mechanics literature and has been used to good effect to derive the non-equilibrium form of the second law in \cite{Mallick:2010su}. We will have occasion to use it in our discussion below in \S\ref{sec:mmo}, \S\ref{sec:mmoEx} (we call it the MMO algebra). }
\item The advantage of having an explicit gauge symmetry is that it allows for  {\sf CPT} breaking to emerge dynamically rather than being imposed a-priori in the formalism. Demonstration of this of course requires that the gauge dynamics admits such vacua, but thus far we do not see any obstacle (as we describe in detail below) for this to be the case.
\end{itemize}

\begin{figure}[t!]
\begin{center}
\includegraphics[width=.6\textwidth]{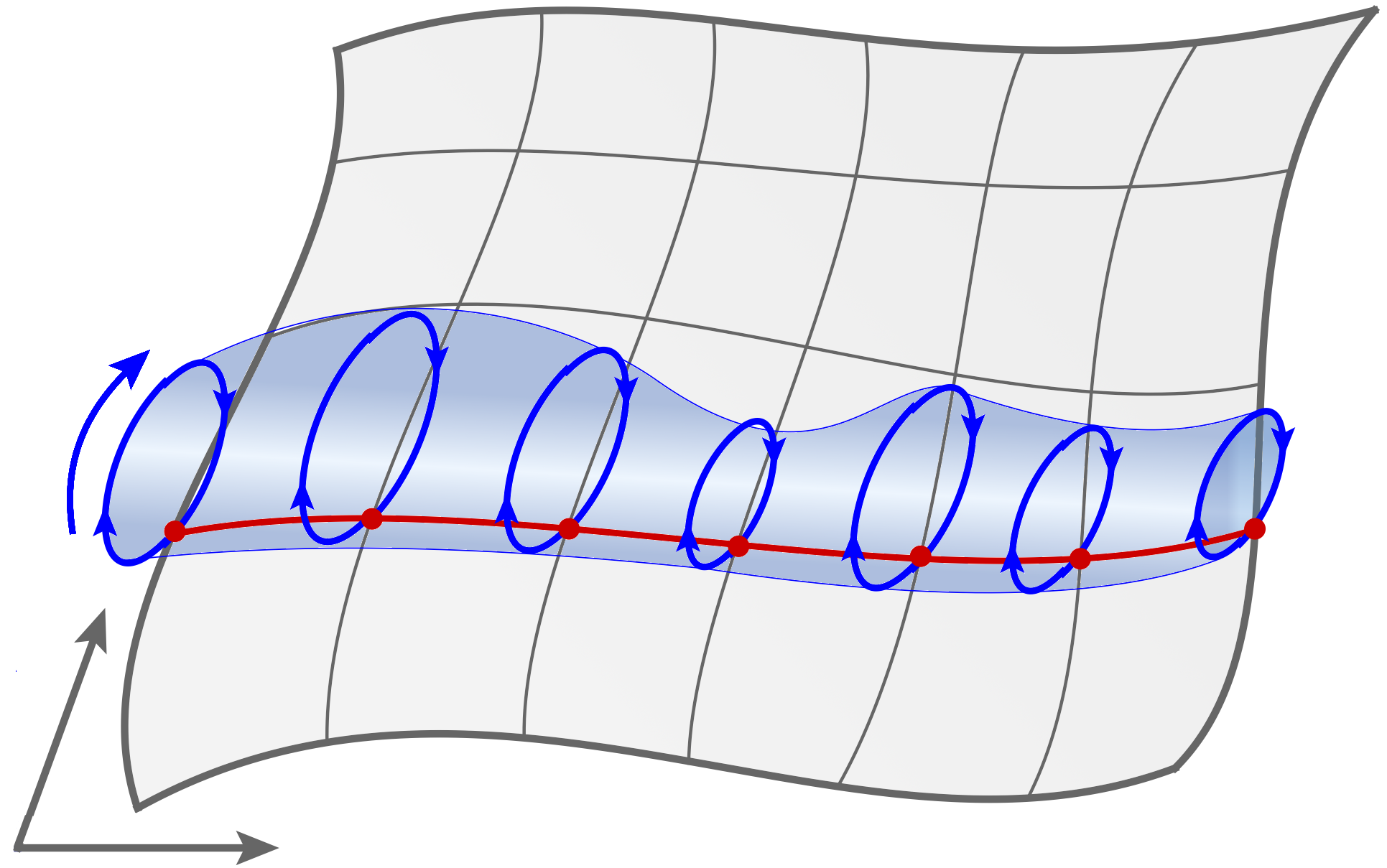}
\begin{picture}(0.3,0.4)(0,0)
\put(-265,93){\makebox(0,0){{\color{blue}Im$(t)$}}}
\put(-268,45){\makebox(0,0){{\color{black}Re$(t)$}}}
\put(-207,5){\makebox(0,0){{\color{black}$x^i$}}}
\end{picture}
\caption{Illustration of the spacetime picture associated with the proposed KMS gauge invariance, with the figure taken 
from \cite{Haehl:2016uah}. We upgrade the spacetime manifold (a Lorentzian geometry, depicted in gray,  with a typical Cauchy slice indicated in red), on which our quantum system resides to a thermal fibre bundle. We further assume  local thermal equilibrium (as in, e.g., hydrodynamics) at each spacetime point which guarantees existence of a thermal vector $\Kref^\mu$. We draw this vector field as a circle fibration with a thermal circle whose size is set by the local temperature. The KMS transformations we seek implement  equivariance with respect to thermal translations along this local imaginary time circle. Restricting to a gauge slice corresponds to picking a Lorentzian section of this fibration. Note that the size of the thermal circle is exaggerated; our arguments assume the high temperature limit where the size of the thermal circle is much smaller than the fluctuation scale.}
\label{fig:SKspacetime}
\end{center}
\end{figure}
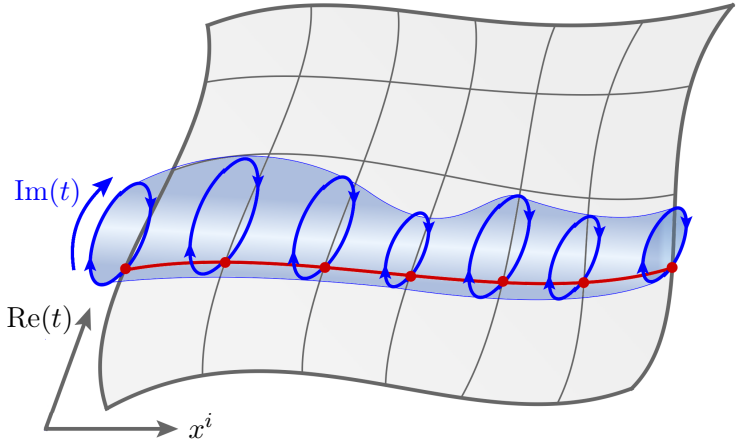

Usually, thermal equilibrium is viewed in terms of statistical mechanics (in the Euclidean formulation of the QFT), owing to the fact that the thermal circle can be considered to be fibered non-trivially over the spatial geometry. The thermal equivariance formalism, heuristically can be viewed as extending this fibration to the physical spacetime (see Fig.\ref{fig:SKspacetime}). In this extended spacetime, we ask that the resulting physics be independent of the choice of section of the thermal circle fibration, which amounts to imposing a gauge invariance in the system corresponding to thermal translations in Euclidean time. We expect this picture to make sense in the low energy limit, when 
frequencies and momenta are low compared to the thermal scale $\omega\,T, k\,T \ll 1$, when we can approximate the discrete thermal circle translations by a continuous (and thence infinitesimal) translation. The reader will recognize that we are postulating an emergent low energy gauge symmetry; the information encoded in the KMS conditions is captured by the KMS-$\UT$ symmetry of thermal translations.\footnote{ It is an interesting question as to what this structure implies away from the low energy or higher temperature limit. While we have computed features of the SK-KMS algebra away from this limit from microscopic Ward identities, we have not been able to distill them into a useful principle (cf., \cite{Haehl:2016uah}).  It is also unclear where one should impose a global KMS condition in the hydrodynamic state which is only in local equilibrium -- our discussion only requires a local KMS condition be imposed. 
It is worth mentioning here that the analysis of \cite{Jensen:2017kzi} argues that this can be done. We are however unsure if this implementation, which involves writing down actions involving mutually non-local degrees of freedom with no direct coupling between them, is compatible with the influence functionals derived in simple models of quantum dissipation like those discussed in \cite{Feynman:1963fq,Caldeira:1982iu}.} 

The full BRST superalgebra relevant for our considerations is an extended equivariant cohomology algebra, which is a graded algebra with $4$ Grassmann odd and $2$ Grassmann even generators. The generators are categorized into 
exterior derivations $\{\QSK, \QSKb\}$, interior contractions $\{\QKMS, \QKMSb, \Qzero\}$, and a Lie derivation $\Qbeta$. The exterior derivations and interior contractions are nilpotent, and the Lie derivation actually implements translations along the thermal circle; operators are Lie dragged along the thermal vector $\Kref^\mu$ consistent with their tensorial structure. The full algebra is presented in  Appendix \ref{sec:skkmsalg}.

\subsection{Superspace and the thermal gauge multiplet}
\label{sec:superg1}

We now describe how to implement our proposed symmetries, and give references to more detailed discussions.

The practical way of implementing thermal equivariance in explicit constructions is to pass onto a superspace description as described in \cite{Haehl:2016pec,Haehl:2016uah}. We introduce two Grassmann odd coordinates  
$\{\theta,\thetab\}$, identifying  $\{\QSK,\QSKb\}\sim \{\partial_{\thetab},\partial_\theta\}$, and promote fields ${\cal Y}$ to superfields (denoted by a circle "$\ \mathring{ }\ $" accent): 
\begin{equation}
{\cal Y} \;\to\; \SF{{\cal Y}} = 
	{\cal Y} +  \theta\,  {\cal Y}_{\bar \psi} + \thetab \, {\cal Y}_\psi + \thetab\,\theta\, \tilde{{\cal Y}}
	 \equiv 
	\frac{1}{2} \left({\cal Y}_\skR +{\cal Y}_\skL\right) +  \theta\, {\cal Y}_{\bar \psi}  
	+ \thetab \, {\cal Y}_\psi  + \thetab\,\theta\, ({\cal Y}_\skR - {\cal Y}_\skL) \,.
\label{eq:xsfld}
\end{equation}	
The top ($\thetab\theta$) components of the superfields  represent the difference operators while the ${\cal Y}_\psi,
{\cal Y}_{\bar \psi} $ are the ghost super-partners carrying same spin but opposite Grassmann parity as  ${\cal Y}$. Individual components are recovered by taking suitable $\theta$, $\thb$ derivatives  and projecting onto ordinary space; we denote this projection as $\big| \equiv \big|_{\theta = \thb  =0}$.  This structure is sufficient to describe the Schwinger-Keldysh formalism in generic initial states (see \cite{Geracie:2017uku} for an example). We refer to Appendix \ref{sec:skkmsalg} for a brief review of the algebraic construction.

To describe our macroscopic gauge theory at a certain temperature we introduce a background timelike vector superfield  $\SF{\Kref}^I(z)$.\footnote{ We work in conventions where  $z^I \equiv \{\sigma^a,\theta,\thetab\}$ is the collection of superspace coordinates with lower case Latin alphabet indexing the ordinary spacetime coordinates, reserving uppercase Latin indices for superspace. The notation is adapted for sigma models of interest in our discussion later. $z^I$ will coordinatize the worldvolume directions, while physical spacetime coordinates will be denoted as $X^{\uA}=\{X^\mu,\Theta, {\bar \Theta}\}$ and indexed with lowercase Greek alphabet for ordinary indices with accented (breve) Greek lowercase reserved for  super-indices. \label{fn:indexnote}}
We will use some of the superdiffeomorphism invariance to simplify the thermal super-vector:
\begin{equation}
 \SF{\Kref}^\theta=\SF{\Kref}^{\thetab} = 0=\partial_\theta \SF{\Kref}^a=  \partial_{\thetab} \SF{\Kref}^a\,.
\label{eq:betagauge}
\end{equation}	
 We will consider below only that subset of superdiffeomorphisms which respect this gauge choice for the background thermal super-vector $\SF{\Kref}^I$.\footnote{ We can also introduce a thermal twist $\Lref(\sigma)$ which encodes the chemical potential for general ensembles with additional conserved charges. The twist is the phase entering the thermal periodicity conditions in a particular flavour symmetry gauge. For simplicity we will not elaborate further on this.} Appendix \ref{sec:gaugesdiff} contains a detailed discussion.

The \emph{$\UT$ super-gauge transformations} can be parameterized by an adjoint superfield gauge parameter $\LamS$. They  
act on a  superfield $\SF{{\cal Y}}$ by Lie dragging it along $\LamS\, \SF{\Kref}^I$ and can be represented by a 
 \emph{thermal bracket},
\begin{equation}
(\LamS,\SF{{\cal Y}})_\Kref=\LamS  \,\lieD_\Kref \SF{{\cal Y}} \,,
\label{eq:betabrkX}
\end{equation}	
where  $\lieD_\Kref$ denotes the super-Lie derivative along $\SF{\Kref}^I$. The 
infinitesimal gauge transformation is thus given by
\begin{equation}
\SF{{\cal Y}}\mapsto \SF{{\cal Y}}+ (\LamS,\SF{{\cal Y}})_\Kref \,.
\label{eq:utx}
\end{equation}	
For scalar $\SF{{\cal Y}}_\text{scalar}$ this is  just a thermal translation \begin{align}
(\LamS,\SF{{\cal Y}}_\text{scalar})_\Kref=\LamS \;\SF{\Kref}^I \partial_I \SF{{\cal Y}}_\text{scalar}\,.
\label{}
\end{align}
The Jacobi identity  then fixes the action of  thermal bracket on adjoint superfields, so that under successive $\UT$ transformation $\LamS' \mapsto \LamS' + (\LamS,\LamS')_\Kref $  with
\begin{align}
(\LamS,\LamS')_\Kref &=\LamS\lieD_\Kref \LamS'-\LamS' \lieD_\Kref \LamS \,.
\label{eq:adbetabrk}
\end{align}

We introduce a $\UT$ \emph{gauge superfield one-form} as a triplet
 $\As_I(z)\equiv\{\As_a(z),\SF{\mathscr{A}}_{\theta}(z),\SF{\mathscr{A}}_{\thetab}(z)\}$
\begin{align}
\As_I(z)\, dz^I = \As_a(z)\, d\sigma^a +\SF{\mathscr{A}}_{\theta}(z)\, d\theta + \SF{\mathscr{A}}_{\thetab}(z) \, d\bar{\theta}
\label{eq:Aform}
\end{align}
  whose gauge transformation is like an adjoint superfield except for an inhomogeneous term, viz., 
\begin{equation}
 \As_I \mapsto \As_I +(\LamS,\As_I)_\Kref  - \partial_I \LamS \,,
\label{eq:aut}
\end{equation}	
with the thermal bracket as in \eqref{eq:adbetabrk}.  One can further define as usual a 
\emph{covariant derivative}\footnote{ The covariant derivative $\Dut$ introduced in \eqref{eq:covDI} implements covariance under $\UT$ transformations on a flat superspace ${\mathbb R}^{d-1,1|2}$ in Cartesian coordinates. Later on we will encounter a super-covariant derivative $\Dwv$ that will also involve additional contributions from the background geometric connection. } 
\begin{equation}
 \Dut_I = \partial_I + (\As_I,\ \cdot\, )_\Kref\,,
\label{eq:covDI}
\end{equation}	
and an associated field strength
\begin{equation}
\Fs_{IJ}  \equiv(1-\frac{1}{2}\, \delta_{IJ}) \left(\partial_I\, \As_J - (-)^{IJ} \,\partial_J\, \As_I  + (\As_I,\As_J)_\Kref \right) ,
\label{eq:fdef}
\end{equation}	
where $(-)^{IJ}$ is the mutual Grassmann parity of the two indices involved (see below).
Given the low-energy superfields $\SF{{\cal Y}}$, the theory of macroscopic fluctuations is given as the general superspace action invariant under  $\UT$ gauge transformations. We explain some essential features of the gauge  algebra, the structure of the multiplets and the various fields involved, and the attendant representation theory in the Appendices \ref{sec:UTreps}, \ref{sec:Adjoint}, \ref{sec:GaugeMult}.  Some of the background material has already been detailed in \cite{Haehl:2016uah}, so we also refer the reader there for further details.

The final symmetry we implement is {\sf CPT}. It is important that the implementation of {\sf CPT} not conjugate the initial hydrodynamic state (which is crucial for example when chemical potentials are turned on). We also require that such a symmetry be present even when we discuss non-relativistic systems with no microscopic {\sf CPT} symmetry. As such any anti-linear involution respecting these criteria will suffice for our purposes and can justifiably called {\sf CPT}.\footnote{ The manner we implement the anti-linear involution was inspired by Veltman's diagrammatic rules for the SK path integral for the vacuum initial state \cite{tHooft:1973pz}, which acts by exchanging incoming and outgoing states in a scattering process.} We want to encode the information that the SK path integral be invariant under the combined  {\sf CPT} transformation of the initial state and the sources. The anti-unitary nature of {\sf CPT} allows us to translate these requirements  into a reality condition for the SK path integral, viz., 
\begin{equation}
\mathcal{Z}^*_{SK}[J_\skL,J_\skR] =  \mathcal{Z}_{SK}[J_\skR,J_\skL]\,.
\end{equation} 

Apart from the usual action on $\sigma^a$ coordinates, {\sf CPT} exchanges $\thetab\leftrightarrow\theta$ and hence acts as an R-parity on the superspace. This is necessitated by  our requirement that the $\thetab\theta$ component of the superfields be identified with difference operators (and the exchange of $\text{R} \leftrightarrow \text{L}$ under {\sf CPT}). 
In addition we have a conserved ghost number $\text{gh}$, which assigns $\gh{\theta} =1$, $\gh{\thb} = -1$. 
To understand these assignments, and our rationale for calling the anti-linear involution as {\sf CPT}, note that the fluid equations are not by themselves {\sf CPT} invariant. One can add an additional spurion field to make them so -- this is the role $i\, \Fs_{\theta\thb}|$ will play in our formalism. We therefore want our anti-linear involution symmetry of the SK functional integral be such that $i\, \Fs_{\theta\thb}|$ is odd under it. 
As argued in section 8 of \cite{Haehl:2016pec}, the aforementioned assignment will do the job.

\subsection{Super-index contractions}
\label{sec:dewitt}

We also are now at a point, where we should specify our super-index contraction conventions. We follow the conventions described in 
the book by DeWitt \cite{DeWitt:1992cy}, which says that when adjacent super-indices are contracted from south-west to north-east $\nearrow$, there is no sign, but one picks up a Grassmann sign when indices are contracted from 
north-west to south-east $\searrow$. 

We will shortly introduce a metric superfield $\SF{\gref}_{IJ}$ which will be used to raise and lower indices. The thermal super-vector and gauge supermultiplet of course do not rely on the existence of a metric, but all of these will enter into our constructions below. The index contraction convention is simplest to intuit from the orthogonality condition imposed on metric and its inverse. We have:
\begin{equation}
\SF{\gref}_{IJ} \, \SF{\gref}^{JK} = \delta_I^{\ K} \,, \qquad 
(-)^J\;\SF{\gref}^{IJ} \, \SF{\gref}_{JK} = \delta^I_{\ K} \,.
\label{eq:deWitt1}
\end{equation}	
Note that the index placements are all important since any swap of indices ends up leading to extra signs. In general, contraction of separated indices is carried out by checking what the relative sign would be to bring the two indices being contracted next to each other:
\begin{equation}
\SF{\mathcal{T}}^{I_1 \cdots I_m  K_1 \,\cdots K_n}{}_{L_1\,\cdots L_p  \, M_1\,\cdots M_q}   = 
(-)^{J \left(1+\sum_i K_i + \sum_j L_j\right)}
\; \SF{\mathcal{T}}^{I_1  \cdots  I_m J \, K_1 \,\cdots K_n}{}_{L_1\,\cdots L_p \, J \, M_1\,\cdots M_q}  \label{eq:dwconven}
\end{equation}	
where the sign is such that the contraction of the tensor is still a tensor. 
We will however leave this sign implicit in various formulae below, to keep them reasonably readable. The reader should exercise care in contracting indices (we will give examples in \S\ref{sec:mmo}).

\newpage
\part{Hydrodynamic effective field theories}
\section{The hydrodynamic degrees of freedom}
\label{sec:hydrodof}

Hydrodynamics is the long-wavelength effective field theory for systems in local, but not global thermal equilibrium. The natural variables in terms of which hydrodynamics is presented are the fluid velocity $u^\mu(x)$ (normalized $u^\mu\, u_\mu =-1$ ), the local temperature $T(x)$, and other intensive parameters such as chemical potentials $\mu_i(x)$ when additional conserved charges are present. It is useful to combine the basic data into a thermal vector $\Kref^\mu = u^\mu/T$ and thermal twists $\Lambda_i  = \mu_i/T - u^\mu A_\mu$ as explained in \cite{Haehl:2015pja}. We allow our fluid to be subject to arbitrary (though slowly varying) external fields $g_{\mu\nu}(x)$ and $A^{(i)}_\mu(x)$, respectively. The thermal vector  picks out the direction of the local inertial frame in which the fluid is locally at rest, and its length sets the size of the thermal circle. The choice of the thermal vector/twist was inspired by the thermal circle fibration which proved useful for the analysis of equilibrium partition functions.

Before proceeding, it is also helpful to understand qualitatively how hydrodynamics emerges from a 
microscopic perspective. As noted in the introduction \S\ref{sec:intro}, the hydrodynamic observables are response functions as those described in Fig.~\ref{fig:hydrosk}. The response functions, which schematically are observables of the form $\vev{\mathcal{T}_{SK} \SKAv{O}^1 \SKDif{O}^2 \cdots \SKDif{O}^n} $, are the leading non-vanishing observables in the Schwinger-Keldysh formalism. They  in turn are related by the KMS relations to fluctuation correlators involving multiple average operators, which are schematically combinations of correlators of the form $\vev{\mathcal{T}_{SK} \SKAv{O}^1 \SKAv{O}^2 \cdots \SKDif{O}^n}$.  

The dynamics in hydrodynamics is simply conservation of charges, which are the slow modes surviving once the transient effects die down in any system that has been perturbed away from global equilibrium. This leads to conservation of energy-momentum (and charges if any). The low energy theory which then involves conserved current operators should capture the IR limit of both the response functions and the fluctuations. The former is what classical hydrodynamic constitutive relations tries to encode, but the latter is necessary for the system to be aware of its microscopic origins. This entails that an effective field theory should have adequate degrees of freedom to go beyond the classical hydrodynamic limit, and be able to predict the structure of fluctuations.  Thus, once we identify the relevant IR modes that can lead to the correct hydrodynamic equations, we then need to figure out how to upgrade them to include degrees of freedom that capture fluctuations. Fortunately, the hard work is already done in the preceding discussion: the Schwinger-Keldysh superspace and thermal equivariance makes this a simple task.

\subsection{The pion fields of hydrodynamics}
\label{sec:pions}

Let us  start by identifying the classical variables in the hydrodynamic effective field theory.
In general, conservation laws follow from a Noether construction while dynamics is dictated by a variation principle. The latter usually implies the former, but it is not always true that conservation laws encapsulate the entirety of dynamics. This can happen only if there is an additional reparameterization symmetry, whence the reparameterization invariance of the dynamical fields essentially implements diffeomorphisms (or flavour rotations) which result in the conservation laws appearing. This is precisely the situation in hydrodynamics.

The natural framework to implement such a reparameterization invariance is to view the degrees of freedom in terms of a parameterized sigma model, much like how we describe string or brane dynamics in string theory. We pick an auxiliary space, \emph{the worldvolume}, with coordinates $\sigma^a$, and equipped with the reference thermal vector $\Kref^a$. The physical spacetime coordinates $X^\mu$ are viewed as maps from the worldvolume to the target space, $X^\mu(\sigma)$. 
The rigid worldvolume reference thermal vector pushes forward to the physical thermal vector $\Kref^\mu = \Kref^a \partial_a X^\mu$ in spacetime. The latter becomes dynamical through the push forward. At the same time the spacetime metric $g_{\mu\nu}$ (which we recall is non-dynamical) pulls back to give a worldvolume metric $\gref_{ab} = g_{\mu\nu} \,\partial_a X^\mu\, \partial_b X^\nu$. Spacetime diffeomorphisms operate as translations of the sigma model fields $X^\mu \mapsto X^\mu + \xi^\mu$.

\begin{figure}[ht!]
\begin{center}
\includegraphics[width=0.8\textwidth]{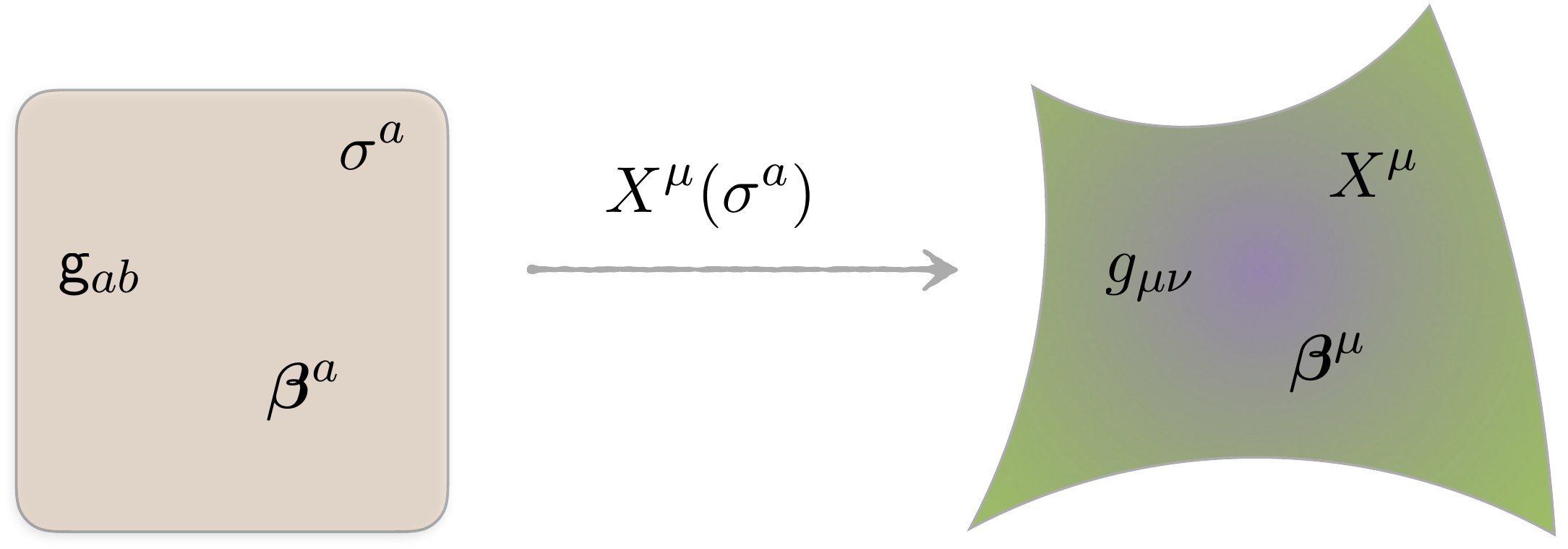}
\caption{Illustration of the data for hydrodynamic sigma models as described in \cite{Haehl:2015pja}. The physical degrees of freedom are captured in the target space maps $X^\mu(\sigma)$ which are the hydrodynamic pion fields. The worldvolume geometry is equipped with a reference thermal vector field $\Kref^a$, which pushes forward to the physical thermal vector in spacetime, while the spacetime metric  pulls back to the worldvolume metric $\gref_{ab}$.}
\label{fig:sigmamap1}
\end{center}
\end{figure}

Demanding worldvolume reparameterization invariance under such transformations we end up with the correct variational principle to single out spacetime energy momentum conservation as the dynamical equations of motion, as explained in \cite{Haehl:2015pja} (extensions to flavour charges can be found in the aforementioned reference).\footnote{ As noted in \cite{Haehl:2015pja} the Lagrangian fluid variables utilized in \cite{Dubovsky:2011sj} which are predominantly employed in the analysis of \cite{Crossley:2015evo,Glorioso:2017fpd} can be recovered by gauge fixing the thermal vector. In our experience there are some disadvantages of working with this set of variables. Not only is one sacrificing manifest covariance, but also there are some accidental symmetries (such as the volume preserving diffeomorphisms of \cite{Dubovsky:2011sj}). Furthermore, these variables are ill-adapted to circumstances where the entropy current is non-trivial, and hence we chose to move to a more covariant formalism that also adapts naturally to equilibrium analysis in \cite{Haehl:2015pja} and subsequent works (see also \cite{Haehl:2017zac}).\label{fn:thermalvec}} A worldvolume effective action for the sigma model degrees of freedom described above gives rise to a class of Landau-Ginzburg sigma models which capture a fraction of adiabatic transport in any relativistic fluid. Such sigma models capture $2$ of the $8$  admissible physical classes of transport (Classes $\PS$ and $\LS$ in the classification of \cite{Haehl:2015pja}). In order to obtain the remainder, including the dissipative Class D, we need additional structures. This is where the thermal equivariance enters the discussion.

Let us start by elaborating on the sigma model fields $X^\mu(\sigma)$. These are the classical variables encoding the information of how the fluid is behaving locally, since the physical data contained in the thermal vector is entirely captured in them. To an extent, we may simply view the reference vector $\Kref^a $ as a means to remove the local inhomogeneities (as can be seen by the gauge fixing alluded to in footnote \ref{fn:thermalvec}; see also \cite{Hongo:2016mqm}). This structure is illustrated in Fig.~\ref{fig:sigmamap1}.

As noted at the outset these classical fields should be subject to thermal/quantum fluctuations from the dissipative effects inherent in the fluid. Invoking the microscopic Schwinger-Keldysh construction, we intuit that the hydrodynamic degrees of freedom are Goldstone modes for spontaneously broken off-diagonal diffeomorphism (and flavour) present in the microscopic description \cite{Haehl:2015pja}. This philosophy was first made clear in \cite{Nickel:2010pr} and has been used in other attempts to construct hydrodynamic actions \cite{Kovtun:2014hpa}. 

As a useful intuition building exercise, consider a probe particle immersed in the fluid. Such a particle, buffeted by the fluctuations in the fluid, will undergo stochastic Brownian motion. In addition to its classical position we should also keep track of its fluctuations. In the standard discussion of the Langevin effective action one therefore introduces the classical position $X$ and the quantum/stochastic/fluctuation field $\tx$. The dynamics of the Brownian particle can then be described by a worldline BRST symmetric action \cite{Martin:1973zz,Parisi:1982ud}, which as explained in \cite{Haehl:2015foa,Haehl:2016uah} is the simplest example of the thermal equivariant sigma model. 
This led us to describe the dynamics of \emph{Brownian branes}, which are  brane like objects of various codimension which, when immersed in the fluid, undergo generalized Brownian motion. Hydrodynamics then is the theory of a space-filling Brownian brane (while Langevin dynamics corresponds to ${\sf B}0$-brane dynamics).

The proposal then is to invoke the underlying Schwinger-Keldysh intuition and view the target space maps $X^\mu(\sigma)$, the  pion fields of the sigma model, as  vector Goldstone modes arising from broken difference diffeomorphisms of the doubled construction. It is worth noting that such a description is necessary; the structural part of hydrodynamics is `universal', since it is agnostic of the microscopic constituents of the quantum system. The symmetry breaking pattern should reflect this fact. The details of the quantum system matter in determining the actual values of the hydrodynamic transport data (the analog of the pion coupling constants in the chiral Lagrangian). This perspective is also clear from the fluid/gravity correspondence (see \cite{Bhattacharyya:2008jc,Hubeny:2011hd}).  The fluctuation fields will then arise from the modes that couple to average diffeomorphisms in the doubled construction.

\subsection{Symmetries of the hydrodynamic sigma models}
\label{sec:symhydro}

Let us now take stock of the symmetries inherent in the thermal Schwinger-Keldysh construction, encapsulated within the notion of thermal equivariance as described in \S\ref{sec:thermaleq}, and upgrade the hydrodynamic Goldstone dynamics of \S\ref{sec:pions} to be cognizant of them. We have seen that the symmetries arising from the microscopic picture, viz., $\UT$ gauge invariance, together with CPT and ghost number conservation are easily encoded in superspace. 

This entails that we should first upgrade the worldvolume to superspace parameterized by $z^I = \{\sigma^a, \theta, \thb\}$. The thermal vector will be uplifted to a thermal super-vector $\SF{\Kref}^I$. Since we have worldvolume diffeomorphisms that upgrade themselves to superdiffeomorphisms, we will exploit some of the freedom to gauge fix components of the reference super-thermal vector as indicated in \eqref{eq:betagauge}.\footnote{ We will examine worldvolume symmetries more precisely  in Appendix~\ref{sec:gaugesdiff}, but note that the worldvolume superdiffeomorphisms we allow are simply $z^I \mapsto z^I + f^I(\sigma^a)$.}

\begin{figure}[ht!]
\begin{center}
\includegraphics[width=0.8\textwidth]{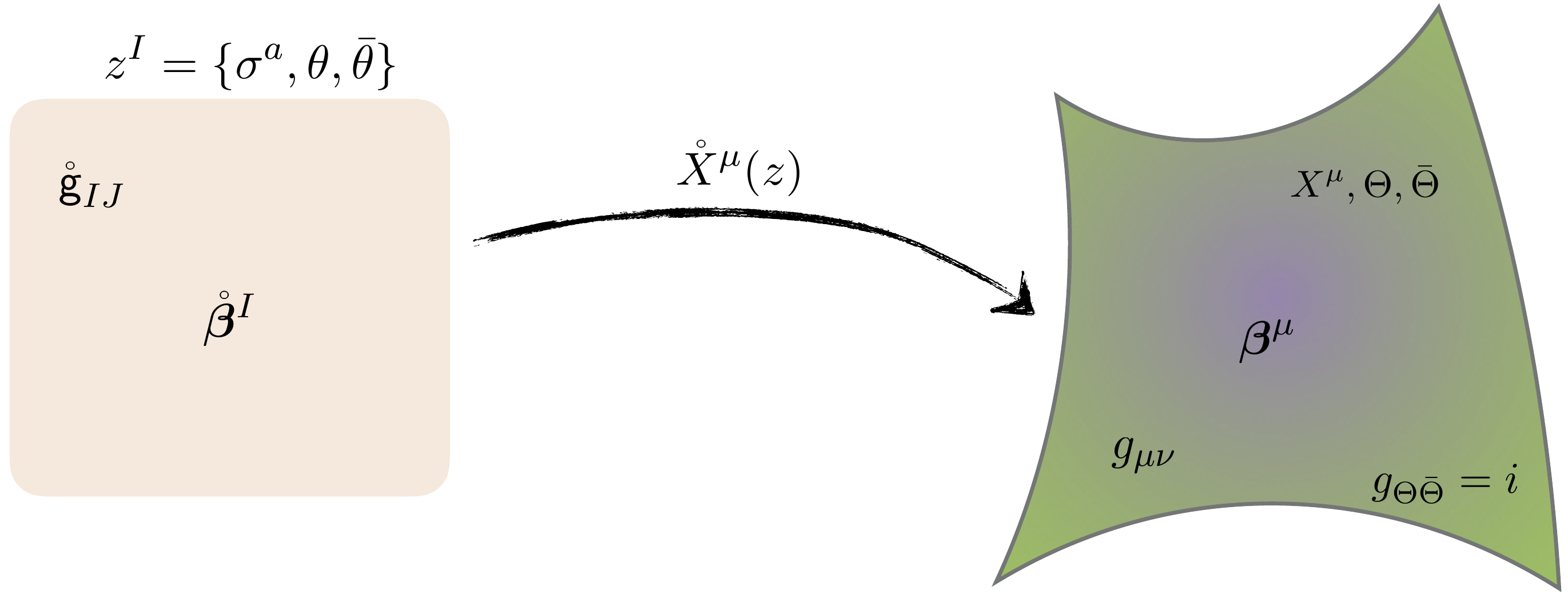}
\caption{Illustration of the data for hydrodynamic sigma models. The physical degrees of freedom are captured in the target space maps $\SF{X}^\mu(\sigma,\theta,\thb)$, along with the gauge condition which aligns the Grassmann coordinates in target space with their worldvolume counterparts $\Theta = \theta$ and ${\bar \Theta} =\thb$.  The worldvolume geometry is equipped with a reference super-vector field $\SF{\Kref}^I$, which pushes forward to the physical thermal vector in spacetime, while the spacetime metric $g_{\mu\nu}$ (with $g_{\Theta {\bar\Theta}} = i$)  pulls back to the worldvolume metric $\SF{\gref}_{IJ}$. The pullbacks and pushfowards are $\UT$ gauge covariant, since the worldvolume dynamics is constrained by this symmetry.}
\label{fig:sigmamap2}
\end{center}
\end{figure}

Furthermore, we realize that the target space maps which are sigma model fields should be upgraded to superfields following \eqref{eq:xsfld}: $X^\mu(\sigma) \mapsto \SF{X}^\mu(z)$, viz., 
\begin{equation}
\SF{X}^\mu = X^\mu + \theta \, \xpsib^\mu + \thb\, \xpsi^\mu + \thb\theta \left(\tx^\mu -\Gamma^\mu_{\rho \sigma} \,\xpsib^\rho \, \xpsi^\sigma\right)
\label{eq:xsf0}
\end{equation}	
The contribution from target space connection to the top component of the superfield (where the fluctuation field $\tx^\mu$ resides) can be understood from covariance of the pullback map and is explained for completeness in Appendix~\ref{sec:gaugesdiff}.

 However, not only do the bosonic hydrodynamic pions get upgraded to a superfield, but we also should obtain the spacetime Grassmann odd partners, leading to a spacetime triplet of superfields:\footnote{  Notational conventions for indices is summarized in footnote \ref{fn:indexnote}.}
\begin{equation}
\SF{X}^{\uA}(z) =  \{\SF{X}^\mu(z), \SF{\Theta}(z),\SF{\bar{\Theta}}(z)\} \,.
\label{eq:}
\end{equation}	
Note that this structure is enforced by the way we wish to implement the Schwinger-Keldysh construction. Even in the physical spacetime we have to allow for the superspace structure, since after all it is there that our quantum system resides (and it is the quantum operators that get uplifted to super-operators). We illustrate the superspace upgrade of the hydrodynamic pion fields, which we henceforth work with, in Fig.~\ref{fig:sigmamap2}.

The physical spacetime is equipped with a background metric $g_{\uA\uB}$ which being a background source, we are free to pick at will. We will use this freedom to demand 
\begin{equation}
g_{\mu\Theta} = g_{\mu \bar{\Theta}} =0 \,, \qquad g_{\Theta \bar{\Theta}} = -g_{\bar{\Theta} \Theta} = i \,.
\label{eq:gsptgauge}
\end{equation}	
We only turn on bottom components for the spacetime metric (thus enabling us to dispense with superfield notational contrivances) which will suffice for the rest of the discussion. As usual with sigma models we choose the target space data $g_{\uA \uB}(X^{\uC})$ and its metric compatible connection first and then upgrade the resulting expressions to functions on the worldvolume with the $z^I$ dependence induced from the embedding $\SF{X}^{\uA}(z)$.

The symmetries alluded to above, viz.,  superdiffeomorphisms, CPT, and ghost number symmetries act as usual on these. In addition the action of $\UT$ is given as in \eqref{eq:betabrkX}; for instance
\begin{equation}
(\LamS,\SF{X}^\mu)_\Kref = \LamS \, \SF{\Kref}^I \partial_I \SF{X}^\mu\,, 
 \label{eq:xut}
\end{equation}
and similarly for $\{\SF{\Theta}(z),\SF{\bar{\Theta}}(z)\}$.  We will refer to this action as the action of $\UT$ on fundamental representation (or $0$-adjoints).\footnote{ Representations of the $\UT$ thermal diffeomorphism symmetry are worked out in Appendix~\ref{sec:formal}.}

On the worldvolume, the $\UT$ gauge symmetry requires that we have in addition to the reference super-vector (which mainly picks out the reference frame) the $\As_I$ super-form superfield. Since $\{\SF{X}^\mu, \SF{\Theta}, \SF{\bar{\Theta}} \}$ carry non-trivial $\UT$ charges, gauge invariance on the worldvolume requires that we work with suitable gauge covariant objects. To this end, the pullback map onto the worldvolume will be implemented in a $\UT$ covariant manner. As a result the worldvolume metric  $\gref_{ab}$ gets uplifted to a superfield  $\SF{\gref}_{IJ}$
\begin{align}
\SF{\gref}_{IJ}(z) = g_{\mu\nu}(\SF{X}(z))\,\Dut_I \SF{X}^\mu\,\Dut_J \SF{X}^\nu  + g_{\Theta \bar{\Theta}} \left( 
\Dut_I \SF{\Theta} \, \Dut_J \, \SF{\bar{\Theta}}  - \Dut_I \SF{\bar{\Theta}} \, \Dut_J \, \SF{\Theta}\right) \,.
\label{eq:gref}
\end{align}
where we have already incorporated our gauge condition \eqref{eq:gsptgauge} to simplify this expression, and have stuck to DeWitt conventions \cite{DeWitt:1992cy} for super-index contractions as noted in \S\ref{sec:dewitt}.

Our goal will be to construct a topological sigma model governing the dynamics of the fields  $\{\SF{X}^\mu, \SF{\Theta}, \SF{\bar{\Theta}} \}$. With the symmetries at hand, such a theory has been engineered to capture the constraints arising from the Schwinger-Keldysh construction in thermal states. However, the physical fluid dynamical theory is not a topological field theory; fluids have non-trivial dynamics. To get the physical hydrodynamic data we should deform away from the topological limit. This can be easily achieved by de-aligning the sources for the left and right fields. This implies turning on the difference metric to  obtain the physical the energy-momentum tensor.

We obtain the energy-momentum tensor on the world-volume and then push it forward to the physical spacetime, and so will turn on a  difference source $\source{h_{IJ}}$ on the world-volume, i.e.,
\begin{align}
\SF{\gref}_{IJ}(z) \;\rightarrow\; \SF{\gref}_{IJ}(z) + \source{\bar{\theta} \,\theta\; h_{IJ}(\sigma)}\,.
\label{eq:grefh}
\end{align}
Given the worldvolume Lagrangian, varying it with respect to the source deformation $\source{\shref_{IJ}}$ will give us the (worldvolume) fluid dynamical stress tensor $\TEMref^{IJ}$. The dynamics for the fields will be obtained by variation with respect to the fields $\SF{X}^\mu$; for the classical field $X^\mu$, the dynamics is obtained by varying the fluctuation field $\tilde{X}^\mu$ and will end up being the conservation equation for the stress-tensor  pushed-forward to the physical target space, $T^{\mu\nu}$.

The reader may be wondering what about the spacetime Grassmann fields, which were introduced to incorporate the Schwinger-Keldysh superspace structure directly in the physical spacetime. However, here target space symmetries come to our rescue. What used to be ordinary diffeomorphisms in spacetime, now get upgraded to target space superdiffeomorphisms that act on  $\{\SF{X}^\mu,\SF{\Theta},\SF{\bar{\Theta}} \}$. Furthermore, a fluid dynamical effective field theory is required to respect this spacetime symmetry; fluids cannot have potentials in physical spacetime. Consequently, this superdiffeomorphism symmetry can be exploited to fix a form of super-static gauge. We gauge fix:
\begin{equation}
\SF{\Theta}=\theta\,, \qquad \SF{\bar{\Theta}}=\thetab \,,
\label{eq:superstatic}
\end{equation}	
to simplify our discussion. As a consequence, in many formulae we will be sloppy about indicating the full target super-tensor structure, so often the reader will encounter isolated target spacetime components (indexed by lowercase Greek).

Let us take stock, now that we have assembled all the ingredients. The data for the hydrodynamic effective field theories is captured by the following:
\emph{
\begin{itemize}
\item A space filling Brownian-brane with intrinsic coordinates $z^I = \{\sigma^a, \theta, \thb\}$ on the worldvolume.
\item A $\UT$ gauge super-multiplet captured by $\As_I$, and a reference super-vector $\SF{\Kref}^I$, which has been partially gauge fixed to have only its $\Kref^a$ component non-zero, cf., \eqref{eq:betagauge}.
\item Target space superfields $\SF{X}^{\uA} = \{\SF{X}^\mu,\SF{\Theta},\SF{\bar{\Theta}} \}$, which are the dynamical degrees of freedom in the theory and corresponding sources for conserved currents. For neutral fluids we have a spacetime super-metric, which has been gauge fixed to satisfy \eqref{eq:gsptgauge}. 
\item Target spacetime superdiffeomorphisms are exploited to set $\SF{\Theta} = \theta$ and $\SF{\bar{\Theta}} = \thb$, leaving only $\SF{X}^\mu$ as the physical degrees of freedom. They transform as in \eqref{eq:xut} under the worldvolume $\UT$ gauge symmetry. 
\end{itemize} }
We have summarized this information after taking various gauge fixings into account in a tabular form in Table~\ref{tab:multiplets}.

\begin{table}[th!]
\centering
\begin{tabular}{|| c || c | c | c ||c|| c||}
\hline\hline
  {\shadeR{ghost}} & {\shadeB{Faddeev-Popov}}& {\shadeB{Vafa-Witten ghost}} & {\shadeB{Vector}} & {\shadeB{Position}} \\
  {\shadeR{charge}} & {\shadeB{ghost triplet}}  & {\shadeB{of ghost quintet}} & {\shadeB{quartet}} & {\shadeB{multiplet}} \\
  \hline
  {\shadeR{2}} &          & $\qquad\phiT \qquad\qquad$ &  &  \\
  {\shadeR{1}} & $\GT\qquad$                          & $\qquad\etaT$      & $\lambda_a$ & $X_\psi^\mu$ \\
  \hline
  {\shadeR{0}} & \shadeR{$\qquad \BT$} & \shadeR{$\phizT\qquad$}      &
  \shadeR{ $\quad\mathscr{A}_a \qquad \sAt{a}\quad$} & \shadeR{$X^\mu \qquad \tx^\mu$} \\
  \hline
  {\shadeR{-1}}& $\GbT\qquad$               & $\qquad\bar\etaT$& $\bar{\lambda}_a$   & $X_{\psib}^\mu$    \\
  {\shadeR{-2}}&      & $\phibT \qquad$ &  &\\
\hline\hline
\end{tabular}
\caption{The set of basic fields in the $\mathcal{N}_T = 2$ superalgebra along with their respective ghost number assignments, in terms of which the hydrodynamic effective action is written.}
\label{tab:multiplets}
\end{table}

In what follows we will explain how to use this data, the target space and worldvolume symmetries (including target space CPT) to construct hydrodynamic effective field theories as we envisaged in \cite{Haehl:2015uoc}. We will carry out this exercise first somewhat abstractly, indulging in  superspace variational calculus, to extract some general lessons. We then will illustrate this with examples at low order in the gradient expansion, deriving explicit actions involving the classical and fluctuation fields. The last stage of our discussion will be to make contact with the eightfold classification of \cite{Haehl:2015pja}.

In order to keep the presentation reasonable, and to write down formulae in a succinct manner, we have relegated some of the details on how the various multiplets appearing in Table~\ref{tab:multiplets} are constructed to Appendices. The reader interested in details of how the multiplets are organized is invited to consult the following:
\begin{itemize}
\item Appendix~\ref{sec:Adjoint} for the general structure of the adjoint multiplet.
\item Appendix~\ref{sec:GaugeMult} for the general structure of the gauge multiplet.
\item Appendix~\ref{sec:PositionMult} for the multiplet containing the hydrodynamic pion fields.
\item Finally, Appendix~\ref{sec:MetCurvMult} provides details on the construction of the worldvolume metric which plays an important role in the construction of the actions.
\end{itemize}

Readers analyzing these appendices are advised to note that we first develop the structure of the multiplet on a flat worldvolume; we do not endow the worldvolume with an intrinsic metric. For this purpose it suffices to consider the $\UT$ covariant derivative $\Dut$ introduced in \eqref{eq:covDI}. Of course in the physical theory, we want to work with the pullback metric $\SF{\gref}_{IJ}$ and its associated  covariant derivative $\Dwv$. This turns out to be a bit more involved and we describe in Appendix~\ref{sec:wvconnect} how one might go about constructing a worldvolume connection, after describing the symmetries and gauge fixing constraints we impose on our construction in Appendix~\ref{sec:gaugesdiff}.

\subsection{Comments on $\UT$ gauge dynamics} 
\label{sec:caveat}

Before proceeding however we should make one important disclaimer. A complete theory would also involve us giving a prescription for $\UT$ dynamics. While we believe this is possible, and is captured by a topological $BF$-type theory, we have not yet managed to construct all the  machinery necessary to give a satisfactory presentation here. We will therefore ignore the $\UT$ dynamics, for the most part, and effectively treat $\As$ as a background gauge field. 

Furthermore, we will be so bold as to assume that  it is consistent to give thermal expectation values to all the fields of the $\UT$ gauge multiplet such that 
\begin{equation}
\langle \SF{\Ascr}_a \rangle = 0  \,, \qquad 
  \langle \SF{\Ascr}_\thb \rangle  = 0 \,, \qquad  \langle \SF{\Ascr}_\theta \rangle = \thb \, (-i)\,.
 \label{eq:Azero}
\end{equation}
after variations on the worldvolume theory. This amounts to having a non-zero flux for the super-field strength, 
\begin{equation}
\langle \mathscr{F}_{\thb\theta} \rangle = -i
\label{eq:Fththbvev}
\end{equation}	
 We refer to this limit as the MMO limit after Mallick-Moshe-Orland \cite{Mallick:2010su} and will carry out explicit computations in this setting in \S\ref{sec:mmo}. We therefore define the background gauge configuration:
\begin{equation}
\textbf{MMO limit:}  \qquad \boxed{ \SF{\Ascr}_a  = 0  \,, \qquad 
   \SF{\Ascr}_\thb   = 0 \,, \qquad   \SF{\Ascr}_\theta  = \thb \, (-i)\,.}
 \label{eq:Azero}
\end{equation}

We have argued previously \cite{Haehl:2015uoc,Haehl:2016uah} that the reason behind this component of the field strength acquiring an imaginary vev has to do with spontaneous CPT symmetry breaking in systems with dissipation. This spontaneous breaking of CPT underlies the origin of the macroscopic arrow of time, codified into the second law. This picture is inspired by the discussions of Mallick et.al.,  \cite{Mallick:2010su} and especially Gaspard \cite{Gaspard:2012la,Gaspard:2013vl}, who derive the non-equilibrium Jarzynski relation \cite{Jarzynski:1997aa} and the associated work statistics of Crooks \cite{Crooks:1999fk}, using this strategy. We leave it to future work to show that such a mechanism exists.

It is actually not hard to argue that an appropriate $BF$ theory exists, for our construction closely hews to the logic of $\mathcal{N}_\smallT =2$ balanced topological field theories discussed in \cite{Cordes:1994sd,Blau:1991bn,Dijkgraaf:1996tz,Blau:1996bx}. The prototype example, which was the inspiration for much of our work was the topologically twisted version of $\mathcal{N}= 4$ SYM  constructed in \cite{Vafa:1994tf}. The gauge dynamics we want to write down is the generalization of their analysis to thermal $\UT$ gauge symmetry (which is not complicated), and additionally extend it to arbitrary dimensions. The latter is necessary for us, since we want to describe the nature of fluids in any number of spacetime dimensions. While the $BF$ theory does exist with the requisite $\mathcal{N}_\smallT$ symmetries in arbitrary dimensions, the structure of the $\SF{B}$-multiplet changes owing to the fact that it has to capture a spacetime codimension-2 form. It should be possible to work this out in detail (in fact using some of the existing technology, see e.g., \cite{Blau:1996bx}) and demonstrate the aforementioned claim. 

Due to this assumption, we will frequently drop terms in calculations, which will eventually give rise  to expressions proportional to one of the components of the $\UT$ gauge field that we set to zero at the end of the calculation. There is no obstruction to computing all these terms, but they proliferate quickly and obscure some of the physical aspects of the construction. We do keep track of those terms involving Grassmann-odd tensor fields (or components thereof) which play a role in the physical interpretation. 

Once we gauge fix the $\UT$ gauge field as in \eqref{eq:Azero} the $\mathcal{N}_\smallT =2$ thermal equivariant algebra simplifies drastically. One can show (see Appendix~\ref{sec:EqReview}) that the topological symmetries can be captured by two supercharges (these are the Cartan charges in the equivariant construction) $\Q, \Qb$ which are nilpotent, but anti-commute to a thermal translation:
\begin{equation}
\Q^2  = \Qb^2 =0 \,, \qquad \{\Q,\Qb\} = i\, \lieD_\Kref 
\label{eq:mmoQ}
\end{equation}	
Since the Lie derivative acts on scalars as $\lieD_\Kref =  \Kref^a \partial_a$, once we align $\Kref^a = \beta\, \delta^a_0$ we see that $\{\Q, \Qb\} \sim i\, \beta \partial_{\sigma_0}$. This gauge fixed algebra appears to be  well known in the statistical mechanics literature, and is used for example in \cite{Mallick:2010su} in their derivation of the Jarzynski relation. In this form, this algebra also is the one written down in \cite{Crossley:2015evo} in the high-temperature (or as they put it, classical) limit.\footnote{ In interest of completeness, let us note that \cite{Crossley:2015evo} posit that there is a single nilpotent supercharge $\delta$ implementing the Schwinger-Keldysh alignment condition, another emergent nilpotent supercharge $\bar \delta$ arising from the KMS condition, satisfying altogether:
$$ \delta^2 = \bar{\delta}^2  =0 \,, \qquad \{\delta, \bar{\delta}\} = 2\, \tanh\left(\frac{i}{2}\, \Kref\, \partial_t \right) \approx\, i\, \Kref \partial_t \,, $$
where the last approximation holds in the high temperature limit.}  We will refer to the limit captured by 
\eqref{eq:Azero}, and the resulting superalgebra \eqref{eq:mmoQ}, as the MMO limit henceforth. In this limit, we will recover the constructions described in \cite{Haehl:2015uoc,Crossley:2015evo} (see also \cite{Glorioso:2017fpd,Jensen:2017kzi}). 

Given this, the reader may ask, why do we even care if the $\UT$ symmetry is gauged, since after all we are effectively treating it as a global symmetry, and are picking a suitable background field for our analysis in \eqref{eq:Azero}. There is an important physical distinction which drives our considerations as we explained in \cite{Haehl:2018uqv}, which we elaborate here.

If we stick to a background $\UT$ gauge field $\As_I$, restricted to ensure that the field strength component $\vev{\Fs_{\theta\thb}|} = -i$ as in \eqref{eq:Azero}, {\sf CPT} is broken explicitly in the theory. We will  a-priori have biased the theory towards picking out an arrow of time which leads to entropy production in the fluid. On physical grounds however, we expect {\sf CPT} breaking to emerge dynamically rather than being imposed by fiat from the beginning. This entails that we allow for a framework where dynamics picks out saddle points where $\vev{\Fs_{\theta\thb}|} \neq -i$. This clearly requires for $\As_I$ to be a dynamical field in the problem.\footnote{ In this discussion and for the rest of the paper we implicitly assume that the scale at which the 
$\UT$ gauge symmetry emerges is commensurate with the scale at which {\sf CPT} is broken. We offer some further thoughts on the relative hierarchy of scales between these two phenomena in \S\ref{sec:discussion}. \label{fn:utcpt}}  As we shall see below, the $\UT$ super-Bianchi identity, which is independent of the particulars of the gauge dynamics ensures that the corresponding current is conserved. Absent any fundamental obstructions to gauging the thermal diffeomorphisms and making the superspace gauge fields dynamical, it behooves us to consider the possibility for them to be so.

\newpage
\section{Dissipative effective action and entropy inflow mechanism}
\label{sec:topsigma}

We are now in a position to make our  central claim and write explicit hydrodynamic effective actions. We posit the following:\vspace{.2cm}

\noindent
{\it
All of hydrodynamic transport consistent with the second law is described by effective actions of the form
\begin{equation}\label{eq:fullLag}
\boxed{
S_\text{wv} = \int d^d \sigma \; \Lag_\text{wv} \,,\qquad 
\Lag_\text{wv} = \int\,d\theta\, d\thetab\; \frac{\sqrt{-\SF{\gref}}}{\zsf} \ \SF{\Lagref}[\SF{\gref}_{IJ},\Kref^a,\Dwv_I, \gpsib_{IJ}, \gpsi_{IJ}] \,,
}
\end{equation}
provided the following symmetries are respected:
\begin{enumerate}
\item Invariance under $\UT$ thermal diffeomorphisms and ${\cal N}_T=2$ BRST symmetry.
 \item Physical spacetime superdiffeomorphisms $X^{\uA} \mapsto X^{\uA}\, + \,  \widehat{H}^{\uA}(X)$, with $\widehat{H}^{\uA}(X)$ being a target super-vector.
 \item Worldvolume diffeomorphisms $z^I \mapsto z^I + f^I(\sigma^a)$, with $f^I(\sigma^a)$ being a worldvolume super-vector.
 \item Anti-linear {\sf CPT} involution.
 \item Ghost number conservation.
\end{enumerate}
}
\smallskip
\noindent 
In addition we require that the imaginary part of this action is constrained to be positive $\text{Im}(S_\text{wv}) \geq 0$ (see Appendix A of \cite{Glorioso:2016gsa} for a clear discussion).

To rephrase this statement: {\it any} Lagrangian that is allowed by the ${\cal N}_T = 2$ symmetry of $\UT$ covariant Schwinger-Keldysh formalism is allowed and is consistent with the second law, and these actions are complete vis a vis hydrodynamic transport at all orders in the derivative expansion.\footnote{ This holds in the usual perturbation or effective field theory sense, i.e., we are not claiming to have a non-perturbative theory.} 

Some comments and explanations are in order: 
\begin{itemize}
\item In writing the action we introduced a derivative operator $\Dwv_I$ which upgrades the $\UT$ covariant derivative introduced in \eqref{eq:covDI} to allow for construction of superdiffeomorphism covariant tensors. We will only require that this derivative operator be such that: $(i)$ we can integrate by parts, and $(ii)$ $\Dwv_I \SF{\phi} = \Dut_I \SF{\phi}$ (i.e., the action on scalars agrees with the $\UT$ covariant derivative); we make no further assumption about the connection that specifies it. Importantly, it will not be required to be metric compatible, which upends some of the  standard intuition. In Appendix~\ref{sec:wvconnect} we discuss what classes of connections are compatible with this assumption, and construct and explicit example that we work with for explicit computations in \S\ref{sec:mmo} and \S\ref{sec:mmoEx}. Our choice of connection for explicit computations is summarized in \S\ref{sec:wvconn}.\footnote{ In actual implementation we also require that the commutators of the two derivatives on scalars agrees and closed on $\UT$ field strengths, which can be defined in the absence of a metric connection.}
\item The measure for superspace integration involves the field 
\begin{equation}
\zsf = 1+ \SF{\Kref}^I\, \As_I
\label{eq:zdef}
\end{equation}	
which explicitly depends on the gauge field. Its origins can be traced back to the fact that our pullback maps are implemented using the gauge covariant derivative \eqref{eq:gref}. Given the transformation of the hydrodynamic pions \eqref{eq:xut} it is easy to check that
\begin{equation}
\Dwv_I \SF{X}^\mu = \Dut_I \SF{X}^\mu = \partial_I \SF{X}^\mu + (\As_I, \SF{X}^\mu)_\Kref =   (-)^J\, 
\left(\delta_I^{\ J} + \As_I\, \SF{\Kref}^J\right)\partial_J \SF{X}^\mu \,.
\label{eq:Dparz}
\end{equation}	
Consequently, we end up with factors of $\zsf$ when taking traces, determinants, etc., as noted in  \cite{Haehl:2015uoc}: 
\begin{equation}
 d^dX \,d\Theta\, d\bar{\Theta} \; \sqrt{-\SF{g}} = d^d\sigma\, d\theta\, d\thb \; \sqrt{-\SF{\gref}} \, \frac{\det [\partial_I \SF{X}^{\uA}]}{\det [\Dut_I \SF{X}^{\uA}]}  = d^d\sigma\, d\theta\, d\thb \; \frac{ \sqrt{-\SF{\gref}}}{\SF{{\bf z}}} \,.
\end{equation}
\item We have also introduced the fields 
\begin{equation}
\gpsib_{IJ} \equiv \Dut_\theta \SF{\gref}_{IJ}\,,\qquad \gpsi_{IJ} \equiv \Dut_\thb \SF{\gref}_{IJ} \,.
\end{equation}
These turn out to be covariant 2-tensors consistent with all our symmetries (as explained in Appendix \ref{sec:wvsym}).
\item Finally, note that symmetries 1, 2, and 3 are manifest in our formulation, while 4 and 5 can be trivially implemented.  
\end{itemize}

We now describe how the effective actions of the form \eqref{eq:fullLag} maintain consistency with the second law, give rise to the correct dynamical equations of motion, and note some additional salient features. Most of these are summarized in the companion paper \cite{Haehl:2018uqv}, and non-superspace versions of some statements have already been noted earlier in \cite{Haehl:2015uoc}.

\subsection{Super-adiabaticity from $\UT$ Bianchi identity}
\label{sec:einflow}

Let us first see how the action maintains consistency with the second law. To this end consider the  Ward identity associated with a $\UT$ gauge transformation by a parameter $\SF{\Lambda}$. Define the 
energy-momentum and free-energy  Noether super-currents\footnote{ Note that variation of with respect to $\SF{\gref}_{IJ}$ inside $\{\gpsi_{IJ},\gpsib_{IJ}\}$ is well defined since integration by parts of the Grassmann odd derivatives in $\{\gpsi_{IJ} = \Dut_\thb \SF{\gref}_{IJ},\gpsib_{IJ} = \Dut_\theta \SF{\gref}_{IJ}\}$ is allowed in this case. See Appendix \ref{sec:wvsym}.}
\begin{equation}
\begin{split}
\SF{\TEMref}_{}^{IJ} 
&\equiv 
	\frac{2}{\sqrt{-\SF{\gref}}} \, \frac{\delta}{\delta \SF{\gref}_{IJ}}
	\prn{\sqrt{-\SF{\gref}}\ \SF{\Lagref}}\,,
\\ 
\SF{\Nref}_{}^I 
&\equiv
	 -\frac{\zsf}{\sqrt{-\SF{\gref}}} \, \frac{\delta }{\delta \SF{\Ascr}_{I}}
	 \prn{\frac{\sqrt{-\SF{\gref}}}{\zsf}\ \SF{\Lagref} } \,. 
\end{split}
\label{eq:TNwvdef}
\end{equation}	
The expression for the energy-momentum tensor should be familiar (modulo our upgrade to a super-tensor). The Noether current $\SF{\Nref}^I$ is related to the free-energy super-current $\SF{\Gref}_{}^I$ up to a factor of the temperature:
\begin{equation}
\SF{\Nref}_{}^I = -\frac{\SF{\Gref}_{}^I}{\SF{T}} \,,
\label{eq:NG}
\end{equation}	
generalizing the  construction of free-energy currents in hydrodynamic sigma models \cite{Haehl:2015pja}.

Consider a $\UT$ transformation by $\SF{\Lambda}$. The worldvolume metric inherits the $\UT$ transformation from \eqref{eq:xut} while the gauge field $\As_I$ transforms inhomogeneously as in \eqref{eq:aut}. The reference thermal super-vector $\SF{\Kref}^I$ does not transform. All told, we find that the $\UT$ gauge transformation acts as follows on the action:\footnote{ In \eqref{eq:utcalcfull}  the first term inside the braces has a sign $ \, (-)^{I+J+IJ} $ from super-index contractions, which we have suppressed, while the second term is free of any signs. We will suppress signs from index contractions in all of the present section.\label{fn:signsSA} }
\begin{equation}\label{eq:utcalcfull}
\begin{split}
\delta_{\SF{\Lambda}}  \, S_\text{wv} 
&=
	 \int d^d \sigma \int\,d\theta\, d\thetab\, \frac{\sqrt{-\SF{\gref}}}{\zsf}  
	 \,\left\{\frac{1}{2}  \SF{\TEMref}_{}^{IJ}  \, (\SF{\Lambda}, \SF{\gref}_{IJ})_\Kref  + \Dut_I (\SF{\Lambda})\,  \SF{\Nref}_{}^I   \right\}\\
&=
	 \int d^d \sigma \int\,d\theta\, d\thetab\, \frac{\sqrt{-\SF{\gref}}}{\zsf} \; \SF{\Lambda} \,\left\{ \frac{1}{2} \, \SF{\TEMref}_{}^{IJ}  \, \lieD_\Kref \SF{\gref}_{IJ}  -  \Dwv_I  \SF{\Nref}_{}^I   \right\} \,.
\end{split}
\end{equation}
In writing the second line, we invoked the $\UT$ transformation of the pullback metric, inherited from \eqref{eq:xut},  and  performed an integration by parts in superspace.\footnote{ We reiterate that we have not explicitly specified the worldvolume connection. Importantly, it is not the Christoffel connection of $\SF{\gref}_{IJ}$since pullback is performed respecting $\UT$ covariance. We prove the existence of a measure compatible connection in Appendix~\ref{sec:wvconnect}, which along with the fact that the connection does not contribute to derivatives of worldvolume scalars, 
is all we need for \eqref{eq:utcalcfull}.  }

 The quantity inside the curly braces must vanish due to the  $\UT$ symmetry. We shall refer to this Ward identity as the \emph{super-adiabaticity equation}:
\begin{equation}\label{eq:superadiabatic} 
\boxed{
\Dwv_I  \SF{\Nref}_{}^I-\frac{1}{2} \, \SF{\TEMref}_{}^{IJ}  \, \lieD_\Kref \SF{\gref}_{IJ}    =0  \,.
}
\end{equation}
This equation turns out to embody the physics of entropy production, and allows for a clean parameterization of dissipative contributions. The connection is made through the off-shell adiabaticity equation introduced in \cite{Haehl:2015pja}. We review the basics of that analysis and explain how entropy production arises in terms of a superspace inflow. One can equivalently view this directly as a non-equilibrium detailed balance condition analogous to the Jarzynski relation (as partly explained in 
\cite{Haehl:2015uoc}). Before doing so however it will be helpful to also understand the dynamical content of the superspace action \eqref{eq:fullLag}.

\subsection{Fluid dynamical equations of motion}
\label{sec:heom}

In superspace, we simply derive the equations of motion by varying \eqref{eq:fullLag} with respect to the superfield $\SF{X}^\mu$. Since $\SF{\gref}_{IJ}$ is the dynamical worldvolume field (through its $\SF{X}^\mu$-dependence), this means that the variation will depend on the stress tensor. We thus find the following equation of motion: 
\begin{equation} 
\Dwv_I \left( \SF{\TEMref}^{IJ} \, \Dwv_J \SF{X}^\mu \right) = 0 \,.
\label{eq:supereom}
\end{equation}
At face value this seems not quite what we need, since the equation above involves all the super-components of the energy-momentum tensor. In particular, expanding out in components and distributing derivatives, we infer that\footnote{ Note there is overall $(-)^J$ sign coming from the contraction of indices.}
\begin{equation}
\begin{split}
0 &= 
\underbrace{
	\Dwv_a \left( \SF{\TEMref}^{ab} \, \Dwv_b \SF{X}^\mu \right)  
	+\prn{ \Dwv_\theta\SF{\TEMref}^{a \theta} \;   
	+ \Dwv_\thb   \SF{\TEMref}^{a\thb } } \Dwv_a \SF{X}^\mu  
	 - \SF{\TEMref}^{\theta \thb} \left(\Dwv_\theta  \Dwv_\thb-\Dwv_\thb\Dwv_\theta\right)  \SF{X}^\mu 
	}_\text{classical + fluctuations} 
 \\
& \qquad
+ \underbrace{	 
	 \SF{\TEMref}^{a \theta} \, \Dwv_\theta  \Dwv_a \SF{X}^\mu   - \Dwv_a \left( \SF{\TEMref}^{a \theta} \, \Dwv_\theta \SF{X}^\mu \right)  +  \SF{\TEMref}^{\thb b} \, \Dwv_\thb \Dwv_b \SF{X}^\mu  
	   - \Dwv_a \left( \SF{\TEMref}^{a\theta} \, \Dwv_\thb  \SF{X}^\mu \right) 
	 }_\text{ghost bilinears}
 \\
& \qquad
+ \underbrace{	 	     
	   \Dwv_\thb \SF{\TEMref}^{\theta  \thb} \, \Dwv_\theta \SF{X}^\mu  
	 - \Dwv_\theta  \SF{\TEMref}^{\theta \thb} \; \Dwv_\thb  \SF{X}^\mu 
	 }_\text{ghost bilinears} 
\end{split}
\label{eq:seom1}
\end{equation}
The second and third line are indicated as being ghost bilinears since each term is made of tensors carrying non-vanishing ghost number. These we are free to ignore, so the physical equation of motion that 
arises from \eqref{eq:fullLag} simply can be expressed as the projection of the first line onto ordinary space:
\begin{equation}
\begin{split}
\bigg\{
	\Dwv_a \left( \SF{\TEMref}^{ab} \, \Dwv_b \SF{X}^\mu \right)  
	+\prn{ \Dwv_\theta\SF{\TEMref}^{a \theta} \;   
	+ \Dwv_\thb   \SF{\TEMref}^{a\thb } } \Dwv_a \SF{X}^\mu  
	 - \SF{\TEMref}^{\theta \thb} \left(\Dwv_\theta  \Dwv_\thb-\Dwv_\thb\Dwv_\theta\right)  \SF{X}^\mu 
\bigg\}  \big|= 0 + \text{ghosts.}
\end{split}
\label{eq:seom2}
\end{equation}
Setting the ghost bilinears implicit  in \eqref{eq:seom2} to zero, we will end up with equations that involve both the classical field $X^\mu$ as well the fluctuations $\tx^\mu$. The latter are ignored in the classical hydrodynamic equations, but the advantage of having a full effective action, is that the deformation to the equations of motion owing to the statistical (and quantum) fluctuations is made explicit.

The astute reader might wonder why such fluctuation terms are not encountered in the MSR construction for Langevin dynamics \cite{Martin:1973zz} (see the textbook discussion in \cite{ZinnJustin:2002ru}). In that case there is a stochastic noise term, which is assumed to be Gaussian. Consequently, the fluctuations enter simply as a Lagrange multiplier enforcing the physical dynamical equation, which can also be seen from the explicit superspace construction described in \cite{Haehl:2016uah}. The novelty in hydrodynamics is the non-Gaussianity of the noise. There can be (and in general are) non-trivial noise kernels in the system. As a result the fluctuation variable will no longer enter simply as a Lagrange multiplier. 

The formalism we have outlined here has the power to completely encompass such behaviour. This is somewhat hard to see at the abstract level we are describing, so we will indeed exemplify some of these statements with explicit calculations for dissipative hydrodynamics up to second order in gradients in 
\S\ref{sec:mmoEx}. 

While ignoring the ghost bilinears allows us to drop the second two lines of \eqref{eq:seom1}, we still have to understand the contribution from the  super-components of the energy-momentum tensor. We will now argue that they do not contribute to the leading classical dynamics, and in fact should simply provide the fluctuation terms in the equation of motion. Rewriting the commutator of Grassmann odd derivatives, we claim:
\begin{equation}
\begin{split}
\left\{\prn{ \Dwv_\theta\SF{\TEMref}^{a \theta} \;   
	+ \Dwv_\thb   \SF{\TEMref}^{a\thb } } \Dwv_a \SF{X}^\mu  
	 - \SF{\TEMref}^{\theta \thb} \left(2 \, \Dwv_\theta\Dwv_\thb \SF{X}^\mu- (\Fs_{\theta\thb},\SF{X}^\mu)_\Kref  \right)  
	 \right\} \big| & 
	 = \mathcal{O}(\tx^\mu) \,.
\end{split}
\label{eq:seom3f}
\end{equation}
This turns out to be pretty non-trivial to prove, but is explicitly borne out in examples that we have computed (see \S\ref{sec:mmogeneral}). With this understanding, we can then see that  the classical hydrodynamic equations are contained entirely in the first term of \eqref{eq:seom2}, viz., 
\begin{equation}
\Dwv_a \left( \SF{\TEMref}^{ab} \, \Dwv_b \SF{X}^\mu \right)  \big|  = 0 \,.
\label{eq:seom3}
\end{equation}	
To derive the first equation above we used the pullbacks and the fact that 
$ 0 = \Dwv_I  \SF{e}^I_\mu | = \Dwv_a \SF{e}^a_\mu | $ which follows from our choice of the worldvolume connection be measure compatible, cf., \eqref{eq:modrels}. 

Equations \eqref{eq:seom3f} and \eqref{eq:seom3} should be understood as follows: there is no superspace inflow of energy-momentum modifying the classical hydrodynamic equations of motion. At best there are additional fluctuation contributions arising from the Grassmann-odd directions. We view this as a non-trivial consistency check of our formalism's ability  to incorporate the correct dynamics.

This separation makes explicit the idea that we can capture the classical part of the equations of motion by focusing on the ordinary space components alone. The role of the superspace is to bring in the fluctuations (and associated Grassmann odd terms). Given that the superspace directions control the entropy production, it makes sense for them to capture the fluctuation deformations of the dynamical equations as well.

Let us also convince ourselves that the worldvolume equation of motion \eqref{eq:seom3} indeed gives the correct dynamical equations for the hydrodynamic fields after pushing forward the data to the target space. Writing\footnote{ $\SF{e}_{\uA}^{\ I}$ span the dual basis with 
$(-)^I\, \SF{e}_{\uA}^{\ I}\, \Dwv_I \SF{X}^{\uB} = \delta_{\uA}^{\ \uB}$ and 
$(-)^{\uA}\, \Dwv_J \SF{X}^{\uA}\,\SF{e}_{\uA}^{\ I}=  \delta_{J}^{\ I}$, written as usual in DeWitt conventions \cite{DeWitt:1992cy}.} 
$\SF{\TEMref}^{IJ} = \SF{e}_{\uA}^{\ I} \,\SF{e}_{\uB}^{\ J} \, T^{\uA\uB} $ (contraction signs left implicit) we have 
\begin{equation}
\begin{split}
\Dwv_a \left( \SF{\TEMref}^{ab} \, \Dwv_b \SF{X}^\mu \right)  \big|  
&= \Dwv_a \SF{X}^{\uB} \, \nabla_{\uB}\left(\SF{e}^a_{\uC} \;T^{\uC\mu} \right) | 
 \\
& = \nabla_\nu T^{\mu\nu} + \Dwv_a \SF{e}^a_\nu \; T^{\mu\nu} | + \text{ghost bilinears + fluctuations} 
\\
& = \nabla_\mu T^{\mu\nu} + \text{ghost bilinears + fluctuations}
\end{split}
\label{eq:}
\end{equation}	
where the last equality is ensured by the measure compatibility of the worldvolume connection as can be seen directly from \eqref{eq:modrels}. The issues arising from the $\UT$ covariant pullbacks are taken care of by the fact that \eqref{eq:supereom} is a target space vector, and the contracted indices transform homogeneously in both the worldvolume and target space. We will have use for this statement when we try to understand the super-adiabaticity equation below.

The equations of motion we have written down only involve the variations of the superfield $\SF{X}^\mu$. In target space we also have supercoordinates $\SF{\Theta}$ and $\SF{\bar{\Theta}}$, which we have previously gauged fixed to $\SF{\Theta} =\theta$ and $\SF{\bar \Theta}  = \thb$. We could of course vary with respect to these fields (prior to gauge fixing) and get corresponding equations -- these will be related by the underlying BRST supersymmetry to \eqref{eq:supereom}. What is a-priori clear is that these equations will be Grassmann odd and thus proportional to the ghost fields. Since for the most part we will ignore the ghost partners below, we will not have occasion to analyze these equations of motion in detail here. We should however note that a corresponding analysis for Langevin dynamics has been explicitly carried out earlier in section 7 of  \cite{Haehl:2016uah}.

\subsection{Adiabaticity equation, entropy inflow, and the second law}
\label{sec:adiab2nd}

We recall that the conventional axiomatic formulation of hydrodynamics not only demands that the dynamics be governed by conservation equations, but requires that the local form of the second law be upheld on-shell in every fluid configuration. Usually this is taken to mean that there exists an entropy current $J^\mu_S$  with non-negative divergence. It is however efficacious to take the statement of entropy production off-shell, by introducing a Lagrange multiplier field. Exploiting the field-redefinition freedom in hydrodynamics, this can be recast as the statement \cite{Haehl:2015pja}:
\begin{equation}
\nabla_\mu J^\mu_S + \Kref_\mu \nabla_\nu T^{\mu\nu} = \Delta \geq 0 \,,
\label{eq:eprodoff}
\end{equation}	
where $\nabla_\mu$ is the usual target space metric compatible covariant derivative.
Note that  $\Delta$ is the total (off-shell) entropy production. Legendre transforming the entropy current to  free energy Noether current via 
\begin{equation}
N^\mu = J_S^\mu + \Kref_\nu \, T^{\mu\nu}
\label{eq:Ndef}
\end{equation}	
we arrive at the grand canonical form for the off-shell statement of the second law, which demands that
\begin{equation}\label{eq:ordadiab}
\nabla_\mu N^\mu -\frac{1}{2}\, T^{\mu\nu} \lieD_\Kref g_{\mu\nu}  = \Delta \geq 0 \,,
\end{equation}
where $\lieD_\Kref g_{\mu\nu} = 2\, \nabla_{(\mu} \Kref_{\nu)}$.

The special case where transport produces no entropy was given the name \emph{adiabatic fluids}. In this case
an ordinary-space adiabaticity equation holds:
\begin{equation}
\qquad \frac{1}{2} \, T_{\text{(adiabatic)}}^{\mu\nu}\, \lieD_\Kref g_{\mu\nu} = \nabla_\mu N_{\text{(adiabatic)}}^\mu \,.
\end{equation}

The adiabaticity relation \eqref{eq:ordadiab} was initially discussed in the target space as reviewed above, but we can readily pull-back the  relation to the worldvolume. Firstly, we note that 
\begin{equation}
\begin{split}
\Dwv_a \SF{\Nref}^a \big|
&=    
	\nabla_\mu N^\mu   + \text{ghost bilinears + fluctuations}\\
\SF{\TEMref}^{ab} \lieD_\Kref \SF{\gref}_{ab} \big|
&= 
	T^{\mu\nu} \lieD_\Kref g_{\mu\nu}  + \text{ghost bilinears + fluctuations}
\end{split}
\label{eq:}
\end{equation}	
Modulo some extra terms that we will interpret below, we can indeed view \eqref{eq:superadiabatic} as the worldvolume version of the adiabaticity equation (uplifted to superspace). It is interesting that the intricacies involved with the $\UT$ covariant pull-backs essentially cancel out completely after push-forward.  We will now try to extract the some basic lessons by decomposing the super-adiabaticity equation in more familiar terms by separating out the super-tensor components.

\paragraph{Super-adiabaticity equation as \emph{entropy inflow} from superspace:} 
In the manifestly ${\cal N}_T=2$ supersymmetric formalism, the adiabaticity equation is implied by
\eqref{eq:superadiabatic}.   Writing out the latter and separating the contributions from the ordinary and superspace directions, we end up with:
\begin{equation}
\begin{split}
\underbrace{\left(\Dwv_a \SF{\Nref}^a 
-\frac{1}{2} \, \SF{\TEMref}^{ab}\,\lieD_\Kref\, \SF{\gref}_{ab}\right)\bigg| }_{\text{classical + fluctuations}} 
= -  \underbrace{\left( \Dwv_\theta \SF{\Nref}^\theta + \Dwv_\thb \SF{\Nref}^\thb + \SF{\TEMref}^{a\theta}\,\lieD_\Kref\, \SF{\gref}_{a\theta}
+ \SF{\TEMref}^{a\thb}\,\lieD_\Kref\, \SF{\gref}_{a\thb}
+ \SF{\TEMref}^{\theta \thb }\,\lieD_\Kref\, \SF{\gref}_{\theta \thb} \right)\bigg|}_{\text{entropy inflow}}
\end{split}
\label{eq:superadcomp}
\end{equation}
 Firstly, noting that $\{\SF{\TEMref}^{ab}, \SF{\TEMref}^{\theta\thb}, \SF{\Nref}^a\}$ are the only Grassmann-even tensors (i.e., their bottom components are Grassmann even fields) we project \eqref{eq:superadcomp} down to $\theta=\thb=0$, isolating the part that contains these tensors.\footnote{ 
 This can be inferred by looking at the ghost charges or equivalently noting that  $\SF{\gref}_{a \theta}$ and $\SF{\gref}_{a\thb}$ are clearly Grassmann-odd owing to them being proportional to $\Dwv_\theta \SF{X}^\mu$ and $\Dwv_\thb \SF{X}^\mu$, respectively. } 
 It is worth reiterating here that when we project Grassmann even quantities to their bottom component, we will end up with expressions involving the classical fields $X^\mu$,  the fluctuations $\tx^\mu$,  the zero ghost number element of the Vafa-Witten quintet $\Fs_{\thb \theta}| =\phizT$, or  ghost bilinears (either of the form $\xpsib^\mu \xpsi^\nu$, or involving the gauge sector). 

 Therefore, upto contributions from the fluctuation fields (and ghost bilinears), the terms we have isolated on the l.h.s.\ of \eqref{eq:superadcomp} are precisely the combination appearing on the l.h.s.\ of \eqref{eq:ordadiab}. This naturally suggests interpreting the second set of terms indicated as `entropy inflow' as the part that contributes to $\Delta$. Hence, 
\begin{equation}
\boxed{
\Delta = - \left( \Dwv_\theta \SF{\Nref}^\theta + \Dwv_\thb \SF{\Nref}^\thb + \SF{\TEMref}^{a\theta}\,\lieD_\Kref\, \SF{\gref}_{a\theta}
+ \SF{\TEMref}^{a\thb}\,\lieD_\Kref\, \SF{\gref}_{a\thb}
+ \SF{\TEMref}^{\theta \thb }\,\lieD_\Kref\, \SF{\gref}_{\theta \thb} \right)\Big| \,.
}
\label{eq:Delta0}
\end{equation}	
We interpret this equation as saying that the entropy production is controlled by the Grassmann-odd descendants of the free energy and energy-momentum super-tensors. 

The total entropy produced manifests itself in the physical spacetime, although it originates from the superspace components. This is highly reminiscent of the anomaly inflow mechanism, where one realizes a theory with an anomalous symmetry as the boundary dynamics at the edge of a topological (bulk) field theory.  The connection is not coincidental -- the entropy current for charged fluids with mixed flavor-gravitational anomalies is naturally captured by introducing a thermal gauge field whose chemical potential is the temperature \cite{Jensen:2013rga}. This construction provided the rationale behind the introduction of the $\UT$ gauge field in \cite{Haehl:2014zda,Haehl:2015pja}, but as we see here the full justification is for this is provided by working within the aegis of the thermal equivariance formalism.

 To compare with conventional hydrodynamics, we only need to match the classical part, since this is all that is captured in the familiar presentation of the system in terms of constitutive relations.  Dropping the ghost bilinears and setting $\tx^\mu =0$ in the other terms, we see that we can easily match the first set of terms in \eqref{eq:Delta0} with the entropy production equation, provided we identify the ghost (Grassmann-odd) components of the free energy current,  $\SF{\Nref}^\theta$ and $\SF{\Nref}^\thb$, as encoding the total dissipation:
\begin{equation}
  \Delta   = 
  	 -\left( \mathfrak{D}_\theta \Nref^\theta+  \mathfrak{D}_\thb \Nref^\thb 
  	 \right)  + \text{ghost bilinears}\,.
 \label{eq:Delta1}	 
 \end{equation}
We  have dropped the contribution involving  $\TEMref^{\theta \thb }$ since it vanishes with our super-static gauge choice which sets the bottom component  $\SF{\gref}_{\theta\thb} | = i$.\footnote{ We note that 
$\lieD_\Kref \SF{\gref}_{\theta\thb} =  (-)^I \, \SF{\Kref}^I \partial_I \SF{\gref}_{\theta\thb} - \partial_\thb \SF{\Kref}^I\,\SF{\gref}_{\theta I}  +(-)^I  \,\partial_\theta \SF{\Kref}^I \,\SF{\gref}_{I \thb} $. The last two terms vanish because of our gauge choice on the thermal super-vector \eqref{eq:betagauge}. The first term does not obviously vanish (recall that we are working with a connection that is not metric compatible, so have to exercise care in using usual intuition). However, we have a flat  metric  in the superspace directions ($\SF{\gref}_{\theta\thb}|=i$) leading to a simple result: $\lieD_\Kref \SF{\gref}_{\theta\thb} | =0$.}

\paragraph{Positivity of entropy production:} We are now left with understanding why $\Delta \geq 0$, i.e., the second law.  Let us first note that an  abstract argument for entropy production may be given by noting that the super-adiabaticity equation is equivalent to the Jarzynski relation \cite{Jarzynski:1997aa,Jarzynski:1997ab} which provides the general non-equilibrium second law.  We have previously argued in \cite{Haehl:2015uoc} that $\UT$ invariance and breaking of {\sf CPT} symmetry imply Jarzynski's equation by employing the strategy described for Langevin dynamics in \cite{Mallick:2010su}. We won't  undertake the upgrade of this statement to superspace directly.  Instead we will follow an alternate route and make direct connection between the analysis here and the eightfold classification of hydrodynamic transport given in \cite{Haehl:2015pja}. We can then invoke our earlier results to constrain the contributions to $\Delta$ and show that it is indeed non-negative definite. 

Let us first attempt to gain some intuition for the terms in the r.h.s. of \eqref{eq:Delta1}. We will ignore the contribution from $\TEMref^{\theta a}$ and $\TEMref^{\thb a}$ for we have seen that they contribute only ghost bilinear terms. The physical entropy production will be a functional of the macroscopic degrees of freedom $X^\mu$, though we will often keep track of the fluctuation contribution as well (given by terms involving 
$\tx^\mu$). The free energy Noether current appears in the action (to linear order) as 
\begin{equation}
\begin{split}
S_\text{wv} \supset - \int d^d\sigma d\theta d\thb \; \frac{\sqrt{-\SF{\gref}}}{\zsf} \, \SF{\Nref}^I \, \As_I 
&=
	 - \int d^d\sigma \, \partial_\theta \, \partial_\thb \left( \frac{\sqrt{-\SF{\gref}}}{\zsf}\,\SF{\Nref}^I \, \As_I \right)  \Big{|} \\
&=
	 \int d^d \sigma \,\left( - \mathfrak{D}_\theta \Nref^\theta\; \partial_\thb \Ascr_\theta 
	+ \mathfrak{D}_\thb \Nref^\thb\; \partial_\theta \Ascr_\thb  + \cdots \right)
\end{split}
\label{eq:NAinS}
\end{equation}	
where the ellipses ($\cdots$) denote the ghost bilinear terms that we are dropping. We are also not writing non-linear terms, which are necessary to ensure full covariance, but will not contribute to  $\Delta$. Given \eqref{eq:NAinS} it is clear that $\Delta$ can be extracted by switching on suitable sources for the superfields  $\partial_\thb \Ascr_\theta $ and 
 $\partial_\theta \Ascr_\thb$;  to wit,
\begin{equation}
\Delta = 
	\left(\frac{\delta}{\delta \partial_\thb \Ascr_\theta} 
	- \frac{\delta }{\delta \partial_\theta \Ascr_\thb} \right) S_\text{wv}
\label{eq:Deltavar}
\end{equation}	
An easy way to extract the contribution to $\Delta$ is to turn on sources for the gauge superfield, 
$\As_\thb = \theta \, \sBdel$ and $\As_\theta = - \thb \, \sBdel$. By construction the source $\sBdel$ for $\Delta$ has 
ghost number zero; in fact it appears in the gauge potential in the same place as the bare ghost number zero element of the Faddeev-Popov triplet $\BT$ -- compare the above with \eqref{eq:AthbExp} and \eqref{eq:AthExp}, respectively.\footnote{ 
To be clear, we  turn on the source $\sBdel$  only in the two aforementioned components, so as to aid in extracting $\Delta$ directly from the action. In particular,  we do not perform a full Faddeev-Popov rotation discussed in Appendices~\ref{sec:GaugeMult} and \ref{sec:PositionMult}. The difference amounts to the following: a full FP-rotation by $\BT$ will shift  the top component of the vector multiplet 
$\sAt{a} \mapsto \sAt{a} + D_a \BT$, as well as the fluctuation field, $\tx^\mu \mapsto \tx^\mu - \dbrk{\BT,X^\mu}$, as is clear from \eqref{eq:AaExp} and \eqref{eq:Xsuper}, respectively. These combinations are $\UT$ covariant. However, extracting $\Delta$ requires us to break the $\UT$ covariance since we want to isolate a part of the full adiabaticity equation. A full Faddeev-Popov rotation to introduce $\BT$ will, the reader can check, give back the  super-adiabaticity equation (more precisely: its bottom component).
\label{fn:BSnotBT}} In explicit computations we will encounter non-linear terms in  $\sBdel$, which will be necessary to ensure that the action transforms correctly under  $\UT$ and ${\sf CPT}$. The easiest way to infer the contribution of $\sBdel$ will turn out to be by means of performing a Faddeev-Popov rotation of the fields with a gauge parameter $\SF{\Lambda}_{FP} = \thb \theta \, \sBdel$ and work with the non-covariant versions of $\tx^\mu$ and $\sAt{a}$ as explained in footnote \ref{fn:BSnotBT}.

In summary, we start with the action \eqref{eq:fullLag} and pick out the contributions linear in $\sBdel$:
\begin{equation}
\Delta = -\frac{1}{\sqrt{-\gref}} \frac{\delta S_\text{wv}}{\delta \sBdel } \,.
\end{equation}
 With this interpretation we are now in a position to constrain $\Delta \geq 0$. Since we are interested in contributions to the action proportional to $\partial_\thb \As_\theta $ and  $\partial_\theta \As_\thb$, we conclude that $\sBdel$ will occur in covariant derivatives of the worldvolume metric $\SF{\gref}_{IJ}$, the thermal vector $\SF{\Kref}^I$, the 
 field strengths $\Fs_{JK}$, and their derivatives. We can further isolate contributions by examining the discrete ${\sf CPT}$ symmetry that acts as a $\mathbb{Z}_2$ anti-linear involution (recall that we implement ${\sf CPT}$ as $R$-parity on the worldvolume).

 Firstly, unlike charge currents, $\SF{\Nref}^I$ is even under ${\sf CPT}$, while $\Delta$ is ${\sf CPT}$ odd. Therefore,  the only terms that can contribute to $\Delta$ are those that transform with an additional pair of super-covariant derivatives in the Grassmann odd directions ($\Dut_\theta$ and $\Dut_\thb$), with an overall factor of $i$ to account for the anti-linearity of the ${\sf CPT}$ transformation.  Secondly, while a-priori these Grassmann-odd derivatives can appear inside various other derivatives, we can adapt a trick introduced in \cite{Haehl:2015pja} to simplify the analysis. One can capture extra derivatives acting on the fluid dynamical fields through the action of a differential operator valued tensor. Utilizing these two simple facts, we are lead to the observation that the only sets of terms that contribute non-trivially to $\Delta$ are those where we have a contribution from  $\gpsib_{IJ}=\Dut_\theta \SF{\gref}_{IJ}$ and 
 $\gpsi_{IJ}=\Dut_\thb \SF{\gref}_{KL}$. The free indices have to be contracted against a differential operator valued super-tensor $\SF{\etaref}^{IJKL}[\SF{\gref}_{IJ},\SF{\Kref}^I,\Dwv_I]$. That is, the action takes the form:
\begin{equation}
\Lag_\text{wv, diss} = \int d\theta\, d\thb\; \frac{\sqrt{-\SF{\gref}}}{\zsf}\,\left(-\frac{i}{4}\right) \, \SF{\etaref}^{IJKL} \, \gpsib_{IJ}\, \gpsi_{KL} \,.
\label{eq:Ldiss}
\end{equation}	
The dissipative tensor $\SF{\etaref}^{IJKL}$ is graded symmetric in both its first and second pair of indices. In addition it is required to also be graded symmetric under the exchange of the first and second pair of indices, viz.,  
\begin{equation}
\begin{split}
& \SF{\etaref}^{(IJ)(KL)} = (-)^{IJ}\; \SF{\etaref}^{(JI)(KL)} = (-)^{KL}\,\; \SF{\etaref}^{(IJ)(LK)}  \\
& \SF{\etaref}^{(IJ)(KL)}_{_\text{(D)}}  = (-)^{(I+J)(K+L)}\;  \SF{\etaref}^{(KL)(IJ)}_{_\text{(D)}} \,.\end{split}
\label{eq:etacpt1}
\end{equation} 
See also \S\ref{sec:classD} for more details.

With this ansatz for the dissipative terms it is easy to isolate  the contributions to $\Delta$. To this end we need the  leading ghost free contributions, $\gpsi_{IJ} = \theta\, \dbrk{\sBdel -\phizT, \gref_{IJ}} + \cdots $ and $\gpsib_{IJ} =- \thb\, \dbrk{\sBdel, \gref_{IJ}} +\cdots$, which leads us to:
\begin{equation}
\begin{split}
\Delta = -\frac{1}{\sqrt{-\gref}} \frac{\delta S_\text{wv}}{\delta \sBdel } 
&=
	\frac{1}{4}\, \etaref^{abcd}\, \lieD_\Kref \gref_{ab}\; \lieD_\Kref \gref_{cd}  + \text{fluctuations} + \text{ghost-bilinears} \,,
\end{split}
\label{eq:Delgg}
\end{equation}	
where we have employed the expectation value $\vev{\phizT} = -i $ owing to spontaneous {\sf CPT} breaking which signifies the origin of dissipation as argued for in \cite{Haehl:2015uoc}.
 
We are now in a position to infer that as long as $\etaref^{abcd}$ is a positive definite (derivative operator valued) map from the space of symmetric two-tensors to symmetric two-tensors, the amount of entropy produced is non-negative definite, viz., $\Delta \geq 0$. This is precisely the condition inferred in \cite{Haehl:2015pja}, which generalizes the theorem proven by Sayantani Bhattacharyya in \cite{Bhattacharyya:2013lha,Bhattacharyya:2014bha}. In the present context, this condition simply follows from demanding convergence of the path integral -- if $\etaref^{abcd}$ had negative eigenvalues, it would lead to divergent terms in the path integral. We elaborate further in \S\ref{sec:classD}.

We will now proceed to give explicit examples of various classes of transport for illustration. In \S\ref{sec:SigmaModel} we will return to a general discussion of all classes of hydrodynamic transport and how they are realized as specific forms of the Lagrangian \eqref{eq:fullLag}. This will then prove evidence for our central claim delineated  at the beginning of this section.

\section{The MMO limit: a truncation for explicit calculations}
\label{sec:mmo}

The previous discussion, whilst comprehensive, was quite abstract. To see how the fluid dynamical effective actions work, and to verify the statements made above, it is helpful to consider some explicit examples of terms that contribute to the action. This can, of course, be done explicitly, since we only have to construct suitable $\UT$ invariant actions build from the hydrodynamic fields. The one complicating factor  is  that in terms of explicit fields we have many components -- there are four fields in the target space maps $\SF{X}^\mu$ and the $\UT$ super-gauge field $\As_I$ has 12 component fields (3 of which are non-covariant potentials). The full set of variables, organized by ghost number, is given in Table~\ref{tab:multiplets}. In addition we turn on sources for various currents (see below).

In this section we motivate a truncated sector of the theory which suffices to capture the essential physics, and allows us to make contact with the familiar form of the hydrodynamic constitutive relations. The basic idea is to set all the fields carrying non-vanishing ghost number to zero, and furthermore effectively treat the $\UT$ symmetry as a global symmetry by ignoring the contributions from the gauge potentials. The rationale for the latter comes from our discussion in \S\ref{sec:caveat}, where we have argued for the MMO limit (after Mallick-Moshe-Orland \cite{Mallick:2010su}) where the gauge dynamics is frozen. This reduces the dynamical field content to three fields:
the hydrodynamic pions $X^\mu$, the fluctuation fields $\tx^\mu$, and the $\UT$ field strength $\Fs_{\theta\thb}| \equiv \phizT$, which will eventually be set to $\langle \phizT \rangle = -i$ according to \eqref{eq:Azero} and \eqref{eq:Fththbvev}. In addition we will allow for some sources which will prove useful for extracting currents: $\{\source{h_{ab}},\source{h_{\theta\thb}}\}$ as sources for the bosonic components of the energy-momentum tensor, $\sAt{a}$ for the entropy current, as well as a source $\sBdel$ for the total entropy production. We will indicate sources in a slightly different colour to distinguish them from the dynamical field content. We  illustrate how the hydrodynamic effective actions work with this truncation rather explicitly up to second order in derivatives.

\subsection{Field content}
\label{sec:mmofields}

We define the MMO limit to be the limit where all the ghost fields are set to zero a-priori. In this limit, we truncate to the middle row of Table~\ref{tab:multiplets} and furthermore ignore the gauge potential $\Ascr_a$. In this limit, the superfields simplify drastically and we can write them succinctly as:
\begin{equation}
\begin{split}
\As_a  &=\thb \theta\,   \sAt{a} \,, \qquad\qquad\qquad\;\;
\Fs_{ \theta\thb} = \phizT 
\,,
\qquad\qquad\;\;\, \Fs_{\thb a} = \theta \, \sAt{a} \,, \\
\As_{\thb} &= \theta \, \sBdel  \,, \qquad\qquad\qquad\quad
\Fs_{\thb \thb}  = 0  \,, \qquad \qquad\;\quad \Fs_{\theta a} = -\thb \, \prn{ \sAt{a} + \partial_a \phizT} \,, \\
\As_\theta &= \thb\, (\phizT - \sBdel) \,,\qquad\qquad
\Fs_{\theta \theta}  = 0\,, \qquad\qquad\;\quad 
\Fs_{ab} = \thb \theta  \, \prn{ \partial_a \sAt{b} - \partial_b \sAt{a} } \,,\\
\SF{X}^\mu &= X^\mu  + \thb \theta \,  {\tilde X}^\mu \,.
\end{split}
\label{eq:mmodef}
\end{equation}
The fields we have retained in this limit, should be understood by the following rationale: 
\begin{itemize}
 \item The gauge field component $\sAt{a}$ has been retained as a source for the Noether free energy current $\Nref^a$. The super free energy current is defined in  \eqref{eq:TNwvdef} and its bottom component is isolated by varying with respect to the top component of 
 $\As_a$, viz., $\sAt{a}$. Once we obtain the free energy current, we immediately obtain the entropy current via \eqref{eq:Ndef}.\footnote{ This statement relies on varying with respect to the gauge potential whilst keeping the pullback metric fixed which is natural from the standpoint of \eqref{eq:TNwvdef}. If we also vary with respect to the gauge potentials contributing to the pullback metric itself, see \eqref{eq:gwv1}, we will end up directly with the entropy current.} 
 \item The source $\sBdel$ which couples to the net entropy produced $\Delta$. The general idea for isolating entropy production has been described around \eqref{eq:Deltavar}. We will indeed verify in examples below that variation with respect to this field gives the total entropy production (including fluctuation terms).  
 \item The top component of the sigma-model map, $\tilde{X}^\mu$. This field describes fluid fluctuations and variation with respect to it will give the hydrodynamic equations of motion.
 \item In \eqref{eq:mmodef} we have also kept $\phizT$, which we will evetually set to $-i$. This field plays the role of an order parameter for dissipation.
\end{itemize}

From \eqref{eq:mmodef} we can also immediately get the following expressions for the derivatives of the position superfield:
\begin{equation}
\begin{split}
\Dut_a \SF{X}^\mu 
&=  
	\partial_a X^\mu+ \thb \theta\,
	 \prn{ 
  		\partial_a\tilde{X}^\mu  
		+ ( \sAt{a} , X^\mu)_\Kref }  \\
\Dut_{\thetab} \SF{X}^\mu 
& = 
	\theta\, \prn{ \tilde{X}^\mu + (\sBdel,X^\mu)_\Kref }
	\\
 \Dut_\theta \SF{X}^\mu &=
 	 \thb\,   \prn{(\phizT-\sBdel ,X^\mu)_\Kref -\tilde{X}^\mu  
	 }   
	  \\
 \Dut_{\theta}\Dut_{\thb} \SF{X}^\mu  
 &= 
	\tilde{X}^\mu + (\sBdel,X^\mu)_\Kref   
	+\thb \theta \big(\phizT-\sBdel,\tilde{X}^\mu + (\sBdel,X^\mu)_\Kref  \big)_\Kref 
	\\ 
\Dut_{\thetab}\Dut_{\theta } \SF{X}^\mu 
&=   
	 (\phizT-\sBdel,X^\mu)_\Kref-\tilde{X}^\mu  
	 -\thb\theta \, \big(\sBdel, (\phizT-\sBdel,X^\mu)_\Kref-\tilde{X}^\mu \big)_\Kref
\end{split}
\end{equation}
It is instructive to note that the contribution of $\sBdel$ can be easily inferred by the map $\tx^\mu \mapsto \tx^\mu + \dbrk{\sBdel, X^\mu}$ which follows from $\UT$ covariance, for reasons described in footnote \ref{fn:BSnotBT}. 

Let us now take stock of various other derived superfields which enter into the hydrodynamic effective action.

\paragraph{The world-volume metric:} The pullback computation \eqref{eq:gref} leads to the following expression for the worldvolume metric with the ghosts set to zero:
{\small
\begin{equation}
\begin{split}
\SF{\gref}_{ab} 
&= 
	g_{\mu\nu}  \, \partial_a X^\mu \partial_b X^\nu  + 
	\thb \theta
	\brk{
		\source{h_{ab}} +
		2\, g_{\mu\nu} \,\partial_{(a} X^\mu\left(\partial_{b)} \tilde{X}^\nu  
		+ (\sAt{b)} , X^\nu)_\Kref \right)
        + \tilde{X}^\rho \, \partial_\rho g_{\mu\nu} \; \partial_a X^\mu \, \partial_b X^\nu
	}	 
 \\
 & \equiv \gref_{ab} + \thb\theta \, \prn{ \source{h_{ab}} + \tilde{\gref}_{ab}   } \\
\SF{\gref}_{\thetab a}
& = 
	\theta\,  g_{\mu\nu} \, \prn{  \tilde{X}^\mu + (\sBdel,X^\mu)_\Kref } \, \partial_a X^\nu 
\\
  \SF{\gref}_{\theta a}  
  & = 
  	\thb\, g_{\mu\nu} \prn{(\phizT-\sBdel,X^\mu)_\Kref-\tilde{X}^\mu 
	} \partial_a X^\nu 
 \\
\SF{\gref}_{\theta \thetab} 
=  
 	- \SF{\gref}_{ \thetab\theta} 
	&=
 	 i  + \thb\theta\, 
 	 \brk{  
 	 \source{h_{\theta\thb}} + g_{\mu\nu} \prn{(\phizT-\sBdel,X^\mu)_\Kref-\tilde{X}^\mu
	 } 
 	 \prn{ \tilde{X}^\nu  + (\sBdel,X^\nu)_\Kref }
	  } \\
&\equiv 
 	 i + \thb\theta\, 
 	 \prn{  
 	 \source{h_{\theta\thb}} + \tilde{\gref}_{\theta\thb} }
\end{split}
\label{eq:gwv1}
\end{equation}
}\normalsize
where we used $g_{\Theta\bar{\Theta}} = - g_{\bar{\Theta}\Theta} = i$ in target space and we introduced the following shortcuts: 
\begin{equation}
\begin{split}
\gref_{ab} &= 
	g_{\mu\nu}\, \partial_a X^\mu \, \partial_b X^\nu \,, 
\\
\tilde{\gref}_{ab} &= 
	2\, g_{\mu\nu} \,\partial_{(a} X^\mu\left(\partial_{b)} \tilde{X}^\nu   
		+ (\source{\mathcal{F}_{b)}} , X^\nu)_\Kref \right) + \tilde{X}^\rho \, \partial_\rho g_{\mu\nu} \; \partial_a X^\mu \, \partial_b X^\nu\,,\\
\tilde{\gref}_{\theta\thb} &= 
	  g_{\mu\nu} \prn{(\phizT-\sBdel,X^\mu)_\Kref-\tilde{X}^\mu  
	 } \prn{\tilde{X}^\nu + (\sBdel,X^\nu)_\Kref } \\
&  = 
	g_{\mu\nu} \, \prn{ \mathfrak{J}^\mu - (\sBdel, X^\mu)_\Kref } \prn{ \tilde{X}^\nu + (\sBdel,X^\nu)_\Kref }\,, \\
\mathfrak{J}^\mu &\equiv 	
	(\phizT,X^\mu)_\Kref-\tilde{X}^\mu \,.
\end{split}
\label{eq:gJdefs}
\end{equation}	
Under a {\sf CPT} transformation $\tx^\mu$ gets exchanged with  $\mathfrak{J}^\mu$.  Furthermore, in \eqref{eq:gwv1} we introduced sources $\{\source{h_{ab}},\source{h_{\theta \thb}}\}$. These help taking variations of the actions more explicitly and also aid in departing from the topological limit. For instance, the stress tensor will be obtained by varying the action with respect to $\source{h_{ab}}$ (and then setting the sources to zero as usual).

\paragraph{The inverse metric:} One can work out the inverse of the metric in the MMO limit by the requirement that 
\begin{equation}
(-)^J\;\SF{\gref}^{IJ} \, \SF{\gref}_{JK} = \delta^I_{\ K} \,.
\label{eq:}
\end{equation}	
Note that the sign in this equation corresponds to the standard contraction convention for superspace indices that are contracted as $\searrow$, i.e.,  ``NW to SE'' as described in \cite{DeWitt:1992cy} and explained in \eqref{eq:deWitt1}. One then finds:
\begin{equation}
\begin{split}		
\SF{\gref}^{ab} &= 
	 \gref^{ab} - \thb \theta  \,  \gref^{ac} \,   \gref^{bd} 
	\prn{
		\source{h_{cd}} + \tilde{\gref}_{cd} - 
		\, 2\,i\, 
		 \prn{\mathfrak{J}^\mu - (\sBdel,X^\mu)_\Kref}
		  \, \prn{\tilde{X}^\rho + (\sBdel,X^\rho)_\Kref } \,  \partial_c X_\mu  \, \partial_d X_\rho 
	} \,, \\
\SF{\gref}^{a\thb} &= 
	i\, \thb\; \gref^{ab} \, \prn{ \mathfrak{J}_\mu - (\sBdel,X_\mu)_\Kref }  \, \partial_b X^\mu \,,
	\\
\SF{\gref}^{a\theta} &= 
	-i\,\theta\; \gref^{ab} \, \prn{ \tilde{X}_\mu + (\sBdel,X_\mu)_\Kref } \, \partial_b X^\mu\,,
	\\
\SF{\gref}^{\theta \thb} &= - \SF{\gref}^{\thb \theta} = 
	- i +   \thb \theta\prn{ 
		\source{h_{\theta\thb}}  + \tilde{\gref}_{\theta \thb}  - \prn{ \mathfrak{J}^\mu - (\sBdel,X^\mu)_\Kref }   \, \prn{ \tilde{X}_\mu + (\sBdel,X_\mu)_\Kref }  }\\
&=
	-i + \thb \theta\, \source{h_{\theta\thb}}\,.		
 \end{split}
\label{eq:imetSF}
\end{equation}

This completes our discussion of the basic field content. For convenience, we will now record some composite and derived fields in addition.

\paragraph{The super-determinant:}  We view the metric superfield and its inverse as super-matrices. The super-determinant is then given by the Berezinian of 
the metric, i.e., 
\begin{equation}
\begin{split}
\sqrt{-\SF{\gref}} 
&= 
	\sqrt{-\text{det}( \SF{\gref}_{ab} ) \, \text{det}\prn{\SF{\gref}_{z_1z_2}- 
		\SF{\gref}_{c z_1}\, \SF{\gref}^{cd}\, \SF{\gref}_{dz_2}}^{-1}} \\
&= 		
	\sqrt{-\det{\gref_{ab}}} \, \left[1 + \thb \theta 
	\, \prn{ \frac{1}{2} \, \gref^{cd}\, \prn{ \source{h_{cd}}+ \tilde{\gref}_{cd} } 
	 } \right]\,.
\label{eq:detSF}
\end{split}
\end{equation}	

\paragraph{The {\bf z} factor:} 
The $\UT$ covariant pullback measure in the action involves the field $\zsf$ introduced in \eqref{eq:zdef}, which evaluates to 
\begin{equation}
\SF{{\bf z}}  \equiv 1 + \SF{\Kref}^I  \, \As_I = 1 + \thb \theta\, \Kref^a  \,\sAt{a}  \,.
\label{eq:zSF}
\end{equation}	
One can explicitly check in the MMO limit the useful identity $\Dut_I \SF{{\bf z}} = \Fs_{IJ} \SF{\Kref}^J$.\footnote{ In order to check this, note that $\SF{{\bf z}}$ is $\UT$ 1-adjoint, which implies, e.g., $(\Fs_{IJ},\SF{{\bf z}})_\Kref = \Fs_{IJ} \, \lieD_\Kref \SF{{\bf z}} - \SF{{\bf z}} \, \lieD_\Kref \Fs_{IJ}$.}

\paragraph{The temperature superfield:} The local temperature superfield of the worldvolume description of the fluid is defined as
\begin{equation}
\begin{split}
\SF{T} &= 
	\frac{1}{\sqrt{-\SF{\gref}_{IJ} \, \SF{\Kref}^I \, \SF{\Kref}^J}} 
 = 
	\prn{-\gref_{ab} \Kref^a \Kref^b - \thb \theta\,  \Kref^a \Kref^b
	\prn{\source{h_{ab}} +\tilde{\gref}_{ab}}}^{-\frac{1}{2}} \\
& 	\equiv T + \thb\theta \;   \tilde{T}	\\
\text{where} \quad T &\equiv
	 \frac{1}{\sqrt{-\gref_{ab}\Kref^a \Kref^b}}	\,, \qquad 	 
	\tilde{T} \equiv \frac{1}{2}\, T^3\, \Kref^a \Kref^b
	\prn{\source{h_{ab}} +\tilde{\gref}_{ab}  } 
\end{split}
\label{eq:TempSF}
\end{equation}
Any function of $\SF{T}$ (such as the pressure superfield $p(\SF{T})$) can then be expanded similarly:
\begin{equation}
{\mathfrak{f}}(\SF{T})  
 = \mathfrak{f}(T) + \thb \theta \,\mathfrak{f}'(T)  \, \tilde{T} \,.
\label{eq:fnTSF}
\end{equation}	

\paragraph{The velocity superfield:} The fluid velocity is lifted to the following superfield:
\begin{equation}
\begin{split}
\SF{u}^I &\equiv 
	\SF{T}\, \SF{\Kref^I} \\
	\text{in components: } \quad
\SF{u}^a &= 
	\prn{T + \thb \theta \,   \tilde{T}	}\,\Kref^a \,, \qquad 
\SF{u}^\theta = 0 \,, \qquad \SF{u}^\thb = 0 	
\end{split}
\label{eq:uSF}
\end{equation}

\paragraph{The spatial projector:} The tensor that projects onto the space transverse to $\SF{u}^a$ has the following components in the MMO limit:
\begin{equation}
\begin{split}
\SF{P}_{IJ} &\equiv 
	\SF{\gref}_{IJ} + (-)^{K(1+L+I) + L(1+I+J)}\,  \SF{u}^K\, \SF{u}^L \, \SF{\gref}_{IK} \, \SF{\gref}_{JL}  \,.
\end{split}
\end{equation}
In components: 
\begin{equation}
\begin{split}
\SF{P}_{ab}  
&= 
	P_{ab} + \thb \theta \; P_a^c \, P_b^d \, (\source{h_{cd}} + \tilde{\gref}_{cd}) \\
\SF{P}_{a\theta}
&= 
	\thb \, P_{ab} \, e^b_\mu \, \prn{\mathfrak{J}^\mu - (\sBdel, X^\mu)_\Kref }  \\
\SF{P}_{a\thb} 
&= 
	\theta \, P_{ab} \, e^b_\mu \, \prn{ \tilde{X}^\mu + (\sBdel, X^\mu)_\Kref }   \\
\SF{P}_{\theta\thb}
&= 
	i + \thb\theta \, \prn{ \source{h_{\theta\thb}} + P_{cd} \, e_\mu^c \, e_\nu^d\,\prn{ \mathfrak{J}^\mu - (\sBdel, X^\mu)_\Kref }  \prn{ \tilde{X}^\nu + (\sBdel,X^\nu)_\Kref } }
\end{split}
\label{eq:projSFd}
\end{equation}
The spatial projector with upper indices, $\SF{P}^{IJ} =   
	\SF{\gref}^{IJ} + \SF{u}^I\, \SF{u}^J$, reads as follows:
{\small
\begin{equation}
\begin{split}
\SF{P}^{ab} 
 & = 	P^{ab}
	+  \thb \theta \,
	\Big\{ 2\,  \tilde{T} \, u^a\, \Kref^b - 
	  \brk{
		\source{h^{ab}} + \tilde{\gref}^{ab} -2 \,i\,
		 \prn{\mathfrak{J}_\mu  - (\sBdel, X_\mu)_\Kref }
		 \prn{\tilde{X}_\rho + (\sBdel, X_\rho)_\Kref } \, \partial^a X^\mu \, \partial^b X^\rho
	}\Big\}	\\
& \equiv P^{ab} + \thb \theta\,  \tilde{P}^{ab}  \\	
\SF{P}^{a\theta}  
	&=- i\, \theta\; \gref^{ab} \, \prn{\tilde{X}_\mu + (\sBdel, X_\mu)_\Kref } 
	\, \partial_b X^\mu 
	 	\\
\SF{P}^{a\thb}  
	&= 
	i\,\thb\; \gref^{ab} \, \prn{ \mathfrak{J}_\mu - (\sBdel, X_\mu)_\Kref } \, \partial_b X^\mu  
 \\
\SF{P}^{\theta\thb}  
&=
	 -i  +   \thb \theta\, \FH{\source{h_{\theta\thb}}}   
\end{split}
\label{eq:projSFu}
\end{equation}
}\normalsize 
where $P^{ab} \equiv \gref^{ab} + T^2\, \Kref^a \Kref^b$.
 
\subsection{The worldvolume connection}
\label{sec:wvconn}

We will later also need derivatives of superfields. To define these, we need to discuss the covariant derivative, or the connection, in worldvolume superspace. This is a non-trivial task because of the way in which the $\UT$ symmetry intertwines with diffeomorphisms. 

In particular, one can easily see that consistency of our formalism requires the connection in the worldvolume theory (though not the one in physical space!) to be not metric compatible. The two essential features that the worldvolume covariant derivative $\Dwv_I$ has to satisfy are the {\it commutator condition} 
\begin{equation}
\gradcomm{\Dwv_I}{\Dwv_J} \SF{X}^\mu  = ( \Fs_{IJ}, \SF{X}^\mu)_\Kref\,,
\end{equation}
and the requirement of {\it measure compatibility} (needed in order to be able to integrate by parts as usual):
\begin{equation}
\Dwv_J \left( \frac{\sqrt{-\SF{\gref}}}{\zsf} \right) = 0 \,.
\end{equation}
 A detailed analysis of these requirements is given in Appendix \ref{sec:wvconnect}. There we argue that these conditions fix the connection (up to some ambiguity, which we set to zero since it is irrelevant for the following discussion). The resulting connection is
\begin{equation}\label{eq:connFinal}
\begin{split}
\Cwv^{I}_{\ JK}
& = 
	\frac{1}{2} \, \SF{\gref}^{IL} \left[ (-)^{L(1+J)} \, 
	\Dut_J \SF{\gref}_{LK} + (-)^{L(1+J+K)+JK} \, 
	\Dut_K \SF{\gref}_{JL} - (-)^{L} \, \Dut_L \SF{\gref}_{JK} \right]  \\
&	\quad
+ \frac{1}{\zsf}\left( \Fs^I_{\ J} \,\SF{\Kref}_K + (-)^{JK}\, \Fs^I_{\ K}\, \SF{\Kref}_J \right)
\end{split}
\end{equation}	
In other words the connection we choose is the Christoffel connection built from $\Dut_I \equiv \partial_I + (\SF{\Ascr}_I, \,\cdot\,)_\Kref$ supplemented with thermal equivariance field strength data to ensure that we have correctly accounted for the $\UT$ covariance of the pullback. Note that the connection has no (super-)torsion. The fully covariant derivative then acts as follows on a generic super-tensor $\SF{\mathcal{T}}^J{}_K $:
\begin{equation}
  \Dwv_I \SF{\mathcal{T}}^J{}_K \equiv \Dut_I \SF{\mathcal{T}}^J{}_K + (-)^{IJ} \Cwv^J{}_{IL} \, \SF{\mathcal{T}}^L{}_K - (-)^{IJ}   \SF{\mathcal{T}}^J{}_L \, \Cwv^L{}_{IK} 
\end{equation}
and similarly for tensors with more indices. 
In Appendix \ref{sec:MMOfurther} we write the superspace expansion of the connection $\Cwv^{I}_{\ JK}$ explicitly in the MMO limit. 

This data is sufficient to go ahead with explicit computations (some more objects are computed in the MMO limit in Appendix \ref{sec:MMOfurther}). In what follows we will describe the hydrodynamic effective actions at the lowest three orders in the derivative expansion, demonstrating explicitly the familiar structure expected from earlier analyses in \cite{Haehl:2015pja,Haehl:2015uoc}.

\subsection{General observations}
\label{sec:mmogeneral}

Before exploring examples, let us make some general observations which are now clear from the MMO expansion of the fields.

\paragraph{Equations of motion:}

We can discern from the way that $\tilde{X}^\mu$ appears in the metric superfield \eqref{eq:gJdefs} that the equations of motion will be target space covariant. To see this, consider the way that the energy momentum tensor will generically couple to $\tilde{\gref}_{ab}$ in the action (after performing the integration over Grassmann odd directions):
\begin{equation}
  S_\text{wv} \supset \int \sqrt{-\gref} \, \left( \frac{1}{2} \, \TEMref^{ab} \, \tilde{\gref}_{ab} \right)  \supset \int \sqrt{-\gref} \;  \TEMref^{ab} \partial_{(a}X^\mu \, \left(  \partial_{b)} \tilde{X}^\nu \, g_{\mu\nu} + \frac{1}{2} \, \partial_{b)} X^\nu\, \tilde{X}^\rho \partial_\rho g_{\mu\nu}  \right)\,,
\end{equation}
where $\TEMref^{ab}$ contains both classical and fluctuation terms.
This coupling is generic, as can be seen from the fact that every occurance of the source $\source{h_{ab}}$ (which by definition sources $\TEMref^{ab}$) is accompanied by a $\tilde{\gref}_{ab}$. Push-forward to the physical spacetime and integration yields the following:
\begin{equation}
 \int \sqrt{-g} \, \tilde{X}_\mu \nabla_\nu T^{\mu\nu} \,,
\end{equation}
where $\nabla_\mu$ is the usual Christoffel connection built from the target space metric $g_{\mu\nu}$, and $T^{\mu\nu} = \TEMref^{ab}\, \partial_a X^\mu\, \partial_b X^\nu$. 

\paragraph{Superspace equation of motion:} 

The projected version of the full superspace equation of motion \eqref{eq:supereom} led us to the statement that the energy-momentum inflow contains only fluctuation terms, i.e., \eqref{eq:seom3f}. There is a nice version of this statement in the MMO limit. In the superspace action, the couplings relevant to \eqref{eq:seom3f} take the form:
\begin{equation}
\begin{split}
   S_\text{wv} &= \int d^d\sigma \, d\theta \, d\thb \, \frac{\sqrt{-\SF{\gref}}}{\zsf} \, \left\{ - \SF{\TEMref}^{a\theta} \,  \SF{\gref}_{a\theta} - \SF{\TEMref}^{a\thb} \,  \SF{\gref}_{a\thb} - \SF{\TEMref}^{\theta\thb} \,  \SF{\gref}_{\theta\thb} + \ldots \right\} \\
  &= \int d^d\sigma \sqrt{-\gref} \, \left\{ \partial_\theta \SF{\TEMref}^{a\theta} \, \partial_a X^\mu \, \mathfrak{J}_\mu  -  \partial_\thb \SF{\TEMref}^{a\thb} \, \partial_a X^\mu \, \tilde{X}_\mu  - {\TEMref}^{\theta\thb}  (\mathfrak{J}_\mu \tilde{X}^\mu) + \ldots + \text{ghost bilinears} \right\} \Big{|}
  \end{split}
\end{equation}
where we took the MMO limit in the second line and dropped all terms that are not relevant to the present discussion (such as sources). From this expression we can see what the contribution of the indicated currents to the equations of motion will be: variation with respect to $\tx^\mu$ precisely gives (minus) the combination \eqref{eq:seom3f}.  

Let us interpret this statement: when we vary the action with respect to $\tx^\mu$ while holding $\tilde{\gref}_{ab}$ fixed (where $\tx^\mu$ enters differentiated) then we end up with the above  combination entering the equations of motion. From our explicit analysis below, we will see that this combination will always lead to  fluctuation terms and not modify the dynamical equations of motion for the classical field $X^\mu$ as noted in \S\ref{sec:topsigma}.

\paragraph{Total entropy production:}

In computing $\Delta$ by variation of the action with respect to the source $\sBdel$, note that $\sBdel$ always appears only in three qualitatively different ways: 
\begin{enumerate}
\item[(i).] Every fluctuation field of the position occurs in the combination $\tilde{X}^\mu + (\sBdel,X^\mu)_\Kref$ (see footnote \ref{fn:BSnotBT}). It is clear that variation of the action with respect to these occurrences of $\sBdel$ gives a contribution proportional to equations of motion. 
\item[(ii).] Every metric source appears in the combination $\source{h_{ab}} + (\sBdel, \gref_{ab})_\Kref$; variation with respect to these $\sBdel$ gives a term of the form $\frac{1}{2} \, {\TEMref}^{ab} \, \lieD_\Kref  \gref_{ab}$. 
\item[(iii).] Finally, the source for the free energy current occurs in the combination $\sAt{c} - \partial_c \sBdel$, which leads to the contribution $\nabla_a \Nref^a$ in $\Delta$.
\end{enumerate}
We can thus anticipate on general grounds that the entropy production takes the expected form:
\begin{equation}
    -\frac{1}{\sqrt{-\gref}} \, \frac{\delta S_\text{wv}}{\delta \sBdel} = \underbrace{\nabla_a \Nref^a - \frac{1}{2} \, \TEMref^{ab} \, \lieD_\Kref \gref_{ab} }_{=\Delta}+ \text{equations of motion},
\end{equation}
where $\nabla_a$ is the standard Christoffel connection in the worldvolume after doing all superspace integrals and setting sources to zero.

In order for the above statements to be true, it is useful to observe the following. For any superspace Lagrangian $\SF{\Lagref}$ the effective action after integration over Grassmann directions generically takes the form
\begin{equation}\label{eq:trick1}
\begin{split}
  S_\text{wv} = \int d^d \sigma \sqrt{-\gref} \, \left[ \left( \frac{1}{2} \, \gref^{ab}(\source{h_{ab}}+ \tilde{\gref}_{ab}) - \sAt{a} \Kref^a \right) \Lagref + (\partial_\theta\partial_\thb \SF{\Lagref} ) \right]\,.
\end{split}
\end{equation}
This can always be covariantized in the sense described above by adding a total derivative term of the form
\begin{equation}\label{eq:trick2}
0 = \int d^d \sigma \,  \sqrt{-\gref} \,\nabla_c \left( \Kref^c\, \sBdel \, \Lagref \right) =\int d^d \sigma \, \sqrt{-\gref} \, \left[ \frac{1}{2} \, \gref^{ab} \, (\sBdel, \gref_{ab})_\Kref + \Kref^c \partial_c \sBdel \; \Lagref + (\sBdel, \Lagref)_\Kref \right]\,,
\end{equation}
such that adding \eqref{eq:trick1} and \eqref{eq:trick2} gives an action that only involves the combinations described above.

\section{The MMO limit: examples}
\label{sec:mmoEx}

Having defined the MMO limit, we now work out a few examples at low orders in the derivative expansion to illustrate the formalism.

\subsection{Zeroth order: Ideal fluid}
\label{sec:mmoideal}

The simplest example is the ideal fluid.
The effective superspace action has the form \eqref{eq:fullLag} with the Lagrangian being simply the pressure superfield. We parameterize the action by a scalar superfield ${\mathfrak{f}(\SF{T})}$ and explicitly evaluate it in the MMO limit to be 
\begin{equation}
\begin{split}
&\SF{\Lagref}^{^{\text{(ideal)}}} = \frac{\sqrt{-\SF{\gref}}}{\SF{{\bf z}}} \, \; \SF{\mathfrak{f}}(\SF{T}) \\
&\quad = 
	\sqrt{-\gref} \; \mathfrak{f}(T)\, 
	\prn{1 + \thb \theta \, 
		\brk{
			\frac{1}{2} \,  \gref^{cd}\, 	
			\prn{\source{h_{cd}}+ \tilde{g}_{cd} }   
			- \Kref^a \,\sAt{a} 
			+ \frac{\mathfrak{f}'(T)}{\mathfrak{f}(T)}\, \tilde{T}
		}}\,.
\end{split}
\label{eq:sfL0}
\end{equation}
$\mathfrak{f}(T)$ will turn out to be the standard pressure function, cf., \eqref{eq:pepsfT}.
Here we have kept all classical and fluctuation fields, but have discarded ghost contributions.
After the superspace integration we will end up with:
\begin{equation}
\begin{split}
\Lag_\text{wv}^{^{\text{(ideal)}}} &= \int d\theta d \thb \; \SF{\Lagref}^{^{\text{(ideal)}}} \equiv \partial_\theta \partial_\thb \; \SF{\Lagref}^{^{\text{(ideal)}}} \, \big{|}_{\theta=\thb=0} \\
&=  \sqrt{-\gref} \; \mathfrak{f}(T)\, \prn{\frac{1}{2} \,  \gref^{cd}\, 	
			\prn{\source{h_{cd}}+ \tilde{g}_{cd}  }   
			- \Kref^a \, \sAt{a}  
			+ \frac{\mathfrak{f}'(T)}{\mathfrak{f}(T)}\,  \tilde{T}  
}
\end{split}
\label{eq:wvL0}
\end{equation}
We could add a total derivative term to covariantize the expression as described in \S\ref{sec:mmogeneral}, but this would not affect the subsequent calculations.
From the coefficients of the sources $\source{h_{ab}}$ and $\sAt{a}$ we can easily read off the energy-momentum tensor and free energy current, respectively:
\begin{equation}
\begin{split}
{\TEMref}_{_{\text{(ideal)}}}^{ab} 
&\equiv 
	\frac{2}{\sqrt{-{\gref}}} \, \frac{\delta {\Lag}^{^{\text{(ideal)}}}_\text{wv}}{\delta \source{h_{ab}}}
	= \mathfrak{f}(T)\, \prn{  \gref^{ab} + \frac{\mathfrak{f}'(T)}{\mathfrak{f}(T)}\, T^3\, \Kref^a\, \Kref^b } \,,
\\ 
{\Nref}_{_{\text{(ideal)}}}^a
&\equiv
	- \frac{1}{\sqrt{-{\gref}}} \, \frac{\delta {\Lag}^{^{\text{(ideal)}}}_\text{wv}}{\delta {\sAt{a}}}
	=  \mathfrak{f}(T) \, \Kref^a 
	 \,.
\end{split}
\label{eq:TNwvdefideal}
\end{equation}	
The entropy current is related to the free energy current via \eqref{eq:Ndef} and can be obtained directly from the action by varying also the $\sAt{a}$ inside $\tilde{\gref}_{ab}$ (see \eqref{eq:gJdefs}). The tensor ${\TEMref}_{_{\text{(ideal)}}}^{ab}$ encodes both pressure $p(T)$ and energy density $\varepsilon(T)$, which are obtained from the function $\mathfrak{f}$. We identify 
\begin{equation}
p(T)=\mathfrak{f}(T) \,, \qquad \varepsilon(T)= T\,\mathfrak{f}'(T)-\mathfrak{f}(T) \,.
\label{eq:pepsfT}
\end{equation}	
Note that there is no source term $\source{h_{\theta\thb}}$ in the action, so we conclude ${\TEMref}_{_{\text{(ideal)}}}^{\theta\thb}=0$.

To obtain the equations of motion (i.e., conservation equations), we vary the action with respect to the fluctuation field $\tilde{X}^\mu$. The only occurrence of $\tilde{X}^\mu$ in \eqref{eq:wvL0} is inside $\tilde{\gref}_{cd}$. Carrying out this variation simply yields the expected conservation: 
\begin{equation}
 \nabla_a \, {\TEMref}_{_{\text{(ideal)}}}^{ab}  = 0 \,.
\end{equation}

Finally, there is no explicit dependence on $\sBdel$ in the Lagrangian \eqref{eq:wvL0}. Variation with respect to $\sBdel$ gives the free energy divergence, $\nabla_a \Nref_{_{\text{(ideal)}}}^a$. We therefore get for the total entropy production:
\begin{equation}
\Delta_{_{\text{(ideal)}}} \equiv
	- \frac{1}{\sqrt{-{\gref}}} \frac{\delta {\Lag}^{^{\text{(ideal)}}}_\text{wv}}{\delta \sBdel} = 0\,.
\end{equation}

\subsection{First order: viscosities and dissipation}
\label{sec:mmo1st}

Let us now turn to the next order in the gradient expansion. We already know from earlier analysis (see eg., \cite{Haehl:2015pja}) that there is no non-dissipative contribution at this order. So we should now focus on the dissipative terms, which gives us an opportunity to get some intuition for the general arguments described in \S\ref{sec:einflow}. As noted in \cite{Haehl:2015uoc} and explained around \eqref{eq:Ldiss} the structure of dissipative terms is encapsulated in super-tensors $\SF{\etaref}^{((IJ)|(KL))}_{_\text{(D)}}(\SF{\gref}_{IJ},\SF{\Kref}^I,\Dwv_I)$. Since this tensor will be contracted against terms that contain explicit super-derivatives of the metric, to first order in gradients it suffices to consider it to be built out of zero derivative objects. The natural structure then is simply combinations of the projection super-tensor $\SF{P}^{IJ}$ respecting the conditions \eqref{eq:etacpt1}. All other tensors can be argued to be equivalent once field redefinition ambiguities are accounted for (this is a superspace version of the arguments given in Section 5 of \cite{Haehl:2015pja}). 

These requirements then single out two possible independent tensor structures for first order dissipative transport. These are parameterized by scalar superfields $\SF{\zeta}(\SF{T})$ and $\SF{\eta}(\SF{T})$, respectively:
\begin{equation}
\begin{split}
\SF{\Lagref} &= 
	 - \frac{\sqrt{-\SF{\gref}}}{\SF{{\bf z}}}\, \frac{i}{4}\, (-)^{(I+J)(K+L)  + IJ+KL+K+L}\,\SF{\etaref}^{((IJ)|(KL))}_{_\text{(D)}}\, 
	   \gpsib_{IJ} \; \gpsi_{KL} 
 \\
 \SF{{\bm \eta}}^{IJKL}_{_\text{(D,visc.)}} 	 
 &= 
 	\SF{\zeta}(\SF{T})\, \SF{T}\, \SF{P}^{IJ} \, \SF{P}^{KL} + 2\, \SF{\eta}(\SF{T})\, \SF{T}\, (-)^{K(I+J)}\,\SF{P}^{K\langle I} \, \SF{P}^{J \rangle L}
\end{split}
\label{eq:Lageta}
\end{equation}
where the subscript ``visc.'' denotes the first order in gradients viscous transport contributions.
We have explicitly written the index contraction signs following from \eqref{eq:dwconven}. The transverse traceless index projection in the shear viscosity term will be defined below in \eqref{eq:shearten}.
The bottom components of the scalar superfields $\zeta(T)$ and $\eta(T)$ will end up being the usual bulk and shear viscosity transport coefficients.

We can expand the superspace structure \eqref{eq:Lageta} more explicitly to guide the eye. Using the fact that in the absence of sources carrying non-zero fermion number we can set $\gpsib_{\theta a} \to 0$ and $\gpsi_{\thb a} \to 0$, we find after evaluating the signs and symmetry factors:
\begin{equation}
\begin{split}
\Lag_\text{wv}^{(\text{D})} 
&= 
	 - \frac{i}{4}\,  \partial_\theta \partial_\thb
 	 \left(
	  \frac{\sqrt{-\SF{\gref}}}{\zsf} \, (-)^{(I+J)(K+L)  + IJ+KL+K+L}\,\SF{\etaref}^{((IJ)|(KL))}_{_\text{(D)}}\, 
	   \gpsib_{IJ} \; \gpsi_{KL} 
	 \right) \bigg| \\
&=	-\frac{i}{4}\, \sqrt{-{\gref}}\,\bigg\{
	4
 	 \left( \partial_\theta \partial_\thb\SF{\etaref}^{((\thb c)|(\theta d))}_{_\text{(D)}} + \left[ \frac{1}{2} \, \gref^{ab}(\source{h_{ab}}+ \tilde{\gref}_{ab}) - \sAt{a} \Kref^a \right]\,\SF{\etaref}^{(\thb c)(\theta d)}_{_\text{(D)}}  \right)
	   \gpsib_{\thb c} \; \gpsi_{\theta d} \\
&\qquad	+ 
	 2\,\partial_\theta 
 	 \SF{\etaref}^{((ab)|(\theta d))}_{_\text{(D)}}\,
	 \partial_\thb  \gpsib_{ab} \; \gpsi_{\theta d} 
	 -
	  4\,\partial_\theta 
 	  \SF{\etaref}^{((\theta\thb)|(\theta d))}_{_\text{(D)}}\,
	 \partial_\thb  \gpsib_{\theta\thb} \; \gpsi_{\theta d} 
        \\
&\qquad
	+2\,
	 \partial_\thb 
 	   \SF{\etaref}^{((\thb b)|(cd))}_{_\text{(D)}}\,
	\gpsib_{\thb b} \; \partial_\theta  \gpsi_{cd} 
	-4\,
	 \partial_\thb 
 	 \SF{\etaref}^{((\thb b)|(\theta\thb))}_{_\text{(D)}}\,
	\gpsib_{\thb b} \; \partial_\theta  \gpsi_{\theta\thb} \\	
&\qquad
	+ \Big[
	-2 \, \SF{\etaref}^{((\theta\thb)|(cd))}_{_\text{(D)}} \left( \partial_\thb \gpsib_{\theta\thb} \; \partial_\theta \gpsi_{cd} + \partial_\theta \gpsi_{\theta\thb} \; \partial_\thb \gpsib_{cd} \right) \\
&\qquad\qquad\qquad +4 \,\SF{\etaref}^{((\thb b)|(\theta d))}_{_\text{(D)}} 
	\left( \partial_\theta \partial_\thb  \gpsib_{\thb b} \; \gpsi_{\theta d}  +   \gpsib_{\thb b} \; \partial_\theta \partial_\thb\gpsi_{\theta d}  \right)
	\\
&\qquad\qquad\qquad
	+ \SF{\etaref}^{((ab)|(cd))}_{_\text{(D)}} \,\partial_\thb  \gpsib_{ab} \; \partial_\theta\gpsi_{cd} 
        +4\, \SF{\etaref}^{(\theta\thb)(\theta\thb)}_{_\text{(D)}} \,\partial_\thb  \gpsib_{\theta\thb} \; \partial_\theta\gpsi_{\theta\thb} 
	\Big]
	\bigg\}\bigg| + \text{ghost bilinears.}
\end{split}
\label{eq:Dcalc2}
\end{equation}	
More explicitly, one can plug in the following expressions:
\begin{equation}
\begin{split}
    \gpsib_{\thb b} | &= \left( \tilde{X}^\mu + (\sBdel, X^\mu)_\Kref \right) \partial_b X^\mu \\
    \gpsi_{\theta b} | &= \left( \mathfrak{J}^\mu - (\sBdel, X^\mu)_\Kref \right) \partial_b X^\mu\\
   \partial_\theta \partial_\thb  \gpsib_{\thb b} |&= \left( \phizT-\sBdel , (\tilde{X}_\mu + (\sBdel, X_\mu)_\Kref) \partial_b X^\mu \right)_\Kref\\
   \partial_\theta \partial_\thb  \gpsi_{\theta b} |&=- \left( \sBdel , (\mathfrak{J}_\mu - (\sBdel, X_\mu)_\Kref) \partial_b X^\mu \right)_\Kref \\
    \partial_\thb  \gpsib_{ab} |& = -  \source{h_{ab}} - \tilde{\gref}_{ab} + (\phizT-\sBdel, \gref_{ab})_\Kref  \\
     \partial_\theta  \gpsi_{ab}| & =  \source{h_{ab}} + \tilde{\gref}_{ab} + (\sBdel, \gref_{ab})_\Kref \\
   \partial_\theta\gpsi_{\theta\thb} | &= - \partial_\thb  \gpsib_{\theta\thb}| = \source{h_{\theta\thb}} + \tilde{\gref}_{\theta\thb}\,.
 \end{split}
 \label{eq:gpsisim}
\end{equation}
We will now compute \eqref{eq:Dcalc2} for bulk and shear viscosity separately.

\paragraph{Bulk viscosity term:} By explicit expansion of \eqref{eq:Dcalc2} using the MMO truncation of the superfields as defined in \S\ref{sec:mmofields} and \eqref{eq:gpsisim}, we find:
{
\begin{equation}
\begin{split}	
{\Lag}_{\text{wv}}^{^{\text{(visc.},\zeta)}} &=\int d\theta d \thb \;  \SF{\Lagref}^{^{\text{(visc.},\zeta)}}   \\
&=
	 - \sqrt{-\gref} \, \frac{i}{4} \; T\, \zeta(T)  
	 \bigg[ -i \, P^{ab} \prn{ \source{h_{ab}} + \tilde{\gref}_{ab} -
 	  (\phizT -\sBdel, \gref_{ab})_\Kref }+2 \source{h_{\theta\thb}}  \bigg] \\
& \qquad\qquad\qquad\qquad  \times \bigg[ -i \, P^{ab} \prn{ \source{h_{ab}} + \tilde{\gref}_{ab} +
 	  ( \sBdel, \gref_{ab})_\Kref }  +2 \source{h_{\theta\thb}}  \bigg] 
\end{split}
\label{eq:bvterm3}
\end{equation}
}\normalsize
This expression captures the full contribution of the bulk viscosity in the effective action, including its associated dissipation and  fluctuations. The latter are buried inside the 
$\tilde{\gref}_{ab}$, and we also see for the first time a contribution involving the metric source in the Grassmann-odd directions.

It remains to convince ourselves that the Lagrangian encodes viscous hydrodynamics in the correct way. The constitutive relations are 
\begin{equation}
\begin{split}
{\TEMref}_{_{\text{(visc.},\zeta)}}^{ab} 
&\equiv 
	\frac{2}{\sqrt{-{\gref}}} \, \frac{\delta {\Lag}^{^{\text{(visc.},\zeta)}}_\text{wv}}{\delta \source{h_{ab}}} 
	= \underbrace{ -i\phizT \, \zeta(T) \, \vartheta \, P^{ab}  }_{\text{classical}} \,
	  	+ \underbrace{\,i\,T \, \zeta(T) \,P^{ab} \, P^{cd} \, \tilde{\gref}_{cd} }_{\text{fluctuations}}
\\ 
{\Nref}_{_{\text{(visc.},\zeta)}}^a
&\equiv
	- \frac{1}{\sqrt{-{\gref}}} \, \frac{\delta {\Lag}^{^{\text{(visc.},\zeta)}}_\text{wv}}{\delta \sAt{a}}
 =  0  
 \,, 
\end{split}
\label{eq:TNwvdefideal}
\end{equation}	
where we set the sources to zero after variation, as usual. We also introduced the fluid expansion $\vartheta = \nabla_a u^a$. 
The classical part of the currents is as expected for bulk viscosity, after we set $\langle \phizT \rangle = -i$. We can then write the energy-momentum tensor including the fluctuation contribution somewhat suggestively as:
\begin{equation}
   {\TEMref}_{_{\text{(visc.},\zeta)}}^{ab}  = - \frac{1}{2} \, T \, \zeta(T) \, P^{ab} \, P^{cd} \, \prn{\lieD_\Kref\gref_{cd}  - 2i \, \tilde{\gref}_{cd} } \,.
\end{equation}
The conservation equations for ${\TEMref}_{_{\text{(visc.},\zeta)}}^{ab} $ follow as before by varying the action associated with \eqref{eq:bvterm3} with respect to $\tilde{X}^\mu$. One can verify readily that this will indeed give the expected result. We also recall that there are no contributions to the entropy current (and likewise for the free energy current) at first order.

Furthermore, as noted above, there is also a contribution to the energy-momentum tensor in the mixed Grassmann directions, 
\begin{equation}
{\TEMref}_{_{\text{(visc.},\zeta)}}^{\theta\thb} 
\equiv 
	\frac{2}{\sqrt{-{\gref}}} \, \frac{\delta {\Lag}^{^{\text{(visc.},\zeta)}}_\text{wv}}{\delta {h}_{\theta\thb}} 
	= \, \underbrace{ 2\, \phizT\, \zeta(T) \, \vartheta }_{\text{classical}} \, 
	- \underbrace{ 2\, T \, \zeta(T) \, P^{ab}\,\tilde{\gref}_{ab} }_{\text{fluctuations}}
\end{equation}

The full entropy production is also encoded in the superspace action. Namely, variation with respect to the source $\sBdel$ gives
\begin{equation}
\begin{split}
 \Delta_{_{\text{(visc.,}\zeta)}} &\equiv
	 -\frac{1}{\sqrt{-{\gref}}} \frac{\delta {\Lag}^{^{\text{(visc.,}\zeta)}}_\text{wv}}{\delta \sBdel}   
	  	 = \partial_a \Nref_{_{\text{(visc.,}\zeta)}}^a - \frac{1}{2} \, {\TEMref}_{_{\text{(visc.},\zeta)}}^{ab} \, \lieD_\Kref \gref_{ab} 
	\\
&=  \underbrace{(i\phizT) \,\frac{\zeta}{T} \, \vartheta^2}_{\text{classical}} - \underbrace{\frac{i}{2} \, T \, \zeta\, P^{ab} P^{cd} \, \lieD_\Kref \gref_{ab} \, \tilde{\gref}_{cd}}_{\text{fluctuations}} 
 \end{split}
\end{equation}
which can be verified by explicit calculation. Let us now set $\langle \phizT \rangle = -i$. Then the classical part of entropy production is precisely what is expected for dissipative fluids on general grounds \cite{Haehl:2015pja}. The second law requires $\Delta_{_{\text{(visc.,}\zeta)}} \geq 0$, which holds for arbitrary fluid configurations if and only if $\zeta \geq 0$.

A  first-principles way to argue for the positivity of $\zeta$ is by extracting the imaginary part of the effective action as a functional of the physical source (after turning off the fluctuation terms and  setting $\langle \phizT \rangle = -i$): 
\begin{equation} \label{eq:ImZeta}
 \text{Im}\big({\Lag}_{\text{wv}}^{^{\text{(visc.},\zeta)}}\big) = \sqrt{-\gref} \, \frac{1}{4} \, T \, \zeta(T) \, (P^{ab} \,\source{h_{ab}})^2  \,.
\end{equation}
In the path integral, the effective action occurs in an exponential $e^{i S_\text{wv}}$. Thus, for convergence we require Im$(S_\text{wv})\geq 0$. In the present case, this clearly translates to $\zeta \geq 0$.

\paragraph{The shear viscosity term:} We now wish to simplify the shear viscosity contribution to \eqref{eq:Dcalc2}.  To explain the constellation of indices there, let us write the shear part of the viscosity tensor more explicitly:
\begin{equation}
(-)^{K(I+J)}\,\SF{P}^{K\langle I} \, \SF{P}^{J \rangle L}
 = 
		\frac{1}{2} \prn{ (-)^{JK}
		\SF{P}^{IK}\, \SF{P}^{JL} + (-)^{I(J+K)} \SF{P}^{JK}\, \SF{P}^{IL} } - \frac{1}{d-1} \, \SF{P}^{IJ}\, \SF{P}^{KL} \,.
\label{eq:shearten}
\end{equation}	
By expanding \eqref{eq:Dcalc2} we get\footnote{ As a consistency check, note that all terms in this expression are either {\sf CPT} odd, or occur in {\sf CPT} conjugate pairs with a relative sign. Furthermore,  note that the expression is manifestly $\UT$ invariant since every difference field (dressed with a tilde) occurs with the correct covariant dressing involving $\sBdel$ (see also below). We have added the total derivative term described in \S\ref{sec:mmogeneral}.}
\small{
\begin{equation}
\begin{split}
&{\Lag}_{\text{wv}}^{^{\text{(visc.},\eta)}}  
=  -\sqrt{-\gref}\,\frac{i}{2}\, T\, \eta \, \bigg\{
	- \frac{1}{2} \prn{P^{ac}\, P^{b d}  + P^{bc}\, P^{ad} - \frac{2}{d-1} \, P^{ab} \, P^{cd} } 
	\prn{ \source{h_{ab}} + \tilde{\gref}_{ab} - (\phizT - \sBdel, \gref_{ab})_\Kref}\\
	&\qquad\qquad \qquad\qquad\qquad \times 
	\prn{ \source{h_{cd}} + \tilde{\gref}_{cd} + (\sBdel, \gref_{cd})_\Kref }  
	+   \frac{2\,i}{d-1}  \, P^{ab} \prn{ 2\,\source{h_{ab}} + 2\,\tilde{\gref}_{ab} - (\phizT-2 \sBdel , \gref_{ab})_\Kref  } \source{h_{\theta\thb}}\\	
& \qquad\;\; 
	 	 + 2i \, \brk{ 
		P^{c(a}  \, e_\nu^{b)}  \prn{\mathfrak{J}^{\nu} - (\sBdel,X^{\nu})_\Kref} \prn{\tilde{X}_\mu+(\sBdel,X_{\mu)})_\Kref } \partial_c X^\mu} 
		\prn{\source{h_{ab}} + \tilde{\gref}_{ab} + (\sBdel , \gref_{ab})_\Kref} \\
& \qquad\;\; 
	+ 2i \, \brk{ 
		P^{c(a}  \, e_\nu^{b)}  \prn{\tilde{X}^{\nu} + (\sBdel,X^{\nu})_\Kref } \prn{\mathfrak{J}_\mu  -(\sBdel,X_{\mu})_\Kref}\partial_c X^\mu} 
		\prn{\source{h_{ab}} + \tilde{\gref}_{ab} + (\phizT- \sBdel, \gref_{ab})_\Kref}   	  
\\ & \qquad\;\;
	+ 2i\, P^{ac}\, \partial_a X^\mu \, \partial_c X^\nu \, \left[ \left(\phizT-\sBdel, \tilde{X}_\mu +(\sBdel, X_\mu)_\Kref \right)_\Kref \, \left(\mathfrak{J}_\mu -(\sBdel, X_\mu)_\Kref\right) \right]
\\ & \qquad\;\;
	- 2i\, P^{ac}\, \partial_a X^\mu \, \partial_c X^\nu \, \left[ \left(\sBdel, \mathfrak{J}_\mu -(\sBdel, X_\mu)_\Kref \right)_\Kref \, \left(\tilde{X}_\mu +(\sBdel, X_\mu)_\Kref\right) \right]\\
&\qquad\;\; - 2 \brk{ \frac{d+1}{d-1}\,  \source{h_{\theta\thb}} + P^{\mu\nu} \prn{ \tilde{X}_\mu + (\sBdel, X_\mu)_\Kref } \prn{ \mathfrak{J}_\nu - (\sBdel, X_\mu)_\Kref } } \source{h_{\theta\thb}} 
\bigg\} \\
 &\quad\;\;  +\sqrt{-\gref} \, \bigg\{ T \, \eta(T) \left( \frac{1}{2} \, \gref^{cd} \prn{\source{h_{cd}} + \tilde{\gref}_{cd}+(\sBdel, \gref_{cd})_\Kref} - \Kref^c (\sAt{c} - \partial_c \sBdel)
\right) P^{ab} \\
&\qquad\quad +  (T\, \eta)' \, \prn{\tilde{T} + (\sBdel, T)_\Kref } P^{ab}
+ T\, \eta \, \brk{ 2\prn{\tilde{T} + (\sBdel, T)_\Kref } \, T \, \Kref^a \Kref^b - \source{h^{ab}} - \tilde{\gref}^{ab} - (\sBdel, \gref^{ab})_\Kref  }\bigg\}  \\
&\qquad\quad  \times \partial_a X^\mu \, \partial_b X^\nu \,   \prn{ \tilde{X}_\mu + (\sBdel, X_\mu)_\Kref } \prn{ \mathfrak{J}_\nu - (\sBdel, X_\mu)_\Kref }
\end{split}
\label{eq:L1etaC}
\end{equation}	
}\normalsize
The first line describes (among other things) the classical part and will give rise to the standard energy-momentum tensor for shear viscosity and its associated dissipation. All the remaining terms describe fluctuations.
If desired, \eqref{eq:L1etaC} can be simplified a little further by using the identity
\begin{equation}
P^{\mu\nu} \mathfrak{J}_\mu =  P^{\mu\nu}\, \prn{-i \,e_\mu^a \, \Kref_a -  \tilde{X}_\mu} = - P^{\mu\nu} \tilde{X}_\mu \,,
\label{eq:projJar}
\end{equation}	
exploiting  the fact that the spatial projector is transverse to the velocity (and hence the thermal vector).

We can now read off the following currents:
\begin{equation}
\begin{split}
{\TEMref}_{_{\text{(visc.},\eta)}}^{ab} 
&\equiv 
	\frac{2}{\sqrt{-{\gref}}} \, \frac{\delta {\Lag}^{^{\text{(visc.},\eta)}}_\text{wv}}{\delta \source{h_{ab}}} 
	\\
&= \underbrace{ -2\,i\phizT \, \eta(T) \, \sigma^{ab}  }_{\text{classical}} + \underbrace{ 2i \, T \, \eta \, P^{a\langle c} \, P^{d\rangle b} \, \tilde{g}_{cd} }_{\text{fluctuations}} \\
&\qquad\qquad\qquad\qquad - \underbrace{ \left[T\eta \left(2P^{(ac}P^{b)d} + P^{ab}P^{cd} \right) + T^2 \eta' \, u^au^b P^{cd} \right] \partial_{c}X^\mu \partial_{d}X^\nu \, \tilde{X}^{(\mu} \, \tilde{X}^{\nu)}   }_{\text{fluctuations}}\\
{\TEMref}_{_{\text{(visc.},\eta)}}^{\theta\thb} 
&\equiv 
	\frac{2}{\sqrt{-{\gref}}} \, \frac{\delta {\Lag}^{^{\text{(visc.},\eta)}}_\text{wv}}{\delta {h}_{\theta\thb}} 
	= \, \underbrace{ - \frac{4\,\eta(T)}{d-1} \, \phizT \, (\partial_c u^c)  }_{\text{classical}} \, + \, \underbrace{ 2\,T\, \eta(T) \, \left(  \frac{2}{d-1} 
	 \, P^{ab}\, \tilde{\gref}_{ab}  + 
	i \, P^{\mu\nu} \,\tilde{X}_\mu \tilde{X}_\nu\right)   }_{\text{fluctuations}}\\
 {\Nref}_{_{\text{(visc.},\eta)}}^a
&\equiv
	- \frac{1}{\sqrt{-{\gref}}} \, \frac{\delta {\Lag}^{^{\text{(visc.},\eta)}}_\text{wv}}{\delta \sAt{a}}
 =  \underbrace{  -T \, \eta(T) \big( P^{\mu\nu} \, \tilde{X}_\mu \, \tilde{X}_\nu \big) \Kref^a  }_{\text{fluctuations}}
\end{split}
\label{eq:TNwvdefeta}
\end{equation}	
where we have used \eqref{eq:projJar} to replace $\mathfrak{J}^\nu\mapsto -\tilde{X}^\nu$.
The classical parts give the dissipative constitutive relations for standard hydrodynamics, with $\sigma^{ab}$ being the shear tensor (the symmetrized, trace free derivative of the velocity field). The fluctuation terms proportional to $\tilde{\gref}_{ab}$ are precisely as anticipated in \cite{Haehl:2015uoc} and can be written in terms of a Hubbard-Stratonovich noise field as explained there.

Finally, the entropy production is 
\begin{equation}
\begin{split}
 \Delta_{_{\text{(visc.,}\eta)}} &\equiv
	 -\frac{1}{\sqrt{-{\gref}}} \frac{\delta {\Lag}^{^{\text{(visc.,}\eta)}}_\text{wv}}{\delta \sBdel}  \\
& = \partial_a \Nref_{_{\text{(visc.,}\eta)}}^a - \frac{1}{2} \, {\TEMref}_{_{\text{(visc.},\eta)}}^{ab} \, \lieD_\Kref  \gref_{ab} + \text{e.o.m.}\\
  &= \eta \, \sigma^{ab} \left( \lieD_\Kref \gref_{ab} - 2i \, \tilde{\gref}_{ab} \right) + 2 \eta \, P^{c(a} P^{b)d} \, \nabla_a u_b \, \tilde{X}_c \tilde{X}_d - \eta \, (u.\nabla) \prn{ P^{ab} \tilde{X}_a \tilde{X}_b } \,,
 \end{split}
\end{equation}
where we set $\langle \phizT \rangle =-i$ as before. The classical part of the entropy production is $\eta\, \sigma^{ab} \lieD_\Kref \gref_{ab} =  \frac{2}{T} \, \eta \, \sigma^2$, which manifestly satisfies the second law as long as $\eta \geq 0$. On top of this, our derivation predicts the fluctuation terms in $\Delta_{_{\text{(visc.,}\eta)}}$. In order to establish positivity of $\eta$, we can again compute the imaginary part of the effective action as a functional of the physical source and demand its non-negativity for the path integral to be well-defined. We find:
\begin{equation}
   \text{Im} \left({\Lag}_{\text{wv}}^{^{\text{(visc.},\eta)}}   \right) 
   = \sqrt{-\gref} \,  \frac{1}{2} \, T \, \eta \,\left[ \frac{1}{2} \prn{P^{ac}\, P^{b d}  + P^{bc}\, P^{ad} - \frac{2}{d-1} \, P^{ab} \, P^{cd} }  \source{h_{cd}} \right]^2 \,,
\end{equation}
which is positive semi-definite provided that $\eta \geq 0$. 
 
\subsection{Second order: some examples}
\label{sec:mmo2nd}

Second order hydrodynamics is rich and we will not record all possible terms here. They have been classified in \cite{Haehl:2015pja} and have been obtained for conformal fluids from a superspace effective action in \cite{Haehl:2015uoc}. They are, of course, also encompassed by the general discussion of all order hydrodynamics in \S\ref{sec:SigmaModel}. Here we discuss some examples in the MMO limit to illustrate the formalism. The dedicated reader will see from these examples that the formalism extends to all derivative orders (i.e., Class L of \cite{Haehl:2015pja} is reproduced entirely from our superspace formalism). 

To write an action for second order fluids, we will need derivatives of the velocity superfield.
First we have to decide how to take derivatives since we do not have a metric compatible connection. We will conventionally treat $\SF{\Kref}^I$ as a vector and correspondingly $\SF{u}^I$ is also defined with an upstairs index. We differentiate these objects and then lower indices whenever necessary. The reader can convince themselves that a different convention would not change any of the currents that we compute (in the MMO limit).

\paragraph{Derivatives of $\SF{\Kref}^I$ and $\SF{u}^I$:} The simplest set of derivatives to compute are those of the thermal vector $\SF{\Kref}^I$ on the worldvolume. Given our gauge fixing condition, we simply find
\begin{equation}
\Dut_I \SF{\Kref}^J = \delta_I^b \, \delta^J_a \; \partial_b \Kref^a \,.
\label{eq:KrefD}
\end{equation}
The velocity superfield is given in \eqref{eq:uSF} and we can take derivatives directly to obtain the acceleration superfield $\SF{\acc}^I 
\equiv 
	(-)^J\, \SF{u}^J \Dwv_J \SF{u}^I$:
{\small
\begin{equation}
\begin{split}
\SF{\acc}^a 
& = 
	\acc^a + \thb \theta \bigg\{ 2\,\frac{\tilde{T}}{T} \, \acc^a + u^b \partial_b \left(\frac{\tilde{T}}{T}\right) \, u^a +  \frac{(\Kref.\partial) T}{T} \,(u. \source{\mathcal{F}})\, u^a + \Cwvt^a{}_{bc}u^bu^c \bigg\} \\
\SF{\acc}^\theta  
&= -\frac{i}{2} \, \theta \, \bigg\{ - 2 \, \nabla_{(a}  \brk{ \partial_{b)}X^\mu \prn{\tilde{X}_{\mu} + (\sBdel, X_\mu)_\Kref} } + \source{h_{ab}} + \tilde{\gref}_{ab} + (\sBdel, \gref_{ab})_\Kref -4 \, \sAt{a} \Kref_b\bigg\} u^a u^b\\ 
\SF{\acc}^\thb &= \frac{i}{2} \, \thb \, \bigg\{ - 2 \, \nabla_{(a}\brk{  \partial_{b)}X^\mu \prn{\mathfrak{J}_{\mu} - (\sBdel, X_\mu)_\Kref}  }- \source{h_{ab}} - \tilde{\gref}_{ab} + (\phizT-\sBdel, \gref_{ab})_\Kref + 4 (\sAt{a} + \partial_a \phizT) \Kref_b \bigg\} u^a u^b
\end{split}
\label{eq:accSF}
\end{equation}
}\normalsize
where $\acc^a \equiv u^b \partial_b u^a$ as usual. 
Likewise the expansion follows from:
\begin{equation}
\SF{\vartheta} \equiv \Dwv_I \SF{u}^I  
= \vartheta + \thb \theta \brk{ \frac{\tilde{T}}{T} \, \vartheta + u^b \partial_b \left( \frac{\tilde{T}}{T} \right) + \frac{(\Kref.\partial)T}{T} (u.\source{\mathcal{F}}) + (-)^I\, \Cwvt^I{}_{Ic}\; u^c }
\label{eq:expSF}
\end{equation}	
where we defined the ordinary space expansion as $\vartheta \equiv \partial_a u^a $ and the trace $ (-)^I\, \Cwvt^I{}_{Ic}$ can be found in \eqref{eq:Ctrace}.

From the exhaustive classification of hydrodynamic transport in \cite{Haehl:2015pja} we know that the Lagrangian Class L of transport (i.e., adiabatic transport which can already be captured using a single-copy Lagrangian) is quite rich. At second order in derivatives, neutral parity-even fluids exhibit 5 types of transport that falls in Class L. We will not discuss all of these, but give some examples for illustration. 

\subsubsection{Acceleration and expansion squared}
\label{sec:acc}

We start with a Lagrangian consisting of the square of acceleration, $\SF{\acc}^2$, and the square of expansion, $\SF{\vartheta}^2$ (in the classification of \cite{Haehl:2015pja}, these belong to the hydrostatic Class $\PS$ and non-hydrostatic Class $\LS$, respectively): 
\begin{equation}
\begin{split}	
{\Lag}_{\text{wv}}^{^{(\acc,\vartheta)}} \equiv {\Lag}_{\text{wv}}^{^{(\acc)}} + {\Lag}_{\text{wv}}^{^{(\vartheta)}} &=\int d\theta d \thb \;  \frac{\sqrt{-\SF{\gref}}}{\SF{\bf z}}\;  \brk{
	 \SF{K}_\acc(\SF{T})\; \SF{\gref}_{IJ} \, \SF{\acc}^J \, \SF{\acc}^I  
	 +  \SF{K}_\vartheta(T)\, \SF{\vartheta}^2 }
\end{split}
\label{eq:bvterm3}
\end{equation}
Explicit evaluation of the superspace integral is straightforward and yields
{\small
\begin{equation}
\begin{split}	
{\Lag}_{\text{wv}}^{^{(\acc)}} &=  \sqrt{-{\gref}}\;  \bigg\{ \brk{ \prn{ 2K_\acc+\frac{T}{2}\,\K_\acc'} \acc^2\, u^a u^b + \frac{1}{2} \, K_\acc \, \acc^2 \, g^{ab} + K_\acc \, \acc^a \acc^b } (\source{h_{ab}}+ \tilde{\gref}_{ab})- K_\acc \, \acc^2 \, (\Kref. \source{\mathcal{F}})\\
&+\frac{i}{2} K_\acc\, \brk{ - 2 \,u^a u^b\, \nabla_{(a}\brk{ \partial_{b)}X^\mu \prn{\tilde{X}_{\mu} + (\sBdel, X_\mu)_\Kref} } + u^a u^b\prn{\source{h_{ab}} + \tilde{\gref}_{ab} + (\sBdel, \gref_{ab})_\Kref } + 4 (\Kref.\source{\mathcal{F}})
}  \\
& \quad\times \brk{ - 2 \, u^c u^d\,\nabla_{(c}  \brk{ \partial_{d)}X^\mu \prn{\mathfrak{J}_{\mu} - (\sBdel, X_\mu)_\Kref} } - u^c u^d\prn{\source{h_{cd}} + \tilde{\gref}_{cd} - (\phizT-\sBdel, \gref_{cd})_\Kref } - 4 \Kref^c (\source{\mathcal{F}}_c + \partial_c \phizT)
}  \\
&\quad\;\, + K_\acc\,u^a\,u^b \,\acc^d \,\Big[ 2\nabla_{(a} (\source{h_{b)d}}+\tilde{\gref}_{b)d}) +2 (\source{\mathcal{F}_{(a}}, \gref_{b)d})_\Kref   -  \nabla_d (\source{h_{ab}}+\tilde{\gref}_{ab}) - (\sAt{d},\gref_{ab})_\Kref \Big] \bigg\} 
\\
{\Lag}_{\text{wv}}^{^{(\vartheta)}} &=\sqrt{-{\gref}}\;  \bigg\{
 \brk{ \frac{1}{2} (K_\vartheta+TK_\vartheta') \vartheta^2\, u^a u^b + \frac{1}{2} \, K_\vartheta \, \vartheta^2 \, P^{ab} + 2\,K_\vartheta \, \vartheta \, \acc^{(a} u^{b)} } (\source{h_{ab}}+ \tilde{\gref}_{ab}) \\
&\qquad\quad\;\, + K_\vartheta \, \vartheta \, (u.\nabla) (\source{h_{ab}}+ \tilde{\gref}_{ab}) \, u^a u^b + K_\vartheta\, \vartheta \prn{ 2 (\Kref.\partial) T - \vartheta } \, (\Kref.\source{\mathcal{F}})  - 2i\,K_\vartheta \,\vartheta\,  (u.\partial) \source{h_{\theta\thb}} \\
&\qquad\quad\;\, +  K_\vartheta \, \vartheta\, u^a \, \gref^{bd} \Big[ 2\nabla_{(a} (\source{h_{b)d}}+\tilde{\gref}_{b)d}) +2 (\source{\mathcal{F}_{(a}}, \gref_{b)d})_\Kref   -  \nabla_d (\source{h_{ab}}+\tilde{\gref}_{ab}) - (\sAt{d},\gref_{ab})_\Kref \Big] 
 \bigg\}
\end{split}
\label{eq:bvterm4}
\end{equation}
}\normalsize
Note that the presence of fluctuation terms is only superficial: expanding out $\tilde{\gref}_{ab}$ as in \eqref{eq:gJdefs} shows that all dependence on $\tilde{X}^\mu$ actually cancels. We can then read off the energy-momentum tensor and the free energy current:
\begin{equation}
\begin{split}
{\TEMref}_{_{(\acc,\vartheta)}}^{ab} 
&\equiv 
	\frac{2}{\sqrt{-{\gref}}} \, \frac{\delta {\Lag}^{^{(\acc,\vartheta)}}_\text{wv}}{\delta {h}_{ab}} 
	\\
&= \prn{ 4K_\acc+T\,\K_\acc'} \acc^2\, u^a u^b +  K_\acc \, \acc^2 \, g^{ab} + 2\,K_\acc \, \acc^a \acc^b  +2 \nabla_c \brk{ K_\acc \prn{  u^a u^b \acc^c -2u^c u^{(a} \acc^{b)} }}  \\
&\quad +  K_\vartheta \, \vartheta^2 \, g^{ab}+  \prn{TK_\vartheta'\,\vartheta^2 -2K_\vartheta' \,\vartheta\, (u.\partial)T - 2K_\vartheta (u.\partial)\vartheta }  u^a u^b   - 2 \, \nabla_c \prn{K_\vartheta\, \vartheta \, u^c \, \gref^{ab} }  \\
 {\Nref}_{_{(\acc,\vartheta)}}^a
&\equiv
	- \frac{1}{\sqrt{-{\gref}}} \, \frac{\delta {\Lag}^{^{(\acc,\vartheta)}}_\text{wv}}{\delta {\source{{\cal F}_{a}}}}\\
 &=   K_\acc \, \acc^2 \, \Kref^a + K_\acc \, \prn{ \acc^a \, u^c \, u^d - 2 \, u^a \, u^{(c} \, \acc^{d)} } \lieD_\Kref \gref_{cd} - K_\vartheta \, \vartheta^2 \, \Kref^a  \,.
\end{split}
\label{eq:expsquaredres}
\end{equation}	
where we have set $\phizT=-i$. The above currents then precisely match the ones derived in Appendix F of \cite{Haehl:2015pja}, which provides a non-trivial consistency check on our formalism. 
Note that there are no fluctuation terms in the above currents (as expected in Class L). Finally, we find
\begin{equation}
\begin{split}
{\TEMref}_{_{(\acc,\vartheta)}}^{\theta\thb} 
&\equiv 
	\frac{2}{\sqrt{-{\gref}}} \, \frac{\delta {\Lag}^{^{(\acc,\vartheta)}}_\text{wv}}{\delta {h}_{\theta\thb}} 
	= 4\,i \, \nabla_a(  K_\vartheta \, \vartheta\,u^a ) \,,\\
 \Delta_{_{(\acc,\vartheta)}} &\equiv
	 -\frac{1}{\sqrt{-{\gref}}} \frac{\delta {\Lag}^{^{(\acc,\vartheta)}}_\text{wv}}{\delta \sBdel}  
	 = \text{e.o.m.} 
 \end{split}
\end{equation}
where ``e.o.m.'' denotes a term proportional to the equations of motion.

\subsubsection{Shear squared}
\label{sec:vorticity}

Let us now consider another, slightly more complicated, Class L term: the square of the shear tensor,
\begin{equation}
\begin{split}	
{\Lag}_{\text{wv}}^{^{(\sigma^2)}} &=\int d\theta d \thb \;  \frac{\sqrt{-\SF{\gref}}}{\SF{\bf z}}\;  \brk{
	 \SF{K}_\sigma(\SF{T})\, \SF{\sigma}^2 }
\end{split}
\label{eq:bvterm3}
\end{equation}
where 
\begin{equation}
\begin{split}
\SF{\sigma}^2 &\equiv (-)^{I (1+K+J+L) + K(1+L)+J(1+L) + L} \; \SF{\gref}^{IK} \SF{\gref}^{JL} \; 
	\SF{\sigma}_{IJ}\, \SF{\sigma}_{KL} \\
\SF{\sigma}_{IJ} &\equiv \frac{1}{2}\,(-)^{M(1+J+N) + N} \, \SF{P}_I{}^M\, \SF{P}_J{}^N \prn{ \Dwv_M \SF{u}_N  + (-)^{MN} \Dwv_N \SF{u}_M  -  \frac{2}{d-1} \, \SF{\gref}_{MN} \, \SF{\vartheta}}\,.
\end{split}
\label{eq:}
\end{equation}
That is, $\SF{\sigma}_{IJ}$ is the superspace version of the usual (symmetric, transverse, traceless) shear tensor. The superspace expansion of $\SF{\sigma}_{IJ}$ is computed in Appendix \ref{sec:MMOfurther} (see eq.\ \eqref{eq:supersigma}). Transport associated with an ordinary space Lagrangian of the form $\sigma^2$ was studied in \cite{Haehl:2015pja}.
The anti-symmetric combination $\SF{\Omega}_{IJ}=\Dwv_I \SF{u}_J -  (-)^{IJ} \Dwv_J \SF{u}_I$ (which is closely related to the vorticity tensor) could be treated similarly as in our analysis below.

For $\SF{\sigma}^2$ we find after a straightforward calculation (using \eqref{eq:supersigma}):
\begin{equation}
\begin{split}
\SF{\sigma}^2 
&= 
	\SF{\gref}^{ac} \SF{\gref}^{bd}\, \SF{\sigma}_{ab} \, \SF{\sigma}_{cd}
	-4\SF{\gref}^{a\theta} \SF{\gref}^{bd}\, \SF{\sigma}_{ab} \, \SF{\sigma}_{\theta d}
	-4\SF{\gref}^{a\thb} \SF{\gref}^{bd}\, \SF{\sigma}_{ab} \, \SF{\sigma}_{\thb d}
	-4\SF{\gref}^{\theta \thb} \SF{\gref}^{bd}\, \SF{\sigma}_{\theta b} \, \SF{\sigma}_{\thb d}
	\\
& \qquad - 4 \, \SF{\gref}^{\theta a} \SF{\gref}^{\thb b} \, \SF{\sigma}_{\theta\thb} \, \SF{\sigma}_{ab} - 2 \, \SF{\gref}^{\theta\thb} \SF{\gref}^{\thb \theta}  \, \SF{\sigma}_{\theta\thb} \, \SF{\sigma}_{\thb \theta} \\
& =  \sigma^2  - \frac{2}{(d-1)^2} \, \vartheta^2 
   + \thb\theta\, \Big\{ - 2 (\source{h_{ab}} + \tilde{\gref}_{ab} ) \, \sigma^{ac} \, \sigma^b{}_c 
+	2\, (\partial_\theta\partial_\thb \SF{\sigma}_{ab}|)\, \sigma^{ab}  \\
&\qquad\qquad\qquad\qquad\qquad\quad\;\; +4i\,  \prn{ \mathfrak{J}^a\, \sigma_a{}^c - g^{cd}(\partial_\thb \SF{\sigma}_{\theta d}|) } \prn{ \tilde{X}^b\, \sigma_{bc} - (\partial_\theta \SF{\sigma}_{\thb c}|) }\\
&\qquad\qquad\qquad\qquad\qquad\quad\;\; - \frac{4i}{d-1} \, \brk{  (\partial_\theta\partial_\thb \SF{\sigma}_{\theta\thb}|) + \tilde{X}^a \mathfrak{J}^b \, \sigma_{ab}  + \frac{1}{d-1} \, \vartheta \, \source{h_{\theta\thb}}  } \, \vartheta   \Big\}\,.
\end{split}
\label{eq:sigsq}
\end{equation}
The second line of the r.h.s.\ does not contribute to the currents. The other two lines yield:
\begin{equation}
\begin{split}
{\TEMref}_{_{(\sigma^2)}}^{ab} 
&= K_\sigma \, \sigma^2 \, \gref^{ab} + T \, K_\sigma'\, \sigma^2 \, u^a u^b - 4 \, K_\sigma \, \sigma^{ac}\sigma^b{}_c - 4 \, K_\sigma \, \frac{\vartheta}{d-1} \, \sigma^{ab} + 4\, K_\sigma \, \nabla_c u^{(a} \, \sigma^{b)c}\\
&\qquad + 4 \, K_\sigma \, \acc_c \, \sigma^{c(a} \, u^{b)}+ 2\, K_\sigma \, \sigma^2 \, u^au^b - 2 \, \nabla_c \prn{ K_\sigma \, u^c \, \sigma^{ab} }  \\
&\qquad - \frac{4}{(d-1)^2} \brk{ \prn{ K_\sigma  + \frac{T}{2}K'_\sigma  } \,\vartheta^2\, u^a u^b + \frac{1}{2} \, K_\sigma \,\vartheta^2\, \gref^{ab}- \nabla_c \prn{ K_\sigma \, \vartheta \, P^{ab} u^c }} \,,\\
 {\Nref}_{_{(\sigma^2)}}^a
& =  - K_\sigma \, \sigma^2 \, \Kref^a + \frac{2}{(d-1)^2} \, K_\sigma \, \vartheta^2 \, \Kref^a \,,
 \\
 {\TEMref}_{_{(\sigma^2)}}^{\theta\thb} &= \frac{4\,i\,(d+1)}{(d-1)^2} \, \nabla_a ( K_\sigma \, \vartheta \, u^a ) \,.
\end{split}
\label{eq:shearsqresult}
\end{equation}	
To intuit these constitutive relations, first note that there are no fluctuation terms involving $\tilde{X}^\mu$. This is consistent with the expectation that Class L transport is adiabatic and does not involve fluctuations. For a more detailed understanding, we can decompose the relations as follows: 
\begin{itemize}
\item The first two lines of ${\TEMref}_{_{(\sigma^2)}}^{ab}$, and the first term in $ {\Nref}_{_{(\sigma^2)}}^a$ correspond to the constitutive relations associated with the shear-tensor-squared Class L transport, worked out in detail in Appendix F of \cite{Haehl:2015pja}.\footnote{ Note a typo in Eq.\ (F.9) of \cite{Haehl:2015pja}: the first term in the third line has the wrong sign.} 
\item The last line of ${\TEMref}_{_{(\sigma^2)}}^{ab}$ together with the second term in $ {\Nref}_{_{(\sigma^2)}}^a$ correspond to Class L transport of type $\vartheta^2$. Indeed, one can easily check that these terms take the same form as $\{{\TEMref}_{_{(\vartheta^2)}}^{ab}, {\Nref}_{_{(\vartheta^2)}}^a\}$ in \eqref{eq:expsquaredres} up to normalization. This is true up to a single term in the energy-momentum tensor proportional to $u^{(a} \acc^{b)}$, which is, however, adiabatic by itself.
\end{itemize}
In conclusion, the Lagrangian \eqref{eq:bvterm3} describes a particular type of second order Class L transport in a consistent fashion. In order to match to the classification of second order Class L transport in \cite{Haehl:2015pja}, one has to form some linear combinations of the constitutive relations worked out there. We speculate that is is because superspace prefers a slightly different natural basis of constitutive relations than does the Landau-Ginzburg construction of \cite{Haehl:2015pja}. It would be interesting to work this out in detail by conducting an exhaustive analysis of second order Class L transport using the superspace formalism.

\section{The Eightfold Way from superspace}
\label{sec:SigmaModel}

Having understood the formalism behind the superspace construction of hydrodynamic effective actions, we now turn to the question of completeness, which inspired us to undertake this analysis. This will provide the proof for the theorem stated at the beginning of \S\ref{sec:topsigma}. As described earlier, in \cite{Haehl:2014zda,Haehl:2015pja} we examined the axioms of hydrodynamics as the low energy, near-equilibrium dynamics of conserved currents, subject to the local form of the second law being upheld. We argued there that a complete set of solutions to these axioms results in 8 distinct classes of allowed hydrodynamic transport, in addition to a class of forbidden terms whose presence would be deleterious to the second law. This classification provided an explicit construction of all consistent hydrodynamic constitutive relations at any order in derivatives. In this section we outline how to reproduce this structure from effective actions in superspace.

Our classification scheme, dubbed \emph{the Eightfold Way}, for obvious reasons, can be summarized as follows:
\begin{itemize}
\item Class A: anomalous transport originating from both flavor and Lorentz anomalies in conserved currents.
\item Class B: Berry-like transport such as Hall viscosity in parity odd fluids,  shear-vorticity transport in parity-even systems, 
$T^{\mu\nu} \sim \lambda_2(T)\, \sigma^{\rho (\mu} \omega_{\rho}^{\;\nu)}$, etc.
\item Class C: identically conserved entropy with no physical transport, as occurring in systems with macroscopic ground state degeneracy.
\item Class D: dissipative (entropy producing) transport, including familiar quantities such as shear and bulk viscosities, conductivities, etc.
\item Class $\PS$: hydrostatic scalar transport which can be inferred from thermal equilibrium on curved spacetime with a global timelike Killing vector. It is described by the simplest of effective actions, viz., an equilibrium partition function. Examples include ideal fluids, curvature couplings $T^{\mu\nu} \sim \kappa(T)\, R^{\mu\sigma\nu\rho} u_\sigma\, u_\rho$, vorticity-squared transport $T^{\mu\nu} \sim \lambda_3(T)\, \omega^{\rho (\mu} \omega_{\rho}^{\ \nu)}$, acceleration squared considered in \S\ref{sec:mmo2nd}, etc.
\item Class $\LS$: hydrodynamic (non-dissipative) scalar transport with a free energy current that is generated from a scalar functional. These include terms that vanish in equilibrium but can nevertheless be inferred  from a Landau-Ginzburg effective action. An example is shear-squared transport $T^{\mu\nu} \sim \lambda_1(T)\, \sigma^{\rho (\mu} \sigma_{\rho}^{\ \nu)}$ which we considered in \S\ref{sec:mmo2nd}, cf., \eqref{eq:bvterm3}. 
\item Class $\PV$: hydrostatic transport, with the novelty that the free energy current is transverse to the fluid flow. 
This class of terms arises in the context of anomalous fluids.
\item Class $\GV$: hydrodynamic transport involving a transverse free energy current. This is perhaps the least well-understood class with the simplest examples involving charged fluids at second order in derivative expansion.
\end{itemize} 

In addition we also defined a Class $\PF$ of hydrostatic forbidden terms, which follow from the equilibrium partition function analysis of \cite{Banerjee:2012iz,Jensen:2012jh}; these terms are forbidden from appearing in any physical fluid as they violate the second law.
In practice, since one tends to present fluid dynamical constitutive relations in a particular fluid frame, one usually ends up interpreting these $\PF$ terms as relations between transport coefficients. See Section 14 of \cite{Haehl:2015pja} for a concise summary of the classification and explicit examples.

It is important to note that of the 8 classes described above, only Class D terms lead to entropy production. The remaining 7 classes are adiabatic and should be viewed as non-dissipative transport. They were also argued to be captured by an effective action, dubbed Class $\LT$ action, in \cite{Haehl:2014zda,Haehl:2015pja}, which as we shall see bears a large degree of structural similarity to the superspace construction described here.

In what follows, we will describe in the superspace formalism the effective action for 6 of these 8 classes, eschewing for the present, a discussion of Classes A and $\PV$ which are associated with anomalies and require additional formalism that would obscure the presentation but should be straightforward to implement (see \S\ref{sec:discussion}). This will prove the theorem of \S\ref{sec:topsigma} upto anomalous transport, which we leave for a future investigation. Elements of the following discussion will be an elaboration of concepts explained in \cite{Haehl:2015uoc}.
More specifically, for Classes $\PS$, $\LS$ combined into Class L the effective action is given in  \eqref{eq:Llag}, the Class D Lagrangian is given in \eqref{eq:Dlag2}, while he effective actions of the other non-anomalous Classes B, C, and $\GV$ are given in \eqref{eq:LagB}, \eqref{eq:HVbarLagN}, and \eqref{eq:LagC}, respectively. To the best of our understanding this demonstrates that all of allowed non-anomalous hydrodynamic transport can be embedded into our formalism.

\subsection{Class L (= Class $\PS \ \cup$ Class $\LS$)}
\label{sec:classL}

The scalar classes $\PS$ and $\LS$ are ones where the free energy is aligned with the fluid velocity. Moreover, the free energy can be generated from a single scalar Landau-Ginzburg functional (a generalized pressure term). It  is therefore simplest to discuss them at the same time. 

As in \cite{Haehl:2015pja} we define therefore Class L = Class $\PS \ \cup$ Class $\LS$ and seek a scalar generating functional for all the currents. This case is easy to understand without the need for superspace; all we needed was a scalar Landau-Ginzburg function of the hydrodynamic fields (see \S6 and \S7 of \cite{Haehl:2015pja}). For neutral fluids under discussion we would construct a Lagrangian functional which depends on $\Kref^a$ and $\gref_{ab}$. The upgrade to superspace simply lifts this functional to a superfield. 

For Class L, we consider a general superspace Lagrangian as in \eqref{eq:fullLag} without usage of $\gpsi_{IJ}$ and $\gpsib_{IJ}$:
\begin{equation}
\boxed{
\Lag_\text{wv}^{\text{(L)}}  = 
	\int\,d\theta\, d\thetab\, \frac{\sqrt{-\SF{\gref}}}{\zsf} \ \SF{\Lagref}^{\text{(L)}}[\SF{\gref}_{IJ},
	\SF{\Kref}^I,\Dwv_J] \,.
}
\label{eq:Llag}
\end{equation}
The superspace integral can be performed by variation with respect to the various fields, and subsequent application of appropriate (covariant) derivatives. Keeping track for simplicity only of the variation with respect to the sources $\{\source{h_{ab}},\source{h_{\theta\thb}},\sAt{a},\sBdel\}$ that we turned on in the MMO limit, we find in general
{\small
\begin{equation}
\begin{split}
\Lag_\text{wv}^{\text{(L)}} 
&=
	 \partial_\theta \partial_\thb \left( \frac{\sqrt{-\SF{\gref}}}{\zsf} \ 
	 \SF{\Lagref}^{\text{(L)}}[\SF{\gref}_{IJ}\,, \SF{\Kref}^I,\Dwv_I] \right) \bigg{|} \\
 &\equiv 
 	\sqrt{-{\gref}} \left\{ \frac{1}{2} \, \TEMref_{{_\text{(L)}}}^{ab}  \, \left(\source{h_{ab}} + (\sBdel, \gref_{ab})_\Kref \right)- \Nref_{{_\text{(L)}}}^a \, \left(\sAt{a} - \partial_a \sBdel \right) + \frac{1}{2} \, \TEMref_{{_\text{(L)}}}^{\theta\thb}  \, \source{h_{\theta\thb}} 
	\right\}  
  + \ldots \\
 &= \sqrt{-{\gref}}\left\{ \frac{1}{2} \, \TEMref_{{_\text{(L)}}}^{ab}  \, \source{h_{ab}}- \Nref_{{_\text{(L)}}}^a \, \sAt{a} + \frac{1}{2} \, \TEMref_{{_\text{(L)}}}^{\theta\thb}  \, \source{h_{\theta\thb}} + \sBdel \; \Delta \right\} + \ldots 
\end{split}
\end{equation}
}\normalsize
where the ellipses ``$\ldots$'' denote ghost bilinears and terms at higher orders in sources, and we defined 
\begin{equation}\label{eq:LL}
\begin{split}
\TEMref_{{_\text{(L)}}}^{ab} &\equiv \frac{2\,\zsf}{\sqrt{-\SF{\gref}}}
         \frac{\delta}{\delta \SF{\gref}_{ab}} \left\{ \frac{\sqrt{-\SF{\gref}}\; \SF{\Lagref}^{\text{(L)}} }{\zsf} \right\} \bigg{|} =
	\frac{2}{\sqrt{-\gref}} \, \frac{\delta \Lag_\text{wv}^{\text{(L)}}}{\delta \source{h_{ab}}}\,,\\
\TEMref_{{_\text{(L)}}}^{\theta\thb} &\equiv  \frac{\zsf}{\sqrt{-\SF{\gref}}}
        \frac{\delta}{\delta \SF{\gref}_{\theta\thb}} \left\{ \frac{\sqrt{-\SF{\gref}}\; \SF{\Lagref}^{\text{(L)}} }{\zsf} \right\} \bigg{|} =
	\frac{2}{\sqrt{-\gref}} \, \frac{\delta \Lag_\text{wv}^{\text{(L)}}}{\delta \source{h_{\theta\thb}}} \,,\\ 
\Nref_{{_\text{(L)}}}^a &\equiv
         - \frac{\zsf}{\sqrt{-\SF{\gref}}}\frac{\delta}{\delta \SF{\Ascr}_{a}} 
	\left\{ \frac{\sqrt{-\SF{\gref}}\; \SF{\Lagref}^{\text{(L)}} }{\zsf}  \right\} \bigg{|} = -  \frac{1}{\sqrt{-\gref}} \frac{\delta \Lag_\text{wv}^{\text{(L)}}}{\delta \source{\mathcal{F}_{a}}} 	 \,.
\end{split}
\end{equation}	
In Class L, the Grassmann-odd components of the currents vanish. The adiabaticity equation in this class is therefore simply the ordinary-space adiabaticity equation, up to ghost bilinears as usual: 
\begin{equation}
 \nabla_a \, \Nref_{{_\text{(L)}}}^a = \frac{1}{2} \, \TEMref_{{_\text{(L)}}}^{ab} \, \lieD_\Kref \gref_{ab} + \ldots \,.
\end{equation}
Therefore, $\Delta=0$ as expected. All of this is consistent with the basic intuition that the Class L terms are captured by an ordinary Landau-Ginzburg sigma model for the dynamical part of the hydrodynamic pion $X^\mu$.

\subsection{Class D}
\label{sec:classD}

This is the class of hydrodynamic constitutive relations containing the physical dissipative transport. We have already seen in \S\ref{sec:topsigma} how the origin of entropy production is encoded in a (differential operator valued) super-tensor with certain symmetries. We will now add a few more details to this abstract discussion; the explicit examples of bulk and shear viscosities in \S\ref{sec:mmo1st} should provide the reader with intuition for these manipulations.

The Class D action is parameterized by a rank 4 super-tensor $\SF{\etaref}^{IJKL}$:
\begin{equation}
\boxed{
\begin{split}
\Lag_\text{wv}^{(\text{D})} &= 
	\int\,d\theta\, d\thetab\, \frac{\sqrt{-\SF{\gref}}}{\zsf} \; 
	\SF{\Lagref}_{(\text{D})} \\
	& \text{ with } 
	\quad 
	\SF{\Lagref}_{(\text{D})} =
 	 - \frac{i}{4}\, (-)^{(I+J)(K+L)  + IJ+KL+K+L} \;\SF{\etaref}^{((IJ)|(KL))}_{_\text{(D)}}\, 
 	 \gpsib_{IJ} \; \gpsi_{KL} \,.
	 \end{split}
}
 \label{eq:Dlag2}
\end{equation}
where we remind the reader of the definitions $\gpsib_{IJ} \equiv \Dut_\theta \SF{\gref}_{IJ}$ and $\gpsi_{IJ} \equiv \Dut_\thb \SF{\gref}_{IJ}$.
The dissipation tensor has the following symmetry properties:
\begin{equation}
\begin{split}
& \SF{\etaref}_{_\text{(D)}}^{((IJ)|(KL))} = (-)^{IJ}\; \SF{\etaref}_{_\text{(D)}}^{((JI)|(KL))} = (-)^{KL}\,\; \SF{\etaref}_{_\text{(D)}}^{((IJ)|(LK))}  \\
& \SF{\etaref}_{_\text{(D)}}^{((IJ)|(KL))} = (-)^{(I+J)(K+L)} \; \SF{\etaref}_{_\text{(D)}}^{((KL)|(IJ))} \,.
\end{split}
\label{eq:etacpt}
\end{equation} 
The first line is obvious since the super-tensor $\SF{\etaref}^{IJKL}_{_\text{(D)}}$ must have the correct graded symmetric index pairs to be contracted against the derivatives of the metric tensor. The second line is what distinguishes Class D from Class B. The latter will be discussed in the next subsection and will have the same structure but with a graded anti-symmetric $\etaref$ tensor.

Furthermore, the dissipative tensor satisfies an Onsager relation under {\sf CPT} transformations:
\begin{equation}
\SF{\etaref}^{((IJ)|(KL))}_{_\text{(D)}}  = (-)^{(I+J)(K+L)}\; [ \SF{\etaref}^{((KL)|(IJ))}_{_\text{(D)}} ]^{\text{\sf{\tiny CPT}}}\,.
\end{equation}
To understand this requirement, let us consider the {\sf CPT} transformation of the Lagrangian density \eqref{eq:Dlag2}:
\begin{equation}
\begin{split}
\SF{\Lagref}_{(\text{D})}
&\stackrel{\text{\sf{\tiny CPT}}}{\longmapsto}	
	\frac{i}{4}  (-)^{(I+J)(K+L)  + IJ+KL+K+L} \;  [ \SF{\etaref}^{((IJ)|(KL))}_{_\text{(D)}} ]^{\text{\sf{\tiny CPT}}} \,  \gpsi_{IJ} \; \gpsib_{KL} \\
&=
	\frac{i}{4}  (-)^{(I+J)(K+L)  + IJ+KL+K+L +(I+J+1)(K+L+1)} \;  [ \SF{\etaref}^{((IJ)|(KL))}_{_\text{(D)}} ]^{\text{\sf{\tiny CPT}}} \, \gpsib_{KL}\; \gpsi_{IJ}  \\
& = 
	-\frac{i}{4}(-)^{IJ+KL+K+L} \,  [ \SF{\etaref}^{((KL)|(IJ))}_{_\text{(D)}} ]^{\text{\sf{\tiny CPT}}}\,  \gpsib_{IJ}\; \gpsi_{KL}	
\end{split}
\label{eq:etacptderive}
\end{equation}
from which we see \eqref{eq:etacpt} naturally follows such that $\SF{\Lag}_\text{wv}^{(\text{D})}$ is {\sf CPT} invariant. Note that {\sf CPT} invariance of $\SF{\Lag}_\text{wv}^{(\text{D})}$ is consistent with the Lorentzian path integral: in the action $ i S_\text{wv}$ introduced  in \eqref{eq:fullLag}, which appears in the functional integral, the anti-linear action of {\sf CPT} (sending $i \mapsto -i$) cancels the sign coming from the $R$-parity action on the super-integration measure (sending $d\theta \, d\thb \mapsto - d\theta \, d\thb$).

We have already seen some examples of $\SF{\etaref}^{IJKL}_{_\text{(D)}}$ built using the projector $\SF{P}^{IJ}$ in 
\S\ref{sec:mmo1st}. To second order in the gradient expansion, it suffices to consider super-tensors parameterized by two non-negative definite scalar functions and an arbitrary (i.e., unconstrained) super-tensor involving a derivative operator \cite{Haehl:2015pja}:
\begin{equation}
\begin{split}
\SF{\etaref}^{IJKL}_{_\text{(D)}} &=
	\SF{\zeta}(\SF{T})\, \SF{T}\, \SF{P}^{IJ} \, \SF{P}^{KL} 
	+ 2\, \SF{\eta}(\SF{T})\, \SF{T}\, (-)^{K(I+J)}\,\SF{P}^{K\langle I} \, \SF{P}^{J \rangle L} \\
&\qquad 
	+\frac{1}{2}\left( \SF{\mathcal{N}}^{(IJ)(KL)}[\Dwv] +  (-)^{(I+J)(K+L)}\; 
		\SF{\mathcal{N}}^{(KL)(IJ)}[\Dwv] \right) 
\end{split}
\label{eq:D2ndorder}
\end{equation}	
where the symmetrized traceless projection in superspace is defined in \eqref{eq:shearten}.

There are a couple of important features to note about the effective action \eqref{eq:Dlag2}. Firstly there  is the explicit $i$ in front of the Lagrangian, which signals dissipative contributions. This is required to ensure that the full action $i \, S_\text{wv}^{(\text{D})}$ is {\sf CPT} invariant.  It is instructive to compare this with the Martin-Siggia-Rose construction \cite{Martin:1973zz} for the Langevin dynamics of a Brownian particle. A discussion in the framework of thermal equivariance can be found in \cite{Haehl:2016uah} (see also \cite{Haehl:2015foa}).  Secondly, this is the sector that leads to non-trivial entropy production. As we have argued in \S\ref{sec:topsigma} all we need for the second law to hold is that $\etaref^{(ab)(cd)}_{_\text{(D)}} \equiv \SF{\etaref}^{(ab)(cd)}_{_\text{(D)}}|$ be a non-negative definite quadratic form, when interpreted as a map from $\text{Sym}_2 \mapsto \text{Sym}_2$. This ensures positivity of entropy production (see \eqref{eq:DeltaD} below). We will demonstrate positivity below.

Let us analyze the dissipative Lagrangian more explicitly by doing the superspace integral. For simplicity, we assume that $\SF{\etaref}^{(IJ)(KL)}_{_\text{(D)}}$ is simply a super-tensor (as opposed to a derivative operator valued super-tensor). We can then compute directly from \eqref{eq:Dlag2} the action in ordinary space by taking super-derivatives. The result was already given in \eqref{eq:Dcalc2}, which we reproduce here:
\begin{equation}
\begin{split}
\Lag_\text{wv}^{(\text{D})} 
&=	-\frac{i}{4}\, \sqrt{-{\gref}}\,\bigg\{
	4
 	 \left( \partial_\theta \partial_\thb\SF{\etaref}^{((\thb c)|(\theta d))}_{_\text{(D)}} + \left[ \frac{1}{2} \, \gref^{ab}(\source{h_{ab}}+ \tilde{\gref}_{ab}) - \sAt{a} \Kref^a \right]\,\SF{\etaref}^{(\thb c)(\theta d)}_{_\text{(D)}}  \right)
	   \gpsib_{\thb c} \; \gpsi_{\theta d} \\
&\qquad	+ 
	 2\,\partial_\theta 
 	 \SF{\etaref}^{((ab)|(\theta d))}_{_\text{(D)}}\,
	 \partial_\thb  \gpsib_{ab} \; \gpsi_{\theta d} 
	 -
	  4\,\partial_\theta 
 	  \SF{\etaref}^{((\theta\thb)|(\theta d))}_{_\text{(D)}}\,
	 \partial_\thb  \gpsib_{\theta\thb} \; \gpsi_{\theta d} 
        \\
&\qquad
	+2\,
	 \partial_\thb 
 	   \SF{\etaref}^{((\thb b)|(cd))}_{_\text{(D)}}\,
	\gpsib_{\thb b} \; \partial_\theta  \gpsi_{cd} 
	-4\,
	 \partial_\thb 
 	 \SF{\etaref}^{((\thb b)|(\theta\thb))}_{_\text{(D)}}\,
	\gpsib_{\thb b} \; \partial_\theta  \gpsi_{\theta\thb} \\	
&\qquad
	+ \Big[
	-2 \, \SF{\etaref}^{((\theta\thb)|(cd))}_{_\text{(D)}} \left( \partial_\thb \gpsib_{\theta\thb} \; \partial_\theta \gpsi_{cd} + \partial_\theta \gpsi_{\theta\thb} \; \partial_\thb \gpsib_{cd} \right) \\
&\qquad\qquad\qquad +4 \,\SF{\etaref}^{((\thb b)|(\theta d))}_{_\text{(D)}} 
	\left( \partial_\theta \partial_\thb  \gpsib_{\thb b} \; \gpsi_{\theta d}  +   \gpsib_{\thb b} \; \partial_\theta \partial_\thb\gpsi_{\theta d}  \right)
	\\
&\qquad\qquad\qquad
	+ \SF{\etaref}^{((ab)|(cd))}_{_\text{(D)}} \,\partial_\thb  \gpsib_{ab} \; \partial_\theta\gpsi_{cd} 
        +4\, \SF{\etaref}^{(\theta\thb)(\theta\thb)}_{_\text{(D)}} \,\partial_\thb  \gpsib_{\theta\thb} \; \partial_\theta\gpsi_{\theta\thb} 
	\Big]
	\bigg\}\bigg| + \text{ghosts}
\end{split}
\label{eq:Dcalc3}
\end{equation}	
We can now isolate the terms which describe classical hydrodynamics and fluctuations. The classical parts of the currents only get a contribution from the first term in the last line of \eqref{eq:Dcalc3}:
{\small
\begin{equation}
\begin{split}
\TEMref_{_\text{(D,class.)}}^{ab} &= -\frac{i \, \phizT}{2} \, \etaref^{((ab)|(cd))}_{_\text{(D)}} \, \lieD_\Kref \, \gref_{cd} \,, \\
\Nref_{_\text{(D,class.)}}^a &= 0 \,.
\end{split}
\label{eq:Dconst}
\end{equation}	
}\normalsize
These classical terms (not involving $\tilde{X}$) are consistent with the general analysis of Class D transport in classical hydrodynamics. This structure was already anticipated in \cite{Haehl:2015uoc}.\footnote{ Note in particular that a vanishing current $\Nref_{_\text{(D,class.)}}^a$ is consistent with our assumption that $\SF{\etaref}^{IJKL}$ is a tensor (not a tensor valued derivative operator).} 

The fluctuation currents can a-priori get contributions from each term in \eqref{eq:Dcalc3} which involves a source. While we will not write this in generality (the answer is obvious from \eqref{eq:Dcalc3}), we note that it can be straightforwardly extracted for any given choice of $\SF{\etaref}^{IJKL}$. The resulting fluctuation currents provide a prediction by our formalism.

\paragraph*{The second law:}
Setting $\langle i\, \phizT \rangle= 1$, we can easily verify the standard adiabaticity equation from \eqref{eq:Dconst}: 
\begin{equation}
\nabla_a \Nref_{_\text{(D,class.)}}^a - \frac{1}{2}\, \TEMref_{_\text{(D,class.)}}^{ab} \, \lieD_\Kref g_{ab} =   \Delta   \,,
\end{equation}
where the total entropy production is 
\begin{equation}
\label{eq:DeltaD}
\Delta \equiv \frac{1}{4} \, \etaref^{((ab)|(cd))}_{_\text{(D)}}  \, \lieD_\Kref \gref_{ab} \; \lieD_\Kref \gref_{cd}  \,.
\end{equation}
Let us finally demonstrate how the second law follows, i.e., $\Delta \geq 0$. From \eqref{eq:DeltaD} it is already clear that $\Delta$ is a quadratic form. Therefore, all that needs to be shown is that $\etaref^{((ab)|(cd))}_{_\text{(D)}}$ is a semi-positive definite map from symmetric two-tensors to symmetric two-tensors. We will now argue that this follows simply from convergence of the path integral. 

In the path integral the effective action appears in the form $e^{i S_\text{wv,diss.}}$. Consequently, a well defined path integral requires Im$(S_\text{wv,diss.}) \geq 0$. From \eqref{eq:Dcalc3} we can extract the imaginary part of the action. For this purpose, we turn off the fluctuations and set the auxiliary sources $\{\source{h_{\theta\thb}},\source{\mathcal{F}_I},\sBdel\}$ to zero. This leaves us with an effective action in terms of the physical source $\source{h_{ab}}$ whose imaginary part is given by
\begin{equation} \label{eq:ImS}
\begin{split}
  \text{Im}(S_\text{wv,diss.} [\source{h_{ab}}] )&= \frac{1}{4} \, \int d^d \sigma \sqrt{-\gref} \ \etaref^{((ab)|(cd))}_{_\text{(D)}} \,\source{h_{ab}}\, \source{h_{cd}} 
 \end{split}
\end{equation}
which is a quadratic form as desired and originates from the last two lines of \eqref{eq:Dcalc3}. Positivity of the above expression then requires that $\etaref^{((ab)|(cd))}_{_\text{(D)}}$ is a semi-positive definite map on symmetric two-tensors. Note that the derivative expansion implies that once we have established the general form of entropy production \eqref{eq:DeltaD} with \eqref{eq:D2ndorder}, it is sufficient to constrain the leading viscous coefficients ($\eta,\zeta\geq 0$). This is a compact version of Bhattacharyya's theorem \cite{Bhattacharyya:2014bha,Bhattacharyya:2013lha}.

Let us note that demanding Im$(S_\text{wv,diss.}) \geq 0$ within the hydrodynamic limit is tantamount to asking for the path integral to converge within the hydrodynamic expansion. A-priori one could have allowed for the possibility that Im$(S_\text{wv,diss.}) $ has any sign within the hydrodynamic regime, but receives non-hydrodynamic corrections that render the full path integral convergent. If this were to be the case, we would find we are missing important information within the effective field theory, and cannot use the effective actions to derive any useful statements per se. Our analysis assumes that the hydrodynamic effective actions are complete in and of themselves, and  do not require any additional non-perturbative 
(from the viewpoint of the gradient expansion) information to complete them. On a related note, one would also argue that these effective actions should not lead to non-perturbative constraints that violate the rules of the gradient expansion. 

To summarize our discussion, we have provided a complete derivation of the two statements that constitute the second law in hydrodynamics: $(i)$ the most general form of the dissipative constitutive relations, and $(ii)$ the positivity of total entropy production, $\Delta \geq 0$. We thus have an effective action description of the entire class of dissipative transport uncovered in \cite{Haehl:2014zda,Haehl:2015pja}.

\subsection{Class B}
\label{sec:classB}

The Class B action is also parameterized by a rank 4 super-tensor $\SF{\etaref}^{[(IJ) | (KL)]}_{_\text{(B)}}$ which is graded symmetric in the first two and second two pairs of indices, but is graded anti-symmetric under swapping the two pairs. This is a straightforward generalization of the discussion in \cite{Haehl:2015pja}. However, now we have to ensure that the Lagrangian density is {\sf CPT} even, which means unlike in the case of the Class D terms we are not going to include explicit factors of $i$. The only other choice we have is to use the ghost number zero field strength component $\Fs_{\thb \theta}$ which leads to the following action for Class B:
\begin{equation}
\boxed{
\begin{split}
\Lag_\text{wv}^{(\text{B})} &= 
	\int\,d\theta\, d\thetab\, \frac{\sqrt{-\SF{\gref}}}{\zsf} \; 
	\SF{\Lagref}_{(\text{B})} \\
	& \text{ with } 
	\quad 
	\SF{\Lagref}_{(\text{B})} =
	 \frac{1}{4}\, (-)^{(I+J)(K+L)  + IJ+KL+K+L}\,\Fs_{\thb\theta} \, \SF{\etaref}^{[(IJ)|(KL)]}_{_\text{(B)}}\, \gpsib_{IJ} \; \gpsi_{KL} \,.
	 \end{split}
}
 \label{eq:LagB}
\end{equation}
The Class B tensor has the following symmetry properties:
\begin{equation}
\begin{split}
& \SF{\etaref}_{_\text{(B)}}^{[(IJ)|(KL)]} = (-)^{IJ}\; \SF{\etaref}_{_\text{(B)}}^{[(JI)|(KL)]} = (-)^{KL}\,\; \SF{\etaref}_{_\text{(B)}}^{[(IJ)|(LK)]}  \\
& \SF{\etaref}_{_\text{(B)}}^{[(IJ)|(KL)]} = -(-)^{(I+J)(K+L)} \; \SF{\etaref}_{_\text{(B)}}^{[(KL)|(IJ)]} \,.
\end{split}
\label{eq:etabcpt}
\end{equation} 
where the second line distinguishes this type of transport from Class D. {\sf CPT} transformations act as
\begin{equation}
\SF{\etaref}_{_\text{(B)}}^{[(IJ)|(KL)]} = (-)^{(I+J)(K+L)} \, [ \SF{\etaref}^{[(KL)|(IJ)]}_{_\text{(B)}} ]^{\text{\sf{\tiny CPT}}}  \,.
\end{equation}

Note that despite this {\sf CPT} transformation, the Lagrangian \eqref{eq:LagB} gives a {\sf CPT} invariant action due to the non-minimal coupling to $\SF{\mathscr{F}}_{\thb\theta}$ (which is by itself {\sf CPT} invariant, but will later on acquire a {\sf CPT} breaking expectation value $\langle \SF{\mathscr{F}}_{\thb\theta} | \rangle = \langle \phizT \rangle = -i$).

An abstract analysis completely analogous to the one we did for Class D in the previous subsection can be performed. In particular, an equation of the same form as \eqref{eq:Dcalc3} holds and leads to the following classical currents:
{\small
\begin{equation}
\begin{split}
\TEMref_{_\text{(B,class.)}}^{ab} &\equiv  - \frac{(i\,\phizT)^2}{2} \, \etaref^{[(ab)|(cd)]}_{_\text{(B)}} \, \lieD_\Kref \, \gref_{cd} \,, \\
\Nref_{_\text{(B,class.)}}^a &\equiv 0 \,.
\end{split}
\label{}
\end{equation}	
}\normalsize
The adiabaticity equation is satisfied for the classical currents with vanishing entropy production: 
\begin{equation}
\nabla_a \Nref_{_\text{(B,class.)}}^a - \frac{1}{2}\, \TEMref_{_\text{(B,class.)}}^{ab} \, \lieD_\Kref g_{ab} =   \Delta = 0  \,,
\end{equation}
due to the antisymmetry $\etaref^{[(ab)|(cd)]}_{_\text{(B)}} = - \etaref^{[(cd)|(ab)]}_{_\text{(B)}}$.

\subsection{Class $\GV$}
\label{sec:classHVbar}

We call as Class $\GV$ the non-dissipative transport\footnote{ The reader should keep in mind that even the classical terms in Class $\GV$ transport are still very much unexplored. The simplest known examples involve charged fluids at second order in the derivative expansion \cite{Haehl:2015pja}. We include a discussion of the expected structures here merely for completeness.} which is described by the following Lagrangian: 
\begin{equation}
\boxed{
\begin{split}
&\Lag_\text{wv}^{(\GV)} = 
	\int\,d\theta\, d\thetab\, \frac{\sqrt{-\SF{\gref}}}{\zsf} \; 
	\SF{\Lagref}_{(\GV)} \\
&\;\; \text{ with} 
	\quad 
	\SF{\Lagref}_{(\GV)} =
	\frac{i}{4}  \,(-)^{M(I+J+K+L)+I(J+K+L)+J(K+L)+K+L+KL} \, \SF{\mathfrak{C}}^{M(IJ)(KL)}\\
&\qquad\qquad\qquad\qquad \times \Big[ \left( \Dwv_M  \gpsib_{IJ} +(\Fs_{\theta M},  \SF{\gref}_{IJ} )_\Kref \right)\ \gpsi_{KL}  -  \gpsib_{IJ} \, \left( \Dwv_M \gpsi_{KL} +  (\Fs_{\thb M},  \SF{\gref}_{KL} )_\Kref \right)\Big] 
	 \end{split}
}
 \label{eq:HVbarLagN}
\end{equation}
with the symmetry $\SF{\mathfrak{C}}^{M(IJ)(KL)}=(-)^{(I+J)(K+L)}\,\SF{\mathfrak{C}}^{M(KL)(IJ)}$. 
Note that the relative minus sign in the bracket of \eqref{eq:HVbarLagN} is essential for distinguishing this action from a Class D action with derivative-valued dissipative tensor $\etaref$. Note further that we have written a combination of, on the one hand, terms involving derivatives of $\{\gpsib,\gpsi\}$, and, on the other hand, terms involving $\UT$ field strengths. Each set of terms would also be individually allowed, but lead to a somewhat unconventional structure of the currents. In order to match with \cite{Haehl:2015pja} we choose the parameterization \eqref{eq:HVbarLagN}.\footnote{ For instance, dropping the terms in involving field strengths would still give us Class $\GV$ energy-momentum tensors, but would shift the free energy entirely into ghost components of the free energy super-current. This would also be a consistent presentation of this class of transport in superspace.}

It is tedious (though easy) to keep track of even the fluctuation contributions in this expression. For illustration, we only write the classical terms. These take the form 
\begin{equation}
\Lag_\text{wv}^{(\GV)} 
= \sqrt{-\gref} \; \left\{ \frac{1}{2} \, \TEMref_{_{(\GV\text{,class.})}}^{ab} \, \source{h_{ab}}  -  \sAt{a}\, \Nref_{_{(\GV\text{,class.})}}^a +\text{fluctuations} + \text{ghosts} \right\} 
\end{equation}
with
\begin{equation}
\begin{split}
  \TEMref_{_{(\GV\text{,class.})}}^{ab}& = \frac{1}{2}\, (i \, \phizT) \left[ \nabla_f {\mathfrak{C}}^{f(ab)(cd)} \, \lieD_\Kref \gref_{cd}  + 2 \, {\mathfrak{C}}^{f(ab)(cd)} \, \nabla_f  \lieD_\Kref \gref_{cd}  \right] 
  + \frac{i}{2} \,  {\mathfrak{C}}^{f(ab)(cd)} \,( \partial_f \phizT) \lieD_\Kref \gref_{ab}
   \,,\\
  \Nref_{_{(\GV\text{,class.})}}^{f} &= \frac{1}{4}\,(i \, \phizT) \,  {\mathfrak{C}}^{f(ab)(cd)} \, \lieD_\Kref \gref_{ab}\, \lieD_\Kref \gref_{cd}    \,.
\end{split}
\end{equation}
These currents precisely reproduce the discussion of \cite{Haehl:2015pja}.\footnote{ As usual, after setting $\langle \phizT \rangle = -i$.}  We can again explicitly verify that the super-adiabaticity equation is upheld: 
\begin{equation}
 \nabla_a \,\Nref_{(\GV\text{,class.})}^a = \frac{1}{2}\,\TEMref_{_{(\GV\text{,class.})}}^{ab} \, \lieD_\Kref \gref_{ab} \,.
\end{equation}
Note that this equation holds even before setting $\langle \phizT \rangle = -i$. This is, of course, required since super-adiabaticity is a statement of $\UT$ invariance of the theory independent of spontaneous symmetry breaking. 

\subsection{Class C}
\label{sec:classC}

Let us finish with the most trivial class of transport, Class C, where there is no energy-momentum transport at all; only the free energy current is non-zero (but identically conserved). The action describing such transport is just:
\begin{equation}\label{eq:LagC}
\boxed{
\Lag_\text{wv}^{(\text{C})} = - \int\,d\theta\, d\thetab\, \frac{\sqrt{-\SF{\gref}}}{\zsf} \, \As_I\, \SF{\Nref}_{_\text{(C)}}^I  
\qquad\text{for any}\quad \Dwv_I \SF{\Nref}_{_\text{(C)}}^I= 0 \,.
}
\end{equation}
This exact conservation guarantees that the action \eqref{eq:LagC} is invariant under the symmetries (in particular under $\UT$). 

Evaluating the superspace integral, we find
\begin{equation}
\begin{split}
\Lag_\text{wv}^{(\text{C})} &= \partial_\theta \partial_\thb \left\{ -\frac{\sqrt{-\SF{\gref}}}{\zsf} \; \As_I\, \SF{\Nref}_{_\text{(C)}}^I  \right\} \bigg{|} \\
&=  \sqrt{-{\gref}}\, \left\{-  \left( \partial_\theta \partial_\thb  \As_a \right)\, {\Nref}_{_\text{(C)}}^a - \left( \partial_\thb \As_\theta \right) \partial_\theta \SF{\Nref}_{_\text{(C)}}^\theta + \left( \partial_\theta \As_\thb \right) \partial_\thb \SF{\Nref}_{_\text{(C)}}^\thb  \right\} \Big{|} + \ldots \\
&= \sqrt{-{\gref}}\, \left\{-  \sAt{a}\,  \Nref_{_\text{(C)}}^a -(\phizT-\sBdel) \,\partial_\theta {\Nref}_{_\text{(C)}}^\theta + \sBdel \, \partial_\thb {\Nref}_{_\text{(C)}}^\thb  \right\}  + \ldots
\end{split}
\end{equation}
where we are again employing the MMO limit and dropping ghost bilinears. The resulting adiabatic constitutive relations simply yield an identically conserved Noether current and a vanishing energy-momentum tensor.

\section{Discussion}
\label{sec:discussion}

In this paper, we have endeavored to provide a comprehensive (and largely self-contained) discussion of hydrodynamic effective field theories. The primary aim was to fill in some of the details in our earlier discussion \cite{Haehl:2015uoc}, and work in a fully superspace covariant formalism. We have shown how to write down the effective action Eq.~\eqref{eq:fullLag} constrained by thermal equivariance which captures not only the classical hydrodynamic transport, but also the fluctuations about the hydrodynamic state.

Explicit computations backing these statements have been carried out in part in a particular limit, the MMO limit inspired from \cite{Mallick:2010su}, where the emergent thermal $\UT$ gauge field is treated as a background gauge field. In this limit, we demonstrated how to construct actions up to second order in the gradient expansion for neutral fluids.  Furthermore, we  have extracted predictions for the hydrodynamic fluctuations that go hand-in-hand with dissipative transport, for the shear and bulk viscosity. Clearly, it would be interesting to verify these statements from a direct computation of the associated correlation functions in a long-wavelength limit. 

One of the important highlights of the construction relates to results that can be obtained abstractly at a structural level. We have argued that entropy production can be viewed as arising from an inflow mechanism operating in superspace. To appreciate this statement better, we remind the reader that in the standard picture of classical hydrodynamics, a local form of the second law is input as an axiom, with its origins left unexplained. While one can readily understand that the dynamics in the near-equilibrium limit should refer to conserved charges, the presence of an entropy current, or its Legendre transform, the free energy current, is quite mysterious. 

 The origins of entropy and how it gets produced in the low energy effective field theory is a challenge from the microscopic perspective. In part, this has to do with the observation that there is no obvious analog of entropy in the fine-grained microscopic sense, but such a quantity has to emerge in the infra-red. This is well illustrated by the fluid/gravity correspondence \cite{Bhattacharyya:2008jc}, where the conserved currents are well-defined (composite) operators in the boundary field theory, but the entropy current is obtained from the pullback of the area form on the black hole event horizon onto the boundary \cite{Bhattacharyya:2008xc}. Despite the locality inherent in the fluid/gravity map  (see \cite{Hubeny:2011hd} for a review). There is still is a need to explain the rationale behind the existence of a local entropy current. In \cite{Haehl:2015foa} we argued that the thermal equivariance provides us with a mechanism for understanding this emergence.

 There was however an important question that was not satisfactorily answered at that point: how could it be that the physical system leads to entropy production in the macroscopic limit, and at the same time the entropy current  be the conserved Noether current of a gauge symmetry? The earlier arguments that this had to work had to do with indirect arguments relating to the Jarzynski relation, etc. However, the cleanest answer to this question is elegantly captured by the \emph{entropy inflow} mechanism that we have described here. As argued the thermal equivariance symmetry leads to a conservation law, the super-adiabaticity equation which operates at the level of full superspace. Decomposing this equation into the physical spacetime directions and the super-directions leads to a presentation where one can explicitly see the origins of entropy production -- it is sourced in superspace. 

 We have referred to this picture as an inflow mechanism for entropy in analogy with the manner inflow operates in the context of anomalies \cite{Callan:1984sa}. Starting with an anomalous theory, one constructs an anomaly-free theory, by viewing the physical theory as the boundary dynamics of a topological field theory (eg., the chiral edge states at the boundary of a quantum Hall insulator). In the hydrodynamic effective field theories, the upgrade to superspace with thermal equivariance, introduces a topological sector and associated Grassmann-odd partners of the entropy current, which end up in effect sourcing physical entropy production, when a field strength background is turned on. This suggests that the Class D terms are  analogous to Wess-Zumino-Witten terms in superspace. It would be interesting to try to make this analogy more precise. Entropy inflow serves as an important consistency check of our formalism.  Importantly, it guarantees that there is no in-principle obstacle to gauging the thermal $\UT$ symmetry, which was one of our primary hypotheses. 

 Somewhat curiously there is no corresponding inflow of energy-momentum from superspace at the classical level. This is of course necessary for the formalism to reproduce the correct dynamics. We  find that the non-fluctuation contributions cancel amongst themselves when we sum over all the Grassmann-odd directions in the equations of motion. These components of the energy-momentum tensor however do end up contributing fluctuation terms. This supports the hypothesis that the thermal $\UT$ dynamics controls primarily entropy production and  associated fluctuations. It would be good to understand this at an abstract level more generally. 

 There are several questions that we have not addressed in our formalism so far. One obvious extension is to consider the presence of other conserved charges.  For flavour symmetry with source $A_\mu$ in the physical spacetime, the Goldstone modes include a flavour group element $c(\sigma)$, which maps points on the flavour bundle of the worldvolume onto the physical target flavour bundle. The chemical potential is defined by pushing-forward the worldvolume thermal twist $\Lref$. Incorporating the desired super-transformations  one upgrades  $c(\sigma)$ to a superfield $\SF{c}(z)$.  This gives  a worldvolume  pull-back flavour  gauge field $\SF{\Aref}_I$ defined by the map $\{\SF{X}^\mu, \SF{c} \}$ which may further be deformed by unaligning sources. The crucial item to note is the  $\UT$ transformation on the flavour superfield, which is given by $(\LamS,\SF{c})_\Kref = \LamS \, \SF{c}\, \prn{\SF{\Lambda}_\Kref + \SF{\Kref}^I \SF{\Aref}_I}$. 
Coupled with target space flavour symmetry acting on $\SF{c}$, we should be able to include this into our effective actions along the lines of the sigma models constructed in \cite{Haehl:2015pja}. 

The extension to flavour charges is necessary to address two of the classes of transport that we have not tackled in our discussion: those relating to anomalies (Classes A and $\PV$). We note here that the construction of anomalous effective actions provided us with the first hints of the existence of the doubling of degrees of freedom \cite{Haehl:2013hoa}. A shadow of the $\UT$ symmetry was also seen in the construction of hydrostatic partition functions for mixed anomalies in \cite{Jensen:2013rga}. The close similarity between the constructions of the topological sigma models for hydrodynamics, and the empirical analysis of Class $\LT$ Lagrangians in \cite{Haehl:2015pja} leads us to believe that the extension to recover the two classes should be straightforward; it would however be useful to undertake the exercise explicitly. 

The other missing element in our discussion is the precise nature of the gauge dynamics for the $\UT$ gauge field. We have given arguments in favour of this being described by a topological $BF$ system, which would be interesting to verify in detail. From a physical perspective, the dynamical nature of the $\UT$ gauge field allows us to probe phases of the theory where the {\sf CPT} symmetry is not spontaneously broken. Understanding the dynamics is of course necessary to also argue that the physical hydrodynamic system with dissipation is picked out by the vacuum that dynamically leads to the thermal gauge field strength $\Fs_{\theta\thb}| $ acquiring a non-trivial expectation value, $\vev{ i\, \Fs_{\theta\thb}| } =1$. 

As mentioned in the introduction, the formalism we have described in this work raises many
interesting and intriguing questions in holography and black hole physics. What is the AdS counterpart
of the emergent gauge super fields in CFT fluid dynamics? What would be the gravity description of the 
fact that {\sf CPT} is spontaneously broken by the expectation value of a super field strength which then in turn drives the entropy inflow? It is tantalizing to speculate that the super field strength is in fact the order parameter characterizing the presence of black hole horizons (with the opposite expectation saddle points  describing the white hole saddle points of gravity Schwinger-Keldysh path integral). It would then follow that the superspace directions are the counterpart of  black-hole interior and the entropy inflow is the in-falling of lowest quasinormal mode. It would be nice to make these ideas precise.

As mentioned in footnote \ref{fn:utcpt} we have implicitly conflated the scale at which the $\UT$ symmetry emerges in the low energy theory with the scale at which {\sf CPT} symmetry is broken and dissipation ensues. It is however, more reasonable to anticipate that the $\UT$ symmetry arises at a higher energy scale. Once it does the notion of a local entropy current becomes sensible (even if it is conserved). Dissipation then will arise at a lower scale and at this point we have entropy production and {\sf CPT} symmetry  breaking. The relative separation between the scales is most likely given by the leading order dissipative coefficients like viscosity (made suitably dimensionless, say $\eta/s$)  This picture is suggested by the fact that the adiabaticity equation with $\Delta =0$ already holds in equilibrium, while dissipation is a higher order effect involving superspace inflow. Further support also can be gleaned from the fluid/gravity correspondence \cite{Hubeny:2011hd} -- the $\UT$ symmetry is naturally associated with the emergence of a horizon, while {\sf CPT} symmetry breaking relates to a choice of causal boundary condition on the horizon, cf., \cite{Haehl:2015foa}. A detailed verification of this picture is highly desirable, and it should be possible once we understand $\UT$  gauge dynamics. 

Our discussion was primarily focused on the high temperature or stochastic limit, which is sufficient for hydrodynamics. It would be interesting to ask how to understand the thermal equivariance structure at finite temperatures. Likewise, one can imagine a parallel discussion involving modular Hamiltonians for density matrices which are not necessarily of the Gibbs form. In both cases one is forced to confront a symmetry structure that has some intrinsic non-locality, which makes it an interesting challenge. The fluid/gravity correspondence and recent discussions between entanglement and geometry also suggests that there ought to be a similar structure hidden in the semi-classical gravitation theory for real-time observables, as speculated in \cite{Haehl:2015foa}.

Finally, we note that the class of fluctuations that we have understood about hydrodynamics are those that are computed by the Schwinger-Keldysh contour, which are thus only a subset of the full fluctuation relations. The response functions studied in hydrodynamics are related by the thermal KMS relations to other Wightman functions, which are not computed on the Schwinger-Keldysh path integral contour, but rather on a multi-timefolded, $k$-OTO contour (the Schwinger-Keldysh contour is the $1$-OTO contour) \cite{Haehl:2017eob}.  How the hydrodynamic effective field theory incorporates this information, and whether by a suitable upgrade one can construct effective field theories that capture more than the post local-thermalization behaviour  (eg., by having effective descriptions for higher OTO correlators), remain fascinating open questions that are worth exploring.

\acknowledgments

It is a pleasure to thank Veronika Hubeny, Shiraz Minwalla, and Spenta Wadia for useful discussions.
FH gratefully acknowledges support through a fellowship by the Simons Collaboration `It from Qubit', and hospitality by UC Davis, McGill University, University College London, UC Santa Barbara, and UT Austin during the course of this project. RL gratefully acknowledges support from International Centre for Theoretical Sciences (ICTS), Tata institute of fundamental research, Bengaluru. RL would also like to acknowledge his debt to all those who have generously supported and encouraged the pursuit of science in India. MR is supported in part by U.S.\ Department of Energy grant DE-SC0009999 and by the University of California. 
MR would like to thank TIFR, Mumbai and ICTS, Bengaluru, as well as the organizers of the ``It from Qubit workshop'' at Bariloche, and  ``Chaos and Dynamics in Correlated Quantum Matter'' at Max Planck Institute for the Physics of Complex Systems, Dresden for hospitality during the course of this work.

\appendix

\newpage
\part{Thermal supergeometry: Mathematical details}
\section{The thermal SK-KMS superalgebra}
\label{sec:skkmsalg}

The  SK-KMS superalgebra described in \S\ref{sec:thermaleq} is generated by the SK supercharges $\{\QSK,\QSKb\}$ and the thermal supermultiplet of KMS operators $\{\QKMS,\QKMSb,\Qbeta,\Qzero\}$  which satisfy the following algebra:
\begin{gather}
\QSK^2 =\QSKb^2 = \QKMS^2=\QKMSb^2 = 0\ , \nonumber\\
\gradcomm{\QSK}{\QKMS} =  \gradcomm{\QSKb}{\QKMSb} = \gradcomm{\QSKb}{\QSK} = \gradcomm{\QKMS}{\QKMSb} = 0\ ,\nonumber \\
\gradcomm{\QSK}{\QKMSb} =  \gradcomm{\QSKb}{\QKMS} = \Qbeta\,, \label{eq:kmsalg}\\
\gradcomm{\QKMS}{\Qzero} = \gradcomm{\QKMSb}{\Qzero} = 0 \,,\nonumber\\
\gradcomm{\QSK}{\Qzero} =  \QKMS\,, \qquad \gradcomm{\QSKb}{\Qzero} = -\QKMSb\,.\nonumber
\end{gather}
This is a graded superalgebra with a conserved ghost number charge, and can be understood in terms of an equivariant cohomology algebra as described in earlier works \cite{Haehl:2016uah}.

The canonical way to interpret this structure is to extend the operator algebra of the quantum system of interest to an operator superalgebra which lives in a two dimensional superspace (locally ${\mathbb R}^{d-1,1|2}$).  The quadruplet of operators $\{\SKAv{O}, \SKG{O}, \SKGb{O}, \SKDif{O}\}$ associated with a single-copy operator $\OpH{O}$ is encapsulated into a single SK-superfield
\begin{equation}\label{eq:OpO}
\SF{\Op{O}} \equiv \SKAv{O} + \theta \, \SKGb{O} + \bar\theta \, \SKG{O} + \bar\theta\theta \, \SKDif{O} \,.
\end{equation}
Here $\theta, \thb$ are the Grassmann coordinates in superspace with $\gh{\theta} =1$ and $\gh{\thb}=-1$, respectively.

 In superspace, the operators $\{\QSK,\QSKb\}$, which are intrinsic to the Schwinger-Keldysh formalism, are realized as derivations along the Grassmann-odd directions:
\begin{equation}
\QSK \;\longrightarrow\;  \partial_{\thb}\,, \qquad \QSKb \; \longrightarrow \;  \partial_\theta\,.
\end{equation}
These derivations induce the action of $\{\QSK,\QSKb\}$ component-wise on superfields; for example, we have $\gradcomm{\QSK}{\SF{\mathbb{O}}} = \partial_\thb \SF{\mathbb{O}} = \SKGb{O} + \theta\, \SKDif{O}$. We refer to \cite{Haehl:2016pec,Haehl:2016uah} for more details. The KMS operators can be best understood in terms of a quadruplet of thermal translation operators. Let, 
\begin{equation}\label{eq:SKsuperops}
\begin{split}
   \IKMSzero &= \Qzero + \thb \, \QKMS - \theta \, \QKMSb + \thb \theta \, \Qbeta \,,\\
   \IKMS &= \QKMS + \theta \, \Qbeta \,,\\
   \IKMSb &= \QKMSb + \thb \, \Qbeta \,,\\
   \LKMS &=  \Qbeta \,,
\end{split}
\end{equation}
such that 
\begin{equation}
\IKMSzero \, \SF{\mathbb{O}} = \IKMS \, \SF{\mathbb{O}} = \IKMSb \, \SF{\mathbb{O}}= 0 \,.
\end{equation}

The connection between the SK-KMS algebra and the equivariant cohomology algebras (see Appendix \ref{sec:EqReview}) can be made manifest by noting that the thermal operators transform under the super-derivations into each other via, 
\begin{equation}
\begin{tikzcd}
&\IKMSzero \arrow{ld}{\partial_\thb} \arrow{rd}[below]{\partial_\theta}    &   \\
\IKMS\arrow{rd}{\partial_\theta} & &  -\IKMSb \arrow{ld}[above]{\!\!\!\!\!\!\!\!\!\!\!\!\!\!-\partial_\thb}\\
&   \LKMS &
\end{tikzcd}
\label{eq:QzeroDiag}
\end{equation}
One can equivalently express this directly in terms of the super-commutation relations as: 
\begin{gather}
\partial_\thb^2 =\partial_\theta^2 = (\IKMS)^2=  (\IKMSb)^2= 0\ , \nonumber\\
\gradcomm{\QSK}{\IKMS} =  \gradcomm{\QSKb}{\IKMSb} = \gradcomm{\QSKb}{\QSK} = \gradcomm{\IKMS}{ \IKMSb}= 0\ ,\nonumber \\
\gradcomm{\QSK}{ \IKMSb} =  \gradcomm{\QSKb}{\IKMS} = \LKMS \,, \label{eq:kmsalg2}\\
\gradcomm{\LKMS}{ \IKMSzero} = \gradcomm{ \IKMSb}{ \IKMSzero} = 0 \,,\nonumber\\
\gradcomm{\QSK}{ \IKMSzero} = \IKMS\,, \qquad \gradcomm{\QSKb}{ \IKMSzero} =  -\IKMSb\,.\nonumber
\end{gather}
%

\section{Review of equivariant cohomology}
\label{sec:EqReview}

We briefly review here (the superspace formulation of) ${\cal N}_\smallT=2$ extended equivariant cohomology mainly to set-up some of the background structure and notation. Our discussion will completely general and simply exemplifies the algebraic structures, which will be relevant for the thermal Schwinger-Keldysh theories. These concepts were originally developed in \cite{Cordes:1994sd,Blau:1991bn,
Vafa:1994tf,Dijkgraaf:1996tz,Blau:1996bx} and have been reviewed in our earlier paper \cite{Haehl:2016uah} where the reader can find more background and references. 

\subsection{Weil model}
\label{sec:eqweil}

The essential idea in equivariant cohomology is to develop an algebraic formalism to tackle cohomology on spaces with a group action (so that one is often interested in non-trivial gauge invariant data). The first step is to construct appropriate cohomological charges and representations for the gauge sector, which can be done by working in the universal covering space of the group (parameterized algebraically by the space of gauge connections). Exploiting the isomorphism between differential forms and fermions, one can package the relevant fields of the  gauge sector (the Weil complex) into a super gauge field with components: 
\begin{subequations}
\begin{align}
\As_a \equiv&\
	\Ascr_a + \thb  \bigbr{ \lambda_a + D_a G  } + \theta \bigbr{  \bar{\lambda}_a + D_a \bar{G}  }
\nonumber \\
&
\qquad	+\;\thb   \theta \;\Bigl\{ \sAt{a} +D_a B + \comm{G}{\bar{\lambda}_a}
	- \comm{\bar{G}}{\lambda_a}
	+\half \comm{G}{D_a \bar{G} }-  \half \comm{\bar{G} }{D_a G} \Bigr\}
\label{eq:AaExp2} \\
\As_\thb \equiv&\
	 G + \thb  \bigbr{ \phi-  \half \comm{G}{G}  } + \theta \bigbr{ B-  \half \comm{\bar{G} }{G} }
\nonumber \\
&
\qquad	 -\; \thb  \theta \bigbr{ \eta +  \comm{\bar{G} }{\phi} - \comm{G}{B }+ \half \comm{G}{\comm{\bar{G}}{G}}   }
\label{eq:AthbExp2}  \\
\As_\theta \equiv&\
	\bar{G}  + \theta  \bigbr{ \bar{\phi}-  \half \comm{\bar{G} }{\bar{G} }  }
	+ \thb  \bigbr{\phi^0- B-  \half \comm{\bar{G} }{G} }
\nonumber \\
&
\qquad	+\;\thb   \theta \bigbr{\etab +  \comm{G}{\bar{\phi}} -  \comm{\bar{G} }{\phi^0- B}
	+ \half  \comm{\bar{G} }{ \comm{\bar{G} }{G}}  } \,.
\label{eq:AthExp2}
\end{align}
\end{subequations}
All components of these gauge fields are valued in the Lie algebra of the equivariant gauge group. For example, $G \equiv G^i t_i$ with $\{t_i\}_{i=1,\ldots,n}$ the generators of the Lie algebra.

The algebra is generated by two Weil charges, which are represented on superspace as generators of translations in the Grassmann odd directions:
\begin{equation}
\QW = \partial_{\thb}(\ldots)| \,, \qquad \QWb = \partial_\theta(\ldots)|\,.
\label{}
\end{equation}
Further, we have interior contractions $\{{\cal I}_k,{\cal I}_k^0,\Ibar_k\}$, which project in the gauge directions. The defining equations for these operators are the following non-vanishing actions:
\begin{align}
\Ibar_k G^j = - \delta^j_k  \,, &\qquad \Ibar_k B^j = -\frac{1}{2}\, f^j_{k\ell} \,\bar{G}^\ell \,,\\
{\mathcal{I}}_k \bar{G}^j = - \delta^j_k  \,, &\qquad  {\mathcal{I}}_k B^j = \frac{1}{2} \, f^j_{k\ell}\, G^\ell \,,\\
\mathcal{I}^0_k &B^j  = -\delta^j_k \,.
\label{}
\end{align}
These equations should be interpreted as follows: the {\it `Faddeev-Popov ghosts'}  $\{G^i,B^i,\bar{G}^i\}$ are generators of the gauge directions with different ghost numbers ($\{1,0,-1\}$, respectively); the interior contractions project along them.
A Lie derivative operation (which generates gauge transformations) can now be defined in terms of the interior contractions and the Weil charges:
\begin{equation}
  \mathcal{L}_j \equiv \gradcomm{\QW}{\Ibar_j} = \gradcomm{\,\QWb}{{\mathcal{I}}_j}  \,.
\end{equation}
The generators introduced above generate a graded algebra:
\begin{equation}
\begin{split}
\QW^2 = \QWb^2  &=\gradcomm{\QW}{\QWb} = 0 \\
\gradcomm{\QW}{\Ibar_j} = \gradcomm{\,\QWb}{{\mathcal{I}}_j}
= \mathcal{L}_j  \,, & \qquad
\gradcomm{\QW}{{\mathcal{I}}_j} = \gradcomm{\,\QWb}{\Ibar_j} = 0\\
\gradcomm{\QW}{\mathcal{I}^0_j} = {\mathcal{I}}_j\,, & \qquad
  \gradcomm{\,\QWb}{\mathcal{I}^0_j} =- {\Ibar}_j \\
 \gradcomm{\QW}{\mathcal{L}_j} &= \gradcomm{\,\QWb}{\mathcal{L}_j} = 0 \\
 \gradcomm{\Ibar_i}{{\mathcal{I}}_j} &= f_{ij}^k\, \mathcal{I}^0_k \\
 \gradcomm{\mathcal{L}_i}{\Ibar_j} = -f_{ij}^k\,\Ibar_{k}\,, \qquad
   \gradcomm{\mathcal{L}_i}{{\mathcal{I}}_j} &= -f_{ij}^k\,{\mathcal{I}}_{k}\,, \qquad
   \gradcomm{\mathcal{L}_i}{\mathcal{I}^0_j} = - f_{ij}^k\,\mathcal{I}^0_{k}
\end{split}
\label{eq:LieSAlg2}
\end{equation}

This algebra can immediately be lifted to superspace by defining the following super-interior contractions and super-Lie derivative:
\begin{equation}\label{eq:NT2superops}
\begin{split}
  \SF{{\cal I}}^0_k & \equiv {\cal I}^0_k  + \thb \, {\cal I}_k- \theta\, \Ibar_k + \thb \theta \, {\cal L}_k \,,\\
  \SF{\cal I}_k &\equiv {\cal I}_k + \theta \, {\cal L}_k \,,\\
  \SF{\Ibar}_k &\equiv \Ibar_k + \thb \, {\cal L}_k \,,\\
  \SF{{\cal L}}_k &\equiv {\cal L}_k \,.
\end{split}
\end{equation}
The graded algebra in superspace can be written as:
\begin{equation}
\begin{split}
 \gradcomm{\QWb}{\SF{{\cal I}}^0_k} &= \partial_\theta \SF{{\cal I}}^0_k = - \SF{\Ibar}_k \,,\qquad
 \gradcomm{\QW}{\SF{{\cal I}}^0_k} = \partial_\thb \SF{{\cal I}}^0_k = \SF{{\cal I}}_k \,, \quad\;\;\\
\gradcomm{\QW}{\SF{\Ibar}_k} &= \partial_\thb \SF{\Ibar}_k = \gradcomm{\QWb}{\SF{{\cal I}}_k} = \partial_\theta \SF{{\cal I}}_k = \SF{\cal L}_k \,, \\
  \gradcomm{\QWb}{\SF{\Ibar}_k} &= \gradcomm{\QW}{\SF{{\cal I}}_k} = \gradcomm{\SF{{\cal I}}^0_j}{\SF{{\cal I}}^0_k} = 0 \,,
\end{split}
\end{equation}
This is simply an explicit way of writing the simple statement that the Weil charges $\{\QW,\QWb\}$ act as super derivations $\{\partial_{\thb},\partial_\theta\}$ on the super operations \eqref{eq:NT2superops}.

\subsection{Cartan model}
\label{sec:eqcartl}

A slightly different model of equivariant cohomology is the Cartan model, where one trades $\{\QW,\QWb\}$ for the covariant Cartan charges $\{\QC,\QCb\}$. In superspace, one would phrase this as the statement that $\{\partial_\thb,\partial_\theta\}$ are being traded for super-covariant derivatives $\{\Dut_\thb,\Dut_\theta\}$, which generate the action of Cartan charges via
\begin{equation}
\QC = \Dut_\thb(\ldots)| \,, \qquad \QCb = \Dut_\theta (\ldots)| \,.
\end{equation}
The super-covariant derivatives are defined with respect to the gauge superfield as usual:
\begin{equation}
 \Dut_I = \partial_I + [\As_I,\ \cdot\;] \,.
\label{eq:covDIa}
\end{equation}
After passing to the Cartan model, it is consistent to fix a gauge where $G^i=\bar{G}^i=B^i=0$. We will refer to this as Wess-Zumino gauge. In this gauge, one can then check that the Cartan charges can be written as:
\begin{equation}
\begin{split}
\QC &= \QW + \phi^k \, \Ibar_k + (\phi^0)^k \, {\cal I}_k +  \eta^k\, \mathcal{I}^0_k \,,\\
\QCb &= \QWb +  \overline{\phi}^k \, {\mathcal{I}}_k \,.
\end{split}
\label{eq:QCQCb}
\end{equation}
From this representation one can easily check the property that they square to gauge transformations along the covariant field strength components:
\begin{equation}
\QC^2 = \phi^k \, {\cal L}_k \,,\qquad
\QCb^2 = \overline{\phi}^k \, {\cal L}_k \,.
\end{equation}
In superspace, the full Cartan algebra can be stated succinctly as follows:
\begin{equation}
\begin{split}
 \Dut_\thb^2 = \SF{\mathcal{L}}_{\SF{\mathscr{F}}_{\thb\thb}} \,, \qquad& \Dut_\theta^2 = \SF{\mathcal{L}}_{\SF{\mathscr{F}}_{\theta\theta}}  \,,
\qquad
\gradcomm{\Dut_\thb}{\Dut_\theta} = \SF{\mathcal{L}}_{\SF{\mathscr{F}}_{\theta\thb}} \,, \\
\gradcomm{\Dut_\theta}{\SF{\mathcal{L}}_{\SF{\Lambda}}} = \SF{\mathcal{L}}_{\Dut_\theta \SF{\Lambda}} \,, \qquad &\gradcomm{\Dut_\thb}{\SF{\mathcal{L}}_{\SF{\Lambda}}} = \SF{\mathcal{L}}_{\Dut_\thb \SF{\Lambda}} \,, \qquad \gradcomm{ \SF{\cal L}_{\SF{\Lambda}_1}}{\SF{\cal L}_{\SF{\Lambda}_2}} = \SF{\mathcal{L}}_{ \comm{ \SF{\Lambda}_1 }{ \SF{\Lambda}_2 } } \,.
\end{split}
\label{eq:sCa2}
\end{equation}
%
 
\section{$\UT$ connection and representation theory}
\label{sec:UTreps}

We now turn to the structure of the $\UT$ symmetry and the machinery associated with thermal equivariance. As in the main text we take $z^I=\{\sigma^a,\theta,\thetab \}$ as the coordinates on the superspace. Here $\{\sigma^a\}$ denotes  Grassmann-even co-ordinates of the ordinary spacetime  whereas  $\{\theta,\thetab\}$ are Grassmann odd. The essential structure that adapts such a superspace to the purposes of thermal field theory is a background super-vector $\SF{\Kref}^I(z)$, which we will call as the \emph{thermal super-vector}. 

The notation and nomenclature are supposed to be suggestive -- it will later turn out that, in fluid dynamics
the thermal vector roughly points along the local fluid velocity and its norm is the inverse of local temperature (viz., the length of local thermal circle after a local Wick rotation in the  fluid rest frame). We take $\SF{\Kref}^I$ to have the same Grassmann parity as $z^I$. Thus, $\SF{\Kref}^a$ is Grassmann-even whereas  $\SF{\Kref}^\theta$ and $\SF{\Kref}^{\thetab}$ are Grassmann odd.

Let us also note that in the presence of flavor, we will need a Grassmann even \emph{thermal twist} $\SLref(z)$ which is the flavor twist imposed as we go around the local thermal circle after Wick rotation. Under a local flavor rotation $\SF{U}(z)$, it transforms  as $\SLref \mapsto  \SF{U}^{-1} \, \SLref  \,\SF{U} - 
 \SF{U}^{-1}\, \SF{\Kref}^I \, \partial_I \SF{U} $. In a flavor basis where the rest frame temporal component of the gauge field is zero,  $\SLref(z)$ can be thought of as the local flavor chemical potential divided by local temperature. For details on how these structures play a role in conventional fluid dynamics, please see  \cite{Haehl:2015pja}.

\subsection{Thermal commutator}

Say we have a tensor superfield $\SF{\mathcal{T}}^{I\ldots}{}_{J\ldots}$ which transforms in the flavor representation. This means under an infinitesimal super-diffeomorphism $\SF{\xi}^I$ and flavor transformation $\SF{\Lambda}^{(F)}$,  $\SF{\mathcal{T}}^{I\ldots}{}_{J\ldots}$ transforms as\footnote{  We will use anti-hermitian basis
for flavors and use $[\ ,\ ]$ to denote flavor adjoint action.}
\begin{equation}
\begin{split}
\SF{\mathcal{T}}^{I\ldots}{}_{J\ldots} 
&\mapsto 
	\SF{\mathcal{T}}^{I\ldots}{}_{J\ldots} -
	 \bigbr{ \lieD_{\SF{\xi}} \SF{\mathcal{T}}^{I\ldots}{}_{J\ldots} 
	 - [\SF{\Lambda}^{(F)}, \SF{\mathcal{T}}^{I\ldots}{}_{J\ldots}] }
 \end{split}
\end{equation}
where $\lieD_\xi$ is the left super-Lie derivative  (see \cite{DeWitt:1992cy} for details).

We can then introduce a \emph{thermal bracket} which acts via infinitesimal thermal translations and flavor twist, viz.,
\begin{equation}
\label{eq:tensorTrf}
(\SF{\Lambda},\SF{\mathcal{T}}^{I\ldots}{}_{J\ldots})_\Kref \equiv \SF{\Lambda} \bigbr{ \lieD_\Kref \SF{\mathcal{T}}^{I\ldots}{}_{J\ldots} - [\SLref, \SF{\mathcal{T}}^{I\ldots}{}_{J\ldots}] }
\end{equation}
where the superfield $\SF{\Lambda}=\SF{\Lambda}(z)$ is the infinitesimal parameter which parameterizes the thermal translation.

The result of a thermal bracket $(\SF{\Lambda},\SF{\mathcal{T}}^{I\ldots}{}_{J\ldots})_\Kref $ is another tensor superfield  which transforms in the same tensor/flavor representation as $\SF{\mathcal{T}}^{I\ldots}{}_{J\ldots}$. 
For example, under flavor transformations if we have $\SF{\mathcal{T}}^{I\ldots}{}_{J\ldots}\mapsto \SF{U}^{-1}\SF{\mathcal{T}}^{I\ldots}{}_{J\ldots}\SF{U}$, then it follows that $(\SF{\Lambda},\SF{\mathcal{T}}^{I\ldots}{}_{J\ldots})_\Kref \mapsto \SF{U}^{-1}(\SF{\Lambda},\SF{\mathcal{T}}^{I\ldots}{}_{J\ldots})_\Kref  \SF{U}$.

The set of such infinitesimal thermal translations form a super Lie algebra, viz.,
\begin{equation}
\prn{\SF{\Lambda}_1,(\SF{\Lambda}_2,\SF{\mathcal{T}}^{I\ldots}{}_{J\ldots})_\Kref}_\Kref
-(-)^{\SF{\Lambda}_1\SF{\Lambda}_2}\prn{\SF{\Lambda}_2,(\SF{\Lambda}_1,\SF{\mathcal{T}}^{I\ldots}{}_{J\ldots})_\Kref}_\Kref
= \prn{(\SF{\Lambda}_1,\SF{\Lambda}_2)_\Kref,\SF{\mathcal{T}}^{I\ldots}{}_{J\ldots}}_\Kref
\end{equation}
where
\begin{equation}\label{eq:u1T_closure}
(\SF{\Lambda}_1,\SF{\Lambda}_2)_\Kref \equiv \SF{\Lambda}_1 \lieD_\Kref \SF{\Lambda}_2 -(-)^{\SF{\Lambda}_1\SF{\Lambda}_2}\  \SF{\Lambda}_2 \lieD_\Kref \SF{\Lambda}_1\ .
\end{equation}
Further, we can show that a super-Jacobi identity holds, viz.,
\begin{equation}\label{eq:u1T_assoc}
\prn{\SF{\Lambda}_1,(\SF{\Lambda}_2,\SF{\Lambda}_3)_\Kref}_\Kref
-(-)^{\SF{\Lambda}_1\SF{\Lambda}_2}\prn{\SF{\Lambda}_2,(\SF{\Lambda}_1,\SF{\Lambda}_3)_\Kref}_\Kref
= \prn{(\SF{\Lambda}_1,\SF{\Lambda}_2)_\Kref,\SF{\Lambda}_3}_\Kref\ ,
\end{equation}
which can be rewritten in a cyclic form:
\begin{equation}\label{eq:u1T_Jacobi}
(-)^{\SF{\Lambda}_3\SF{\Lambda}_1}\prn{\SF{\Lambda}_1,(\SF{\Lambda}_2,\SF{\Lambda}_3)_\Kref}_\Kref
+(-)^{\SF{\Lambda}_1\SF{\Lambda}_2}\prn{\SF{\Lambda}_2,(\SF{\Lambda}_3,\SF{\Lambda}_1)_\Kref}_\Kref
+(-)^{\SF{\Lambda}_2\SF{\Lambda}_3} \prn{\SF{\Lambda}_3,(\SF{\Lambda}_1,\SF{\Lambda}_2)_\Kref}_\Kref =0\ .
\end{equation}
We will denote this group of infinitesimal thermal translations as $\UT$.

Note that from \eqref{eq:u1T_closure} we see that the $\UT$ parameter superfield $\SF{\Lambda}$ does not quite transform like a scalar superfield under thermal translations. In fact, it transforms like a $\UT$ thermal twist analogous to how $\Lref$  transforms as a flavor thermal twist with $\Lref\mapsto U^{-1} \Lref U- U^{-1}\lieD_\Kref U $ under a flavor transformation $U$. We will call a field which transforms as \eqref{eq:u1T_closure}  under $\UT$ as a $\UT$ \emph{adjoint superfield}. 

\subsection{$\UT$ connection}

The first item in our agenda is to understand such $\UT$ adjoint superfields and formulate a $\UT$ connection, curvature etc. We begin by introducing a superspace one-form gauge field $\SF{\Ascr}_I$ associated with $\UT$:
\begin{equation}
 \As = \As_I \, dz^I = \As_a \, d\sigma^a + \As_\theta \, d\theta + \As_{\thb} \, d\thb\,.
\label{eq:A1fsf}
\end{equation}
We will also define the corresponding two-form super-field strengths $\SF{\mathscr{F}}_{IJ}$  via
\[\SF{\mathscr{F}}_{IJ} \equiv (1-\half \delta_{IJ}) \bigbr{ \partial_I \SF{\Ascr}_J - (-)^{IJ}  \partial_J \SF{\Ascr}_I + (\SF{\Ascr}_I,\SF{\Ascr}_J)_\Kref}\ .\]
Here, the thermal bracket is taken by treating  $\SF{\Ascr}_I$ as a $\UT$ adjoint superfield, i.e.,
\begin{equation}
(\SF{\Ascr}_I,\SF{\Ascr}_J)_\Kref \equiv\SF{\Ascr}_I \lieD_\Kref \SF{\Ascr}_J -(-)^{IJ}\  \SF{\Ascr}_J \lieD_\Kref \SF{\Ascr}_I\ .
\label{eq:F1fsf}
\end{equation}
The extra factor of half in \eqref{eq:F1fsf} is introduced relative to the usual definition to save us various factors of half later on. Note that the factor of half is relevant only when  both the superspace indices are Grassmann odd and equal. Often it is useful to write this definition in the form
\[ (1+ \delta_{IJ})\SF{\mathscr{F}}_{IJ} \equiv  \partial_I \SF{\Ascr}_J - (-)^{IJ}  \partial_J \SF{\Ascr}_I + (\SF{\Ascr}_I,\SF{\Ascr}_J)_\Kref\ .\]
Finally, we define a non-covariant combination, which will sometimes be useful:
\begin{align}
\Bs _{\theta\thb} \equiv \partial_\theta\As_{\thb} + \frac{1}{2}\, (\As_\theta,\As_{\thb})_\Kref \,.
\end{align}

Let us pause for a moment to count the number of independent (real space) fields that we have introduced. Clearly, the total number of fields in the super gauge connection \eqref{eq:A1fsf} is 12. Let us hence give 12 names to these independent fields so we can refer to them more easily. We organize these 12 fields as follows:

\begin{itemize}
\item {\bf Faddeev-Popov triplet:} out of the 12 fields, there are 3 combinations which are not $\UT$ covariant. We identify them as
\begin{align}
\As_{\thb}|  \equiv  \GT \,, \qquad \As_\theta| \equiv \GbT \,, \qquad \Bs_{\theta\thb}| \equiv \BT \,.
\label{eq:triplet}
\end{align}
\item {\bf Vafa-Witten quintet:} a further 5 fields are given by
\begin{align}
& \Fs_{\thb\thb}| \equiv \phi\,,\qquad \Fs_{\theta\theta}|\equiv \overline{\phi} \,,\qquad
  \Fs_{\theta\thb}| \equiv \phi^0 \,,
\nonumber \\
&\qquad\Dut_{\thb} \Fs_{\theta\thb}| \equiv \etaT\,, \qquad
\Dut_{\theta} \Fs_{\theta\thb}| \equiv \etabT\,,
\label{eq:quintet}
\end{align}
where we defined the $\UT$ covariant derivative (on a flat superspace ${\mathbb R}^{d-1,1|2}$)
\begin{equation}
\Dut_I \SF{\mathscr{F}}_{JK}  \equiv \partial_I \SF{\mathscr{F}}_{JK}
+ ({\SF{\Ascr}}_I, \SF{\mathscr{F}}_{JK} )_\Kref \,.
\end{equation}
This covariant derivative has a number of useful properties which we will discuss later in great detail.
\item {\bf Vector quartet:} Finally, we have 4 covariant vectors:
\begin{align}
\As_a| \equiv \Ascr_a \,, \qquad \Fs_{\thb a}| \equiv \lambda_a \,, \qquad \Fs_{\theta a}| \equiv \overline{\lambda}_a \,, \qquad
\Dut_\theta \Fs_{\thb a}| \equiv \sAt{a} \,.
\label{eq:quartet}
\end{align}
\end{itemize}
This provides a brief and abstract discussion of the components of the gauge superfield that we use. 
We will return to it in more details in Appendix~\ref{sec:GaugeMult}. 

However, before turning to this discussion, it is useful to understand features of the $\UT$ representations and adjoint fields a bit better.  We will for the most part be interested only in objects transforming as in \eqref{eq:tensorTrf} which we refer to as $0$-adjoints, and the gauge multiplet, which will transform as $1$-adjoints or simply adjoints (see below). However, other representations labeled by an integer appear to be admissible, and are detailed in Appendix~\ref{sec:formal}.

\section{Superspace representations: adjoint superfield}
\label{sec:Adjoint}

In this appendix we work out the detailed superspace expansions of adjoint superfields and their derivatives, elaborating in part on the material described in \cite{Haehl:2016uah}.

\subsection{Basic definitions and component fields}

Given a $\UT$ adjoint superfield $\SF{\mathfrak{F}}$, we can define a covariant derivative on it via
\begin{equation}\label{eq:DFs}
 \Dut_I \SF{\mathfrak{F}} \equiv \partial_I \SF{\mathfrak{F}} + (\SF{\Ascr}_I, \SF{\mathfrak{F}})_\Kref\,.
\end{equation}
The superfield $\Dut_I \SF{\mathfrak{F}} $ is a $\UT$ adjoint valued super-one-form.
The  thermal bracket for adjoint superfields  is defined via
\begin{equation}
\begin{split}
 (\SF{\mathfrak{F}},\SF{\mathfrak{F}}')_\Kref &\equiv \SF{\mathfrak{F}}\ \lieD_\Kref \SF{\mathfrak{F}}' -(\lieD_\Kref \SF{\mathfrak{F}})\ \SF{\mathfrak{F}}'\,.
\end{split}
\end{equation}
for arbitrary $\UT$ adjoint fields $\SF{\mathfrak{F}},\SF{\mathfrak{F}}'$.  The bracket appearing in \eqref{eq:DFs} acts in the same way.

We can use this to define the covariant components of $\SF{\mathfrak{F}}$ via
\begin{equation}
\begin{split}
 \mathfrak{F} &\equiv \SF{\mathfrak{F}}|\ ,\quad
 \mathfrak{F}_{\psib} \equiv \Dut_\theta \SF{\mathfrak{F}} | \ ,\quad
 \mathfrak{F}_\psi \equiv \Dut_{\thetab } \SF{\mathfrak{F}} |\ ,\quad
 \tilde{\mathfrak{F}} \equiv \Dut_\theta \Dut_{\thetab } \SF{\mathfrak{F}}|\ .
 \end{split}
 \label{eq:covadj}
\end{equation}
We will often focus just on these covariant components of superfields and write them in a vectorial form:
\begin{equation}
 \SF{\mathfrak{F}} \ :\ \prn{\begin{array}{c} \mathfrak{F} \\ \mathfrak{F}_\psi \\\mathfrak{F}_{\psib} \\ \tilde{\mathfrak{F}}
 \end{array}} \,.
 \label{eq:adjcol}
\end{equation}
The action of various covariant derivative operators can then simply be written as a map between such column vectors of covariant superfield components. The conversion of such a vector into an explicit superspace expression is given by
\begin{equation}
\begin{split}
\SF{\mathfrak{F}} &=\mathfrak{F} + \thetab\,  \bigbr{\mathfrak{F}_\psi - (\GhT,\mathfrak{F})_\Kref   }
+ \theta\, \bigbr{\mathfrak{F}_{\psib}- (\GhbT,\mathfrak{F})_\Kref  } \\
&\quad + \thetab   \theta\ \bigbr{\tilde{\mathfrak{F}} - (\BT,\mathfrak{F})_\Kref+ (\GhT,\mathfrak{F}_{\psib})_\Kref-(\GhbT,\mathfrak{F}_\psi)_\Kref +\half (\GhbT,(\GhT,\mathfrak{F})_\Kref)_\Kref - \half (\GhT,(\GhbT,\mathfrak{F})_\Kref)_\Kref\ } \,.
\end{split}
\label{eq:adjsf}
\end{equation}
The two representations of the superfield $\SF{\mathfrak{F}}$ given above, are equivalent.

Given two $\UT$ adjoint superfields $\mathfrak{F}$ and $\mathfrak{F}'$, their thermal bracket
$(\mathfrak{F},\mathfrak{F}')_\Kref$ is also a $\UT$ adjoint superfield with components
\begin{equation}
\begin{split}
 (\SF{\mathfrak{F}}, \SF{\mathfrak{F}}')_\Kref \ :\ \prn{\begin{array}{c} (\mathfrak{F},\mathfrak{F}')_\Kref \\  (\mathfrak{F},\mathfrak{F}_{\psi}')_\Kref+(\mathfrak{F}_{\psi},\mathfrak{F}')_\Kref \\
 (\mathfrak{F},\mathfrak{F}_{\psib}')_\Kref+(\mathfrak{F}_{\psib},\mathfrak{F}')_\Kref \\
 (\mathfrak{F},\tilde{\mathfrak{F}}')_\Kref + (\tilde{\mathfrak{F}},\mathfrak{F}')_\Kref+(\mathfrak{F}_{\psib},\mathfrak{F}_\psi')_\Kref -(\mathfrak{F}_{\psi},\mathfrak{F}_{\psib}')_\Kref
 \end{array}}
\end{split}
\end{equation}
The covariant derivative $\Dut_I $ then distributes over this thermal bracket, viz.,
\begin{equation}
\begin{split}
 \Dut_I(\SF{\mathfrak{F}},\SF{\mathfrak{F}}')_\Kref =  (\Dut_I\SF{\mathfrak{F}},\SF{\mathfrak{F}}')_\Kref +
(-)^{I\mathfrak{F}}(\SF{\mathfrak{F}}, \Dut_I\SF{\mathfrak{F}}')_\Kref .
\end{split}
\end{equation}
%

\subsection{Gauge transformation and WZ gauge}

Next, we are interested in how this superfield transforms under $\UT$ super gauge transformations.  Given a gauge parameter superfield $\SF{\Lambda}$ we will define the transformation of the $\UT$ adjoint superfield to be
\begin{equation}
\begin{split}
 \SF{\mathfrak{F}} &\mapsto  
 	\text{exp}\left[ (\SF{\Lambda},\,\cdot\,)_\Kref \right] \SF{\mathfrak{F}} \\
  & \equiv 
	  \SF{\mathfrak{F}}+(\SF{\Lambda},\SF{\mathfrak{F}})_\Kref + \frac{1}{2!}\prn{\SF{\Lambda}, (\SF{\Lambda},\SF{\mathfrak{F}})_\Kref}_\Kref
	   + \frac{1}{3!}\prn{\SF{\Lambda},\prn{\SF{\Lambda}, (\SF{\Lambda},\SF{\mathfrak{F}})_\Kref}_\Kref}_\Kref + \ldots
\end{split}
\end{equation}
The gauge transformation parameter $\SF{\Lambda}$ is again an adjoint superfield whose superspace expansion is built atop its covariant components $\{ \Lambda, \Lambda_\psi, \Lambda_{\psib}, \tilde{\Lambda} \}$, analogously to \eqref{eq:adjsf}.
 
Let us, for simplicity, consider first the subset of gauge transformation with the bottom component of 
$\SF{\Lambda}$ vanishing, i.e., $\Lambda=0$. Then, the super-gauge transformation of the adjoint superfield truncates to
\begin{equation}
\SF{\mathfrak{F}} \mapsto \SF{\mathfrak{F}}+(\SF{\Lambda},\SF{\mathfrak{F}})_\Kref + \frac{1}{2!}\prn{\SF{\Lambda}, (\SF{\Lambda},\SF{\mathfrak{F}})_\Kref}_\Kref  .
\end{equation}
since 
$ \prn{\SF{\Lambda},\prn{\SF{\Lambda}, (\SF{\Lambda},\ldots)_\Kref }_\Kref  }_\Kref =0$ when $\Lambda=0$.  
This is easily worked out to give
\begin{equation}
\begin{split}
 \SF{\mathfrak{F}} &\mapsto \SF{\mathfrak{F}}\bigbr{\begin{array}{c} \GhT\mapsto \GhT- \Lambda_\psi, \\
 \GhbT\mapsto \GhbT- \Lambda_{\psib},\\ \BT\mapsto \BT- \tilde{\Lambda}
 + \half (\Lambda_{\psib},\GhT)_\Kref- \half (\Lambda_{\psi},\GhbT)_\Kref \end{array} }
 \quad \text{with} \quad \prn{\begin{array}{c} \mathfrak{F} \\ \mathfrak{F}_{\psi} \\ \mathfrak{F}_{\psib} \\ \tilde{\mathfrak{F}} \end{array}}
  \quad \text{fixed.}
\end{split}
\end{equation}

Thus the super-gauge transformed superfield can be obtained by absorbing the gauge transformation into  shifts of Faddeev-Popov (FP) fields  $\{\GhT,\GhbT,\BT\}$ appearing in the superfield expansion of $ \SF{\mathfrak{F}}$ keeping the basic quartet fields fixed. 

For this reason, we will  henceforth refer this subset of transformations as \emph{FP boosts}. We will call the complementary subset of gauge transformations  which leave invariant $\{\GhT,\GhbT,\BT\}$ as the \emph{FP rotations}. This is akin to how  the Poincar\'e group action on massive spinning particle states can be constructed by thinking of it as composed of a little group of rest-frame rotations dressed with translations/boosts. Analogously, the group of $\UT$  super-gauge transformations can be thought of as a FP rotations having $\Lambda\neq0$ dressed with FP boosts that shift $\{\GhT,\GhbT,\BT\}$. More precisely, we think of the transformations of super-fields  using the method of induced representations whereby FP rotations that leave $\{\GhT,\GhbT,\BT\}$ invariant can be used to induce the whole group of $\UT$ gauge transformations.

The first step in doing this is to use the FP boosts to set the fields $\{\GhT,\GhbT,\BT\}=0$. This amounts to a partial gauge  fixing and we will call a gauge with  $\{\GhT,\GhbT,\BT\}=0$ as a \emph{Wess-Zumino gauge}.  This step is analogous to using the boosts to  go to the rest frame in the Poincar\'e group example. 
Given a $\{\GhT,\GhbT,\BT\}$, the simplest FP boost which takes us to Wess-Zumino gauge has
$\{\Lambda=0,\Lambda_\psi = \GhT,\Lambda_{\psib} = \GhbT, \tilde{\Lambda} = \BT\}$
which gives
\begin{equation}
\Lambda_{FP}\equiv \theta\ \GhbT + \thetab \ \GhT + \thetab  \theta\ \BT
\label{eq:LamFP}
\end{equation}

This is the analogue of the boost which takes a given momentum to its rest frame. In Poincar\'e case, such a boost is obviously not unique: one can always add in a rest frame rotation from the little group and the new boost is as good as the old. By a similar logic, $\Lambda_{FP}$ can be composed with any little $\UT$ gauge transformation which leaves $\{\GhT,\GhbT,\BT\}$ invariant and it would still give a Wess-Zumino gauge. In the Wess-Zumino gauge obtained using $\Lambda_{FP}$, the adjoint superfield has a simple superspace expansion:
\begin{equation}
\begin{split}
(\SF{\mathfrak{F}})_{WZ} &=\mathfrak{F} + \thetab\ \mathfrak{F}_\psi
+ \theta\ \mathfrak{F}_{\psib} + \thetab   \theta\ \tilde{\mathfrak{F}}
\end{split}
\end{equation}
This shows that basic quartet fields are analogous to  components of a tensor measured after going to the rest frame using a standard set of boosts.

The next step in the method of induced representations is to construct the `little group', viz., the set of all FP rotations which preserve $\{\GhT,\GhbT,\BT\}=0$. This is easily seen to be the transformations with  $\{\Lambda\neq 0,\Lambda_\psi = 0,\Lambda_{\psib} = 0, \tilde{\Lambda} = 0\}$. A more super-covariant way of characterizing FP rotations is to define them as the set of $\SF{\Lambda}$ with 
$$ \mathcal{D}_\theta \SF{\Lambda} |= \Dut_{\thb} \SF{\Lambda} |
= \mathcal{D}_\theta \Dut_{\thb}\SF{\Lambda} | =0  \,.$$ 
Note however that $\Dut_{\thb}\mathcal{D}_\theta \SF{\Lambda} | =(\phizT,\Lambda)_\Kref \neq 0$ as we shall explicitly see below. Thus, If $\alpha$ represents one of the basic quartet fields $\{\mathfrak{F},\mathfrak{F}_\psi,\mathfrak{F}_{\psib},\tilde{\mathfrak{F}}\}$,  its transformation under FP rotations is
given by  the infinite series
\begin{equation}
\begin{split}
\alpha&\mapsto \alpha+(\Lambda,\alpha)_\Kref + \frac{1}{2!}\prn{\Lambda, (\Lambda,\alpha)_\Kref}_\Kref
 + \frac{1}{3!}\prn{\Lambda,\prn{\Lambda, (\Lambda,\alpha)_\Kref}_\Kref}_\Kref
+ \ldots
\end{split}
\end{equation}
We show in Appendix~\ref{sec:UTBCH} that the algebra of gauge transformations closes and work out  the composition rule for $\UT$ transformations.

\subsection{Covariant derivatives and Bianchi identities}

Having understood gauge transformations and the WZ gauge $\GhT=\GhbT=\BT=0$, we can now without loss of generality work in this gauge. This has the salubrious feature that the non-covariant gauge potentials appearing in the FP-triplet are no longer floating around in the computations.

Superspace covariant derivatives acting on a covariant adjoint superfield give a covariant superfield of the same $\UT$ adjoint  type with a similar superspace expansion built out of a new quartet of fields.  It is thus convenient to represent the action of  superspace covariant derivatives by how they change the basic quartets
rather than to repeat the full superspace expansion again and again. We let 
$$ D_a = \partial_a + (A_a, \cdot)_\Kref \,.$$
The action of any super-covariant derivative maps a superfield into a new superfield with covariant components as follows:
\begin{subequations}
\begin{align}
&\Dut_a\ :\ 
	\prn{\begin{array}{c} 
		\mathfrak{F} \\ \mathfrak{F}_{\psi} \\ \mathfrak{F}_{\psib} \\ \tilde{\mathfrak{F}} 
	\end{array}} 
	\mapsto
	\prn{\begin{array}{c} 
		D_a\mathfrak{F} \\ 
		D_a\mathfrak{F}_{\psi}+(\lambda_a,\mathfrak{F})_\Kref  \\
		D_a \mathfrak{F}_{\psib}+(\bar{\lambda}_a,\mathfrak{F})_\Kref  \\
		D_a \tilde{\mathfrak{F}} + (\sAt{a},\mathfrak{F})_\Kref  
			+(\bar{\lambda}_a,\mathfrak{F}_\psi)_\Kref-(\lambda_a,\mathfrak{F}_{\psib})_\Kref
	\end{array}}\ , 
 \label{eq:DamapAdj}	\\
& \Dut_{\thetab}\ :\ 
	\prn{\begin{array}{c} 
	\mathfrak{F} \\ \mathfrak{F}_{\psi} \\ \mathfrak{F}_{\psib} \\ \tilde{\mathfrak{F}} 
	\end{array}} 
	\mapsto
	\prn{\begin{array}{c}  
		\mathfrak{F}_\psi \\ 
		(\phiT,\mathfrak{F})_\Kref \\ 
		\tilde{\mathfrak{F}}\\   
		(\phiT,\mathfrak{F}_{\psib})_\Kref - (\etaT,\mathfrak{F})_\Kref
	 \end{array}} \,, 
 \label{eq:DthbmapAdj}	 \\
& \Dut_\theta\ :\ 
	\prn{\begin{array}{c} 
	\mathfrak{F} \\ \mathfrak{F}_{\psi} \\ \mathfrak{F}_{\psib} \\ \tilde{\mathfrak{F}} 
	\end{array}} 
	\mapsto
	 \prn{\begin{array}{c}  
		 \mathfrak{F}_{\psib} \\
		  (\phizT,\mathfrak{F})_\Kref-\tilde{\mathfrak{F}}\\
		  (\phibT,\mathfrak{F})_\Kref \\
		(\phizT,\mathfrak{F}_{\psib})_\Kref-(\phibT,\mathfrak{F}_{\psi})_\Kref  + (\etabT,\mathfrak{F})_\Kref
	  \end{array}}\,.
 \label{eq:DthmapAdj}
 \end{align}
 \end{subequations}

\afterpage{\clearpage
\begin{sidewaystable}[H]
\small
\centering
\begin{tabular}{|| c || c | c || c | c || c | c ||}
\hline\hline
{\shadeB } &  {\shadeB }      & {\shadeB  } &
	{ \shadeB } & { \shadeB } &
	{ \shadeB } & { \shadeB } \\
{\shadeB Type} &  {\shadeB Field}      & {\shadeB Superspace } &
	{ \shadeB $\QSK\equiv \partial_{\thetab }(\ldots)|$} & { \shadeB $\QSKb\equiv \partial_{\theta}(\ldots)|$} &
	{ \shadeB $\Q\equiv \Dut_{\thetab }(\ldots)|$} & { \shadeB $\Qb\equiv \Dut_{\theta}(\ldots)|$} \\
{\shadeB } &  {\shadeB }      & {\shadeB  Definition} &
	{ \shadeB } & { \shadeB } &
	{ \shadeB } & { \shadeB } \\
\hline
\hline
{\shadeR Adjoint} &{\shadeR $\mathfrak{F}$}      & {\shadeR $\SF{\mathfrak{F}}|$} &
{\shadeR $\mathfrak{F}_\psi - (\GhT,\mathfrak{F})_\Kref $} & {\shadeR $\mathfrak{F}_{\psib}- (\GhbT,\mathfrak{F})_\Kref $} &  {\shadeR $\mathfrak{F}_\psi$} &  {\shadeR $\mathfrak{F}_{\psib}$}  \\
\cdashline{2-7}[2pt/2pt]
{\shadeR superfield} &{\shadeR $\mathfrak{F}_{\psib}$}      & {\shadeR $\Dut_\theta \SF{\mathfrak{F}}|$} &
{\shadeR $(\phizT,\mathfrak{F})_\Kref-\tilde{\mathfrak{F}}- (\GhT,\mathfrak{F}_{\psib})_\Kref$} & {\shadeR $ (\phibT,\mathfrak{F})_\Kref - (\GhbT,\mathfrak{F}_{\psib})_\Kref$}
&  {\shadeR $(\phizT,\mathfrak{F})_\Kref-\tilde{\mathfrak{F}}$} & {\shadeR $ (\phibT,\mathfrak{F})_\Kref $}  \\
\cdashline{2-7}[2pt/2pt]
 {\shadeR }&{\shadeR $\mathfrak{F}_{\psi}$}      & {\shadeR $\Dut_{\thetab} \SF{\mathfrak{F}}|$} &
 {\shadeR $ (\phiT,\mathfrak{F})_\Kref-(\GhT,\mathfrak{F}_{\psi})_\Kref$}
& {\shadeR$\tilde{\mathfrak{F}}- (\GhbT,\mathfrak{F}_{\psi})_\Kref$} & {\shadeR $ (\phiT,\mathfrak{F})_\Kref$} & {\shadeR $\tilde{\mathfrak{F}}$}  \\
\cdashline{2-7}[2pt/2pt]
{\shadeR }  &{\shadeR $\tilde{\mathfrak{F}}$}      & {\shadeR $\Dut_\theta \Dut_{\thetab } \SF{\mathfrak{F}}|$ } &
{\shadeR $(\phizT,\mathfrak{F}_\psi)_\Kref  + (\etaT,\mathfrak{F})_\Kref -(\phiT,\mathfrak{F}_{\psib})_\Kref$} &
{\shadeR $(\phibT,\mathfrak{F}_\psi)_\Kref - (\GhbT,\tilde{\mathfrak{F}})_\Kref$} &
{\shadeR $(\phizT,\mathfrak{F}_\psi)_\Kref  + (\etaT,\mathfrak{F})_\Kref -(\phiT,\mathfrak{F}_{\psib})_\Kref$} &
 {\shadeR $(\phibT,\mathfrak{F}_\psi)_\Kref$ } \\
 {\shadeR }  &{\shadeR }      & {\shadeR } &
{\shadeR $\qquad- (\GhT,\tilde{\mathfrak{F}})_\Kref$} &
{\shadeR } &
{\shadeR } &
 {\shadeR  } \\
\hline
\hline
\end{tabular}
\caption{Four fields of the adjoint multiplet and the action of Weil and Cartan charges.}
\label{tab:Adjoint}
\end{sidewaystable}
\clearpage }

These maps can be used to evaluate graded commutators of $\UT$ super-covariant derivatives to give
\begin{equation}
\begin{split}
\Dut_I \Dut_J \SF{\mathfrak{F}} 
-(-)^{IJ}\,  \Dut_J \Dut_I 
\SF{\mathfrak{F}}
&= (1+ \delta_{IJ}) \dbrk{\SF{\mathscr{F}}_{IJ},\, \SF{\mathfrak{F}}}.
\end{split}
\end{equation}
or more explicitly
\begin{align}
[\Dut_a,\Dut_b] \SF{\mathfrak{F}} &= \dbrk{\SF{\mathscr{F}}_{ab},\SF{\mathfrak{F}}}   \,, & 
[\Dut_\theta,\Dut_a] \SF{\mathfrak{F}} &=\dbrk{\SF{\mathscr{F}}_{\theta a},\, \SF{\mathfrak{F}}}  \,, &
[\Dut_{\thetab},\Dut_a] \SF{\mathfrak{F}} &= \dbrk{\SF{\mathscr{F}}_{\thetab a},\,  \SF{\mathfrak{F}}}\,,
\nonumber \\
\{\Dut_\theta,\Dut_{\thetab }\} \SF{\mathfrak{F}} &= \dbrk{\SF{\mathscr{F}}_{\theta\thetab } ,\,  \SF{\mathfrak{F}} } \,, &
\Dut_\theta^2 \SF{\mathfrak{F}} &= \dbrk{\SF{\mathscr{F}}_{\theta\theta} ,\, \SF{\mathfrak{F}} }  \,, &
\Dut_{\thetab}^2 \SF{\mathfrak{F}} &= \dbrk{\SF{\mathscr{F}}_{\bar\theta\bar\theta}  ,\,  \SF{\mathfrak{F}}} \,.
\end{align}
To verify this, it is useful to directly compose the above maps, which gives:
\begin{equation}
\begin{split}
\Dut_\theta\Dut_{\thetab}\ :\ \prn{\begin{array}{c} \mathfrak{F} \\ \mathfrak{F}_{\psi} \\ \mathfrak{F}_{\psib} \\ \tilde{\mathfrak{F}} \end{array}} &\mapsto
\prn{\begin{array}{c} \tilde{\mathfrak{F}}  \\ (\phizT,\mathfrak{F}_\psi)_\Kref-(\phiT,\mathfrak{F}_{\psib})_\Kref+(\etaT,\mathfrak{F})_\Kref  \\ (\phibT,\mathfrak{F}_\psi)_\Kref\\
(\phizT,\tilde{\mathfrak{F}})_\Kref - (\phibT,(\phiT,\mathfrak{F})_\Kref)_\Kref + (\etabT,\mathfrak{F}_\psi)_\Kref
 \end{array}}
\\
\Dut_{\thb}\Dut_\theta\ :\ \prn{\begin{array}{c} \mathfrak{F} \\ \mathfrak{F}_{\psi} \\ \mathfrak{F}_{\psib} \\ \tilde{\mathfrak{F}} \end{array}} &\mapsto
 \prn{\begin{array}{c} (\phizT,\mathfrak{F})_\Kref - \tilde{\mathfrak{F}} \\  (\phiT,\mathfrak{F}_{\psib})_\Kref \\
(\phizT,\mathfrak{F}_{\psib})_\Kref-(\phibT,\mathfrak{F}_{\psi})_\Kref  + (\etabT,\mathfrak{F})_\Kref \\
(\phiT,(\phibT,\mathfrak{F})_\Kref)_\Kref - (\etaT,\mathfrak{F}_{\psib})_\Kref
  \end{array}}
  \\
\Dut_{\theta}^2\ :\ \prn{\begin{array}{c} \mathfrak{F} \\ \mathfrak{F}_{\psi} \\ \mathfrak{F}_{\psib} \\ \tilde{\mathfrak{F}} \end{array}} &\mapsto
 \prn{\begin{array}{c} (\phibT,\mathfrak{F})_\Kref  \\  (\phibT,\mathfrak{F}_{\psi})_\Kref - (\etabT,\mathfrak{F})_\Kref \\
(\phibT,\mathfrak{F}_{\psib})_\Kref \\
((\phizT,\phibT)_\Kref,\mathfrak{F})_\Kref + (\phibT,\tilde{\mathfrak{F}})_\Kref + (\etabT,\mathfrak{F}_{\psib})_\Kref
  \end{array}}
    \\
\Dut_{\thb}^2\ :\ \prn{\begin{array}{c} \mathfrak{F} \\ \mathfrak{F}_{\psi} \\ \mathfrak{F}_{\psib} \\ \tilde{\mathfrak{F}} \end{array}} &\mapsto
 \prn{\begin{array}{c} (\phiT,\mathfrak{F})_\Kref  \\  (\phiT,\mathfrak{F}_{\psi})_\Kref \\
(\phiT,\mathfrak{F}_{\psib})_\Kref - (\etaT,\mathfrak{F})_\Kref \\
 (\phiT,\tilde{\mathfrak{F}})_\Kref - (\etaT,\mathfrak{F}_{\psi})_\Kref
  \end{array}}
\end{split}
\end{equation}
and similarly for double derivatives involving $\Dut_a$. 
The essential identities are summarized in Table~\ref{tab:Adjoint}.

\subsection{SK-KMS superalgebra}

Finally, it is useful to define an abstract set of operations that illustrate explicitly the structure of the  $\mathcal{N}_\smallT =2$ superalgebera  we are working with. This can be done by writing down a set of  operators that map adjoint-superfields as follows:
\begin{subequations}
\begin{equation}
\begin{split}
\dSK\ :\ \prn{\begin{array}{c} \mathfrak{F} \\ \mathfrak{F}_{\psi} \\ \mathfrak{F}_{\psib} \\  \tilde{\mathfrak{F}} \end{array}} & \mapsto
\prn{\begin{array}{c}  \mathfrak{F}_\psi \\ 0\\  \tilde{\mathfrak{F}}\\   0
 \end{array}} \,,\qquad\,
\dSKb\ :\ \prn{\begin{array}{c} \mathfrak{F} \\ \mathfrak{F}_{\psi} \\ \mathfrak{F}_{\psib} \\  \tilde{\mathfrak{F}} \end{array}} \mapsto
 \prn{\begin{array}{c}  \mathfrak{F}_{\psib} \\ - \tilde{\mathfrak{F}}\\ 0\\     0   \end{array}}  \,.
\end{split}
 \label{eq:N1LadjWZa}
\end{equation}
\begin{equation}
\begin{split}
 \iKMSb :
	\prn{\begin{array}{c} \mathfrak{F} \\ \mathfrak{F}_{\psi} \\ \mathfrak{F}_{\psib} \\  \tilde{\mathfrak{F}}
	\end{array}}
	&\mapsto
	\prn{\begin{array}{c}  0 \\\relax -\lieD_\Kref\mathfrak{F} \\ 0 \\ -\lieD_\Kref\mathfrak{F}_{\psib} \end{array}}\ ,
\; \quad\;\;\,
 \iKMSzero:
	\prn{\begin{array}{c} \mathfrak{F} \\ \mathfrak{F}_{\psi} \\ \mathfrak{F}_{\psib} \\  \tilde{\mathfrak{F}} \end{array}}
	\mapsto
	\prn{\begin{array}{c}  0 \\\relax 0 \\ 0  \\ \lieD_\Kref\mathfrak{F}  \end{array}}
 \\
  \iKMS :
	\prn{\begin{array}{c} \mathfrak{F} \\ \mathfrak{F}_{\psi} \\ \mathfrak{F}_{\psib} \\  \tilde{\mathfrak{F}}\end{array}}
	&\mapsto
	\prn{\begin{array}{c}  0 \\\relax 0 \\-\lieD_\Kref\mathfrak{F} \\ \lieD_\Kref\mathfrak{F}_\psi \end{array}}\ ,
 \qquad
 \lKMS:
	\prn{\begin{array}{c} \mathfrak{F} \\ \mathfrak{F}_{\psi} \\ \mathfrak{F}_{\psib} \\  \tilde{\mathfrak{F}}\end{array}}
	\mapsto
	\prn{\begin{array}{c} \lieD_\Kref\mathfrak{F} \\\relax \lieD_\Kref\mathfrak{F}_\psi \\ \lieD_\Kref\mathfrak{F}_{\psib} \\ \lieD_\Kref\tilde{\mathfrak{F}}  \end{array}}
\end{split}
 \label{eq:N1LadjWZb}
\end{equation}
\end{subequations}
The notation is meant to suggest the origins of these map: SK standing for Schwinger-Keldysh and KMS for the corresponding thermal periodicity conditions.

These maps generate an extended equivariant algebra, called the {\it SK-KMS superalgebra}, which takes the following form:
\begin{equation}
\begin{split}
\dSK^2 = \dSKb^2  &=\gradcomm{\dSK}{\dSKb} = 0 \\
\gradcomm{\dSK}{\iKMSb} = \gradcomm{\dSKb}{\iKMS}
= - \lKMS \,, & \qquad
\gradcomm{\dSK}{\iKMS} = \gradcomm{\dSKb}{\iKMSb} = 0\\
\gradcomm{\dSK}{\iKMSzero} = - \iKMS\,, & \qquad
  \gradcomm{\dSKb}{\iKMSzero} =  \iKMSb \\
 \gradcomm{\dSK}{\lKMS} &= \gradcomm{\dSKb}{\lKMS} = 0 \\
 \gradcomm{\iKMSb}{\iKMS} =
 \gradcomm{\lKMS}{\iKMSb} &=
   \gradcomm{\lKMS}{\iKMS} =   \gradcomm{\lKMS}{\iKMSzero} = 0\,.
\end{split}
\label{eq:LieSAlg3}
\end{equation}
The $\UT$ covariant derivative maps can be written as the following linear combinations of \eqref{eq:N1LadjWZa} and \eqref{eq:N1LadjWZb}: 
\begin{equation}\label{eq:DSKexpl2}
\begin{split}
\DSK &= \dSK - \phiT \,\iKMSb  - \etaT \,\iKMSzero \,,\\
\DSKb &= \dSKb - \phizT \,\iKMSb- \phibT\, \iKMS + \etabT \,\iKMSzero \,.
\end{split}
\end{equation}
From this representation it is also straightforward to check the following:
\begin{equation}
\begin{split}
& \DSK^2 = \phiT \, \lKMS - \etaT \, \iKMS \,,  \qquad
\DSKb^2 = \phibT \, \lKMS - \etabT \, \iKMSb + (\phizT,\phibT)_\Kref \, \iKMSzero \,,\\
&\quad\gradcomm{\DSK}{\DSKb} = \phizT \, \lKMS + \etaT \,\iKMSb + \etabT \, \iKMS + (\phiT,\phibT)_\Kref \, \iKMSzero \,.
\end{split}
\end{equation}
%

\section{Superspace representations: gauge multiplet}
\label{sec:GaugeMult}

We now turn to a detailed analysis of the superspace $\UT$ gauge field with components $\{\As_a,\As_\thb,\As_\theta\}$. This analysis will involve understanding the covariant field strength derived from the gauge field. Since the field strengths transform as adjoint superfields, the results of Appendix \ref{sec:Adjoint} will be useful.

\subsection{Component fields }
We now want to define the component fields by projecting the superfields and their Grassmann odd derivatives (projecting Grassmann even derivatives
just gives the derivatives of the component fields and does not yield new component fields). Following earlier analysis of the SK-KMS superalgebra \cite{Haehl:2016uah}  (also see above) we can associate a supercharge with every Grassmann odd derivative  as we will explain shortly. As explained in Appendix~\ref{sec:EqReview}, the Weil charges $\{\QSK,\QSKb\}$ are associated with ordinary derivatives $\{\partial_{\thetab },\partial_\theta\}$, while  $\{\Q,\Qb\}$ associated with covariant derivatives $\{\Dut_{\thetab },\Dut_\theta\}$ are the Cartan supercharges.

We will define the action of these supercharges on a particular component field as follows. First, we lift the component field to superspace, i.e., construct a superspace expression that when projected down to $\{\theta=0,\thetab =0\}$ subspace gives the component field of interest. Next, we apply the
appropriate Grassmann odd derivative on the superspace expression and project back to the  $\{\theta=0,\thetab =0\}$  subspace. Thus, by this procedure we construct  the Weil supercharges that are nilpotent but not covariant, whereas Cartan supercharges are covariant but not nilpotent. As is common with
covariant derivatives, Cartan charges square to $\UT$ transformations along projected super-field strengths. If one thinks of topological invariance as a twisted supersymmetry \emph{a la} Parisi-Sourlas, then the Cartan supercharges can alternately be interpreted as twisted supercharges.

It is convenient to choose a complete set of component fields by projecting down covariant expressions in superspace. Once such a set is chosen the action of supercharges can directly be read off using the superspace identities derived in Appendix~\ref{sec:Adjoint}. We summarize the results of this exercise in  Table~\ref{tab:Dodecuplet}.

\afterpage{\clearpage
\begin{sidewaystable}[H]
\small
\centering
\begin{tabular}{|| c || c | c || c | c || c | c ||}
\hline\hline
{\shadeB } &  {\shadeB }      & {\shadeB  } &
	{ \shadeB } & { \shadeB } &
	{ \shadeB } & { \shadeB } \\
{\shadeB Type} &  {\shadeB Field}      & {\shadeB Superspace } &
	{ \shadeB $\QSK\equiv \partial_{\thetab }(\ldots)|$} & { \shadeB $\QSKb\equiv \partial_{\theta}(\ldots)|$} &
	{ \shadeB $\Q\equiv \Dut_{\thetab }(\ldots)|$} & { \shadeB $\Qb\equiv \Dut_{\theta}(\ldots)|$} \\
{\shadeB } &  {\shadeB }      & {\shadeB  Definition} &
	{ \shadeB } & { \shadeB } &
	{ \shadeB } & { \shadeB } \\
\hline
{\shadeR Faddeev Popov} &{\shadeR $\GhT$}      & {\shadeR $\Athb|$} &
	{\shadeR $\phiT -  \half (\GhT,\GhT)_\Kref $} & {\shadeR $\BT -  \half (\GhbT,\GhT)_\Kref $} &  {\shadeR $0$} &  {\shadeR $0$}  \\
	\cdashline{2-7}[2pt/2pt]
{\shadeR Ghost } &{\shadeR $\BT$}      & {\shadeR $\SF{\mathscr{B}}_{\theta\thetab }|$} &
	{\shadeR $\eta+ \half (\GhbT,\phiT)_\Kref   $} &
	 {\shadeR $  -\half (\GhT,\phibT)_\Kref $}
	&  {\shadeR $0$} &  {\shadeR $0$}\\
{\shadeR Triplet } &{\shadeR }      & {\shadeR } &
	{\shadeR $\quad  -  \half (\GhT,\BT+\phizT+\half (\GhbT,\GhT)_\Kref )_\Kref $} &
	 {\shadeR $ \quad  -  \half (\GhbT,\BT-\half (\GhbT,\GhT)_\Kref )_\Kref$}
	&  {\shadeR } &  {\shadeR }\\
	\cdashline{2-7}[2pt/2pt]
{\shadeR  }  &{\shadeR $\GhbT$}      & {\shadeR $\Ath|$} &
        {\shadeR $\phizT-\BT - \half (\GhbT,\GhT)_\Kref $} & {\shadeR $\phibT - \half (\GhbT,\GhbT)_\Kref $} &  {\shadeR $0$} &  {\shadeR $0$}   \\
\hline
\hline
{\shadeR Vector} &{\shadeR $\SF{\Ascr}_a$}      & {\shadeR $\As_a|$} &
	{\shadeR $\lambda_a +D_a \GhT$} & {\shadeR $\bar{\lambda}_a +D_a \GhbT$} & {\shadeR $\lambda_a$} & {\shadeR $\bar{\lambda}_a $}    \\
	\cdashline{2-7}[2pt/2pt]
{\shadeR Quartet} &{\shadeR $\lambda_a$}      & {\shadeR $\SF{\mathscr{F}}_{\thetab a}|$} &
	{\shadeR $-D_a \phiT-(\GhT,\lambda_a)_\Kref$} & {\shadeR $\sAt{a}-(\GhbT,\lambda_a)_\Kref$} & {\shadeR $-D_a \phiT$}  & {\shadeR $\sAt{a}$ }    \\
	\cdashline{2-7}[2pt/2pt]
{\shadeR } &{\shadeR $\sAt{a}$}      & {\shadeR $\Dut_\theta \SF{\mathscr{F}}_{\thetab  a}|$} &
	{\shadeR $-D_a \eta-(\GhT,\sAt{a})_\Kref$} & {\shadeR $-(\phibT,\lambda_a)_\Kref-(\GhbT,\sAt{a})_\Kref$}  &  {\shadeR $-D_a \eta$}  & {\shadeR $-(\phibT,\lambda_a)_\Kref$}  \\
	\cdashline{2-7}[2pt/2pt]
{\shadeR } &{\shadeR $\bar{\lambda}_a$}      & {\shadeR $\SF{\mathscr{F}}_{\theta a}|$} &
	{\shadeR $-D_a \phizT-\sAt{a}-(\GhT,\bar{\lambda}_a)_\Kref$} & {\shadeR $-D_a \phibT-(\GhbT,\bar{\lambda}_a)_\Kref$}  &  {\shadeR $-D_a \phizT-\sAt{a}$}  & {\shadeR $-D_a \phibT$}  \\
\hline
\hline
{\shadeR Vafa Witten} &{\shadeR $\phiT$}      & {\shadeR $\SF{\mathscr{F}}_{\thetab \thetab }|$} &
	{\shadeR $- (\GhT,\phiT)_\Kref$} & {\shadeR $-\eta- (\GhbT,\phiT)_\Kref$} & {\shadeR $0$}&  {\shadeR $-\eta$} \\
	\cdashline{2-7}[2pt/2pt]
{\shadeR Ghost of Ghost} &{\shadeR $\eta$}      & {\shadeR $\Dut_{\thetab }\SF{\mathscr{F}}_{\theta\thetab }|$} &
	{\shadeR $(\phiT,\phizT)_\Kref- (\GhT,\eta)_\Kref$} & {\shadeR $(\phiT,\phibT)_\Kref- (\GhbT,\eta)_\Kref$} &{\shadeR $(\phiT,\phizT)_\Kref$} & {\shadeR $(\phiT,\phibT)_\Kref$ }  \\
	\cdashline{2-7}[2pt/2pt]{\shadeR  Quintet} &{\shadeR $\phizT$}      & {\shadeR $\SF{\mathscr{F}}_{\theta\thetab }|$} &
	{\shadeR $\eta- (\GhT,\phizT)_\Kref $} & {\shadeR$\bar{\eta}- (\GhbT,\phizT)_\Kref$}  &{\shadeR $\eta$} &{\shadeR $\bar{\eta}$} \\
	\cdashline{2-7}[2pt/2pt]
{\shadeR } &{\shadeR $\bar{\eta}$}      & {\shadeR $\Dut_{\theta}\SF{\mathscr{F}}_{\theta\thetab }|$} &
	{\shadeR $(\phibT,\phiT)- (\GhT,\bar{\eta})_\Kref$} & {\shadeR  $(\phibT,\phizT)- (\GhbT,\bar{\eta})_\Kref$}  &{\shadeR $(\phibT,\phiT)_\Kref$} & {\shadeR $(\phibT,\phizT)_\Kref$}\\
	\cdashline{2-7}[2pt/2pt]
{\shadeR } &{\shadeR $\phibT$}      & {\shadeR $\SF{\mathscr{F}}_{\theta\theta}|$} &
	{\shadeR $-\bar{\eta}- (\GhT,\phibT)_\Kref$} & {\shadeR $- (\GhbT,\phibT)_\Kref$}  &{\shadeR  $-\bar{\eta}$} & {\shadeR $0$ }\\
\hline
 \hline
\end{tabular}
\caption{{Twelve fields of the gauge dodecuplet and the action of Weil and Cartan supercharges on the dodecuplet.} }
\label{tab:Dodecuplet}
\end{sidewaystable}
\clearpage }

Let us explain how this table is constructed with an example. Say we would like to figure out how the Weil supercharge $\QSK$ acts on the field $\bar{\lambda}_a$. As described above, the Weil supercharge $\QSK$ is associated with $\partial_{\thetab }$ and using the second row of Table~\ref{tab:Dodecuplet}, $\lambda_a$ is defined by the projection $\bar{\lambda}_a\equiv \SF{\mathscr{F}}_{\theta a}|$ . Thus, we conclude $\QSK$ acting on  $\bar{\lambda}_a$ gives $\partial_{\thetab } \SF{\mathscr{F}}_{\theta a}|$. In order to write this in terms of component fields, we have to rewrite $\partial_{\thetab } \SF{\mathscr{F}}_{\theta a}|$ in terms of defining combinations appearing  in the  second row of Table~\ref{tab:Dodecuplet} (with possible Grassmann even derivatives acting on the defining combinations):
\begin{equation}
\begin{split}
\QSK \bar{\lambda}_a \equiv \partial_{\thetab } \SF{\mathscr{F}}_{\theta a}| 
&= 
	\Dut_{\thetab } \SF{\mathscr{F}}_{\theta a}- (\Athb,\SF{\mathscr{F}}_{\theta a})_\Kref| \\
&=  
	-\Dut_a \SF{\mathscr{F}}_{\theta \thetab }-\Dut_{\theta} \SF{\mathscr{F}}_{\thetab  a}
	- (\Athb,\SF{\mathscr{F}}_{\theta a})_\Kref| \\
&= 
	-D_a \phizT - \sAt{a} - (\GhT,\bar{\lambda}_a)_\Kref  \,.
\end{split}
\end{equation}
In the first line we have used the definition of the covariant derivative and in the second line we have invoked the Bianchi identity. The final line  then follows from definitions in  the second row of Table~\ref{tab:Dodecuplet}. Every entry in that table can be derived similarly by using superspace
definitions and identities from Appendix~\ref{sec:Adjoint}.

There is a more direct but tedious method of evaluating  $\partial_{\thetab } \SF{\mathscr{F}}_{\theta a}|$ using superspace Taylor expansion. For this we will have to first compute the superspace Taylor expansion of the gauge fields in terms of the defining combinations. Let us show how this is done by working out the superspace Taylor expansion of $\As_a$. We need
\begin{equation}\label{eq:AththbExp1}
\begin{split}
\As_a | 
&\equiv 
	\SF{\Ascr}_a \\
\partial_{\thetab } \As_a | 
&= 
	\SF{\mathscr{F}}_{\thetab a} + \partial_a \Athb + (\As_a, \Athb)_\Kref | =  \lambda_a + D_a \GhT\\
\partial_{\theta} \As_a |
 &= 
 	\SF{\mathscr{F}}_{\theta a} + \partial_a \Ath + (\As_a, \Ath)_\Kref | = \bar{\lambda}_a + D_a \GhbT\\
\partial_{\theta} \partial_{\thetab } \As_a | 
&=
	\partial_{\theta} \brk{ \SF{\mathscr{F}}_{\thetab a} + \partial_a \Athb + (\As_a, \Athb)_\Kref } | \\
&= 
	\Dut_{\theta} \SF{\mathscr{F}}_{\thetab a} - (\Ath,\SF{\mathscr{F}}_{\thetab a})_\Kref
	+ \partial_a \partial_{\theta} \Athb + (\As_a, \partial_{\theta}\Athb)_\Kref + ( \Athb, \partial_\theta \As_a)_\Kref | \\
&= 
	\sAt{a} - (\GhbT,\lambda_a) +D_a [\BT- \half (\GhbT,\GhT)_\Kref ]+ (\GhT,\bar{\lambda}_a+ D_a \GhbT)_\Kref  \\
&= 
	\sAt{a} +D_a \BT+ (\GhT,\bar{\lambda}_a)_\Kref- (\GhbT,\lambda_a) +\half(\GhT, D_a \GhbT)_\Kref -  \half (\GhbT,D_a \GhT)_\Kref\ .
\end{split}
\end{equation}
Similar computation can be repeated for superfields $\Athb$ and $\Ath$. 
Only the $\thetab  \theta$ component computations are slightly involved. They evaluate to
\begin{equation}\label{eq:AththbExp1}
\begin{split}
\partial_{\theta} \partial_{\thetab } \Athb |
&=
	\partial_{\theta} \brk{ \SF{\mathscr{F}}_{\thetab \thetab } -\half  (\Athb, \Athb)_\Kref } |
	= \Dut_{\theta} \SF{\mathscr{F}}_{\thetab \thetab } - (\Ath,\SF{\mathscr{F}}_{\thetab \thetab })_\Kref
	+ (\Athb, \partial_{\theta} \Athb)_\Kref  | \\
&=  
	- \Dut_{\thetab } \SF{\mathscr{F}}_{\theta\thetab } - (\Ath,\SF{\mathscr{F}}_{\thetab \thetab })_\Kref+ (\Athb, \partial_{\theta} \Athb)_\Kref  | \\
&= 
	 - \bigbr{ \etaT+  (\GhbT,\phiT)_\Kref - (\GhT,\BT-  \half (\GhbT,\GhT)_\Kref  )_\Kref }\\
&=
	 -\bigbr{ \etaT+  (\GhbT,\phiT)_\Kref - (\GhT,\BT )_\Kref + \half (\GhT,(\GhbT,\GhT)_\Kref )_\Kref  \ }\ ,\\
 \partial_{\theta} \partial_{\thetab } \Ath | 
 &=
	 \partial_{\theta} \brk{ \SF{\mathscr{F}}_{\theta\thetab } - (\Ath, \Athb)_\Kref } |  \\
&=	
	 \Dut_{\theta} \SF{\mathscr{F}}_{\theta\thetab } -(\Ath,\SF{\mathscr{F}}_{\theta\thetab })_\Kref
	+(\Ath,\partial_{\theta} \Athb)_\Kref +(\Athb,\SF{\mathscr{F}}_{\theta\theta} -\half  (\Ath, \Ath)_\Kref )_\Kref   | \\
&= 
	\Dut_{\theta} \SF{\mathscr{F}}_{\theta\thetab } +(\Athb,\SF{\mathscr{F}}_{\theta\theta})_\Kref
	 -(\Ath,\SF{\mathscr{F}}_{\theta\thetab }-\partial_{\theta} \Athb- (\Ath, \Athb)_\Kref)_\Kref | \\
&=  
	\etabT +  (\GhT,\phibT)_\Kref -  (\GhbT, \phizT- \BT- \half (\GhbT,\GhT)_\Kref )_\Kref\\
&= 
	\etabT +  (\GhT,\phibT)_\Kref -  (\GhbT, \phizT- \BT)_\Kref+ \half  (\GhbT, (\GhbT,\GhT)_\Kref )_\Kref\ .
\end{split}
\end{equation}
We can now put it all together to get the superspace Taylor expansion. For any superfield $\SF{\mathfrak{f}}$, we have
\begin{equation}
\begin{split}
\SF{\mathfrak{f}} = 
	\mathfrak{f}| + \theta \prn{ \partial_\theta \mathfrak{f}|} + \thetab  \prn{ \partial_{\thetab } \mathfrak{f}|} 
	+ \thetab  \theta  \prn{ \partial_\theta\partial_{\thetab } \mathfrak{f}|} \,.
\end{split}
\end{equation}
Thus, we can use the above calculations to deduce the expressions for the $\UT$ gauge superfield:
\begin{subequations}\label{eq:AExp}
\begin{align}
\As_a  
&\equiv
	\ {\Ascr}_a + \thetab  \bigbr{ \lambda_a + D_a \GhT  } + \theta \bigbr{  \bar{\lambda}_a + D_a \GhbT }  
\nonumber \\	
&
	\quad+\thetab   \theta \Bigl\{\source{\mathcal{F}_a} +D_a \BT+ (\GhT,\bar{\lambda}_a)_\Kref- (\GhbT,\lambda_a) 
	+\half(\GhT, D_a \GhbT)_\Kref -  \half (\GhbT,D_a \GhT)_\Kref \Bigr\}
\label{eq:AaExp}\\
\As_{\thb} 
&\equiv
	\ \GhT + \thetab  \bigbr{ \phiT-  \half (\GhT,\GhT)_\Kref  } + \theta \bigbr{ \BT-  \half (\GhbT,\GhT)_\Kref } 
\nonumber \\	
&	
	\quad -\thetab   \theta \bigbr{ \etaT+  (\GhbT,\phiT)_\Kref - (\GhT,\BT )_\Kref + \half (\GhT,(\GhbT,\GhT)_\Kref )_\Kref  \ }
\label{eq:AthbExp} \\
\As_\theta 
&\equiv \ 
	\GhbT + \theta  \bigbr{ \phibT-  \half (\GhbT,\GhbT)_\Kref  } + \thetab  \bigbr{\phizT- \BT-  \half (\GhbT,\GhT)_\Kref  } 
\nonumber \\
& 
	\quad +\thetab   \theta \bigbr{\etabT +  (\GhT,\phibT)_\Kref -  (\GhbT, \phizT- \BT)_\Kref+ \half  (\GhbT, (\GhbT,\GhT)_\Kref )_\Kref  }
\label{eq:AthExp}	
\end{align}
\end{subequations}
We can then compute $\partial_{\thetab } \SF{\mathscr{F}}_{\theta a}|$ by feeding this Taylor expansion into the definition of 
$\SF{\mathscr{F}}_{\theta a}$ and getting its superspace Taylor expansion. While this is straightforward but tedious to do by hand, we have checked that 
it can be easily coded in  Mathematica and checked that the data reported in Table~\ref{tab:Dodecuplet} is reproduced.

Clearly, the component maps defined in Appendix~\ref{sec:Adjoint} for adjoint superfields, provide a much more efficient way of writing down and manipulating superspace Taylor expansions. Let us therefore now introduce component map notation for the field strengths associated with the gauge field.

\subsection{Gauge transformation and WZ gauge}

We now consider a super-connection $\As_I$ whose gauge transformation is given by the series
\begin{equation}
\begin{split}
\As_I &\mapsto
	 \As_I\equiv \As_I+(\SF{\Lambda},\As_I)_\Kref- \partial_I \SF{\Lambda}+ \frac{1}{2!}\prn{\SF{\Lambda}, (\SF{\Lambda},\As_I)_\Kref
	 - \partial_I \SF{\Lambda}}_\Kref  \\
&\qquad \qquad \quad 
	+ \frac{1}{3!}\prn{\SF{\Lambda},\prn{\SF{\Lambda}, (\SF{\Lambda},\As_I)_\Kref- \partial_I \SF{\Lambda} }_\Kref  }_\Kref \\
&\qquad \qquad\quad 
	+ \frac{1}{4!}\prn{\SF{\Lambda},\prn{\SF{\Lambda},\prn{\SF{\Lambda}, (\SF{\Lambda},\As_I)_\Kref
	- \partial_I \SF{\Lambda} }_\Kref  }_\Kref }_\Kref +\ldots\\
\end{split}
\end{equation}
Here, we defined the thermal brackets  on $\As_I$ as if they were $\UT$ adjoint superfields, since the inhomogeneous transformation of the potential is already accounted for.

We will again begin by focusing  on the FP boosts with a gauge parameter \eqref{eq:LamFP} (i.e., setting $\Lambda=0$). 
Then the above series truncates to
\begin{equation}
\begin{split}
\As_I &\mapsto
	 \As_I+(\SF{\Lambda},\As_I)_\Kref- \partial_I \SF{\Lambda}+ \frac{1}{2!}\prn{\SF{\Lambda}, (\SF{\Lambda},\As_I)_\Kref
	 - \partial_I \SF{\Lambda}}_\Kref  
	 \\
&\qquad \quad 
	+ \frac{1}{3!}\prn{\SF{\Lambda},\prn{\SF{\Lambda}, (\SF{\Lambda},\As_I)_\Kref- \partial_I \SF{\Lambda} }_\Kref  }_\Kref .
\end{split}
\end{equation}
since $ \prn{\SF{\Lambda},\prn{\SF{\Lambda}, (\SF{\Lambda},\ldots)_\Kref }_\Kref  }_\Kref =0$ when $\Lambda=0$. We apply this to the superspace expansions given in \eqref{eq:AExp} and learn that the transformation amounts to the shifts:
\begin{equation}
\begin{split}
\As_I &\mapsto \As_I\bigbr{\begin{array}{c} \GhT\mapsto \GhT- \Lambda_\psi, \\
 \GhbT\mapsto \GhbT- \Lambda_{\psib},\\ \BT\mapsto \BT- \tilde{\Lambda}
 + \half (\Lambda_{\psib},\GhT)_\Kref- \half (\Lambda_{\psi},\GhbT)_\Kref \end{array} }
 \quad \text{with} \quad \prn{\begin{array}{c} \SF{\Ascr}_a  \\ \lambda_a \\ \bar{\lambda}_a \\ \sAt{a} \end{array}}
\& \prn{\begin{array}{c} \phiT \\ \etaT\\ \phizT \\ \etabT \\ \phibT \end{array}}
  \quad \text{fixed.}
\end{split}
\end{equation}
As before,  the super-gauge transformed $\As_I$ can be obtained by absorbing the gauge transformation into shifts of $\{\GhT,\GhbT,\BT\}$ appearing in the superfield expansion of $\As_I$. Note that these shifts match with the shifts found in the adjoint superfield discussion (Appendix~\ref{sec:Adjoint}),  thus ensuring the consistency of our framework. The basic vector quartet and Vafa-Witten quintet fields do not transform under these FP boosts.

As in the case of adjoint superfield, one can pass to Wess-Zumino gauge via 
\begin{equation}
\begin{split}
(\As_I)_{WZ} &\equiv \As_I+(\Lambda_{FP},\As_I)_\Kref- \partial_I \Lambda_{FP} + \frac{1}{2!}\prn{\Lambda_{FP}, (\Lambda_{FP},\As_I)_\Kref- \partial_I \Lambda_{FP} }_\Kref  \\
&\qquad+ \frac{1}{3!}\prn{\Lambda_{FP},\prn{\Lambda_{FP}, (\Lambda_{FP},\As_I)_\Kref- \partial_I \Lambda_{FP} }_\Kref  }_\Kref
\end{split}
\end{equation}
with $\SF{\Lambda}_{FP}$ given in \eqref{eq:LamFP} such that we get  $\{\GhT,\GhbT,\BT\}_{WZ}=0$. In such a gauge, the superspace expansion again simplifies and we obtain
\begin{equation}
\begin{split}
(\As_a)_{WZ} &\equiv\ \SF{\Ascr}_a + \thetab \  \lambda_a + \theta\  \bar{\lambda}_a  +\thetab  \theta\ \sAt{a} \ , \\
(\Athb)_{WZ} &\equiv\  \thetab \ \phiT  -\thetab  \theta\ \etaT\ , \\
(\Ath)_{WZ} &\equiv\  \thetab \ \phibT+ \thetab \ \phizT +\thetab   \theta\ \etabT\ .
\end{split}
\label{eq:AExpWZ}
\end{equation}

In the Wess-Zumino gauge, FP rotations act on $\SF{\Ascr}_a$ as
\begin{equation}
\begin{split}
\SF{\Ascr}_a &\mapsto \SF{\Ascr}_a +(\Lambda,\SF{\Ascr}_a )_\Kref -\partial_a \Lambda+ \frac{1}{2!}\prn{\Lambda, (\Lambda,\SF{\Ascr}_a )_\Kref -\partial_a \Lambda}_\Kref
 + \frac{1}{3!}\prn{\Lambda,\prn{\Lambda, (\Lambda,\SF{\Ascr}_a )_\Kref -\partial_a \Lambda}_\Kref}_\Kref
+ \ldots
\end{split}
\end{equation}
The fields $\{\lambda_a,\bar{\lambda}_a,\sAt{a}\}$ and the Vafa-Witten quintet are $\UT$ adjoint superfields which transform via the series we encountered before, viz.,
\begin{equation}
\begin{split}
\alpha&\mapsto \alpha+(\Lambda,\alpha)_\Kref + \frac{1}{2!}\prn{\Lambda, (\Lambda,\alpha)_\Kref}_\Kref
 + \frac{1}{3!}\prn{\Lambda,\prn{\Lambda, (\Lambda,\alpha)_\Kref}_\Kref}_\Kref
+ \ldots
\end{split}
\end{equation}
The fields $\{\GhT,\GhbT,\BT\}_{WZ}$ do not transform under FP rotations and remain zero.

\subsection{Super field strengths}

For the rest of this section, we will conveniently work in WZ gauge. Also, we will focus on component vectors of covariant superfields, such that transformations are simple.

We repeat for convenience the definitions of the covariant field strengths and of the non-covariant combination that defines $B_\smallT$:
\begin{equation}
\begin{split}
(1+ \delta_{IJ})\,\SF{\mathscr{F}}_{IJ} &\equiv  \partial_I \SF{\Ascr}_J - (-)^{IJ}  \partial_J \SF{\Ascr}_I + (\SF{\Ascr}_I,\SF{\Ascr}_J)_\Kref\,,\\
\SF{\mathscr{B}}_{\theta\thetab }  &\equiv \partial_{\theta} \Athb +\half (\Ath,\Athb)_\Kref \,,
\end{split}
\end{equation}
and remind the reader that indices $I,J$ take values in $\{a,\theta,\thb\}$.
Note that under an exchange of $\theta$ and $\thetab $, we have $\SF{\mathscr{B}}_{\theta\thetab } \mapsto \SF{\mathscr{F}}_{\theta\thetab }-\SF{\mathscr{B}}_{\theta\thetab } $.

The individual components of the above superfields can be simply summarized as column vectors:
\begin{subequations}
\begin{equation}
\begin{split}
\SF{\mathscr{F}}_{ab}\ :\ \prn{\begin{array}{c} \SF{\mathscr{F}}_{ab} \\ D_a\lambda_b-D_b\lambda_a \\ D_a\bar{\lambda}_b-D_b \bar{\lambda}_a \\ D_a\sAt{b}-D_b\sAt{a}+(\lambda_a,\bar{\lambda}_b)_\Kref-(\bar{\lambda}_a,\lambda_b)_\Kref  \end{array}}
 \end{split}
\end{equation}
\begin{equation}
\begin{split}
\SF{\mathscr{F}}_{\thetab a}\ :\ \prn{\begin{array}{c}  \lambda_a \\ - D_a\phiT  \\  \sAt{a} \\D_a \etaT+(\phiT,\bar{\lambda}_a)_\Kref \end{array}} \quad ,\quad
\SF{\mathscr{F}}_{\theta a}\ :\ \prn{\begin{array}{c}  \bar{\lambda}_a  \\ - D_a\phizT- \sAt{a}  \\  -D_a\phibT \\(\phizT,\bar{\lambda}_a)_\Kref  -  D_a \etabT -(\phibT,\lambda_a)_\Kref  \end{array}}
 \end{split}
\end{equation}
\begin{equation}
\begin{split}
\SF{\mathscr{F}}_{\thetab \thetab }\ :\ \prn{\begin{array}{c}  \phiT \\ 0 \\  -\etaT \\ 0 \end{array}} \quad ,\quad
\SF{\mathscr{F}}_{\theta \thetab }\ :\ \prn{\begin{array}{c}  \phizT \\ \etaT\\  \etabT \\ (\phiT,\phibT)_\Kref \end{array}} \quad ,\quad
\SF{\mathscr{F}}_{\theta \theta }\ :\ \prn{\begin{array}{c}  \phibT \\   -\etabT \\ 0 \\ (\phizT,\phibT)_\Kref \end{array}}
 \end{split}
\end{equation}
\end{subequations}
These expressions can be worked out in a similar fashion as described for the gauge field components in the last subsection. See also Table~\ref{tab:Dodecuplet} for a summary of the same.

\subsection{Covariant derivatives and Bianchi identities}

We now turn to the derivatives of $\As_I$ which give the super field strengths. The super field strengths and their super-covariant derivatives transform as $\UT$ adjoint superfields and the discussion of Appendix~\ref{sec:Adjoint} applies to them without any changes.  Thus, we will just quote the basic components that make up the relevant superfields and refer the reader to the previous subsection for details about superspace expansion,
gauge transformation, going to the Wess-Zumino gauge etc. 

We then define covariant derivatives  of these field strengths in the same way as we did for adjoint superfields:
\begin{equation}
\Dut_I \SF{\mathscr{F}}_{JK}  \equiv \partial_I \SF{\mathscr{F}}_{JK}
+ (\SF{\Ascr}_I, \SF{\mathscr{F}}_{JK} )_\Kref \,.
\end{equation}
These covariant derivatives are not all independent as  they must satisfy Bianchi identities which take the form
\begin{equation}
 (-)^{KI}(1+ \delta_{JK})\Dut_I \SF{\mathscr{F}}_{JK} +(-)^{IJ}(1+ \delta_{KI}) \Dut_J \SF{\mathscr{F}}_{KI} + (-)^{JK}(1+ \delta_{IJ}) \Dut_K \SF{\mathscr{F}}_{IJ}   =0\,.
 \end{equation}
We can also write these identities component-wise:
\begin{equation}\label{eq:superBianchiF}
\begin{split}
\Dut_c \SF{\mathscr{F}}_{ab} + \Dut_a \SF{\mathscr{F}}_{bc}+ \Dut_b \SF{\mathscr{F}}_{ca} &= 0
= \Dut_{\theta} \SF{\mathscr{F}}_{\thetab a}+ \Dut_{\thetab } \SF{\mathscr{F}}_{\theta a} +\Dut_a \SF{\mathscr{F}}_{\theta\thetab }\\
\Dut_{\thetab } \SF{\mathscr{F}}_{ab} - \prn{\Dut_a \SF{\mathscr{F}}_{\thetab b}-\Dut_b \SF{\mathscr{F}}_{\thetab a}} & =0
= \Dut_{\theta} \SF{\mathscr{F}}_{ab} - \prn{\Dut_a \SF{\mathscr{F}}_{\theta b}-\Dut_b \SF{\mathscr{F}}_{\theta a}}\\
\Dut_{\thetab } \SF{\mathscr{F}}_{\thetab a} +  \Dut_a \SF{\mathscr{F}}_{\thetab \thetab }  &
=  0 = \Dut_{\theta} \SF{\mathscr{F}}_{\theta a} +  \Dut_a \SF{\mathscr{F}}_{\theta \theta} \ ,
\\
\Dut_{\thetab } \SF{\mathscr{F}}_{\theta\thetab }+\Dut_{\theta} \SF{\mathscr{F}}_{\thetab \thetab }
&=0=\Dut_{\theta} \SF{\mathscr{F}}_{\theta\thetab }+ \Dut_{\thetab } \SF{\mathscr{F}}_{\theta\theta}\ ,\\
\Dut_{\thetab } \SF{\mathscr{F}}_{\thetab \thetab } &=0=\Dut_{\theta} \SF{\mathscr{F}}_{\theta\theta}\ .
\end{split}
\end{equation}
If the non-covariant field $\SF{\mathscr{B}}_{\theta\thb}$ is involved, we find that there are also identities of the form
\begin{equation}\label{eq:superBianchiB}
\begin{split}
 \partial_\theta \SF{\mathscr{B}}_{\theta\thetab } + \half (\Ath, \SF{\mathscr{B}}_{\theta\thetab }+\half (\Ath,\Athb)_\Kref )_\Kref+\half (\Athb, \SF{\mathscr{F}}_{\theta\theta})_\Kref   & = 0\,, \\
 \partial_{\thetab } (\SF{\mathscr{F}}_{\theta\thetab } -\SF{\mathscr{B}}_{\theta\thetab })
 + \half (\Athb, \SF{\mathscr{F}}_{\theta\thetab } - \SF{\mathscr{B}}_{\theta\thetab }
 +\half (\Ath,\Athb)_\Kref )_\Kref
 +\half (\Ath, \SF{\mathscr{F}}_{\thetab \thetab })_\Kref    & = 0 \,,
\end{split}
\end{equation}
satisfied by derivatives of $\SF{\mathscr{B}}_{\theta\thetab }$. The second identity can equivalently be written as
\begin{equation}\label{eq:superBianchiB2}
\begin{split}
\Dut_{\thetab } \SF{\mathscr{F}}_{\theta\thetab }
    & = \partial_{\thetab } \SF{\mathscr{B}}_{\theta\thetab }
 + \half (\Athb, \SF{\mathscr{F}}_{\theta\thetab } + \SF{\mathscr{B}}_{\theta\thetab }-\half (\Ath,\Athb)_\Kref )_\Kref -\half (\Ath, \SF{\mathscr{F}}_{\thetab \thetab })_\Kref \,.
\end{split}
\end{equation}

It will be convenient to write simple component maps for these operations. Using the fact that field strengths transform as adjoint superfields, their derivative components follow immediately from Eqs.\ \eqref{eq:DamapAdj}, \eqref{eq:DthbmapAdj}, \eqref{eq:DthmapAdj}:
\begin{subequations}
\begin{align}
& \Dut_\theta\SF{\mathscr{F}}_{\thetab a}
\ :\ 
	\prn{\begin{array}{c} \sAt{a} \\ (\phizT,\lambda_a)_\Kref- D_a \etaT
	-(\phiT,\bar{\lambda}_a)_\Kref\\  (\phibT,\lambda_a)_\Kref \\ (\phizT,\sAt{a} )_\Kref 
	+ (\phibT,D_a\phiT)_\Kref +(\etabT,\lambda_a)_\Kref  \end{array}} \ , 
\\ 
& \Dut_{\thetab } \SF{\mathscr{F}}_{\theta\thetab }=-\Dut_{\theta } \SF{\mathscr{F}}_{\thetab\thetab } 
\ :\
	 \prn{\begin{array}{c}  \etaT \\ (\phiT,\phizT)_\Kref \\  (\phiT,\phibT)_\Kref \\ (\phiT,\etabT)_\Kref-(\etaT,\phizT)_\Kref \end{array}}  \\
& \Dut_{\theta} \SF{\mathscr{F}}_{\theta\thetab }=-\Dut_{\thetab } \SF{\mathscr{F}}_{\theta\theta }
 :\ 
 	\prn{\begin{array}{c}  \etabT \\-(\phiT,\phibT)_\Kref \\  (\phibT,\phizT)_\Kref \\ -(\phibT,\eta)_\Kref\ \end{array}}
\\
&\Dut_{\thetab}\SF{\mathscr{F}}_{\theta a}
\ :\
	 \prn{\begin{array}{c} -D_a \phizT- \sAt{a} \\   (\phiT,\bar{\lambda}_a)_\Kref \\ (\phizT,\bar{\lambda}_a)_\Kref
	 - D_a \etabT -(\phibT,\lambda_a)_\Kref\\ - (\phiT,D_a\phibT)_\Kref -(\etaT,\bar{\lambda}_a)_\Kref  \end{array}} \ .
\end{align}
\end{subequations}
These derivatives satisfy the super Bianchi identities enumerated in \eqref{eq:superBianchiF}. Further $\UT$ super-covariant derivatives and identities can all be read off as such from our
discussion for a general $\UT$ adjoint superfield.

\subsection{Double covariant derivatives of super field strengths }

We can proceed further and  define additional covariant derivatives  of the covariant field strengths via $\Dut_I \Dut_J \SF{\mathscr{F}}_{KL}  \equiv \partial_I (\Dut_J\SF{\mathscr{F}}_{KL})+ (\SF{\Ascr}_I, \Dut_J\SF{\mathscr{F}}_{KL})_\Kref $. The double covariant derivatives in Grassmann directions
can be worked out explicitly by using
\begin{equation} \Dut_I \Dut_J \SF{\mathscr{F}}_{KL} -(-)^{IJ} \Dut_J \Dut_I \SF{\mathscr{F}}_{KL} = (1+ \delta_{IJ}) (\SF{\mathscr{F}}_{IJ}, \SF{\mathscr{F}}_{KL} )_\Kref \,.
\end{equation}
Let us evaluate these relations again explicitly for future reference.
First, consider the derivatives on $\SF{\mathscr{F}}_{\theta\theta},\SF{\mathscr{F}}_{\thetab \thetab }$ and $\SF{\mathscr{F}}_{\theta\thetab }$:
\begin{equation}\label{eq:doubleDerivF1}
\begin{split}
\Dut_{\theta}^2 \SF{\mathscr{F}}_{\theta\theta} &= 0
= \Dut_{\thetab }^2 \SF{\mathscr{F}}_{\theta\theta} - (\SF{\mathscr{F}}_{\thetab \thetab },\SF{\mathscr{F}}_{\theta\theta})_\Kref \\
\Dut_{\thetab } \Dut_{\theta} \SF{\mathscr{F}}_{\theta\theta} &=0
=\Dut_{\theta} \Dut_{\thetab } \SF{\mathscr{F}}_{\theta\theta}   - (\SF{\mathscr{F}}_{\theta\thetab },\SF{\mathscr{F}}_{\theta\theta})_\Kref \\
\Dut_{\theta}^2 \SF{\mathscr{F}}_{\thetab \thetab }- (\SF{\mathscr{F}}_{\theta\theta},\SF{\mathscr{F}}_{\thetab \thetab })_\Kref &= 0
= \Dut_{\thetab }^2 \SF{\mathscr{F}}_{\thetab \thetab }  \\
\Dut_{\thetab } \Dut_{\theta} \SF{\mathscr{F}}_{\thetab \thetab } - (\SF{\mathscr{F}}_{\theta\thetab },\SF{\mathscr{F}}_{\thetab \thetab })_\Kref &=0
=\Dut_{\theta} \Dut_{\thetab } \SF{\mathscr{F}}_{\thetab \thetab }  \\
\Dut_{\theta}^2 \SF{\mathscr{F}}_{\theta\thetab }-(\SF{\mathscr{F}}_{\theta\theta},\SF{\mathscr{F}}_{\theta\thetab })_\Kref  &= 0
= \Dut_{\thetab }^2 \SF{\mathscr{F}}_{\theta\thetab } - (\SF{\mathscr{F}}_{\thetab \thetab },\SF{\mathscr{F}}_{\theta\thetab })_\Kref \\
\Dut_{\thetab } \Dut_{\theta} \SF{\mathscr{F}}_{\theta\thetab } - (\SF{\mathscr{F}}_{\theta\theta},\SF{\mathscr{F}}_{\thetab \thetab })_\Kref  &=0
=\Dut_{\theta} \Dut_{\thetab } \SF{\mathscr{F}}_{\theta\thetab }   - (\SF{\mathscr{F}}_{\thetab \thetab },\SF{\mathscr{F}}_{\theta\theta})_\Kref\,.
\end{split}
\end{equation}
Similarly, the double derivative acting on $\SF{\mathscr{F}}_{\theta a}$ and $\SF{\mathscr{F}}_{\thetab  a}$ result in
\begin{equation}\label{eq:doubleDerivF2}
\begin{split}
\Dut_{\theta}^2 \SF{\mathscr{F}}_{\theta a} -(\SF{\mathscr{F}}_{\theta\theta},\SF{\mathscr{F}}_{\theta a})_\Kref  &= 0
= \Dut_{\thetab }^2 \SF{\mathscr{F}}_{\theta a} - (\SF{\mathscr{F}}_{\thetab \thetab },\SF{\mathscr{F}}_{\theta a})_\Kref \\
\Dut_{\thetab } \Dut_{\theta} \SF{\mathscr{F}}_{\theta a}
&=\Dut_a \Dut_\theta  \SF{\mathscr{F}}_{\theta\thetab } +(\SF{\mathscr{F}}_{\theta\theta},\SF{\mathscr{F}}_{\thetab  a})_\Kref
=-\Dut_{\theta} \Dut_{\thetab } \SF{\mathscr{F}}_{\theta a}
+ (\SF{\mathscr{F}}_{\theta\thetab },\SF{\mathscr{F}}_{\theta a})_\Kref
 \\
\Dut_{\theta}^2 \SF{\mathscr{F}}_{\thetab  a}- (\SF{\mathscr{F}}_{\theta\theta},\SF{\mathscr{F}}_{\thetab  a})_\Kref &= 0
= \Dut_{\thetab }^2 \SF{\mathscr{F}}_{\thetab  a} - (\SF{\mathscr{F}}_{\thetab \thetab },\SF{\mathscr{F}}_{\thetab  a})_\Kref   \\
-\Dut_{\thetab } \Dut_{\theta} \SF{\mathscr{F}}_{\thetab  a}
+ (\SF{\mathscr{F}}_{\theta\thetab },\SF{\mathscr{F}}_{\thetab  a})_\Kref
&=\Dut_{\theta} \Dut_{\thetab } \SF{\mathscr{F}}_{\thetab  a}
=
\Dut_a \Dut_{\thetab }  \SF{\mathscr{F}}_{\theta\thetab } +(\SF{\mathscr{F}}_{\thetab \thetab },\SF{\mathscr{F}}_{\theta a})_\Kref   \,,
\end{split}
\end{equation}
and the double derivative acting on $ \SF{\mathscr{F}}_{ab}$ evaluates to
\begin{equation}\label{eq:doubleDerivF3}
\begin{split}
\Dut_{\theta}^2 \SF{\mathscr{F}}_{ab} -(\SF{\mathscr{F}}_{\theta\theta},\SF{\mathscr{F}}_{ab})_\Kref  &= 0
= \Dut_{\thetab }^2 \SF{\mathscr{F}}_{ab} - (\SF{\mathscr{F}}_{\thetab \thetab },\SF{\mathscr{F}}_{ab})_\Kref \\
 - \Dut_{\thetab } \Dut_{\theta} \SF{\mathscr{F}}_{ab}+ (\SF{\mathscr{F}}_{\thetab \thetab },\SF{\mathscr{F}}_{ab})_\Kref
&=\Dut_{\theta} \Dut_{\thetab } \SF{\mathscr{F}}_{ab}
= \Dut_a \Dut_\theta  \SF{\mathscr{F}}_{\thetab  b}+(\SF{\mathscr{F}}_{\theta a},\SF{\mathscr{F}}_{\thetab  b})_\Kref\,.
- (a\leftrightarrow b)
\end{split}
\end{equation}

To get the corresponding component vectors, we can again use the adjoint maps  \eqref{eq:DamapAdj}, \eqref{eq:DthbmapAdj}, \eqref{eq:DthmapAdj}. This yields for example: 
\begin{equation}
\begin{split}
\Dut_\theta^2\SF{\mathscr{F}}_{\thetab a}
\ :\
\prn{\begin{array}{c} (\phibT,\lambda_a)_\Kref \\ - (\phibT,D_a\phiT)_\Kref -(\etabT,\lambda_a)_\Kref \\ (\phibT,\sAt{a})_\Kref  \\ ((\phizT,\phibT)_\Kref ,\lambda_a)_\Kref  + (\phibT,  D_a \etaT+(\phiT,\bar{\lambda}_a)_\Kref)_\Kref + (\etabT,\sAt{a})_\Kref \end{array}} .
\end{split}
\end{equation}
%

\section{Properties of the $\UT$ KMS gauge group}
\label{sec:formal}

In this appendix we collect results of various investigations undertaken to understand the structure of the $\UT$ gauge group. Since $\UT$ describes thermal diffeomorphisms, we expect a very close relation to the representation theory of Diff$({\bf S}^1)$. 

The discussion below is in two distinct parts. In \S\ref{sec:UTBCH} we argue that despite the intricate structure of the $\UT$ algebra and the connection to diffeomorphisms, one can follow the usual logic behind the Baker-Campbell-Hausdorff formula, which eventually allows us to compose $\UT$ transformations.  In \S\ref{sec:repthy} we then describe our attempts to understand the general representations of the $\UT$ algebra and demonstrate some connections to usual ideas in deformation quantization. As we will see, there are various suggestive  connections, though we should add that we have not fully appreciated all the details here.  We include this as part of our discussion mainly for completeness (and hope that it may prove useful for others undertaking related  investigations).

\subsection{Composing $\UT$ transformations }
\label{sec:UTBCH}

The aim of this section is to give a Baker-Campbell-Hausdorff formula for $\UT$ transformations acting on adjoint superfields. We begin with the $\UT$ transformation on the adjoint superfield:
\begin{equation}
\begin{split}
 \SF{\mathfrak{F}} &\mapsto \SF{\mathfrak{F}}+(\SF{\Lambda},\SF{\mathfrak{F}})_\Kref + \frac{1}{2!}\prn{\SF{\Lambda}, (\SF{\Lambda},\SF{\mathfrak{F}})_\Kref}_\Kref
 + \frac{1}{3!}\prn{\SF{\Lambda},\prn{\SF{\Lambda}, (\SF{\Lambda},\SF{\mathfrak{F}})_\Kref}_\Kref}_\Kref + \ldots\\
 &\equiv \exp\prn{ \Ad_{\SF{\Lambda}} }  \cdot \SF{\mathfrak{F}}\,,
\end{split}
\end{equation}
where we have defined $\Ad_{\SF{\Lambda}}$ as an operator which acts via thermal bracket: $\Ad_{\SF{\Lambda}} \cdot \SF{\mathfrak{F}} \equiv (\SF{\Lambda},\SF{\mathfrak{F}})_\Kref$.
Jacobi identity then gives the commutation relation
\begin{equation}
\Ad_{\SF{\Lambda}} \cdot \Ad_{\SF{\Lambda}'}
 -  \Ad_{\SF{\Lambda}'} \cdot \Ad_{\SF{\Lambda}} = \Ad_{(\SF{\Lambda} , \SF{\Lambda}')_\Kref}  =\Ad_{\Ad_{\SF{\Lambda}} \cdot \SF{\Lambda}'} \,.
\end{equation}	
We have  taken both $\SF{\Lambda}$ and $ \SF{\Lambda}'$ to be Grassmann even in the above expression.

 Say we perform a $\UT$ transformation with a parameter $\SF{\Lambda}^{(2)}$ followed by a $\UT$ transformation with a parameter $\SF{\Lambda}^{(1)}$. The question we want to address is the following: can the result be written as a $\UT$ transformation along a parameter $\SF{\Lambda}^{(1\oplus 2)}$? That is, we want a $\SF{\Lambda}^{(1\oplus 2)}$ satisfying
\begin{equation}
\exp\prn{ \Ad_{\SF{\Lambda}^{(1)}} } \cdot \exp\prn{  \Ad_{\SF{\Lambda}^{(2)}} } \cdot \SF{\mathfrak{F}} = \exp\prn{ \Ad_{\SF{\Lambda}^{(1\oplus 2)}}}  \cdot \SF{\mathfrak{F}} \,.
\end{equation}	

 We will answer this question in affirmative and derive a Baker-Campbell-Hausdorff (BCH) formula for $\SF{\Lambda}^{(1\oplus 2)}$ which takes the form
\begin{equation}
\begin{split}
\SF{\Lambda}^{(1\oplus 2)} &= 
	\SF{\Lambda}^{(1)}+ \SF{\Lambda}^{(2)} + \half (\SF{\Lambda}^{(1)}, \SF{\Lambda}^{(2)})_\Kref
	 + \frac{1}{12}\prn{\SF{\Lambda}^{(1)}-\SF{\Lambda}^{(2)}, (\SF{\Lambda}^{(1)}, \SF{\Lambda}^{(2)})_\Kref }_\Kref \\
&\qquad \quad 
	 - \frac{1}{24}\prn{\SF{\Lambda}^{(2)}, (\SF{\Lambda}^{(1)},(\SF{\Lambda}^{(1)}, \SF{\Lambda}^{(2)})_\Kref)_\Kref }_\Kref +\ldots
\end{split}
\end{equation}
We have checked this form  explicitly in Mathematica to the order shown.

We will now give a derivation of this expression. The derivation here adopts the standard derivation of BCH formula for Lie algebra commutators for the thermal brackets, sidestepping some of the manipulations in the standard Lie algebra proof that have no analogues in thermal brackets. We will begin with the following theorem.
\begin{theorem}
The following identities hold:
\begin{equation}
\begin{split}
\exp\prn{ \Ad_{\SF{\Lambda}} }  \cdot \exp\prn{- \Ad_{\SF{\Lambda}} }  \cdot \SF{\mathfrak{F}} &= \SF{\mathfrak{F}}\ , \\
\exp\prn{ \Ad_{\SF{\Lambda}} }  \cdot  \Ad_{\SF{\Lambda}'} \cdot \exp\prn{ -\Ad_{ \SF{\Lambda}} }  \cdot \SF{\mathfrak{F}} &=  \Ad_{\exp\prn{ \Ad_{\SF{\Lambda}} }\cdot\SF{\Lambda}'} \cdot \SF{\mathfrak{F}}\ , \\
\exp\prn{- \Ad_{\SF{\Lambda}(t) } }  \cdot \frac{d}{dt} \exp\prn{ \Ad_{\SF{\Lambda}(t) } }  \cdot \SF{\mathfrak{F}} &=
  \Ad_{ \brk{ \frac{1-\exp\prn{ - \Ad_{ \SF{\Lambda}(t) } }}{ \Ad_{ \SF{\Lambda}(t) } }  }  \cdot \frac{d}{dt} \SF{\Lambda}(t)} \cdot \SF{\mathfrak{F}} .
\end{split}
\end{equation}
In the last identity, $\SF{\Lambda}(t)$ denotes a parameterized set of adjoint superfields.
\end{theorem}
\begin{proof} The main elements of the proofs will be to use the definition of the adjoint action, commutation relations and certain simple identities.

\paragraph{1.} A direct expansion of LHS in the first identity gives
\begin{equation}
\begin{split}
\exp\prn{ \Ad_{\SF{\Lambda}} }  \cdot \exp\prn{- \Ad_{\SF{\Lambda}} }  \cdot \SF{\mathfrak{F}}
&= \sum_{p,q=0}^\infty \frac{(-)^q}{p!q!} \prn{ \Ad_{\SF{\Lambda}} }^{p+q} \cdot \SF{\mathfrak{F}} \\
&= \SF{\mathfrak{F}}\ .
\end{split}
\end{equation}
Here we have used identity $\sum_{p,q=0}^\infty \frac{(-)^q x^{p+q}}{p!\,q!}  =1$. It is straightforward to generalize this to
\begin{equation}
\begin{split}
\exp\prn{\tau_1 \,\Ad_{\SF{\Lambda}} }  \cdot \exp\prn{\tau_2 \,\Ad_{\SF{\Lambda}} } 
 \cdot \SF{\mathfrak{F}}
=\exp\prn{(\tau_1+\tau_2) \, \Ad_{\SF{\Lambda}} }  \cdot \SF{\mathfrak{F}}  \,. \qed
\end{split}
\end{equation}

\paragraph{2.} Next we turn to  the second identity. We have
\begin{equation}
\begin{split}
\exp\prn{ \Ad_{\SF{\Lambda}} }  \cdot  \Ad_{\SF{\Lambda}'} \cdot \exp\prn{ -\Ad_{ \SF{\Lambda}} }  \cdot \SF{\mathfrak{F}} 
&=
 \sum_{p,q=0}^\infty \frac{(-)^q}{p!\,q!} 
 \prn{ \Ad_{\SF{\Lambda}} }^p\cdot \Ad_{\SF{\Lambda}'} \cdot \prn{ \Ad_{\SF{\Lambda}} }^q .
\end{split}
\end{equation}
Using the commutation relation for $\Ad$ operators, $ \Ad_{\SF{\Lambda}} \cdot \Ad_{\SF{\Lambda}'}
 = \Ad_{\SF{\Lambda}'} \cdot \Ad_{\SF{\Lambda}} + \Ad_{\Ad_{\SF{\Lambda}} \cdot \SF{\Lambda}'} $ ( repeatedly, we can show that
\begin{equation}
\begin{split}
 \prn{ \Ad_{\SF{\Lambda}} }^p\cdot \Ad_{\SF{\Lambda}'}
 &=\sum_{k=0}^p\binom{p}{k} \ \Ad_{\prn{\Ad_{\SF{\Lambda}} }^{k} \cdot \SF{\Lambda}'} \cdot \prn{ \Ad_{\SF{\Lambda}} }^{p-k} \\
\end{split}
\end{equation}
so that
\begin{equation}
\begin{split}
\exp\prn{ \Ad_{\SF{\Lambda}} }  &\cdot  \Ad_{\SF{\Lambda}'} \cdot \exp\prn{ -\Ad_{ \SF{\Lambda}} }  \cdot \SF{\mathfrak{F}}\\
 &=
 \sum_{k,p,q=0}^\infty \frac{(-)^q}{k!(p-k)!q!} \ \Ad_{\prn{\Ad_{\SF{\Lambda}} }^{k} \cdot \SF{\Lambda}'} \cdot \prn{ \Ad_{\SF{\Lambda}} }^{p-k+q} \cdot \SF{\mathfrak{F}}\\
  &=
 \sum_{k=0}^\infty \frac{1}{k!} \ \Ad_{\prn{\Ad_{\SF{\Lambda}} }^{k} \cdot \SF{\Lambda}'} \cdot
\sum_{p,q=0}^\infty \frac{(-)^q}{p!q!}
 \prn{ \Ad_{\SF{\Lambda}} }^{p+q} \cdot \SF{\mathfrak{F}}\\
 &=  \Ad_{\exp\prn{ \Ad_{\SF{\Lambda}} }\cdot\SF{\Lambda}'} \cdot \SF{\mathfrak{F}}\ .  \qed
\end{split}
\end{equation}

\paragraph{3.} To prove the third identity, we begin by examining
\begin{equation}
\begin{split}
\int_0^1 d\tau\  &\exp\prn{ (1-\tau) \Ad_{\SF{\Lambda}(t)} }  \cdot  \Ad_{  \frac{d}{dt} \SF{\Lambda}(t) } \cdot \exp\prn{ \tau \Ad_{\SF{\Lambda}(t)} } \cdot
\SF{\mathfrak{F}} \\
&= \int_0^1 d\tau\ \sum_{p,q=0}^\infty \frac{(1-\tau)^p\tau^q }{p!q!} \prn{ \Ad_{\SF{\Lambda}(t)} }^p \cdot  \Ad_{  \frac{d}{dt} \SF{\Lambda}(t) } \cdot \prn{ \Ad_{\SF{\Lambda}(t)} }^q\cdot
\SF{\mathfrak{F}} \\
&=  \sum_{p,q=0}^\infty \frac{1 }{(p+q+1)!} \prn{ \Ad_{\SF{\Lambda}(t)} }^p \cdot  \Ad_{  \frac{d}{dt} \SF{\Lambda}(t) } \cdot \prn{ \Ad_{\SF{\Lambda}(t)} }^q \cdot
\SF{\mathfrak{F}} \\
&=   \sum_{k=1}^{\infty} \sum_{p=0}^{k-1} \frac{1}{k!} \prn{ \Ad_{\SF{\Lambda}(t)} }^p \cdot  \Ad_{  \frac{d}{dt} \SF{\Lambda}(t) } \cdot \prn{ \Ad_{\SF{\Lambda}(t)} }^{k-p-1}\cdot
\SF{\mathfrak{F}} \\
&=    \frac{d}{dt}\sum_{k=1}^{\infty}  \frac{1}{k!} \prn{ \Ad_{\SF{\Lambda}(t)} }^k \cdot
\SF{\mathfrak{F}} = \frac{d}{dt} \exp\prn{ \Ad_{\SF{\Lambda}(t) } }  \cdot \SF{\mathfrak{F}}\ .
\end{split}
\end{equation}
Applying on both sides $\exp\prn{ -\Ad_{\SF{\Lambda}(t) } }  \cdot $ and then using the 
second identity, we get
\begin{equation}
\begin{split}
\int_0^1 d\tau\  \Ad_{\exp\prn{ -\tau\ \Ad_{\SF{\Lambda}(t)} }  \cdot    \frac{d}{dt} \SF{\Lambda}(t) } \cdot
\SF{\mathfrak{F}}
&= \exp\prn{ -\Ad_{\SF{\Lambda}(t) } }  \cdot \frac{d}{dt} \exp\prn{ \Ad_{\SF{\Lambda}(t) } }  \cdot \SF{\mathfrak{F}}\ .
\end{split}
\end{equation}
Performing the $\tau$ integral then gives the required identity.
\end{proof}

Using the above derived identities we will now show by an explicit construction that there is a $\SF{\Lambda}^{(1\oplus 2)}(t)$ satisfying
\begin{equation}
\exp\prn{t_1 \ \Ad_{\SF{\Lambda}^{(1)}} } \cdot \exp\prn{ t_2\ \Ad_{\SF{\Lambda}^{(2)}} } \cdot \SF{\mathfrak{F}} = \exp\prn{ \Ad_{\SF{\Lambda}^{(1\oplus 2)}(t) }}  \cdot \SF{\mathfrak{F}} \,.
\label{eq:genBHC}
\end{equation}	
Setting $t_1=t_2=1$ would then give us the required BCH formula.  We have
\begin{equation}
\begin{split}
\exp\prn{- \Ad_{\SF{\Lambda}^{(1\oplus 2)}(t) } }  \cdot \frac{\partial}{\partial t_i} \exp\prn{ \Ad_{\SF{\Lambda}^{(1\oplus 2)}(t) } }  \cdot \SF{\mathfrak{F}} &=
  \Ad_{ \brk{ \frac{1-\exp\prn{ - \Ad_{ \SF{\Lambda}^{(1\oplus 2)}(t) } }}{ \Ad_{ \SF{\Lambda}^{(1\oplus 2)}(t) } }  }  \cdot \frac{\partial}{\partial t_i} \SF{\Lambda}^{(1\oplus 2)}(t)} \cdot \SF{\mathfrak{F}} \,.
 \end{split}
\end{equation}
A direct evaluation gives
\begin{equation}\label{eq:BCHdiff}
\begin{split}
 \frac{\partial}{\partial t_1}  \SF{\Lambda}^{(1\oplus 2)}(t)  &=
\brk{ \frac{
 \Ad_{ \SF{\Lambda}^{(1\oplus 2)}(t) }
 }{
\exp\prn{  \Ad_{ \SF{\Lambda}^{(1\oplus 2)}(t) } } -1
 } }   \cdot \SF{\Lambda}^{(1)}      \\
 \frac{\partial}{\partial t_2}  \SF{\Lambda}^{(1\oplus 2)}(t)  &=
\brk{ \frac{
 \Ad_{ \SF{\Lambda}^{(1\oplus 2)}(t) }
 }{
 1-\exp\prn{ - \Ad_{ \SF{\Lambda}^{(1\oplus 2)}(t) } }
 } }   \cdot \SF{\Lambda}^{(2)} \,.
\end{split}
\end{equation}
These are the basic generating differential equations for BCH series. For example, setting either $t_1$ or $t_2$ to unity and integrating, we can write down an integral version of the BCH formula as
\begin{equation}\label{eq:BCHInt}
\begin{split}
 \SF{\Lambda}^{(1\oplus 2)}  &= \SF{\Lambda}^{(1)} + \int_0^1 dt_2 \brk{ \frac{ \log\bigbr{\exp\prn{ \Ad_{\SF{\Lambda}^{(1)}} } \cdot \exp\prn{ t_2\ \Ad_{\SF{\Lambda}^{(2)}} } } }
{1-\exp\prn{ - t_2\ \Ad_{ \SF{\Lambda}^{(2)} } } \cdot \exp\prn{ -\Ad_{\SF{\Lambda}^{(1)}} } }  } \cdot \SF{\Lambda}^{(2)} \\
 &= \SF{\Lambda}^{(2)} + \int_0^1 dt_1 \brk{ \frac{ \log\bigbr{\exp\prn{t_1 \Ad_{\SF{\Lambda}^{(1)}} } \cdot \exp\prn{ \ \Ad_{\SF{\Lambda}^{(2)}} } } }
{\exp\prn{  t_1\ \Ad_{ \SF{\Lambda}^{(1)} } } \cdot \exp\prn{ \Ad_{\SF{\Lambda}^{(2)}} } -1 }  } \cdot \SF{\Lambda}^{(1)}
\end{split}
\end{equation}

We can find more useful expressions for example by working in the limit when the initial transformations are along infinitesimal parameters. For example if  a linear approximation in either $\SF{\Lambda}^{(1)}$ or $\SF{\Lambda}^{(2)}$ is required, we can directly use the differential equations in \eqref{eq:BCHdiff} to write
\begin{equation}\label{eq:BCHLin}
\begin{split}
 \SF{\Lambda}^{(1\oplus 2)}  &=  \SF{\Lambda}^{(2)}  +
\brk{ \frac{
 \Ad_{ \SF{\Lambda}^{(2)} }
 }{
\exp\prn{  \Ad_{ \SF{\Lambda}^{(2)} } } -1
 } }   \cdot \SF{\Lambda}^{(1)}  + O\prn{\SF{\Lambda}^{(1)} }^2   \\
 &=  \SF{\Lambda}^{(1)}  +
\brk{ \frac{
 \Ad_{ \SF{\Lambda}^{(1)} }
 }{
 1-\exp\prn{ - \Ad_{ \SF{\Lambda}^{(1)} } }
 } }   \cdot \SF{\Lambda}^{(2)}   + O\prn{\SF{\Lambda}^{(2)} }^2
\end{split}
\end{equation}

For our purposes, it is desirable to seek a \emph{quadratic} approximation that would be exact if either of $\SF{\Lambda}^{(k)}$'s turn to be a FP boost. This can be achieved by feeding the above linear approximation back into  \eqref{eq:BCHdiff}.  We get
\begin{equation}\label{eq:BCHQuad1}
\begin{split}
 \SF{\Lambda}^{(1\oplus 2)}  &= \SF{\Lambda}^{(1)} + \int_0^1 d\tau \brk{ \frac{
 \Ad_{ \SF{\Lambda}^{(1),Lin}(\tau) }
 }{
 1-\exp\prn{ - \Ad_{ \SF{\Lambda}^{(1),Lin}(\tau) } }
 } } \cdot \SF{\Lambda}^{(2)}  + O\prn{\SF{\Lambda}^{(2)} }^3\\
\text{with} &\qquad
 \SF{\Lambda}^{(1),Lin}(\tau) \equiv \SF{\Lambda}^{(1)}  + \tau
\brk{ \frac{
 \Ad_{ \SF{\Lambda}^{(1)} }
 }{
 1-\exp\prn{ - \Ad_{ \SF{\Lambda}^{(1)} } }
 } }   \cdot \SF{\Lambda}^{(2)}
\end{split}
\end{equation}
and
\begin{equation}\label{eq:BCHQuad2}
\begin{split}
 \SF{\Lambda}^{(1\oplus 2)}  &= \SF{\Lambda}^{(2)} + \int_0^1 d\tau \brk{ \frac{
 \Ad_{ \SF{\Lambda}^{(2),Lin}(\tau) }
 }{
 \exp\prn{ \Ad_{ \SF{\Lambda}^{(2),Lin}(\tau) } }-1
 } } \cdot \SF{\Lambda}^{(1)}  + O\prn{\SF{\Lambda}^{(1)} }^3\\
\text{with} &\qquad
 \SF{\Lambda}^{(2),Lin}(\tau) \equiv \SF{\Lambda}^{(2)}  + \tau
\brk{ \frac{
 \Ad_{ \SF{\Lambda}^{(2)} }
 }{
 \exp\prn{ \Ad_{ \SF{\Lambda}^{(2)} } }-1
 } }   \cdot \SF{\Lambda}^{(1)}
\end{split}
\end{equation}

We turn now to the case where both $\SF{\Lambda}^{(1)}$ and $\SF{\Lambda}^{(2)}$ are small and we want to organize BCH formula in terms of
the number of  $\SF{\Lambda}^{(1)}$ or $\SF{\Lambda}^{(2)}$ that occurs in every term. This can be done by setting $t_1=t_2=\tau$ such that
\begin{equation}\label{eq:BCHdiffsym}
\begin{split}
 \frac{d}{d\tau}  \SF{\Lambda}^{(1\oplus 2)}(\tau)  &=
 \brk{ \frac{\partial}{\partial t_1}  \SF{\Lambda}^{(1\oplus 2)}(t) + \frac{\partial}{\partial t_2}  \SF{\Lambda}^{(1\oplus 2)}(t) }_{t_1=t_2=\tau}\\
&=
\brk{ \frac{
 \Ad_{ \SF{\Lambda}^{(1\oplus 2)}(\tau) }
 }{
\exp\prn{  \Ad_{ \SF{\Lambda}^{(1\oplus 2)}(\tau) } } -1
 } }   \cdot \SF{\Lambda}^{(1)}      +
\brk{ \frac{
 \Ad_{ \SF{\Lambda}^{(1\oplus 2)}(\tau) }
 }{
 1-\exp\prn{ - \Ad_{ \SF{\Lambda}^{(1\oplus 2)}(\tau) } }
 } }   \cdot \SF{\Lambda}^{(2)}
\end{split}
\end{equation}
Invoke now the expansions
\begin{equation}
\begin{split}
\frac{z}{\exp(z)-1} &= 1-\frac{z}{2} + \sum_{k=1}^\infty \frac{B_{2k}}{(2k)!} \, z^{2k} \\
\frac{z}{1-\exp(-z)} &= 1+\frac{z}{2} + \sum_{k=1}^\infty \frac{B_{2k}}{(2k)!} \, z^{2k}
\end{split}
\end{equation}
where $B_{2k}$ are the $(2k)^{\rm th}$ Bernoulli numbers.\footnote{ The first few Bernoulli numbers are $B_2=\frac{1}{6}$, $B_4= -\frac{1}{30}$, $B_6= \frac{1}{42}$,$B_8= -\frac{1}{30}$, and $B_{10}= \frac{5}{66}$.} This  then gives
\begin{equation}\label{eq:BCHdiffsym}
\begin{split}
 \frac{d}{d\tau}  \SF{\Lambda}^{(1\oplus 2)}(\tau)
&= \SF{\Lambda}^{(1)}  + \SF{\Lambda}^{(2)} -\half  \Ad_{ \SF{\Lambda}^{(1\oplus 2)}(\tau) } \cdot \prn{ \SF{\Lambda}^{(1)}  - \SF{\Lambda}^{(2)} }\\
&\qquad +  \sum_{k=1}^\infty \frac{B_{2k}}{(2k)!} \prn{\Ad_{ \SF{\Lambda}^{(1\oplus 2)}(\tau) } }^{2k} \cdot  \prn{ \SF{\Lambda}^{(1)}  + \SF{\Lambda}^{(2)} }
\end{split}
\end{equation}

The BCH formula is just the solution of this equation with the initial condition 
$ \SF{\Lambda}^{(1\oplus 2)}(\tau=0)=0$. Setting
 $ \SF{\Lambda}^{(1\oplus 2)}(\tau) =\sum_{j=1}^\infty \tau^j  \Lambda_{(S), j}^{(1\oplus 2)}$ in the above expression and comparing the  factors of $\tau^j$ on  either side, we get a recursion relation
\begin{equation}
\begin{split}
(j+1)\ \Lambda_{(S), j+1}^{(1\oplus 2)}
&=\prn{ \SF{\Lambda}^{(1)}  + \SF{\Lambda}^{(2)} } \delta_{j,0} - \half  \Ad_{ \Lambda_{(S), j}^{(1\oplus 2)}} \cdot \prn{ \SF{\Lambda}^{(1)}  - \SF{\Lambda}^{(2)} }\\
&\qquad +  \sum_{k=1}^{2k\leq j} \sum_{p_i \in \mathbb{Z}^+}
 \frac{B_{2k}}{(2k)!}   \delta_{j,\sum_{i=1}^{2k} p_i} \prn{ \prod_{i=1}^{2k}\Ad_{ \Lambda_{(S), p_i}^{(1\oplus 2)} } } \cdot  \prn{ \SF{\Lambda}^{(1)}  + \SF{\Lambda}^{(2)} }
\end{split}
\end{equation}
which, when solved recursively, gives
\begin{equation}
\begin{split}
 \Lambda_{(S), 1}^{(1\oplus 2)} &=  \SF{\Lambda}^{(1)}  + \SF{\Lambda}^{(2)} \\
 \Lambda_{(S), 2}^{(1\oplus 2)} &=  -\quarter  \Ad_{\SF{\Lambda}^{(1)}  + \SF{\Lambda}^{(2)} } \cdot \prn{ \SF{\Lambda}^{(1)}  - \SF{\Lambda}^{(2)} }
 =  \half (\SF{\Lambda}^{(1)}, \SF{\Lambda}^{(2)})_\Kref \\
  \Lambda_{(S), 3}^{(1\oplus 2)} &=  -\frac{1}{6}  \Ad_{ \half (\SF{\Lambda}^{(1)}, \SF{\Lambda}^{(2)})_\Kref } \cdot \prn{ \SF{\Lambda}^{(1)}  - \SF{\Lambda}^{(2)} } +
   \frac{B_2}{3\times2!} \prn{\Ad_{ \SF{\Lambda}^{(1)}  + \SF{\Lambda}^{(2)} } }^2 \cdot  \prn{ \SF{\Lambda}^{(1)}  + \SF{\Lambda}^{(2)} } \\
 &=   \frac{1}{12}\prn{\SF{\Lambda}^{(1)}-\SF{\Lambda}^{(2)}, (\SF{\Lambda}^{(1)}, \SF{\Lambda}^{(2)})_\Kref }_\Kref \\
 \Lambda_{(S), 4}^{(1\oplus 2)} &=  -\frac{1}{8}  \Ad_{\frac{1}{12}\prn{\SF{\Lambda}^{(1)}-\SF{\Lambda}^{(2)}, (\SF{\Lambda}^{(1)}, \SF{\Lambda}^{(2)})_\Kref }_\Kref } \cdot \prn{ \SF{\Lambda}^{(1)}  - \SF{\Lambda}^{(2)} } \\
 &\qquad +
   \frac{B_2}{4\times2!} \Ad_{ \SF{\Lambda}^{(1)}  + \SF{\Lambda}^{(2)} } \cdot  \Ad_{ \half (\SF{\Lambda}^{(1)}, \SF{\Lambda}^{(2)})_\Kref } \cdot  \prn{ \SF{\Lambda}^{(1)}  + \SF{\Lambda}^{(2)} } \\
 &=   \frac{1}{24} \times \quarter\prn{\SF{\Lambda}^{(1)}-\SF{\Lambda}^{(2)},\prn{\SF{\Lambda}^{(1)}-\SF{\Lambda}^{(2)}, (\SF{\Lambda}^{(1)}, \SF{\Lambda}^{(2)})_\Kref }_\Kref }_\Kref\\
&\qquad  -  \frac{1}{24} \times \quarter\prn{\SF{\Lambda}^{(1)}+\SF{\Lambda}^{(2)},\prn{\SF{\Lambda}^{(1)}+\SF{\Lambda}^{(2)}, (\SF{\Lambda}^{(1)}, \SF{\Lambda}^{(2)})_\Kref }_\Kref }_\Kref \\
&= -\frac{1}{48} \prn{\SF{\Lambda}^{(2)}, (\SF{\Lambda}^{(1)},(\SF{\Lambda}^{(1)}, \SF{\Lambda}^{(2)})_\Kref)_\Kref }_\Kref -\frac{1}{48} \prn{\SF{\Lambda}^{(1)}, (\SF{\Lambda}^{(2)},(\SF{\Lambda}^{(1)}, \SF{\Lambda}^{(2)})_\Kref)_\Kref }_\Kref\\
&= -\frac{1}{24} \prn{\SF{\Lambda}^{(2)}, (\SF{\Lambda}^{(1)},(\SF{\Lambda}^{(1)}, \SF{\Lambda}^{(2)})_\Kref)_\Kref }_\Kref
\end{split}
\end{equation}
Thus, we obtain the first few terms of BCH formula quoted in the beginning of this section.

\subsection{$\UT$ representation theory}
\label{sec:repthy}
The basic  action of $\UT$ was explained at the level of the Lie algebra in Appendix~\ref{sec:UTreps}. We briefly comment now on some representations of this symmetry algebra. 

We start by defining a tower of $\UT$ representations which we call as $k$-adjoint representations. A generic $k$-adjoint superfield $\SF{\phi}_{(k)}$ transforms as 
\begin{equation}
 (\SF{\Lambda},\SF{\phi}_{(k)})_\Kref \equiv \SF{\Lambda}\, \lieD_\Kref \, \SF{\phi}_{(k)} - k \, \SF{\phi}_{(k)} \lieD_\Kref \SF{\Lambda} \,.
\end{equation}
We will refer to $k$ as the {\it adjoint weight}. We will refer to fields transforming in the $k$-adjoint representation as {\it thermal primaries of level $k$}. 

Let us briefly mention some examples. 
Comparison with Eq.\ \eqref{eq:u1T_closure} reveals that the $\UT$ gauge parameters are $1$-adjoint as one would expect. Similarly, field strengths are $1$-adjoint. The generic tensor superfields studied in \eqref{eq:tensorTrf} were $0$-adjoint. Another $1$-adjoint field which plays an important role in the discussion is $\zsf$ defined as
\begin{equation}
\zsf \equiv 1+ \SF{\Kref}^I \As_I \,.
\end{equation}
We remind the reader that this field enters into our measure owing to the $\UT$ covariant pullbacks employed. One can easily check that $\zsf$ transforms covariantly in the $1$-adjoint representation. 

The adjoint weight is additive in the sense that $\SF{\phi}_{(k)} \SF{\phi}_{(\ell)}$ transforms in the $(k+\ell)$-adjoint and $\UT$ gauge transformations act distributively on such products:
\begin{equation}
\begin{split}
(\SF{\Lambda},\SF{\phi}_{(k)}\SF{\phi}_{(\ell)})_\Kref &=  (\SF{\Lambda},\SF{\phi}_{(k)})_\Kref \SF{\phi}_{(\ell)} + (\SF{\Lambda},\SF{\phi}_{(\ell)})_\Kref \SF{\phi}_{(k)}\\
& =  \SF{\Lambda} \, \lieD_\Kref\left( \SF{\phi}_{(k)}\SF{\phi}_{(\ell)} \right)- (k+\ell) \SF{\phi}_{(k)}\SF{\phi}_{(\ell)} \,  \lieD_\Kref \, \SF{\Lambda}\,.
\end{split}
\end{equation}
We have written these expressions assuming that $\SF{\phi}_{(k)}$ and $\SF{\phi}_{(\ell)}$ are Grassmann-even superfields and thus refrained from writing various Grassmann parity induced signs in the above expressions.

\paragraph{Properties of thermal bracket:}
We define the thermal bracket between $k$-adjoint and $\ell$-adjoint fields as
\begin{equation} \label{eq:GeneralBrk}
(\SF{\phi}_{(k)} , \SF{\phi}_{(\ell)} )_\Kref \equiv k \, \SF{\phi}_{(k)} \lieD_\Kref \, \SF{\phi}_{(\ell)} - \ell \, \SF{\phi}_{(\ell)} \lieD_\Kref \, \SF{\phi}_{(k)} \,.
\end{equation}
One can easily verify that $(\SF{\phi}_{(k)} , \SF{\phi}_{(\ell)} )_\Kref$ transforms as a $(k+\ell-1)$-adjoint. This definition of the thermal bracket satisfies the following three properties (which are the defining axioms of a Poisson bracket): 
\begin{enumerate}
\item[$(i)$] $(\SF{\phi}_{(k)},\SF{\phi}_{(\ell)})_\Kref = - (\SF{\phi}_{(\ell)},\SF{\phi}_{(k)})_\Kref$ ,
\item[$(ii)$] $(\SF{\phi}_{(k)},(\SF{\phi}_{(\ell)},\SF{\phi}_{(m)})_\Kref)_\Kref + (\SF{\phi}_{(\ell)},(\SF{\phi}_{(m)},\SF{\phi}_{(k)})_\Kref)_\Kref  + (\SF{\phi}_{(m)},(\SF{\phi}_{(k)},\SF{\phi}_{(\ell)})_\Kref)_\Kref  = 0 $ ,
\item[$(iii)$] $(\SF{\phi}_{(k)},\SF{\phi}_{(\ell)}\SF{\phi}_{(m)})_\Kref = (\SF{\phi}_{(k)},\SF{\phi}_{(\ell)})_\Kref \;\SF{\phi}_{(m)} + (\SF{\phi}_{(k)},\SF{\phi}_{(m)})_\Kref\; \SF{\phi}_{(\ell)}$ .
\end{enumerate}
It is natural to define an operator $\Adj$, which gives the adjoint weight of any thermal primary, i.e., 
\begin{equation}
 \Adj \SF{\phi}_{(k)} = k \, \SF{\phi}_{(k)} \,.
\end{equation} 
In terms of this operator, the formal similarity between thermal bracket and a Poisson bracket can be made more manifest:
\begin{equation}
 (\SF{\phi}_{(k)},\SF{\phi}_{(\ell)})_\Kref = \left( \Adj \SF{\phi}_{(k)} \right) \left( \lieD_\Kref \SF{\phi}_{(\ell)} \right) - \left( \lieD_\Kref \SF{\phi}_{(k)} \right)\left( \Adj \SF{\phi}_{(\ell)} \right)   \,.
\end{equation}

Let us now consider covariant derivatives on $k$-adjoint fields. 
We define the $\UT$ covariant derivative in the $k$-adjoint representation as
\begin{equation}\label{eq:CovDerAdj}
\Dut_I \SF{\phi}_{(k)} \equiv \partial_I \SF{\phi}_{(k)} + (\As_I,\SF{\phi}_{(k)} )_\Kref \,,
\end{equation}
where $\As_I$ can be viewed as a $1$-adjoint for the purpose of defining the bracket on the right hand side according to the general definition \eqref{eq:GeneralBrk}. The $\UT$ covariant derivative respects the adjoint weight, for  $\Dut_I \SF{\phi}_{(k)}$ is again a $k$-adjoint, and it furthermore distributes over both the thermal bracket and usual products: 
\begin{equation}
\begin{split}
 \Dut_I (\SF{\phi}_{(k)} , \SF{\phi}_{(\ell)} )_\Kref& = (\mathcal{D}_I\SF{\phi}_{(k)} , \SF{\phi}_{(\ell)} )_\Kref + (\SF{\phi}_{(k)} ,\Dut_I\SF{\phi}_{(\ell)} )_\Kref \,,\\
 \Dut_I (\SF{\phi}_{(k)}\SF{\phi}_{(\ell)} ) &=  \SF{\phi}_{(\ell)}\Dut_I \SF{\phi}_{(k)} + \SF{\phi}_{(k)}\Dut_I \SF{\phi}_{(\ell)} \,.
  \end{split}
\end{equation}

\paragraph{Thermal descendants:}
Given a $k$-adjoint field $\SF{\phi}_{(k)}$, we would like to understand the properties of $\lieD_\Kref \SF{\phi}_{(k)}$. Let us try to figure out the induced action of $\Adj$ on $\lieD_\Kref \SF{\phi}_{(k)}$ by computing the following object in two ways: 
\begin{equation}
\begin{split}
 \Adj (\SF{\phi}_{(k)},\SF{\phi}_{(\ell)})_\Kref &= \Adj \left[ k \SF{\phi}_{(k)} \, \lieD_\Kref \, \SF{\phi}_{(\ell)} - \ell \SF{\phi}_{(\ell)} \, \lieD_\Kref \, \SF{\phi}_{(k)} \right]\\
 &= k^2 \SF{\phi}_{(k)} \, \lieD_\Kref \, \SF{\phi}_{(\ell)} - \ell^2 \SF{\phi}_{(\ell)} \, \lieD_\Kref \, \SF{\phi}_{(k)} + k \SF{\phi}_{(k)} \, \Adj \lieD_\Kref \, \SF{\phi}_{(\ell)} - \ell \SF{\phi}_{(\ell)} \, \Adj \lieD_\Kref \, \SF{\phi}_{(k)} \,,\\
 \Adj (\SF{\phi}_{(k)},\SF{\phi}_{(\ell)})_\Kref  &= (k+\ell-1) (\SF{\phi}_{(k)},\SF{\phi}_{(\ell)})_\Kref  \\
 &= k^2 \SF{\phi}_{(k)} \, \lieD_\Kref \, \SF{\phi}_{(\ell)} - \ell^2 \SF{\phi}_{(\ell)} \, \lieD_\Kref \, \SF{\phi}_{(k)} + k (\ell-1) \SF{\phi}_{(k)} \,  \lieD_\Kref \, \SF{\phi}_{(\ell)} - \ell (k-1) \SF{\phi}_{(\ell)} \, \lieD_\Kref \, \SF{\phi}_{(k)} \,.
\end{split}
\end{equation}
By comparison of these two expressions, we conclude that the following should hold:
\begin{equation}
 \Adj \lieD_\Kref \, \SF{\phi}_{(k)} = (k-1) \lieD_\Kref \, \SF{\phi}_k + \alpha \, k \, \SF{\phi}_{(k)} \,,
\end{equation}
where $\alpha$ is some undetermined object (which must be independent of $k$). It appears tempting to fix $\alpha =0$, but we have not managed to convince ourselves that this is the only possibility (eg.,  adjoint weight counting would also allow for the possibility of $\alpha = \zsf^{-1}$).  We leave this a curious observation that should be better understood.

\paragraph{Deformation quantization:}
For amusement, let us try to quantize the thermal bracket. We first write the thermal bracket in suggestive notation as follows: 
\begin{equation}
(\SF{\phi}_{(k)},\SF{\phi}_{(\ell)})_\Kref = \SF{\phi}_{(k)} \left( \overleftarrow{\Adj} \,\overrightarrow{\lieD_\Kref} - \overleftarrow{\lieD_\Kref} \,\overrightarrow{\Adj} \right) \SF{\phi}_{(\ell)} \,.
\end{equation}
Next, we introduce a star product: 
\begin{equation}
  \SF{\phi}_{(k)} \star \SF{\phi}_{(\ell)} \equiv \SF{\phi}_{(k)} \, \exp \left[ \frac{i\hbar}{2} \left( \overleftarrow{\Adj} \,\overrightarrow{\lieD_\Kref} - \overleftarrow{\lieD_\Kref} \,\overrightarrow{\Adj} \right) \right] \SF{\phi}_{(\ell)} 
\end{equation}
This allows us to define the following deformation of the thermal bracket: 
\begin{equation}
\begin{split}
 (\SF{\phi}_{(k)}, \SF{\phi}_{(\ell)})^\star_\Kref 
 &\equiv  \frac{1}{i\hbar} \left( \SF{\phi}_{(k)} \star \SF{\phi}_{(\ell)} -  \SF{\phi}_{(\ell)} \star \SF{\phi}_{(k)} \right) \\
 &\equiv  \frac{2}{\hbar} \, \SF{\phi}_{(k)} \, \sin \left[ \frac{\hbar}{2} \left( \overleftarrow{\Adj} \,\overrightarrow{\lieD_\Kref} - \overleftarrow{\lieD_\Kref} \,\overrightarrow{\Adj} \right) \right] \SF{\phi}_{(\ell)} \\
 &= (\SF{\phi}_{(k)}, \SF{\phi}_{(\ell)})_\Kref - \frac{1}{3!} \left( \frac{\hbar}{2} \right)^2  \SF{\phi}_{(k)} \, \left( \overleftarrow{\Adj} \,\overrightarrow{\lieD_\Kref} - \overleftarrow{\lieD_\Kref} \,\overrightarrow{\Adj} \right)^3 \SF{\phi}_{(\ell)} + \mathcal{O}(\hbar^5) \,.
\end{split}
\end{equation}
By associativity of the star product, this deformed bracket continues to satisfy the Jacobi identity. 

It is interesting to speculate that this deformed thermal bracket plays a role in the thermal equivariant algebra at finite temperatures $\beta \sim \hbar$ (recall from \S\ref{sec:thermaleq} that we primarily focus on the high temperature limit). This idea has been previously considered in the deformation quantization literature, see for eg., \cite{Basart:1984aa} and subsequent developments in \cite{Bordemann:1998aa,Bordemann:1999aa}.

\newpage
\part{Hydrodynamic sigma models: Mathematical details}
\section{Superspace representations: position multiplet}
\label{sec:PositionMult}

As we saw before, the $\UT$ gauge theory naturally comes with a $\UT$ gauge field quartet, a Faddeev-Popov ghost triplet, and a gauge parameter quintet. We will now formulate the gauged sigma-model of a Brownian brane. This requires introducing a position multiplet consisting of an embedding coordinate of the brane in the target space and its completion into a superfield.

\subsection{Basic definitions}
\label{sec:PositionMultDefs}

As indicated in the main text we should upgrade physical spacetime to a super-geometry with Grassmann odd coordinates ${\Theta}, {\bar{\Theta}}$. We collect these together with the position superfield into 
$$ {X}^{\uA} = \{{X}^\mu, {\Theta}, {\bar{\Theta}}\} \,.$$ 
We will take the spacetime to be endowed with  a (non-dynamical) metric $g_{\uA\uB}({X^{\uC}})$ and a  Christoffel connection $\Gamma^{\uA}{}_{\uB\uC}({X}) $ with a curvature 
$R^{\uA}{}_{\uB\uC\uD}({X})$.  We will furthermore fix a-priori some metric components. Our choice is described in the main text \eqref{eq:gsptgauge} which we reproduce here for convenience:
\begin{equation}
 \qquad 
g_{\mu\Theta} = g_{\mu \bar{\Theta}} =0 \,, \qquad g_{\Theta \bar{\Theta}} = -g_{\bar{\Theta} \Theta} = i
\,, \qquad g_{\mu\nu} = g_{\mu\nu}(X^\rho) \,.
\label{eq:gsptgauge1}
\end{equation}	
This is a canonical choice in the literature on supermanifolds \cite{DeWitt:1992cy}. The target space super-connection $\nabla_{\uA}$ is simply derived from the usual Christoffel symbol (upgraded to superspace).

In hydrodynamics, we think of the physical spacetime as the target space of the Brownian brane parameterized by a worldvolume spacetime with coordinates $\sigma^a$.  The thermal supergeometry structure discussed earlier implies that we  should upgrade the  Brownian particle worldvolume to a superspace with super-coordinates  $z^I = \{\sigma^a,\theta,\thetab \}$ such that $X^{\uA}=X^{\uA}(z^I)$. Note that the Brownian brane preserves the full twisted supersymmetry of the target space. This is in contrast to usual supersymmetry where at least half the supersymmetries are broken by the worldline. This is also related to the fact that there are no additional goldstino-like  Grassmann odd modes living on the Brownian brane.

Let us work in RNS formalism where we ignore the target space twisted supersymmetry along with $\{\Theta,\bar{\Theta}\}$ and focus only on worldvolume supersymmetry and superspace. Equivalently, we can make the super-static gauge choice \eqref{eq:superstatic}, viz.,  $\{\Theta(\sigma^a,\theta,\thetab )= \theta,\bar{\Theta}(\sigma^a,\theta,\thetab )= \thetab \}$, thus aligning the target space Grassmann odd directions to the Grassmann odd directions in the world-volume. We will work out the implications of this super-static gauge choice in Appendix~\ref{sec:gaugesdiff}. With this assumption, pullbacks along $\Theta$'s are trivial. We can then develop an RNS formalism involving just the worldvolume twisted superspace.  

For readers familiar with string theory, some clarifying comments are in order.
In usual perturbative string theory in RNS formalism, target space supersymmetry is eventually restored by
a series of tricks on the worldsheet: doing a GSO projection, constructing spin-fields, picture-changing, summing over  spin structures etc. The corresponding procedure in D-branes (or for that matter, any RR background) is however unknown. D-branes are usually dealt with by taking target space supersymmetry seriously from the beginning and shifting to Green-Schwarz (GS) formalism instead. For strings, GS formalism brings its own problems: it obscures worldsheet supersymmetry (or more precisely, it buries it into $\kappa$-symmetry) and in turn, target space Lorentz invariance which then needs to be dealt with  by either passing to light-cone or pure-spinor gauge. It is unclear, as of now, how these issues translate to the twisted sigma-models we are interested in. We will just plough ahead and work with Brownian branes in RNS formalism hoping that the  usual obstructions are alleviated when considering twisted supersymmetry.

Target space tensors can viewed as worldvolume superfields once we have the super-embedding $\SF{X}^{uA}$. The bottom components of these tensors when are projected  to $\{\theta=0,\thetab =0\}$ subspace will be the basic target tensor superfield. To simplify notation we will routinely drop their arguments, for example we write,  $g_{\mu\nu}\equiv g_{\mu\nu}({X})|=g_{\mu\nu}({X}|)$ and similarly for any other target super-tensor field.

In general, we define the Lie derivative along the thermal super-vector, $\SF{\Kref}$ as 
\begin{equation}
\lieD_\Kref \SF{T}^{\uA\cdots}{}_{\uC\cdots} \equiv 	
	\SF{\Kref}^I \partial_I\left( \SF{T}^{\uA\cdots}{}_{\uC\cdots} \right)
\end{equation}
for $\SF{T}^{\uA\cdots}{}_{\uC\cdots}$ being any partially pulled back tensor superfield which is not in 
$\UT$ adjoint. The superspace Lie derivative is the standard one defined in \cite{DeWitt:1992cy}. We define the $\UT$ super-covariant derivative of the sigma-model position superfield 
$\SF{X}^{\uA}$ via
\begin{equation}
\Dut_I  \SF{X}^{\uA} \equiv  \partial_I \SF{X}^{\uA} + \dbrk{\As_I,\SF{X}^{\uA}} = \partial_I \SF{X}^{\uA} + \As_I \,\SF{\Kref}^J\partial_J \SF{X}^{\uA} \,.
\end{equation}
An advantage of the super-static gauge choice \eqref{eq:superstatic} together with the gauge-fixing of the thermal super-vector \eqref{eq:betagauge} is that $\Dut_\theta \Theta = \Dut_{\thb} \bar{\Theta} =1$ and all other derivatives of the target space Grassmann coordinates vanish. As a result for the most part we can 
simplify the discussion to focus only target space tensors with only ordinary space components. 

On a general partially pulled back target space tensor superfield $\SF{T}^{\uA\cdots}{}_{\uC\cdots}$ the fully covariant derivative will be defined as\footnote{ There are various signs to keep track of in actual computations since we have now Grassmann indices in both target spacetime and the worldvolume. We will refrain from writing these down explicitly unless absolutely necessary.} 
\begin{equation} \label{eq:DcovVS}
\begin{split}
  \Dut_I \SF{T}^{\uA\cdots}{}_{\uC\cdots} &\equiv  \partial_I \SF{T}^{\uA\cdots}{}_{\uC\cdots} 
  + \Gamma^{\uA}{}_{\uB\uD}(\SF{X})\ \Dut_I \SF{X}^{\uD} \; \SF{T}^{\uB\cdots}{}_{\uC\cdots} + \ldots \\
   &\qquad\qquad\quad\; - \Gamma^{\uB}{}_{\uC\uD}(\SF{X})\ \Dut_I \SF{X}^{\uD}\; \SF{T}^{\uA\cdots}{}_{\uB\cdots} + \ldots
   + (\SF{\Ascr}_I \,, \SF{T}^{\uA\cdots}{}_{\uC\cdots})_\Kref \\
       &= \SF{\Dref}_I \SF{T}^{\uA\cdots}{}_{\uC\cdots} + \dbrk{\As_I ,\,\SF{T}^{\uA\cdots}{}_{\uC\cdots}} \,,
\end{split}
\end{equation}
where target space connection terms are absorbed into the partial pullback derivative:
\begin{equation}\label{eq:CovDref}
 \SF{\Dref}_I\SF{T}^{\uA\cdots}{}_{\uC\cdots} \equiv 
 \partial_I \SF{T}^{\uA\cdots}{}_{\uC\cdots} + \Gamma^{\uA}{}_{\uB\uD}(\SF{X})\ 
 \partial_I \SF{X}^{\uD} \ \SF{T}^{\uB\cdots}{}_{\uC\cdots} + \ldots  - \Gamma^{\uB}{}_{\uC\uD}(\SF{X})\ \partial_I \SF{X}^{\uD} \ \SF{T}^{\uA\cdots}{}_{\uB\cdots} + \ldots
\end{equation}
Furthermore, the thermal bracket $\dbrk{\cdot\,,\cdot}$ is defined on any tensors other than $\SF{X}^{\uA}$ by
\begin{equation}
\begin{split}
 \dbrk{\SF{\Lambda} ,\,\SF{T}^{\uA\cdots}{}_{\uC\cdots}} &\equiv  
 \SF{\Lambda}\  \SF{\Kref}^I \SF{\Dref}_I \SF{T}^{\uA\cdots}{}_{\uC\cdots}
\end{split}
\end{equation}
for $\SF{\Lambda}$ in the adjoint of $\UT$ the tensor $\SF{T}^{\uA\cdots}{}_{\uC\cdots}(\SF{X})$ is partially pulled back tensor superfield which is not in the adjoint of $\UT$.

Note that the bracket $\dbrk{\,\cdot\,,\,\cdot}$ defined with covariant derivatives as in 
\eqref{eq:CovDref} still satisfies a Jacobi identity, despite involving a covariant derivative.  That is, for super tensor fields $\SF{A},\SF{B},\SF{C}$ all in the adjoint of $\UT$ we have the thermal bracket (leaving Grassmann index and contraction  signs implicit):
\begin{equation}
\dbrk{\SF{A}^{\uA\cdots}_{\uC\cdots},\,\SF{B}^{\uB\cdots}_{\uD\cdots}} 
\equiv 
	\SF{A}^{\uA\cdots}_{\uC\cdots} \prn{ \SF{\Kref}^J\SF{\Dref}_J \SF{B}^{\uB\cdots}_{\uD\cdots}} -
	 (-)^{AB} \, \SF{B}^{\uB\cdots}_{\uD\cdots} \prn{ \SF{\Kref}^J\SF{\Dref}_J \SF{A}^{\uA\cdots}_{\uC\cdots} } \,.
\end{equation}
which implies:
\begin{equation}
\begin{split}
 (-)^{AC} \, \dbrk{ \SF{A}^{\uA_1\cdots}_{\uA_2\cdots} ,\, \dbrk{\SF{B}^{\uB_1\cdots}_{\uB_2\cdots} ,\, \SF{C}^{\uC_1\cdots}_{\uC_2\cdots}}} 
&+(-)^{CB} \, \dbrk{\SF{C}^{\uC_1\cdots}_{\uC_2\cdots}  ,\, \dbrk{\SF{A}^{\uA_1\cdots}_{\uA_2\cdots} ,\, \SF{B}^{\uB_1\cdots}_{\uB_2\cdots}}} \\
& +(-)^{BA} \, \dbrk{\SF{B}^{\uB_1\cdots}_{\uB_2\cdots}  ,\,\dbrk{ \SF{C}^{\uC_1\cdots}_{\uC_2\cdots} ,\, \SF{A}^{\uA_1\cdots}_{\uA_2\cdots}}}  = 0 \,.
\end{split}
\end{equation}	

It is useful to note that adding the $\Gamma^{\uA}{}_{\uB\uC}(\SF{X})$ terms in the definition of $\Dut_I$ makes $\Dut_I$ target-space metric compatible, viz.,
\begin{equation}
\begin{split}
\Dut_I g_{\uA\uB}(\SF{X}) 
&\equiv 
	\partial_I g_{\uA\uB}(\SF{X}) + \SF{\Ascr}_I\, \SF{\Kref}^J\partial_J g_{\uA\uB}(\SF{X}) \\
&\qquad 
	 -  g_{\uA\uC}(\SF{X})\, \Gamma^{\uC}{}_{\uB\uD}(\SF{X}) \,  \Dut_I \SF{X}^{\uD}
	 -  g_{\uC\uB}(\SF{X})\, \Gamma^{\uC}{}_{\uA\uD}(\SF{X})  \, \Dut_I \SF{X}^{\uD}\\
&= 
	\Dut_I \SF{X}^{\uD} \bigbr{ \nabla_{\uD} g_{\uA\uB}(\SF{X}) } = 0
\end{split}
\end{equation}
where we used target space metric-compatibility of $\nabla_{\uD}$ in the last step. In fact, on any target space tensor we observe that $\Dut_I$ acts via $(\Dut_I \SF{X}^{\uB}) \nabla_{\uB}$, where $\nabla$ is the conventional, metric compatible, covariant derivative on the target space. 
By similar reasoning we also find:
\begin{equation}
\SF{\Dref}_I g_{\uA\uB} (\SF{X}) = \Dref_I \SF{X}^{\uD} \prn{ \nabla_{\uD} g_{\uA\uB}(\SF{X}) } = 0 \,.
\end{equation}
%

\subsection{Component fields the WZ gauge}

We now turn to the component fields of $\SF{X}^{\uA}$. Working in super-static gauge has the advantage that the Grassmann coordinates have simple expansions given in \eqref{eq:superstatic}.  So one only needs to worry about the superspace expansion of $\SF{X}^\mu$.  They are covariantly defined via
\begin{equation}
\begin{split}
X^\mu &\equiv \SF{X}^\mu|\ ,\quad
 \xpsib^\mu \equiv \Dut_\theta \SF{X}^\mu  | \ ,\quad
 \xpsi^\mu \equiv \Dut_{\thetab } \SF{X}^\mu |\ ,\quad
\tilde{X}^\mu \equiv \Dut_\theta \Dut_{\thetab } \SF{X}^\mu|\ .
 \end{split}
\end{equation}
The superspace expansion of these identities reads 
\begin{equation} \label{eq:Xsuper}
\begin{split}
\SF{X}^\mu &= X^\mu + \thetab  \bigbr{\xpsi^\mu -  \GhT \Kref.\partial X^\mu} + \theta \bigbr{\xpsib^\mu -  \GhbT \Kref.\partial X^\mu} \\
&\quad + \thetab   \theta \Bigl\{\tilde{X}^\mu -\Gamma^\mu_{\nu\lambda} \xpsib^\nu \xpsi^\lambda
+ \GhT \Kref.\partial (\xpsib^\mu-  \GhbT \Kref.\partial X^\mu)-\GhbT \Kref.\partial (\xpsi^\mu- \GhT \Kref.\partial X^\mu)  \Bigr.\\
&\qquad \Bigl. \qquad
-\BT\ \Kref.\partial X^\mu - \GhbT \GhT \Kref.\partial (\Kref.\partial X^\mu) \ \Bigr\}
\end{split}
\end{equation}
We note that these expressions are written for simplicity after imposing the gauge condition \eqref{eq:betagauge} on the thermal super-vector 
$\SF{\Kref}^I$ and hence, $\Kref.\partial \equiv \Kref^a \partial_a$ here and below.

The super-gauge transformation on the position superfield  is given by
\begin{equation}
\begin{split}
\SF{X}^\mu &\mapsto \SF{X}^\mu+\lieD_{\SF{\Lambda}\Kref}\SF{X}^\mu + \frac{1}{2!} \lieD_{\SF{\Lambda}\Kref}\lieD_{\SF{\Lambda}\Kref}\SF{X}^\mu
  + \frac{1}{3!}  \lieD_{\SF{\Lambda}\Kref}\ \lieD_{\SF{\Lambda}\Kref}\lieD_{\SF{\Lambda}\Kref}\SF{X}^\mu
+ \ldots
\end{split}
\end{equation}
where the Lie-derivative acts via  $\lieD_{\SF{\Lambda}\Kref}\SF{X}^\mu \equiv \Kref^a \partial_a \SF{X}^\mu$. This series truncates  when $\SF{\Lambda}$ is a FP boost:
\begin{equation}
\begin{split}
\SF{X}^\mu &\mapsto \SF{X}^\mu+\lieD_{\SF{\Lambda}\Kref}\SF{X}^\mu + \frac{1}{2!} \lieD_{\SF{\Lambda}\Kref}\lieD_{\SF{\Lambda}\Kref}\SF{X}^\mu.
\end{split}
\end{equation}
This works out to
\begin{equation}
\begin{split}
 \SF{X}^\mu &\mapsto \SF{X}^\mu\bigbr{\begin{array}{c} \GhT\mapsto \GhT- \Lambda_\psi, \\
 \GhbT\mapsto \GhbT- \Lambda_{\psib},\\ \BT\mapsto \BT- \tilde{\Lambda}
 + \half (\Lambda_{\psib},\GhT)_\Kref- \half (\Lambda_{\psi},\GhbT)_\Kref \end{array} }
 \quad \text{with} \quad \prn{\begin{array}{c} X^\mu \\ \xpsi^\mu \\ \xpsib^\mu \\ \tilde{X}^\mu \end{array}}
  \quad \text{fixed.}
\end{split}
\end{equation}
Thus, FP boosts shift only the FP triplets as expected. In particular, in WZ gauge, we have
\begin{equation}
\begin{split}
(\SF{X}^\mu)_{WZ} &\equiv X^\mu + \thetab \ \xpsi^\mu + \theta\ \xpsib^\mu  +  \thetab \theta\ (\tilde{X}^\mu-\Gamma^\mu_{\nu\lambda} \xpsib^\nu \xpsi^\lambda )
\end{split}
\end{equation}
When we discuss the target space symmetries we will see that the Christoffel connection piece is induced in the top-component for reasons of spacetime covariance under the partial pullback; see the end of Appendix~\ref{sec:gaugesdiff}.

\subsection{Covariant derivatives and Bianchi identities}

The superfields $\Dut_I \SF{X}^\mu$ and its further derivatives transform as target space vector and our previous discussion of such fields can directly be  used to give their superspace expansion, gauge transformation etc. The only required data is the basic
quartet of components which make up the target space vector superfield under question. We  give below the building components for  $\Dut_I \SF{X}^\mu$:
\begin{subequations}
\begin{align}
&\Dut_a \SF{X}^\mu:\;  
	\prn{\begin{array}{c} 
	D_a X^\mu \\ 
	D_a \xpsi^\mu + (\lambda_a,X^\mu)_\Kref \\ 
	D_a{\xpsib}^\mu+ (\bar{\lambda}_a,X^\mu)_\Kref \\
 	D_a\tilde{X}^\mu + (\sAt{a},X^\mu)_\Kref + (\bar{\lambda}_a,\xpsi^\mu)_\Kref 
 		-(\lambda_a,\xpsib^\mu)_\Kref+ R^\mu{}_{\nu\alpha\sigma}\xpsib^\alpha (D_aX^\sigma)\xpsi^\nu
 	\end{array}} \ ,
 \\
& \Dut_{\thetab} \SF{X}^\mu:\;
 	\prn{\begin{array}{c} 
 		\xpsi^\mu \\  
 		(\phiT,X^\mu)_\Kref \\ 
 		\tilde{X}^\mu \\
		(\phiT,\xpsib^\mu)_\Kref -(\etaT,X^\mu)_\Kref
		\end{array}} \ ,
 \\
 &\Dut_\theta \SF{X}^\mu :\; 
 	\prn{\begin{array}{c} 
 		\xpsib^\mu \\  
 		(\phizT,X^\mu)_\Kref -\tilde{X}^\mu \\
 		(\phibT,X^\mu)_\Kref \\
 		 (\phizT,\xpsib^\mu)_\Kref+ R^\mu{}_{\nu\alpha\sigma}\xpsib^\alpha \xpsi^\sigma {\xpsib}^\nu
 			 -(\phibT,\xpsi^\mu)_\Kref+(\etabT,X^\mu)_\Kref
 	\end{array}} \ .
\end{align}
\end{subequations}

Further super-covariant derivatives of  target space vector superfields $\Dut_I \SF{X}^\mu$ give vector superfields. Using \eqref{eq:DcovVS}, we thus define the double covariant derivative on the position superfield via
\begin{equation}
\Dut_I \Dut_J \SF{X}^\mu \equiv \partial_I (\Dut_J \SF{X}^\mu)+ \dbrk{\SF{\Ascr}_I, \Dut_J \SF{X}^\mu}
+ \Gamma^\mu{}_{\nu\alpha}(\SF{X})\   (\partial_I \SF{X}^\alpha) (\Dut_J \SF{X}^\nu)\,,
\end{equation}
These derivatives can then be computed directly, and we have for example
\begin{subequations}
\begin{align}
\Dut_{\theta}\Dut_{\thetab } \SF{X}^\mu &: \,  
	\prn{\begin{array}{c}
		 \tilde{X}^\mu \\
		(\phizT,\xpsi^\mu)_\Kref+R^\mu{}_{\nu\alpha\sigma}\xpsib^\alpha\xpsi^\sigma \xpsi^\nu-(\phiT,\xpsib^\mu)_\Kref+(\etaT,X^\mu)_\Kref\\
		 (\phibT,\xpsi^\mu)_\Kref+\half R^\mu{}_{\nu\alpha\sigma}\xpsib^\alpha\xpsib^\sigma \xpsi^\nu \\
		(\phizT,\tilde{X}^\mu)_\Kref+R^\mu{}_{\nu\alpha\sigma}\xpsib^\alpha\xpsi^\sigma\tilde{X}^\nu -(\phibT,(\phiT,X^\mu)_\Kref)_\Kref
		-\half R^\mu{}_{\nu\alpha\sigma}\xpsib^\alpha\xpsib^\sigma (\phiT,X^\nu)_\Kref+(\etabT,\xpsi^\mu)_\Kref\\
		\qquad \quad 
			+\brk{(\xpsib.\nabla R^\mu{}_{\nu\alpha\sigma})\xpsib^\alpha \xpsi^\sigma
			-  R^\mu{}_{\nu\alpha\sigma}\xpsi^\alpha (\phibT,X^\sigma)_\Kref
 			+ R^\mu{}_{\nu\alpha\sigma}\xpsib^\alpha\prn{(\phizT,X^\sigma)_\Kref-\tilde{X}^\sigma} } \xpsi^\nu
 	\end{array}} \\
 \Dut_{\thetab}\Dut_{\theta } \SF{X}^\mu &:\,  
	\prn{\begin{array}{c} 
		(\phizT,X^\mu)_\Kref-\tilde{X}^\mu \\
 		(\phiT,\xpsib^\mu)_\Kref+\half R^\mu{}_{\nu\alpha\sigma}\xpsi^\alpha\xpsi^\sigma \xpsib^\nu \\
		(\phizT,\xpsib^\mu)_\Kref+R^\mu{}_{\nu\alpha\sigma}\xpsib^\alpha\xpsi^\sigma \xpsib^\nu-(\phibT,\xpsi^\mu)_\Kref+(\etabT,X^\mu)_\Kref\\
		(\phiT,(\phibT,X^\mu)_\Kref)_\Kref +\half R^\mu{}_{\nu\alpha\sigma}\xpsi^\alpha\xpsi^\sigma (\phibT,X^\nu)_\Kref
		 -(\etaT,\xpsib^\mu)_\Kref\qquad \qquad \qquad \\
		\hspace{.5cm}
			 -\brk{(\xpsi.\nabla R^\mu{}_{\nu\alpha\sigma})\xpsib^\alpha \xpsi^\sigma
			- R^\mu{}_{\nu\alpha\sigma} \xpsib^\alpha(\phiT,X^\sigma)_\Kref
			+R^\mu{}_{\nu\alpha\sigma}\xpsi^\alpha\tilde{X}^\sigma  }\xpsib^\nu
	\end{array}}
\end{align}
\end{subequations}

These double super-covariant derivatives obey the commutation relations
\begin{equation}
 \Dut_I \Dut_J \SF{X}^\mu -(-)^{IJ} \Dut_J \Dut_I \SF{X}^\mu = (1+ \delta_{IJ}) (\SF{\mathscr{F}}_{IJ}, \SF{X}^\mu )_\Kref 
\label{eq:doubleDXs0}
\end{equation}	
or more explicitly
\begin{equation}\label{eq:doubleDXs}
\begin{split}
\Dut_{\theta}^2 \SF{X}^\mu - (\SF{\mathscr{F}}_{\theta\theta},\SF{X}^\mu)_\Kref &= 0
= \Dut_{\thetab }^2 \SF{X}^\mu - (\SF{\mathscr{F}}_{\thetab \thetab },\SF{X}^\mu)_\Kref \\
\Dut_{\thetab } \Dut_{\theta} \SF{X}^\mu  + \Dut_{\theta} \Dut_{\thetab } \SF{X}^\mu  &= (\SF{\mathscr{F}}_{\theta\thetab },\SF{X}^\mu)_\Kref \\
\Dut_{\thetab } \Dut_a \SF{X}^\mu  - \Dut_a \Dut_{\thetab } \SF{X}^\mu  &= (\SF{\mathscr{F}}_{\thetab a},\SF{X}^\mu)_\Kref \\
\Dut_{\theta} \Dut_a \SF{X}^\mu  - \Dut_a \Dut_{\theta} \SF{X}^\mu  &= (\SF{\mathscr{F}}_{\theta a},\SF{X}^\mu)_\Kref \\
\end{split}
\end{equation}

Also the triple covariant derivative on the position superfield is defined via \eqref{eq:DcovVS}. This immediately yields the identity
\begin{equation}
\begin{split}
\Dut_I \Dut_J &\Dut_K \SF{X}^\mu -(-)^{IJ} \Dut_J \Dut_I \Dut_K \SF{X}^\mu \\
&= (1+ \delta_{IJ}) \dbrk{\SF{\mathscr{F}}_{IJ},\, \Dut_K \SF{X}^\mu}
+ R^\mu{}_{\nu\alpha\beta}(\SF{X})\   (\Dut_I \SF{X}^\alpha) (\Dut_J \SF{X}^\beta) (\Dut_K \SF{X}^\nu).
\end{split}
\end{equation}
Explicitly, we find, for example,\footnote{ When dealing with curvatures, we often use the Bianchi identity in the form
\begin{equation}
\frac{1}{2} R^\mu{}_{\nu\alpha\beta} \xpsi^\alpha \xpsi^\beta \xpsib^\nu = - R^\mu{}_{\nu\alpha\beta} \xpsi^\alpha \xpsib^\beta \xpsi^\nu
\end{equation}
and similarly for other Grassmann-odd contractions of the same schematic form.
}
\begin{equation}\label{eq:tripleXs}
\begin{split}
  [\Dut_a,\Dut_b] \Dut_K \SF{X}^\mu &= \dbrk{\SF{\mathscr{F}}_{ab},\, \Dut_K \SF{X}^\mu}+ R^{\mu}{}_{\nu\alpha\beta}(\SF{X})\, (\Dut_a \SF{X}^\alpha) (\Dut_{b}\SF{X}^\beta) (\Dut_{K}\SF{X}^\nu)    \,,\\
  [\Dut_a,\Dut_\theta] \Dut_K \SF{X}^\mu &=\dbrk{\SF{\mathscr{F}}_{a\theta},\, \Dut_K \SF{X}^\mu}+ R^{\mu}{}_{\nu\alpha\beta}(\SF{X})\, (\Dut_a \SF{X}^\alpha) (\Dut_{{\theta}}\SF{X}^\beta) (\Dut_{K}\SF{X}^\nu)    \,,\\
  [\Dut_a,\Dut_{\bar \theta}] \Dut_K \SF{X}^\mu &= \dbrk{\SF{\mathscr{F}}_{a\thetab },\, \Dut_K \SF{X}^\mu}+  R^{\mu}{}_{\nu\alpha\beta}(\SF{X})\, (\Dut_a \SF{X}^\alpha) (\Dut_{\thetab }\SF{X}^\beta) (\Dut_{K}\SF{X}^\nu) \,,\\
  \{\Dut_\theta,\Dut_{\thetab }\} \Dut_K \SF{X}^\mu &= \dbrk{\SF{\mathscr{F}}_{\theta\thetab } ,\, \Dut_K \SF{X}^\mu }
        +  R^{\mu}{}_{\nu\alpha\beta}(\SF{X})\, (\Dut_{{\theta}}\SF{X}^\alpha) (\Dut_{{\bar\theta}}\SF{X}^\beta) (\Dut_{K}\SF{X}^\nu) \,,\\
  \Dut_\theta^2\Dut_{K} \SF{X}^\mu &= \dbrk{\SF{\mathscr{F}}_{\theta\theta} ,\,\Dut_{K} \SF{X}^\mu } + \frac{1}{2} \, R^{\mu}{}_{\nu\alpha\beta}(\SF{X})\, (\Dut_{{\theta}}\SF{X}^\alpha) (\Dut_{{\theta}}\SF{X}^\beta) (\Dut_{K}\SF{X}^\nu) \,,\\
  \Dut_{\bar \theta}^2\Dut_K \SF{X}^\mu &= \dbrk{\SF{\mathscr{F}}_{\bar\theta\bar\theta}  ,\, \Dut_K \SF{X}^\mu } + \frac{1}{2} \, R^{\mu}{}_{\nu\alpha\beta}(\SF{X})\, (\Dut_{\thetab }\SF{X}^\alpha) (\Dut_{\thetab }\SF{X}^\beta) (\Dut_{K}\SF{X}^\nu) \,,\\
 \Dut_{\bar\theta} \Dut_\theta \Dut_{\thetab } \SF{X}^\mu &= \dbrk{\SF{\mathscr{F}}_{\theta\bar\theta} ,\,\Dut_{\bar \theta} \SF{X}^\mu } - \dbrk{\SF{\mathscr{F}}_{\thetab \thetab } ,\,\Dut_{\theta} \SF{X}^\mu }  - \Dut_\theta \SF{\mathscr{F}}_{\bar\theta\bar\theta} \, \Kref.\partial \SF{X}^\mu\\
    &\quad  + R^{\mu}{}_{\nu\alpha\beta}(\SF{X})\, (\Dut_{{\theta}}\SF{X}^\alpha) (\Dut_{{\bar\theta}}\SF{X}^\beta) (\Dut_{{\bar\theta}}\SF{X}^\nu) \,,
\end{split}
\end{equation}

Let us now turn to the case of four covariant derivatives acting on $\SF{X}^\mu$. The basic commutator identity in this case is
\begin{equation}
\begin{split}
\Dut_I \Dut_J &\Dut_K \Dut_L \SF{X}^\mu -(-)^{IJ} \Dut_J \Dut_I \Dut_K \Dut_L \SF{X}^\mu \\
&= (1+ \delta_{IJ}) \dbrk{\SF{\mathscr{F}}_{IJ}, \Dut_K \Dut_L \SF{X}^\mu }
+ R^\mu{}_{\nu\alpha\beta}(\SF{X})\   (\Dut_I \SF{X}^\alpha) (\Dut_J \SF{X}^\beta) (\Dut_K \Dut_L \SF{X}^\nu).
\end{split}
\end{equation}
In particular,
\begin{equation}\label{eq:quadrupleXs}
\begin{split}
 \Dut_\theta^2 (\Dut_\theta \Dut_{\thetab } \SF{X}^\mu) &= \dbrk{\SF{\mathscr{F}}_{\theta\theta} ,\, \Dut_\theta \Dut_{\bar\theta} \SF{X}^\mu }
    + \half R^\mu{}_{\nu\alpha\beta}(\SF{X}) (\Dut_{{\theta}} \SF{X}^\alpha) (\Dut_{{\theta}}\SF{X}^\beta)(\Dut_\theta \Dut_{\thetab } \SF{X}^\nu)\,  \,,\\
 \Dut_{\thetab }^2 (\Dut_\theta \Dut_{\thetab } \SF{X}^\mu) &= \dbrk{\SF{\mathscr{F}}_{\bar\theta\thetab } ,\, \Dut_\theta \Dut_{\thetab } \SF{X}^\mu }
    + \half R^\mu{}_{\nu\alpha\beta}(\SF{X}) \, (\Dut_{\thetab } \SF{X}^\alpha) (\Dut_{\thetab }\SF{X}^\beta)(\Dut_\theta \Dut_{\thetab } \SF{X}^\nu) \,,\\
  \{\Dut_\theta,\Dut_{\thetab }\} (\Dut_\theta \Dut_{\thetab } \SF{X}^\mu) &= \dbrk{ \SF{\mathscr{F}}_{\theta\thetab } ,\, \Dut_\theta \Dut_{\thetab } \SF{X}^\mu} + R^\mu{}_{\nu\alpha\beta}(\SF{X}) \, (\Dut_{\thetab } \SF{X}^\alpha) (\Dut_{{\theta}} \SF{X}^\beta)(\Dut_\theta \Dut_{\thetab } \SF{X}^\nu)\,.
\end{split}
\end{equation}
These identities can be used to simplify various expressions in the effective action as we compute higher derivative terms.
 
\section{Superspace representations: metric and curvature}
\label{sec:MetCurvMult}

The target space metric and curvature tensor can be viewed as superfields by means of expanding their dependence on $\SF{X}^\mu$. 
We pull them back to the worldvolume to obtain corresponding data which enters explicitly into the hydrodynamic effective actions.  We give the corresponding superspace expansions for the pullback metric superfield and curvature, which were used in the explicit analysis of 
\S\ref{sec:mmo}.

\subsection{Metric superfield}
The worldvolume metric superfield, we recall, is covariantly pulled back as in \eqref{eq:gref}. We will give the components in the super-static gauge 
where $\Theta(z) = \theta$ and $\bar{\Theta}(z) = \thb$. Note also the symmetry properties are those of a super-tensor:
\begin{equation} 
\SF{\gref}_{IJ} = (-)^{IJ} \, \SF{\gref}_{JI} \;\; \Longrightarrow \;\; \SF{\gref}_{ \theta \thetab} = - \SF{\gref}_{ \thetab\theta} 
\label{eq:grefsym}
\end{equation}	
Working out the components using the superfield expansion of $\Dut_I \SF{X}^\mu$ we end up with:
\begin{subequations}
\begin{align}
&\SF{\gref}_{ab} :  
 	\prn{
 	\begin{array}{l} 
 	\scriptstyle{
 	 	g_{\mu\nu} (D_a X^\mu)(D_a X^\nu)  }\\
 	 \scriptstyle{
		g_{\mu\nu} (D_a X^\mu) (D_b \xpsi^\nu + (\lambda_b,X^\nu)_\Kref) +g_{\mu\nu} (D_b X^\mu) (D_a \xpsi^\nu + (\lambda_a,X^\nu)_\Kref) }\\
 	 \scriptstyle{
		g_{\mu\nu} (D_a X^\mu) (D_b{\xpsib}^\nu+ (\bar{\lambda}_b,X^\nu)_\Kref) + g_{\mu\nu} (D_b X^\mu) (D_a{\xpsib}^\nu
 		+ (\bar{\lambda}_a,X^\nu)_\Kref) }\\
 	 \scriptstyle{
		\source{h_{ab}} +\; g_{\mu\nu} (D_a X^\mu) \prn{D_b\tilde{X}^\nu + (\sAt{b},X^\nu)_\Kref + (\bar{\lambda}_b,\xpsi^\nu)_\Kref 
 		-(\lambda_b,\xpsib^\nu)_\Kref+ R^\nu{}_{\kappa\alpha\sigma}\xpsib^\alpha (D_bX^\sigma)\xpsi^\kappa} }\\
 	\qquad 
 	 \scriptstyle{
  		+ \;g_{\mu\nu} (D_b X^\mu) \prn{D_a\tilde{X}^\nu + (\sAt{a},X^\nu)_\Kref + (\bar{\lambda}_a,\xpsi^\nu)_\Kref -(\lambda_a,\xpsib^\nu)_\Kref+ R^\nu{}_{\kappa\alpha\sigma}\xpsib^\alpha (D_aX^\sigma)\xpsi^\kappa}} \\
 	\qquad  
  	 \scriptstyle{
		+ g_{\mu\nu} (D_a{\xpsib}^\mu+ (\bar{\lambda}_a,X^\mu)_\Kref)(D_b \xpsi^\nu + (\lambda_b,X^\nu)_\Kref) 
 		+ g_{\mu\nu} (D_b{\xpsib}^\mu+ (\bar{\lambda}_b,X^\mu)_\Kref)(D_a \xpsi^\nu + (\lambda_a,X^\nu)_\Kref) }
	\end{array}  }
\\
& \SF{\gref}_{\thetab a}:  
	\prn{\begin{array}{l} 
 	 \scriptstyle{
		g_{\mu\nu} \xpsi^\mu (D_a X^\nu)  }\\
 	 \scriptstyle{
		-g_{\mu\nu} \xpsi^\mu (D_a \xpsi^\nu + (\lambda_a,X^\nu)_\Kref) +g_{\mu\nu} (\phiT,X^\mu)_\Kref D_a X^\nu  }\\
 	 \scriptstyle{
		 -g_{\mu\nu} \xpsi^\mu (D_a{\xpsib}^\nu+ (\bar{\lambda}_a,X^\nu)_\Kref) +g_{\mu\nu}\tilde{X}^\mu D_a X^\nu }\\
  	 \scriptstyle{
		\source{h_{\thb a}} +\;g_{\mu\nu} \xpsi^\mu \prn{D_a\tilde{X}^\nu + (\sAt{a},X^\nu)_\Kref + (\bar{\lambda}_a,\xpsi^\nu)_\Kref 
		-(\lambda_a,\xpsib^\nu)_\Kref+ R^\nu{}_{\kappa\alpha\sigma}\xpsib^\alpha (D_aX^\sigma)\xpsi^\kappa}  }\\
	\qquad 
 	 \scriptstyle{
		+ g_{\mu\nu} ((\phiT,\xpsib^\mu)_\Kref-(\etaT,X^\mu)_\Kref) D_a X^\nu 
		-g_{\mu\nu}\tilde{X}^\mu (D_a \xpsi^\nu + (\lambda_a,X^\nu)_\Kref) }\\
	\qquad 
 	 \scriptstyle{
		+ g_{\mu\nu} (\phiT,X^\mu)_\Kref  (D_a{\xpsib}^\nu+ (\bar{\lambda}_a,X^\nu)_\Kref)}
 	\end{array}} \\
& \SF{\gref}_{\theta a} : 
	\prn{\begin{array}{l} 
	\scriptstyle{
		g_{\mu\nu} \xpsib^\mu (D_a X^\nu)  }\\
 	 \scriptstyle{
		-g_{\mu\nu} \xpsib^\mu (D_a \xpsi^\nu + (\lambda_a,X^\nu)_\Kref) +g_{\mu\nu} \prn{(\phizT,X^\mu)_\Kref-\tilde{X}^\mu} D_a X^\nu  }\\
 	 \scriptstyle{
		- g_{\mu\nu} \xpsib^\mu (D_a{\xpsib}^\nu+ (\bar{\lambda}_a,X^\nu)_\Kref) +g_{\mu\nu}(\phibT,X^\mu)_\Kref  D_a X^\nu }\\
 	 \scriptstyle{
		 \source{h_{\theta a}} +\; g_{\mu\nu} \xpsib^\mu \prn{D_a\tilde{X}^\nu + (\sAt{a},X^\nu)_\Kref 
		 + (\bar{\lambda}_a,\xpsi^\nu)_\Kref -(\lambda_a,\xpsib^\nu)_\Kref
		 + R^\nu{}_{\kappa\alpha\sigma}\xpsib^\alpha (D_aX^\sigma)\xpsi^\kappa}  }\\
	\qquad 
	\scriptstyle{
		+ g_{\mu\nu} \prn{ (\phizT,\xpsib^\mu)_\Kref+ R^\mu{}_{\kappa\alpha\sigma}\xpsib^\alpha \xpsi^\sigma {\xpsib}^\kappa
		-(\phibT,\xpsi^\mu)_\Kref+(\etabT,X^\mu)_\Kref} D_a X^\nu }\\
  	\scriptstyle{ 
		-g_{\mu\nu}(\phibT,X^\mu)_\Kref  (D_a \xpsi^\nu + (\lambda_a,X^\nu)_\Kref)\qquad \qquad\qquad }\\
	\qquad 
 	\scriptstyle{
		+ g_{\mu\nu} \prn{(\phizT,X^\mu)_\Kref-\tilde{X}^\mu}  (D_a{\xpsib}^\nu+ (\bar{\lambda}_a,X^\nu)_\Kref)}
 \end{array}}
  \end{align}
 \begin{align}
& \SF{\gref}_{\theta \thetab} :  
	\prn{\begin{array}{l} 
 	 \scriptstyle{
		i  \;+\;  g_{\mu\nu} \xpsib^\mu \xpsi^\nu  }\\
 	 \scriptstyle{
		-g_{\mu\nu} \xpsib^\mu (\phiT,X^\nu)_\Kref +g_{\mu\nu} \prn{(\phizT,X^\mu)_\Kref-\tilde{X}^\mu} \xpsi^\nu  }\\
 	 \scriptstyle{
		- g_{\mu\nu} \xpsib^\mu \tilde{X}^\nu+g_{\mu\nu}(\phibT,X^\mu)_\Kref  \xpsi^\nu }\\
 	 \scriptstyle{
		\source{h_{\theta\thb}} +\; g_{\mu\nu} \xpsib^\mu ((\phiT,\xpsib^\nu)_\Kref-(\etaT,X^\nu)_\Kref) }\\
 	\qquad
 	 \scriptstyle{
		 + g_{\mu\nu} \prn{ (\phizT,\xpsib^\mu)_\Kref
		 + R^\mu{}_{\kappa\alpha\sigma}\xpsib^\alpha \xpsi^\sigma {\xpsib}^\kappa-(\phibT,\xpsi^\mu)_\Kref+(\etabT,X^\mu)_\Kref} \xpsi^\nu }\\
 	 \qquad
 	 \scriptstyle{
		 -g_{\mu\nu}(\phibT,X^\mu)_\Kref  (\phiT,X^\nu)_\Kref  + g_{\mu\nu} \prn{(\phizT,X^\mu)_\Kref-\tilde{X}^\mu} \tilde{X}^\nu}
 \end{array}}
\end{align}
\end{subequations}
In the above expressions we have indicated the source deformations explicitly, since will end up varying with respect to these to obtain the physical components of the energy-momentum tensor. Ignoring all the terms with non-zero ghost number, we end up with the expressions used in \S\ref{sec:mmo}.

\subsection{Curvature superfield}
\label{sec:curvature}

We define the partially pulled back curvature superfield as:
\begin{equation}
\begin{split}
 (1+ \delta_{IJ}) R^\mu{}_{\nu \text{\tiny IJ}}(\SF{X})
\equiv R^\mu{}_{\nu\alpha\beta}(\SF{X})\   (\Dut_I \SF{X}^\alpha) (\Dut_J \SF{X}^\beta).
\end{split}
\end{equation}

It admits a superfield expansion which can be inferred from the component maps (we only write out the superspace components):
\begin{subequations}
\begin{align}
& R^\mu{}_{\nu\thetab\thetab}:\ 
	\prn{\begin{array}{l}
 	 \scriptstyle{
		   \half R^\mu{}_{\nu\alpha\sigma}\xpsi^\alpha \xpsi^\sigma  }\\
 	 \scriptstyle{
		 - R^\mu{}_{\nu\alpha\sigma} \xpsi^\alpha(\phiT,X^\sigma)_\Kref }\\
 	 \scriptstyle{
		 -\brk{(\xpsi.\nabla R^\mu{}_{\nu\alpha\sigma})\xpsib^\alpha \xpsi^\sigma
		+R^\mu{}_{\nu\alpha\sigma}\xpsi^\alpha\tilde{X}^\sigma } }\\
 	 \scriptstyle{
		 -(\xpsib.\nabla R^\mu{}_{\nu\lambda\sigma})\ \xpsi^\lambda (\phiT,X^\sigma)_\Kref
		 -R^\mu{}_{\nu\lambda\sigma}\ \brk{\tilde{X}^\lambda (\phiT,X^\sigma)_\Kref+\xpsi^\lambda (\etaT,X^\sigma)_\Kref
		 -(\phiT,\xpsib^\lambda )_\Kref \xpsi^\sigma} }
\end{array}}
 \\
& R^\mu{}_{\nu\theta\theta}:\
	 \prn{\begin{array}{l}	  
	 \scriptstyle{
   		\half R^\mu{}_{\nu\alpha\sigma}\xpsib^\alpha \xpsib^\sigma  }\\
	\scriptstyle{
		 -\brk{(\xpsib.\nabla R^\mu{}_{\nu\alpha\sigma})\xpsib^\alpha \xpsi^\sigma
		+R^\mu{}_{\nu\alpha\sigma}\xpsib^\alpha\prn{(\phizT,X^\sigma)_\Kref-\tilde{X}^\sigma} }}\\ 
 	 \scriptstyle{
		- R^\mu{}_{\nu\alpha\sigma} \xpsib^\alpha(\phibT,X^\sigma)_\Kref }\\
 	 \scriptstyle{
		 (\phizT,\half R^\mu{}_{\nu\alpha\sigma}\xpsib^\alpha \xpsib^\sigma )_\Kref
		- (\phibT,R^\mu{}_{\nu\alpha\sigma}\xpsib^\alpha \xpsi^\sigma )_\Kref
		+\half R^\mu{}_{\kappa\alpha\sigma}  R^\kappa{}_{\nu\gamma\xi}
		\prn{\xpsib^\alpha \xpsi^\sigma  \xpsib^\gamma \xpsib^\xi-\xpsib^\alpha \xpsib^\sigma\xpsib^\gamma \xpsi^\xi  } }\\
 	\qquad
 	 \scriptstyle{
		 -(\xpsib.\nabla R^\mu{}_{\nu\lambda\sigma})\ \xpsi^\lambda ( \phibT, X^\sigma)_\Kref  
		  -R^\mu{}_{\nu\lambda\sigma}\ \brk{\tilde{X}^\lambda ( \phibT,X^\sigma)_\Kref-\xpsib^\lambda (\etabT,X^\sigma)_\Kref
		 -( \phibT,\xpsib^\lambda )_\Kref \xpsi^\sigma}}
 \end{array}}
\\ 
& R^\mu{}_{\nu\theta\thetab}:\ 
	\prn{\begin{array}{l}   
 	 \scriptstyle{
		R^\mu{}_{\nu\alpha\sigma}\xpsib^\alpha \xpsi^\sigma  }\\
 	 \scriptstyle{
		(\xpsi.\nabla R^\mu{}_{\nu\alpha\sigma})\xpsib^\alpha \xpsi^\sigma
		 - R^\mu{}_{\nu\alpha\sigma} \xpsib^\alpha(\phiT,X^\sigma)_\Kref 
		 - R^\mu{}_{\nu\alpha\sigma}\xpsi^\alpha\prn{(\phizT,X^\sigma)_\Kref-\tilde{X}^\sigma}  }\\
 	 \scriptstyle{
		(\xpsib.\nabla R^\mu{}_{\nu\alpha\sigma})\xpsib^\alpha \xpsi^\sigma
		 - R^\mu{}_{\nu\alpha\sigma} \xpsi^\alpha(\phibT,X^\sigma)_\Kref
		 -R^\mu{}_{\nu\alpha\sigma}\xpsib^\alpha\tilde{X}^\sigma  }\\
 	 \scriptstyle{
		 (\phiT, \half R^\mu{}_{\nu\alpha\sigma}\xpsib^\alpha \xpsib^\sigma)_\Kref
		 -(\phibT, \half  R^\mu{}_{\nu\alpha\sigma}\xpsi^\alpha \xpsi^\sigma)_\Kref  
		 +\quarter R^\mu{}_{\kappa\alpha\sigma}  R^\kappa{}_{\nu\gamma\xi}
 		\prn{\xpsi^\alpha \xpsi^\sigma \xpsib^\gamma \xpsib^\xi  -\xpsib^\alpha \xpsib^\sigma \xpsi^\gamma \xpsi^\xi  }\qquad \qquad\qquad  }\\
 	\qquad
 	 \scriptstyle{
		 -(\xpsib.\nabla R^\mu{}_{\nu\lambda\sigma})\ \xpsi^\lambda (\phizT,X^\sigma)_\Kref 
		 -R^\mu{}_{\nu\lambda\sigma}\ \brk{\tilde{X}^\lambda (\phizT,X^\sigma)_\Kref+\xpsib^\lambda (\etaT,X^\sigma)_\Kref 
		 -\xpsi^\lambda (\etabT,X^\sigma)_\Kref
		 -(\phizT,\xpsib^\lambda )_\Kref \xpsi^\sigma}}
\end{array}}
\end{align}
\end{subequations}

\section{Gauge fixing and super-diffeomorphisms}
\label{sec:gaugesdiff}

We now turn to the worldvolume and target space symmetries which we have used in our construction. We will begin with a discussion of the worldvolume and then examine the target space issues. In both cases we will explain the rationale behind some of the choices we made in the main text, and how they serve to determine the structure of our hydrodynamic effective actions.

\subsection{Worldvolume symmetries}
\label{sec:wvsym}

We would naively anticipate that the theory we write down should have worldvolume superdiffeomorphisms as a symmetry. This would imply  that the worldvolume coordinates $z^I=(\sigma^a,\theta,\thb)$ transform as
$z^I  \mapsto z^I + \SF{\mathfrak{F}}^I(z^J)$ for some arbitrary superfield $\SF{\mathfrak{F}}$. We however find that this is too constraining, and relax the symmetries slightly to allow only worldvolume reparameterizations of the form:
\begin{equation}
z^I  \mapsto z^I + f^I(\sigma)\,.
\label{eq:wvsdiff}
\end{equation}	
This turns out to be sufficient to allow for a clean representation of the $\mathcal{N}_\smallT = 2$ thermal equivariance structures. While it would be interesting to explore why the full set of super-diffeomorphisms need to be truncated thus, we will not do so here. We will now see, however, that our gauge choices almost enforce this constraint.

As explained in the text we are going to use some of this reparameterization freedom to gauge fix the thermal super-vector $\SF{\Kref}^I$ as in \eqref{eq:betagauge} which we reproduce here for convenience:
\begin{equation}
 \SF{\Kref}^\theta=\SF{\Kref}^{\thetab} = 0=\partial_\theta \SF{\Kref}^a=  \partial_{\thetab} \SF{\Kref}^a\,,
\label{eq:betagauge2}
\end{equation}	
that is, $\SF{\Kref}^I = (\Kref^a,0,0)$.
We therefore can only allow $f^I(\sigma)$ that leave this choice invariant. Since we have a distinguished timelike vector in ordinary space, we expect the full set of usual diffeomorphisms to reduce to those that preserve the foliation induced by the vector field, along with additional constraints in superspace. 

Let us see this explicitly. Note that under \eqref{eq:wvsdiff} the thermal super-vector transforms via a Lie drag which reads
\begin{equation}
\SF{\Kref}^I \mapsto \SF{\Kref}^I + (-)^J \left(  f^J \partial_J \SF{\Kref}^I - \SF{\Kref}^J \partial_J \, f^I \right)\,.
\label{eq:}
\end{equation}	
The first two equations of \eqref{eq:betagauge2} which refer to the superspace components of $\SF{\Kref}^I$ imply that 
\begin{equation}
\Kref^a \partial_a f^\theta  = \Kref^a \partial_a f^\thb = 0 \;\; \Longrightarrow \;\; 
f^\theta= f^\theta(\sigma^\perp) \,,\; f^\thb= f^\thb(\sigma^\perp)
\label{eq:fconst1}
\end{equation}	
where $\sigma^\perp$ coordinatize the space perpendicular to the thermal vector $\Kref^a$. In fact, had we further gauge fixed 
$\Kref^a = \delta^a_0$, then we would conclude that $\sigma^\perp = \sigma^i$ and \eqref{eq:fconst1} would imply that
$f^\theta=f^\theta(\sigma^i) \,, f^\thb=f^\thb(\sigma^i)$. 

We now turn to the second set of constraints in \eqref{eq:betagauge2}. These imply
\begin{equation}
  \partial_\theta( \lieD_\Kref f^a) = \partial_\thb(\lieD_\Kref  f^a) = 0 \,.
\end{equation}
In fact, demanding invariance of the full superfield $\SF{\Kref}^a$ gives the slightly stronger constraint
\begin{equation}
   \lieD_{\SF{\Kref}} \SF{f}^a = \lieD_{\SF{\Kref}}  \SF{f}^a = 0 \,.
\end{equation}
This is enforced since we want to have this background vector unchanged under coordinate transformations.
All told we simply require that the super-vector field $\SF{f}^a(\sigma)$ commute with $\SF{\Kref}^a$, i.e., $[\SF{\Kref}^a,\SF{f}^b]=0$.

Let us now use this information to see what the covariant data in the worldvolume connection is. As noted in the text the connection is a-priori a new piece of data owing to the fact that the $\UT$ covariant pullback of the target spacetime metric does not uniquely determine a worldvolume connection. We should thus determine which components transform covariantly so that we may employ them in the construction of our actions. A general connection on the worldvolume transforms under \eqref{eq:wvsdiff} as
\begin{equation}
\Cwv^I{}_{JK} \mapsto \Cwv^I{}_{JK} +  (-)^{I(J+K)} \partial_J \partial_K f^I +\text{covariant pieces}\,.
\label{eq:wvCtr}
\end{equation}	
Given our restrictions on the super-diffeomorphism function $f^I(\sigma)$ we learn that the following components transform covariantly: 
\begin{equation}
\Cwv^i_{0J} \,,\Cwv^\theta_{0J} \,, \Cwv^\thb_{0J}\,, \;  \Cwv^I_{\ \theta J} \,, \Cwv^I_{\ \thb J} \,.
\label{eq:wvCcov}
\end{equation}	
We will have use for the fact that the last two terms identified above are covariant in our construction. In fact, 
in the bulk of the discussion we use this to allow the introduction of the fields $ \gpsib_{IJ}$ and $\gpsi_{IJ} $ whose covariance follows from the simple identities (see Eq.~\eqref{eq:DwvDutrel} below):
\begin{equation}
\begin{split}
 \gpsib_{IJ} &\equiv  \Dut_\theta  \SF{\gref}_{IJ} = \Dwv_\theta \SF{\gref}_{IJ} + (-)^{IK} \, \Cwv^K_{\ \theta I} \SF{\gref}_{KJ} +(-)^{K(I+J)+IJ} \, \Cwv^K_{\ \theta J} \SF{\gref}_{IK}  \\
\gpsi_{IJ} &\equiv  \Dut_\thb  \SF{\gref}_{IJ} =  \Dwv_\thb \SF{\gref}_{IJ} + (-)^{IK} \, \Cwv^K_{\ \thb I} \SF{\gref}_{KJ} +(-)^{K(I+J)+IJ} \, \Cwv^K_{\ \thb J} \SF{\gref}_{IK} 
\end{split}
\label{}
\end{equation}

We note that integration by parts of the Grassmann odd derivatives in the definition of $\{\gpsib_{IJ}, \gpsi_{IJ}\}$ is allowed.
In particular, we have $\Dut_\theta \SF{\mathcal{S}} = \Dwv_\theta \SF{\mathcal{S}}$ for any scalar superfield $\SF{S}$, such that the boundary term occurring when integrating by parts is zero. For example, one has:
{\small
 \begin{equation}
 \int d^d \sigma d\theta d\thb \, \frac{\sqrt{-\SF{\gref}}}{\zsf}\,  \gpsi_{IJ}\SF{\mathcal{T}}^{JI} = \underbrace{\int d^d \sigma d\theta d\thb  \,  \frac{\sqrt{-\SF{\gref}}}{\zsf}\,  \Dwv_\thb \big(  \SF{\gref}_{IJ}  \SF{\mathcal{T}}^{JI} \big)}_{=0}
 - (-)^{I+J}\int d^d \sigma d\theta d\thb \,  \frac{\sqrt{-\SF{\gref}}}{\zsf}\,\SF{\gref}_{IJ}  \, \Dut_\thb \big( \SF{\mathcal{T}}^{JI} \big)
 \end{equation}
 }\normalsize
  for any tensor superfield $\SF{\mathcal{T}}^{IJ}$.

\subsection{Target space symmetries}
\label{sec:tarsym}

Let us now turn to the target space symmetries. The target space super-coordinates $\SF{X}^{\uA}=(\SF{X}^\mu,\Theta,\bar{\Theta})$ transform under spacetime super-diffeomorphisms along with the transformation induced by the worldvolume super-diffeomorphisms \eqref{eq:wvsdiff}. One has
\begin{equation}
\SF{X}^{\uA} \mapsto \SF{X}^{\uA} \; + \;  \widehat{H}^{\uA} \,, \qquad 
\widehat{H}^{\uA} = \SF{H}^{\uA} \; + \; (-)^I \, f^I\partial_I \SF{X}^{\uA}
\label{eq:tarsdiff}
\end{equation}	
where the index $\uA \in \{\mu, \Theta, {\bar \Theta}\}$. Under this transformation the target space super-metric transforms as usual via
\begin{equation}
g_{\uA \uB} \mapsto g_{\uA\uB} \;+\;  (-)^{\uB\uC} g_{\uA \uC} \partial_{\uB} \widehat{H}^{\uC} \;+\; (-)^{\uA(\uB+\uC)}  g_{\uB \uC} \partial_{\uA} \widehat{H}^{\uC} 
\;-\; g_{\uA \uB} \partial_{\uC} \widehat{H}^{\uC}\,.
\label{eq:gtarsdiff}
\end{equation}	

As discussed in the text we wish to use this freedom to impose the super-static gauge where we pick:
\begin{equation}
\SF{X}^{\Theta }(z^I) \equiv \SF{\Theta}(z^I) = \theta \,, \quad 
 \SF{X}^{\bar \Theta}(z^I) \equiv  \SF{{\bar \Theta}}(z^I) = \thb \,.
\label{eq:sstatic1}
\end{equation}	
This can always be done by suitably constraining $\SF{H}^{\uA}$ in terms of the worldvolume diffeomorphism parameter $f^I$, and poses no real constraint for our considerations for it demands
\begin{equation}
\widehat{H}^{\Theta} = \widehat{H}^{\bar \Theta} =0 \,.
\label{eq:}
\end{equation}	

Our gauge condition \eqref{eq:gsptgauge1} only constrains $\SF{H}^\mu $ and we find:
\begin{equation}
\partial_\theta \left(\SF{H}^{\mu} + f^I \partial_I \SF{X}^\mu\right) 
= \partial_{\thb} \left(\SF{H}^{\mu} + f^I \partial_I \SF{X}^\mu\right) 
= 0 \,.
\label{eq:}
\end{equation}	
These arise from the first set of gauge fixing conditions for the super-metric in \eqref{eq:gsptgauge1}. The second choice $g_{\Theta {\bar \Theta}} =i$ in addition to the above implies
\begin{equation}
 \partial_\mu \widehat{H}^\mu = 0 \,.
\end{equation}
The fact that the ordinary space metric only has a bottom component is also solved by demanding that $\widehat{H}^\mu$ be independent of $\SF{\Theta},\SF{\bar{\Theta}}$. All told we find that 
\begin{equation}
\boxed{
\widehat{H}^{\Theta} = \widehat{H}^{\bar \Theta} =0 \,,\quad \widehat{H}^{\mu}=\widehat{H}^{\mu} (\SF{X}^\nu) \,,\quad  \partial_\mu \widehat{H}^\mu = 0  \,, \qquad 
\widehat{H}^{\uA} = \SF{H}^{\uA} +f^I \partial_I \SF{X}^{\uA} \,.	
}
\label{eq:dthH}
\end{equation}	

It is instructive to ask how the top component of $\SF{X}^\mu$ transforms under \eqref{eq:tarsdiff}. We have
\begin{equation}
\begin{split}
\tx^\mu &= 
	\partial_\theta \partial_{\thb} \SF{X}^\mu | \\
\tx^\mu  & 
	\mapsto \tx^\mu +  \partial_\theta \partial_{\thb} \widehat{H}^\mu | \\ 
& = 
	\tx^\mu +  
	\partial_\rho \partial_\nu \widehat{H}^\mu \, \partial_\theta \SF{X}^\rho \, \partial_{\thb}\SF{X}^\nu |
	+ \partial_\nu \widehat{H}^\mu \partial_\theta\partial_{\thb} \SF{X}^\nu | \\
& = 
	\tx^\mu + \partial_\nu \widehat{H}^\mu \, \tx^\nu 
	+ \partial_\rho \partial_\nu \widehat{H}^\mu\,  X_{\psib}^\rho\, X_\psi^\nu	
\end{split}
\end{equation}
Now we note that the ghost bilinear term has a non-covariant piece analogous to the one in the transformation of the connection on the target spacetime (compare with \eqref{eq:wvCtr}). While $\widehat{H}^\mu$ has no dependence on the target space Grassmann coordinates, worldvolume Grassmann dependence is  induced through $\SF{X}^\mu$.
All this then implies is that the combination appearing in the top component of  \eqref{eq:xsf0}, viz., 
\begin{equation}
\tx^\mu  - \Gamma_{\rho\sigma}^\mu X_{\psib}^\rho \, X_\psi^\sigma  \,,
\end{equation}	
is covariant under \eqref{eq:tarsdiff}.  One can similarly check that $\tilde{\gref}_{ab}$ transforms covariantly as a worldvolume tensor.

\section{Fixing the worldvolume connection}
\label{sec:wvconnect}

We start with a set of axioms for the $\UT$ covariant worldvolume theory and ask how does one fix the connection  $\Cwv^I_{\ JK}$ on the worldvolume. Our aim is to define a worldvolume covariant derivative $\Dwv_I$ which agrees with $\Dut_I$ when acting on scalar superfields, and furthermore allows for covariant integration by parts.
The key reason why there is no obvious canonical connection owes the fact that we pullback the spacetime metric in a $\UT$ covariant fashion and the target space superfield $\SF{X}^{\uA}$ transforms non-trivially under $\UT$. Recall that we define\footnote{ We suppress all signs associated with super-index contractions in this appendix. That is, we only give Grassmann signs explicitly for uncontracted indices.}
\begin{equation}
\begin{split}
\SF{\gref}_{IJ} &= g_{\uA\uB}\; \Dwv_I  \SF{X}^{\uA} \Dwv_J \SF{X}^{\uB}  \,,  \\
\Dwv_I \SF{X}^{\uA} &=  \Dut_I \SF{X}^{\uA} = \partial_I \SF{X}^{\uA} + (\As_I,\SF{ X}^{\uA})_\Kref\,.
\end{split}
\label{eq:pbg}
\end{equation}
We have made explicit that insofar as $\SF{X}^{\uA}$ is concerned we have no issue as it transforms as a scalar under worldvolume superdiffeomorphisms. Without further ado, let us then record our basic assumptions for the worldvolume superspace. 

\paragraph{Axioms:}
\begin{itemize}
\item[{\bf A1.}] There is a  superspace on the worldvolume with intrinsic coordinates $z^I$ and  the pullback superfield $\SF{X}^{\uA}$  exists.
\item[{\bf A2.}] Using target space superdiffeomorphisms we can fix the gauge 
$\Theta = \theta$, ${\bar \Theta} = {\bar \theta}$. 
\item[{\bf A3.}] $\exists \; \SF{\Kref}^I$ and using worldvolume superdiffeomorphisms we can fix a gauge  such that $\SF{\Kref}^\theta = \SF{\Kref}^\thb = 0$ and  
$\partial_\theta \SF{\Kref}^I  = \partial_{\thb} \SF{\Kref}^I =0$.
\item[{\bf A4.}] $\exists \;\As_I$ on the worldvolume and a spacetime metric $g_{\uA\uB}=g_{\uA\uB}(\SF{X}^{\uA})$. We furthermore gauge fix the spacetime metric to obey \eqref{eq:gsptgauge}. 
\item[{\bf A5.}] Commutator condition: 
\begin{equation}
\gradcomm{\Dwv_I}{\Dwv_J} \SF{X}^{\uA}  = ( \Fs_{IJ}, \SF{X}^{\uA})_\Kref \,.
\label{eq:Comm}
\end{equation}	
\item[{\bf A6.}] Measure compatibility:
\begin{equation}
\Dwv_J \left( \frac{\sqrt{-\SF{\gref}}}{\zsf} \right) = 0 \,.
\label{eq:Measure}
\end{equation}	
\end{itemize}
Given this, our task is to ascertain whether there is a connection on the worldvolume which retains $\UT$ covariance. The first four axioms are reasonably innocuous and don't pose any constraints. The discussion in Appendix \ref{sec:gaugesdiff} makes clear that the gauge fixing constraints can be imposed immediately.

We would like to ensure the existence of a connection satisfying {\bf A5} and {\bf A6}. ${\bf A5}$ simply encodes compatibility with $\UT$ covariance. Note that we are not demanding that the worldvolume connection is compatible with the pullback metric. Clearly, measure and metric compatibility cannot simultaneously hold with $\UT$ covariant pullback as is clear from \eqref{eq:Measure}. Measure compatibility ${\bf A6}$ is all we need in order to be able to perform standard integration by parts inside superspace integrals, as we will show. 

Let us record some useful intermediate lemmas that focus our attention on the salient features of the worldvolume connection.  
\begin{enumerate}
\item[{\bf L1}.] Determining $\SF{\Omega}_{IJ}{}^{\uA} \equiv \Dwv_I \Dwv_J \SF{X}^{\uA} $ is equivalent to prescribing a worldvolume connection. 
\item[{\bf L2}.] Our axioms {\bf A1}-{\bf A6} only require us to fix the (graded) antisymmetric part of the connection and its trace, but leave the (graded) symmetric part undetermined.
\item[{\bf L3}.] Given the pullback metric we can evaluate its derivatives and check that 
\begin{equation}
\begin{split}
\Dwv_A \SF{\gref}_{IJ} &=  g_{\uA\uB }\left(
	\SF{\Omega}_{AI}{}^{\uA} \; \Dwv_J \SF{X}^{\uB} + 
	(-)^{AI}\, \Dwv_I \SF{X}^{\uA} \; \SF{\Omega}_{AJ}{}^{\uB} 
	\right)\\
\gradcomm{\Dwv_A}{\Dwv_B} \SF{\gref}_{IJ} &= 
	g_{\uA\uB} \left( 
	\SF{\Upsilon}_{ABI}{}^{\uA} \; \Dwv_J\SF{X}^{\uB} 
	+ (-)^{IJ }\,  \SF{\Upsilon}_{ABJ}{}^{\uA} \; \Dwv_I\SF{X}^{\uB} 
	\right) \\
\text{where }\SF{\Upsilon}_{ABI}{}^{\uA} & \equiv \Dwv_A \SF{\Omega}_{BI}{}^{\uA}  - (-)^{AB}\Dwv_B \SF{\Omega}_{AI}{}^{\uA}
 \end{split}	
\label{eq:Dg1}
\end{equation}	
\item[{\bf L4}.] If we demand that $\SF{\gref}_{IJ}$ transforms as a 0-adjoint and the commutator condition {\bf A5} holds, then we learn that 
\begin{equation}
	g_{\uA\uB} \left( 
	\SF{\Upsilon}_{ABI}{}^{\uA} \; \Dwv_J\SF{X}^{\uB} 
	+ (-)^{IJ}\,  \SF{\Upsilon}_{ABJ}{}^{\uA} \; \Dwv_I\SF{X}^{\uB} 
	\right)
=	\Fs_{AB} \, \lieD_\Kref \SF{\gref}_{IJ}
\end{equation}	
\end{enumerate}

Let us now unpack and explain these statements one by one and present a prescription to fix the derivative $\Dwv$. On occasion we will also use the basic non-covariant derivative $\Dut$ on the worldvolume. One can consider the restriction of $\Dwv$ as acting on scalars to define the derivative operator $\Dut$, which acts on arbitrary worldvolume tensors  or partially pulled back target space tensors as a covariant derivative where the all connection terms discarded:
\begin{equation}\label{eq:PartialSlash}
\Dut_J \equiv \Dwv_J \big{|}_{\Gamma^{\uC}_{\uA\uB} = \Cwv^I_{JK} = 0} \,.
\end{equation}
Indeed, this definition provides a natural generalization of $\partial_I$ in the $\UT$ covariant theory. It is simply given by the covariant derivative restricted to spacetime with flat connection (which is how one could recover $\partial_a$ from $\nabla_a$ in usual geometry without equivariant gauge group). 
For instance, on arbitrary tensors with worldvolume indices we would write 
$\Dut_I \equiv \partial_I + (\As_I \,,\cdot\,)_\Kref$, i.e., it acts like $\Dwv_I$, but without the connection terms:
\begin{equation}
\boxed{
  \Dwv_I \SF{\mathcal{T}}^J{}_K \equiv \Dut_I \SF{\mathcal{T}}^J{}_K + (-)^{IJ} \Cwv^J{}_{IL} \, \SF{\mathcal{T}}^L{}_K - (-)^{IJ}   \SF{\mathcal{T}}^J{}_L \, \Cwv^L{}_{IK} \,,
  }
\label{eq:DwvDutrel}
\end{equation}
where $\SF{\mathcal{T}}^J{}_K$ is a generic super-tensor.

Consider now the derivative of the pull-back metric \eqref{eq:pbg}:\footnote{ In this calculation, the part where the derivative hits $g_{\uA\uB}(\SF{X})$ evaluates to zero since the target space metric can be chosen to have a metric compatible connection: 
\begin{equation}
\Dwv_I g_{\uA\uB} (\SF{X}) =  \Dwv_I \SF{X}^{\uC} \; \nabla_{\uC} g_{\uA\uB}(\SF{X})  = 0 \,.
\end{equation}
}
\begin{equation}
\Dwv_I \,   \SF{\gref}_{JK} = g_{\uA\uB}( \SF{X}) \,\left( \Dwv_I \Dwv_J \SF{X}^{\uA} \, \Dwv_K \SF{X}^{\uB} 
+ (-)^{IJ}\, \Dwv_J \SF{X}^{\uA} \, \Dwv_I  \Dwv_K \SF{X}^{\uB} \right) \,.
\label{eq:Dgwv}
\end{equation}	

The result \eqref{eq:Dgwv} can be equivalently phrased in terms of the behaviour of the derivative acting on the dual tetrad basis $\SF{e}^I_{\uA}$ which satisfy 
\begin{equation}
  \Dwv_J \SF{X}^{\uA} \; \SF{e}_{\uA}{}^I = \delta_{J}{}^{I}\,, \qquad  \,\SF{e}_{\uA}{}^{I}\; \Dut_{I} \SF{X}^{\uB}  =  \delta_{\uA}{}^{\uB} \,. 
\end{equation}	
Taking a worldvolume covariant derivative of the tetrad and calculating in two different ways we infer the following:
\begin{equation}
\begin{split}
\Dwv_J \SF{e}_{\uA}{}^I & =   (-)^{J\uA}\, \SF{e}_{\uA}{}^K\, \Dwv_J \Dwv_K \SF{X}^{\uB}  \; \SF{e}_{\uB}{}^I \\ 
\Dwv_J \SF{e}_{\uA}{}^I &
	= \Dut_J\, \SF{e}_{\uA}{}^I 
	 - (-)^{J\uA}\,\Gamma^{\uB}{}_{\uA\uC} (\SF{X}) \;  \Dwv_J \SF{X}^{\uC} \;\SF{e}_{\uB}{}^I 
	 + (-)^{I(J+\uA)} \; \Cwv^I{}_{JK}\, \SF{e}_{\uA}{}^K  \\
& =
	 \partial_J \SF{e}_{\uA}{}^I + (\SF{\Ascr}_J, \SF{e}_{\uA}{}^I) -(-)^{J\uA}\, \Gamma^{\uB}{}_{\uA\uC}(\SF{X}) \; 
	 \Dwv_J \SF{X}^{\uC} \;\SF{e}_{\uB}{}^I +  (-)^{I(J+\uA)} \;\Cwv^I{}_{JK}\, \SF{e}_{\uA}{}^K\,,
\end{split}
\label{eq:DivE}
\end{equation}	
where we have indicated for clarity the various derivative operators at our disposal.  We can then solve for the worldvolume connection:
\begin{align}
\Cwv^I_{\ JK} &=  
	\SF{e}_{\uA}{}^I\; \Gamma^{\uA}{}_{\uB\uC}( \SF{X}) \; \Dwv_J \SF{X}^{\uB} \, \Dwv_K \SF{X}^{\uC} \, - (-)^{IJ} \Dut_J \SF{e}_{\uA}{}^I \; \Dwv_K \SF{X}^{\uA}   + (-)^{IJ}\; \Dwv_J \SF{e}_{\uA}{}^I \;  \Dwv_K \SF{X}^{\uA} 
\nonumber \\
	&=   \SF{e}_{\uA}{}^I\;  \Gamma^{\uA}{}_{\uB\uC}( \SF{X}) \; \Dwv_J \SF{X}^{\uB} \, 
	\Dwv_K \SF{X}^{\uC}
	+\SF{e}_{\uB}{}^I\; \Dut_J \Dwv_K \SF{X}^{\uB} - \SF{e}_{\uA}{}^I \; \Dwv_J \Dwv_K \SF{X}^{\uA}
\label{eq:gamref}
\end{align}
We can thus see a relation between $\Cwv^I{}_{JK}$ and $\Dwv_J \Dwv_K  X^{\uA}$: determining the worldvolume connection $\Cwv^I_{\ JK}$ is equivalent to fixing the tetrad condition $\Dwv_J \, \Dwv_K X^{\uA}$, as anticipated in lemma {\bf L1}. In the following, we will have to impose certain consistency conditions to fix either of these two objects.  Often, we will also use \eqref{eq:gamref} in the form: 
\begin{align}\label{eq:DbDcX}
\Dwv_J \Dwv_K  \SF{X}^{\uA}  = \Dut_J \Dwv_K  \SF{X}^{\uA} -  \Cwv^I_{\ JK} \;\Dwv_I  \SF{X}^{\uA} +  \Dwv_J  \SF{X}^{\uB} \, \Dwv_K  \SF{X}^{\uC}\, \Gamma^{\uA}{}_{\uB\uC}( \SF{X}) \,.
\end{align}

Let us first assume that we have a tetrad condition encoded in a new superfield $\SF{\Omega}$, viz., 
\begin{equation}
\Dwv_I \Dwv_J \SF{X}^{\uA} \equiv \SF{\Omega}_{IJ}{}^{\uA} \,.
\label{eq:Omdef}
\end{equation}	
We shall now constrain the form of $\SF{\Omega}$ by imposing the axioms {\bf A5} and {\bf A6}.

\paragraph{The commutator condition:} Imposing that \eqref{eq:Comm} holds will constrain our choice of $\SF{\Omega}$. It is easy to check that 
\begin{equation}
\SF{\Omega}_{IJ}{}^{\uA} -(-)^{IJ} \, \SF{\Omega}_{JI}{}^{\uA} \equiv \Fs_{IJ} \; \Kref^I\, \partial_I \SF{X}^{\uA} = \frac{1}{\zsf} \,  
\Fs_{IJ} \SF{\Kref}^K \Dwv_K \SF{X}^{\uA}  
\label{eq:Om1}
\end{equation}	
where we use \eqref{eq:zdef} to infer the  identities:
\begin{equation}
\Dwv_I  \zsf=  \Fs_{IJ}\, \Kref^J \,, \qquad  \Kref^I\Dwv_I \SF{X}^{\uA} = \zsf \,\Kref^I \partial_I \SF{X}^{\uA}  \,.
\label{eq:zrels}
\end{equation}	

\paragraph{Measure compatibility:}  The axiom \eqref{eq:Measure} can be rewritten as
\begin{equation}\label{eq:MeasureCompAlt}
\begin{split}
 0 &= \Dwv_I \left( \frac{\sqrt{-\SF{\gref}}}{\zsf} \right) \\
 &= \frac{\sqrt{-\SF{\gref}}}{\zsf} \left[  \frac{1}{2}\,g_{\uA\uB}\, \left( \Dwv_I \Dwv_J \SF{X}^{\uA} \; \SF{\gref}^{JK}\, \Dwv_K \SF{X}^{\uB} + \Dwv_J \SF{X}^{\uA}  \; \SF{\gref}^{JK}\, \Dwv_I \Dwv_K X^{\uB} \right) -\frac{1}{\zsf} \, \Fs_{IJ} \, \Kref^J\right]\\
 &=  \frac{\sqrt{-\SF{\gref}}}{\zsf} \left[ \SF{\Omega}_{IJ}{}^{\uA} \;\SF{e}_{\uA}{}^J - \frac{1}{\zsf}\, \Fs_{IJ} \, \Kref^J \right] 
 \,,
\end{split}
\end{equation}
where we used \eqref{eq:zrels}. From this, we infer
\begin{equation}
\SF{\Omega}_{IJ}{}^{\uA} \; \SF{e}_{\uA}{}^J = \frac{1}{\zsf}\, \Fs_{IJ} \, \Kref^J \,.
\label{eq:Om2}
\end{equation}

A general solution for  $\SF{\Omega}_{IJ}{}^{\uA}$  satisfying \eqref{eq:Om1} and \eqref{eq:Om2} may now be parameterized in the following form:
\begin{equation}
\SF{\Omega}_{IJ}{}^{\uA} = \frac{1}{2\, \zsf} \left( \Kref_I \, \Fs_{JM} +  (-)^{IJ} \, \Kref_J \, \Fs_{IM} 
+ (-)^{M(1+J)} \Kref_M \, \Fs_{IJ} \right) \, \SF{\gref}^{MN}\, \Dwv_N \SF{X}^{\uA}  + \SF{\omega}_{IJ}{}^{\uA}
 \label{eq:Oansatz}
\end{equation}	
for some $\SF{\omega}_{IJ}{}^{\uA}$.
The rationale for splitting off the pieces involving the field strength is that these are determined by the 
$\UT$ gauge structure alone, and in fact suffice to ensure that our axioms {\bf A5} and {\bf A6} are met. A short computation reveals that all we need is for the tensor $\SF{\omega}_{IJ}{}^{\uA}$ to be graded symmetric and super-traceless in it worldvolume indices:
\begin{equation}
\begin{split}
&\mathbf{A5} \;\; \Longrightarrow \;\; \SF{\omega}_{IJ}{}^{\uA} =  (-)^{IJ}\, \SF{\omega}_{JI}{}^{\uA}\,, \\
&\mathbf{A6} \;\; \Longrightarrow \;\; \SF{\omega}_{IJ}{}^{\uA}  \SF{e}_{\uA}{}^{J} =0  \,.
\end{split}
\label{eq:omprops}
\end{equation}
With this parameterization one can furthermore check that:
\begin{equation}
\begin{split}
\Dwv_I \, \SF{e}_{\uA}{}^I &= 
	- \SF{e}_{\uB}{}^M\, \SF{\omega}_{MN}{}^{\uB} \, \SF{e}_{\uA}{}^N  = 0 \\
\Dwv_I \SF{\gref}_{JK} & = 
	\frac{1}{\zsf} \bigg( 
		\Fs_{IJ} \Kref_K + (-)^{JK}\, \Fs_{IK}\, \Kref_J \bigg) 
	+ \; g_{\uA\uB} \left(\SF{\omega}_{IJ}{}^{\uA} \, \Dwv_K X^{\uB} + (-)^{JK} \, 
  	\SF{\omega}_{IK}{}^{\uB} \, \Dwv_J X^{\uA}  \right)
\end{split}
\label{eq:modrels}
\end{equation}

We thus see that the general connection satisfying measure compatibility and the commutator condition is parametrized by $\SF{\omega}_{IJ}{}^{\uA}$ satisfying \eqref{eq:omprops}. While we could work with a general connection of this type, it is helpful to make a particular choice for explicit computations. We will make the canonical choice $\SF{\omega}_{IJ}{}^{\uA} =0$ which implies that  
\begin{equation}
\boxed{\Dwv_I \Dwv_J \SF{X}^{\uA} 
= \frac{1}{2\, \zsf} \left( \Kref_I \, \Fs_{JM} +  (-)^{IJ} \, \Kref_J \, \Fs_{IM} 
+ (-)^{M(I+J)} \Kref_M \, \Fs_{IJ} \right) \, \SF{\gref}^{MN}\, \Dwv_N \SF{X}^{\uA}  }
\label{eq:ddXfix}
\end{equation}	
which, using \eqref{eq:DbDcX}, fixes $\Cwv^I_{\ JK}$. Since the difference of two connections is covariant we can easily incorporate non-zero $\SF{\omega}_{IJ}{}^{\uA}$ if necessary. The choice of setting $\SF{\omega}_{IJ}{}^{\uA} =0$ is inspired by the observation that then the connection is determined solely in terms of the $\UT$ data and the reference thermal super-vector on the worldvolume. This appears most natural since the failure of the standard tetrad condition is due to our covariant pullback of target spacetime data. 

In the main text we will work with \eqref{eq:ddXfix} or equivalently with:
\begin{equation}
\begin{split}
\Cwv^I_{\ JK} 
&=  
	 \SF{e}_{\uA}{}^I\; \Gamma^{\uA}{}_{\uB\uC}( \SF{X}) \; \partial_J \SF{X}^{\uB} \; \Dwv_K \SF{X}^{\uC}
	+\SF{e}_{\uB}{}^I\; \Dut_J \Dwv_K \SF{X}^{\uB} \\
&	- \frac{1}{2\, \zsf}\; \SF{e}_{\uA}{}^I
	\left( \Kref_J \, \Fs_{KM} +  (-)^{JK} \, \Kref_K \, \Fs_{JM} 
	+ (-)^{M(1+K)} \Kref_M \, \Fs_{JK} \right) \, \SF{\gref}^{MN}\, \Dwv_N \SF{X}^{\uA}\,.
\end{split}
\label{eq:gamreffix}
\end{equation}

It turns out that there is an alternate way to motivate \eqref{eq:gamreffix} which is closely inspired by the manner in which the Christoffel connection is determined. Consider the following worldvolume derivative:
\begin{equation}
\Dut_I^{^{(\mathscr{F})}} \mathcal{T}^{J}{}_K \equiv \Dut_I \mathcal{T}^{J}{}_K
+ \frac{1}{\zsf}\,  \; (-)^{JM}\, \Fs_{I M} \Kref^J \mathcal{T}^M{}_K  
- \frac{1}{\zsf}\, \; (-)^{L(1+J)+JK}\,\Fs_{I K} \Kref^L \mathcal{T}^J{}_L \,.
\label{eq:}
\end{equation}	
and similarly for tensors with more indices. This strange derivative supplements the flat worldvolume derivative with a connection set in terms of the $\UT$ field strength and thermal vector. One can check that our our choice for $\Cwv^I_{\ JK}$ in \eqref{eq:gamreffix} is equivalent to the choice:
\begin{equation}
\begin{split}
\Cwv^{I}_{\ JK}
&= 
	 \frac{1}{2} \, \SF{\gref}^{IL} \left[ 
	 \Dut^{^{(\mathscr{F})}}_J \SF{\gref}_{LK} + (-)^{JK} \, 
	 \Dut^{^{(\mathscr{F})}}_K \SF{\gref}_{JL} -  \Dut^{^{(\mathscr{F})}}_L \SF{\gref}_{JK} \right] \\
& = 
	\frac{1}{2} \, \SF{\gref}^{IL} \left[ 
	\Dut_J \SF{\gref}_{LK} + (-)^{JK} \, 
	\Dut_K \SF{\gref}_{JL} -  \Dut_L \SF{\gref}_{JK} \right]  	
+ \frac{1}{\zsf}\left( \Fs^I_{\ J} \,\SF{\Kref}_K + (-)^{JK}\, \Fs^I_{\ K}\, \SF{\Kref}_J \right)
\end{split}
\label{eq:wvCalt}
\end{equation}	
That is, we choose the connection which looks like the superspace Christoffel connection supplemented with thermal equivariance data to ensure that we have correctly accounted for the $\UT$ covariance of the pullback. Reinstating Grassmann signs for index contractions yields \eqref{eq:connFinal}.

\paragraph{Integration by parts:} For illustration, we now show how the properties of the worldvolume connection that we imposed above, allow for integration by parts. For a generic 0-adjoint vector $\SF{{\bf V}}^I = \SF{V}^{\uA}\,\SF{e}_{\uA}{}^I $, observe that \eqref{eq:modrels} implies 
{\small
\begin{equation}
\begin{split}
  \int \frac{\sqrt{-\SF{\gref}}}{\SF{{\bf z}}}  \, \Dwv_I \SF{{\bf V}}^I 
  &=  \int \frac{\sqrt{-\SF{\gref}}}{\SF{{\bf z}}}   \, \prn{\Dwv_I \SF{V}^{\uA}} \SF{e}_{\uA}{}^I \\
   &= \int \frac{\sqrt{-\SF{\gref}}}{\SF{{\bf z}}} \,  \left[ \Dwv_I X^{\uB} \, \partial_{\uB} \SF{V}^{\uA} +\Dwv_I X^{\uB} \,\SF{\Gamma}^{\uA}{}_{\uB\uC} \,  V^{\uC} \right] \SF{e}_{\uA}{}^I\\
  &= \int \frac{\sqrt{-\SF{\gref}}}{\SF{{\bf z}}} \, \Dwv_I X^{\uB} \prn{ {\nabla}_{\uB} \SF{V}^{\uA} } \SF{e}_{\uA}{}^I
  = \int \sqrt{-g}  \; {\nabla}_{\mu} \SF{V}^{\mu}  = 0 \,,
\end{split}
\label{eq:intparts}
\end{equation}
}\normalsize
using the partially pulled back derivative defined in \eqref{eq:CovDref}.
In the last step we used the transformation of the measure as one passes from the worldvolume to the target space theory: 
\begin{equation}
 d^dX \,d\Theta\, d\bar{\Theta} \; \sqrt{-\SF{g}} = d^d\sigma\, d\theta\, d\thb \; \sqrt{-\SF{\gref}} \, \frac{\det [\partial_I \SF{X}^{\uA}]}{\det [\Dut_I \SF{X}^{\uA}]}  = d^d\sigma\, d\theta\, d\thb \; \frac{ \sqrt{-\SF{\gref}}}{\SF{{\bf z}}} \,.
\end{equation}
Eq.~\eqref{eq:intparts} is, of course, just a way of stating the ability to integrate by parts in the worldvolume theory.

\section{Superspace expansion of further fields in the MMO limit}
\label{sec:MMOfurther}

We collect here some more involved superspace expansions of superfields in the MMO limit. This is relevant for the calculations in \S\ref{sec:mmo}.

\paragraph{Non-covariant derivatives of the metric:} To construct the worldvolume connection \eqref{eq:wvCalt}, we need the non-covariant derivatives $\Dut_K \SF{\gref}_{IJ}$ of the metric, recalling that $\Dut_I \equiv \partial_I + (\SF{\Ascr}_I, \,\cdot\,)_\Kref$ is in some sense a $\UT$ covariant partial derivative:
\begin{equation}
\begin{split}
\Dut_\theta \SF{\gref}_{ab} 
&= 
	- \thb \,	\prn{ \source{h_{ab}} + \tilde{\gref}_{ab} 
	- (\phizT-\sBdel, \gref_{ab})_\Kref} 
 \\
\Dut_\thb \SF{\gref}_{ab} 
&= 
	\theta \, \prn{ \source{h_{ab}} + \tilde{\gref}_{ab} + (\sBdel, \gref_{ab})_\Kref}  \\
\Dut_\theta \SF{\gref}_{\thetab a}
& = 
	  \prn{\tilde{X}_\mu + (\sBdel, X_\mu)_\Kref } \, \partial_a X^\mu 
	  + \thb\theta \, \prn{\phizT-\sBdel, ( \tilde{X}_\mu +(\sBdel,X_\mu)_\Kref \, 
	   \partial_a X^\mu}_\Kref
\\
\Dut_\thb \SF{\gref}_{\thetab a}
& = 
\Dut_\theta  \SF{\gref}_{\theta a}  = 0
  \\
 D_\thb  \SF{\gref}_{\theta a}  
  & = 
  	\prn{ \mathfrak{J}_\nu - (\sBdel, X_\nu)_\Kref } \, \partial_a X^\nu  - \thb \theta \, ( \sBdel, \mathfrak{J}_\mu - (\sBdel, X_\mu)_\Kref ) \partial_a X^\mu )_\Kref
 \\
\Dut_\theta \SF{\gref}_{\theta \thetab} 
&=  
 	- \thb\, 
 	 \prn{  
 	 \source{h_{\theta\thb}} +\tilde{\gref}_{\theta\thb}}
 \\
\Dut_\thb \SF{\gref}_{\theta \thetab} 
&=  
 	\theta\, 
 	 \prn{  
 	 \source{h_{\theta\thb}} + \tilde{\gref}_{\theta\thb}}
\end{split}
\label{eq:Dtgref}
\end{equation}

\paragraph{The worldvolume connection:} Using the above expansions, we can now write the connection \eqref{eq:wvCalt} in the MMO limit.
In components, the non-zero parts of the connection read as follows:
{\small
\begin{subequations}
\begin{align}
 \Cwv^c{}_{ab} 
&= \Gamma^c{}_{ab} + \thb\theta \bigg\{
\frac{1}{2}\, \gref^{cd} \Big[ 2\nabla_{(a} (\source{h_{b)d}}+\tilde{\gref}_{b)d}) +2 (\source{\mathcal{F}_{(a}}, \gref_{b)d})_\Kref   -  \nabla_d (\source{h_{ab}}+\tilde{\gref}_{ab}) - (\sAt{d},\gref_{ab})_\Kref \Big] \nonumber \\
 &\quad -\frac{i}{2} \, e_\nu^c \prn{\tilde{X}^\nu + (\sBdel,X^\nu)_\Kref} \prn{ 2\, \nabla_{(a}\brk{\partial_{b)}X^\mu \prn{\mathfrak{J}_{\mu} - (\sBdel, X_\mu)_\Kref}} +  \source{h_{ab}} + \tilde{\gref}_{ab} - ( \phizT-\sBdel, \gref_{ab})_\Kref }  \nonumber \\
 &\quad -\frac{i}{2} \,e_\nu^c \prn{\mathfrak{J}^\nu - (\sBdel,X^\nu)_\Kref} \prn{ 2\, \nabla_{(a}\brk{\partial_{b)}X^\mu \prn{\tilde{X}_{\mu} + (\sBdel, X_\mu)_\Kref}} -  \source{h_{ab}} - \tilde{\gref}_{ab} + ( \sBdel, \gref_{ab})_\Kref }  \nonumber \\
 &\quad + 2\brk{
	 i  \,\tx^c \left(\source{\mathcal{F}_{(a}} + \partial_{(a} \phizT \right) - i\, \mathfrak{J}^c \, \source{\mathcal{F}_{(a}} } \Kref_{b)}
+ \gref^{cd} \, \left(\Kref_{(a}\,\nabla_{d} \source{\mathcal{F}_{b)}}  -\Kref_{(a} \,\nabla_{b)} \source{\mathcal{F}_{d}}  \right)
 \bigg\} 
 \nonumber \\
\Cwv^c{}_{a\theta} &= \thb \,\bigg\{ - \frac{1}{2} \, \gref^{cd} \prn{ \source{h_{da}} + \tilde{\gref}_{da} - (\phizT- \sBdel, \gref_{da})_\Kref} + \gref^{cd} \, \partial_{[a} \brk{ \partial_{d]}X^\mu \prn{\mathfrak{J}_{\mu} - (\sBdel, X_\mu)_\Kref} }   \nonumber \\
&\qquad + \frac{i}{2} \, e_\mu^c \prn{\mathfrak{J}^\mu - (\sBdel,X^\mu)_\Kref }\; \partial_a X^\nu (\tilde{X}_\nu+ \mathfrak{J}_\nu)  + \brk{\source{\mathcal{F}^c} + \partial^c \phizT - i\, \mathfrak{J}^c \, \phizT} \Kref_a \bigg\}
 \nonumber \\
\Cwv^c{}_{a\thb} &= \theta \,\bigg\{  \frac{1}{2} \, \gref^{cd} \prn{ \source{h_{da}} + \tilde{\gref}_{da} + (\sBdel, \gref_{da})_\Kref} + \gref^{cd} \, \partial_{[a} \brk{\partial_{d]}X^\mu \prn{\tilde{X}_{\mu} + (\sBdel, X_\mu)_\Kref}  } \nonumber \\
&\qquad - \frac{i}{2} \, e_\mu^c \prn{\tilde{X}^\mu + (\sBdel,X^\mu)_\Kref }\; \partial_a X^\nu (\tilde{X}_\nu + \mathfrak{J}_\nu) + \brk{-\source{\mathcal{F}^c} + i\, \tx^c \, \phizT} \Kref_a 
\bigg\}  \nonumber \\
\Cwv^\theta{}_{ab} &= -\frac{i}{2} \, \theta \, \bigg\{ - 2 \, \nabla_{(a} \brk{ \partial_{b)}X^\mu \prn{\tilde{X}_{\mu} + (\sBdel, X_\mu)_\Kref}  }+ \source{h_{ab}} + \tilde{\gref}_{ab} + (\sBdel, \gref_{ab})_\Kref  \nonumber \\
&\qquad -2\, \brk{\sAt{a}\, \Kref_b + \sAt{b} \, \Kref_a}\bigg\}
\end{align}
\begin{align}
\Cwv^\thb{}_{ab} &= \frac{i}{2} \, \thb \, \bigg\{  - 2 \, \nabla_{(a} \brk{ \partial_{b)}X^\mu \prn{\mathfrak{J}_{\mu} - (\sBdel, X_\mu)_\Kref} }- \source{h_{ab}} - \tilde{\gref}_{ab} + (\phizT-\sBdel, \gref_{ab})_\Kref  \nonumber \\
&\qquad + 2\, \brk{\left( \sAt{a} + \partial_a \phizT\right) 
	\Kref_b +\left(\sAt{b} + \partial_b \phizT\right) \, \Kref_a}\bigg\} 
 \nonumber \\
\Cwv^c{}_{\theta\thb} &= \frac{1}{2} \, e_\mu^c\,\prn{ \tilde{X}^\mu - \mathfrak{J}^\mu + 2 (\sBdel, X^\mu)_\Kref } +  \thb\theta \, \bigg\{ \partial_bX^\mu \, \prn{\tilde{X}_\mu - \mathfrak{J}_\mu + 2(\sBdel,X_\mu)_\Kref } \times  \nonumber \\
&\qquad\quad \times \Big[ -i \, \gref^{ab} (\source{h_{\theta\thb}} + \tilde{\gref}_{\theta\thb}) - \frac{1}{2} \prn{ \source{h^{ab}} + \tilde{\gref}^{ab} -  2i \, e^{(a}_\rho e^{b)}_\sigma \, \prn{ \mathfrak{J}^\rho - (\sBdel,X^\rho)_\Kref } \prn{ \tilde{X}^\sigma + (\sBdel,X^\sigma)_\Kref } } \Big]  \nonumber \\
&\qquad\quad+ \frac{1}{2} \, \gref^{ab} \, \Big[ \prn{\phizT-\sBdel,  (\tilde{X}_\nu + (\sBdel, X_\nu)_\Kref )\partial_b X^\nu }_\Kref + \prn{ \sBdel, (\mathfrak{J}_\nu - (\sBdel, X_\nu)_\Kref ) \partial_b X^\nu  }_\Kref  \nonumber \\
&\qquad\quad - \partial_b (\source{h_{\theta\thb}} + \tilde{\gref}_{\theta\thb}) \Big]
-\brk{
	\source{\mathcal{F}^c} (\mathfrak{J}^d-\tx^d) - \partial^c \phizT \;\tx^d 
	+ i \, \phizT\left(\mathfrak{J}^c \tx^d - \tx^c \; \mathfrak{J}^d \right)  
	}\Kref_d 
	 \bigg\}
 \nonumber \\
\Cwv^\theta{}_{a\theta} &=  \frac{i}{2} \, \phizT\,\Kref_a+  \thb\theta \, \bigg\{ i\, \tilde{X}^c \, \nabla_{[a}  \brk{\partial_{c]}X^\mu \prn{\mathfrak{J}_{\mu} - (\sBdel, X_\mu)_\Kref} }  \nonumber \\
&\qquad\quad - \frac{i}{2} \prn{ \source{h_{ac}} + \tilde{\gref}_{ac} - (\phizT-\sBdel, \gref_{ac})_\Kref } \, e_\nu^c \prn{\tilde{X}^\mu + (\sBdel, X^\nu)_\Kref} \nonumber \\ 
&\qquad\quad  - \frac{i}{2} \, \partial_a (\source{h_{\theta\thb}} + \tilde{\gref}_{\theta\thb} ) + \frac{1}{2} \, \source{h_{\theta\thb}} \, \partial_a X^\nu  \prn{\tilde{X}_\nu + \mathfrak{J}_\nu+ 2(\sBdel,X_\nu)_\Kref } \nonumber \\
&\qquad\quad + \frac{i}{2} \, (\sBdel, \partial_aX^\nu \prn{\mathfrak{J}_\nu - (\sBdel,X_\nu)_\Kref} )_\Kref - \frac{i}{2} (\phizT-\sBdel, \partial_aX^\nu \prn{\tilde{X}_\nu + (\sBdel,X_\nu)_\Kref} )_\Kref   \nonumber \\
&\qquad\quad 
	-\phizT\, \Kref_a \left( \source{h_{\theta\thb}} + i\, \Kref^d \, \sAt{d}\right)
	+i\, \tx^d\left(\sAt{d} + \partial_d \phizT \right)\Kref_a 
	-i\, \mathfrak{J}^d \Kref_d \, \sAt{a}
	+ i\,\phizT\left(\source{h_{ab}} + \tilde{\gref}_{ab} \right)\Kref^b 
	+ i \phizT \, \frac{\tilde{T}}{T} \, \Kref_a 
\bigg\}
 \nonumber \\
\Cwv^\thb{}_{a\thb} &= -\frac{i}{2} \, \phizT\,\Kref_a - \thb\theta \, \bigg\{ -i\, \mathfrak{J}^c \, \nabla_{[a} \brk{ \partial_{c]}X^\mu \prn{\tilde{X}_{\mu} + (\sBdel, X_\mu)_\Kref} }  \nonumber \\
&\qquad\quad  - \frac{i}{2} \prn{ \source{h_{ac}} + \tilde{\gref}_{ac} + (\sBdel, \gref_{ac})_\Kref } \, e_\nu^c \prn{\mathfrak{J}^\nu - (\sBdel,X^\nu)_\Kref }  \nonumber \\ 
&\qquad\quad + \frac{i}{2} \, \partial_a (\source{h_{\theta\thb}} + \tilde{\gref}_{\theta\thb} )+ \frac{1}{2} \, \source{h_{\theta\thb}} \, \partial_a X^\nu  \prn{\tilde{X}_\nu + \mathfrak{J}_\nu+ 2(\sBdel,X_\nu)_\Kref }   \nonumber \\
&\qquad\quad + \frac{i}{2} (\sBdel, \partial_aX^\nu \prn{\mathfrak{J}_\nu - (\sBdel,X_\nu)_\Kref} )_\Kref - \frac{i}{2} \, (\phizT-\sBdel, \partial_aX^\nu \prn{\tilde{X}_\nu + (\sBdel,X_\nu)_\Kref} )_\Kref  \nonumber \\
&\qquad\quad - 
	\phizT\, \Kref_a \left( \source{h_{\theta\thb}} + i\, \Kref^d \, \sAt{d}\right)
	+i\, \mathfrak{J}^d \sAt{d} \; \Kref_a 
	-i\, \tx^d \Kref_d \,\left(\sAt{a} + \partial_a \phizT\right)
	+ i\,\phizT\left(\source{h_{ab}} + \tilde{\gref}_{ab} \right)\Kref^b 
	+ i \phizT \, \frac{\tilde{T}}{T} \, \Kref_a 
	\bigg\}  \nonumber \\
\Cwv^\theta{}_{\theta\thb} &= -\frac{i}{2} \, \theta \, \bigg\{ \prn{\tilde{X}^\mu + (\sBdel, X^\mu)_\Kref } \prn{ \tilde{X}_\mu - \mathfrak{J}_\mu + 2 (\sBdel, X_\mu)_\Kref }  + 2 \, (\source{h_{\theta\thb}} + \tilde{\gref}_{\theta\thb} ) - 2 \, \phizT \, \tx^d \Kref_d \bigg\}  \nonumber \\
\Cwv^\thb{}_{\thb\theta} &= \frac{i}{2} \, \thb \, \bigg\{ -\prn{\mathfrak{J}^\mu - (\sBdel, X^\mu)_\Kref } \prn{ \tilde{X}_\mu - \mathfrak{J}_\mu + 2 (\sBdel, X_\mu)_\Kref }  + 2 \, (\source{h_{\theta\thb}} + \tilde{\gref}_{\theta\thb} ) - 2\, \phizT \, \mathfrak{J}^d \Kref_d\bigg\}
\end{align}
\end{subequations}
}\normalsize
Note the following simple identity for the supertrace of the connection: 
\begin{equation}\label{eq:Ctrace}
\begin{split}
(-)^I\, \Cwv^I{}_{Ic} &\equiv \thb \theta \, (-)^I\, \Cwvt^I{}_{Ic} \,,\\
   (-)^I\, \Cwvt^I{}_{Ic} &= 
   \, \frac{1}{2}\,  \gref^{ab} \Big[ 2\nabla_{(a} (\source{h_{c)b}}+\tilde{\gref}_{c)b}) +2 (\source{\mathcal{F}_{(a}}, \gref_{c)b})_\Kref   -  \nabla_b (\source{h_{ac}}+\tilde{\gref}_{ac}) - (\sAt{b},\gref_{ac})_\Kref \Big]  \\
   &\quad + \Kref^a \, \nabla_{[c} \source{\mathcal{F}}_{a]} + i\,  \partial_c \source{h_{\theta\thb}}   \,.
\end{split}
\end{equation}

\paragraph{Derivatives of field strengths:}

For the discussion of dissipative transport (such as \S\ref{sec:mmo1st}) we need covariant derivatives of field strengths. In the MMO limit, these take the following form: 
{\small
\begin{equation}
\begin{split}
  \Dwv_a \SF{\mathscr{F}}_{\theta b} &= \thb \, \bigg\{ - \nabla_a (\sAt{b} + \partial_b \phizT)- i\,\phizT ( \source{\mathcal{F}_{a}} + \partial_{a}\phizT) \, \Kref_{b} - \frac{i}{2} \, \phizT \Kref_a (\sAt{b} + \partial_b \phizT) \\
  &\qquad + \frac{i}{2} \, \phizT \, \prn{ 2 \nabla_{(a} \brk{ \partial_{b)} X^\mu (\mathfrak{J}_\mu - (\sBdel,X_\mu)_\Kref )}  + \source{h_{ab}} + \tilde{\gref}_{ab} - (\phizT-\sBdel, \gref_{ab})_\Kref } \bigg\} 
  \\
   \Dwv_a \SF{\mathscr{F}}_{\thb b} &= \theta \, \bigg\{  \nabla_a \sAt{b}  - i\,\phizT \ \source{\mathcal{F}_{a}} \, \Kref_{b} - \frac{i}{2} \, \phizT \, \Kref_a \, \sAt{b} \\
  &\qquad - \frac{i}{2} \, \phizT \, \prn{ 2 \nabla_{(a} \brk{ \partial_{b)} X^\mu (\tilde{X}_\mu + (\sBdel,X_\mu)_\Kref )} - \source{h_{ab}} - \tilde{\gref}_{ab} + (\sBdel, \gref_{ab})_\Kref } \bigg\}
  \\
  \Dwv_a \SF{\mathscr{F}}_{\theta\thb} 
  &=  \Dwv_a \SF{\mathscr{F}}_{\thb\theta} 
  =  \partial_a\phizT + \thb \theta \, 
  \bigg\{ (\sAt{a}, \phizT)_\Kref 	+ i \,\phizT\, \partial_a (\source{h_{\theta\thb}} + \tilde{\gref}_{\theta\thb} )  \\
  &\qquad + \sAt{b} \Big[ 
   \frac{i}{2} \, e_\mu^b \prn{\mathfrak{J}^\mu - (\sBdel,X^\mu)_\Kref }\; \partial_a X^\nu (\tilde{X}_\nu+ \mathfrak{J}_\nu)   \Big] \\
  &\qquad - (\sAt{b} + \partial_b \phizT) \Big[ 
    \frac{i}{2} \, e_\mu^b \prn{\tilde{X}^\mu + (\sBdel,X^\mu)_\Kref }\; \partial_a X^\nu (\tilde{X}_\nu + \mathfrak{J}_\nu) 
  \Big] \\
 &\qquad + (\source{\mathcal{F}^c} -   i\,\phizT\, \tilde{X}^c) \, \nabla_{[a}  \brk{\partial_{c]}X^\mu \prn{\mathfrak{J}_{\mu} - (\sBdel, X_\mu)_\Kref} } \\
 &\qquad + (\source{\mathcal{F}^c} + \partial^c \phizT- i\,\phizT\, \mathfrak{J}^c) \, \nabla_{[a}  \brk{\partial_{c]}X^\mu \prn{\tilde{X}_{\mu} + (\sBdel, X_\mu)_\Kref} }\\
&\qquad + \brk{ - \frac{1}{2} \, \source{\mathcal{F}^c} + \frac{i}{2} \, \phizT \, e^c_\nu \prn{ \tilde{X}^\nu + (\sBdel,X^\nu)_\Kref } }\prn{ \source{h_{ac}} + \tilde{\gref}_{ac} - (\phizT-\sBdel, \gref_{ac})_\Kref }\\ 
&\qquad + \brk{ \frac{1}{2} \,(\source{\mathcal{F}^c} + \partial^c \phizT) - \frac{i}{2} \, \phizT \, e^c_\nu \prn{ \mathfrak{J}^\nu - (\sBdel,X^\nu)_\Kref }}\prn{ \source{h_{ac}} + \tilde{\gref}_{ac} + (\sBdel, \gref_{ac})_\Kref }\\ 
&\qquad +\phizT\prn{
	i\, \mathfrak{J}^d \Kref_d \, \sAt{a}
	-i \, \tilde{X}^d \Kref_d  (\sAt{a} + \partial_a \phizT)
	}
\Big]
\bigg\}
\\
  \Dwv_\thb \SF{\mathscr{F}}_{\theta a} &= - (\sAt{a} + \partial_a \phizT) 
   - \frac{i}{2}\, (\phizT)^2 \Kref_a  
  + \thb \theta \bigg\{ (\sBdel,\sAt{a} + \partial_a \phizT)_\Kref -  (\tilde{X}^\mu - \mathfrak{J}^\mu + 2(\sBdel,X^\mu)_\Kref) e_\mu^b \, \partial_{[a} \source{\mathcal{F}}_{b]}
   \\
  &\quad 
  + \frac{i}{2} \, (\sAt{a} + \partial_a \phizT)\brk{\prn{\tilde{X}^\mu + (\sBdel, X^\mu)_\Kref } \prn{ \tilde{X}_\mu - \mathfrak{J}_\mu + 2 (\sBdel, X_\mu)_\Kref }  + 2 \, (\source{h_{\theta\thb}} + \tilde{\gref}_{\theta\thb} ) - 2 \, \phizT \, \tx^d \Kref_d} \\
 &\quad+ \frac{i}{2} \, \sAt{a} \brk{-\prn{\mathfrak{J}^\mu - (\sBdel, X^\mu)_\Kref } \prn{ \tilde{X}_\mu - \mathfrak{J}_\mu + 2 (\sBdel, X_\mu)_\Kref }  + 2 \, (\source{h_{\theta\thb}} + \tilde{\gref}_{\theta\thb} ) - 2\, \phizT \, \mathfrak{J}^d \Kref_d} \\
 &\quad - (\sAt{c} + \partial_c \phizT) \bigg[ 
 \frac{1}{2} \, \gref^{cd} \prn{ \source{h_{da}} + \tilde{\gref}_{da} + (\sBdel, \gref_{da})_\Kref} + \gref^{cd} \, \partial_{[a} \brk{\partial_{d]}X^\mu \prn{\tilde{X}_{\mu} + (\sBdel, X_\mu)_\Kref}  }\\
&\qquad\qquad\qquad\qquad - \frac{i}{2} \, e_\mu^c \prn{\tilde{X}^\mu + (\sBdel,X^\mu)_\Kref }\; \partial_a X^\nu (\tilde{X}_\nu + \mathfrak{J}_\nu) - \brk{\source{\mathcal{F}^c} - i\, \tx^c \, \phizT} \Kref_a 
 \bigg]
 + \phizT \, \Cwvt^\thb{}_{\thb a}
 \bigg\}
  \\
  \Dwv_\theta \SF{\mathscr{F}}_{\thb a} &= \sAt{a} 
  + \frac{i}{2}\, (\phizT)^2 \Kref_a + \thb \theta \bigg\{ (\phizT-\sBdel,\sAt{a} )_\Kref
  + (\tilde{X}^\mu - \mathfrak{J}^\mu + 2(\sBdel,X^\mu)_\Kref) e_\mu^b \, \partial_{[a} \source{\mathcal{F}}_{b]}
  \\
  &\quad 
  - \frac{i}{2} \,\sAt{a} \brk{-\prn{\mathfrak{J}^\mu - (\sBdel, X^\mu)_\Kref } \prn{ \tilde{X}_\mu - \mathfrak{J}_\mu + 2 (\sBdel, X_\mu)_\Kref }  + 2 \, (\source{h_{\theta\thb}} + \tilde{\gref}_{\theta\thb} ) -2 \, \phizT \, \mathfrak{J}^d \Kref_d} \\
 &\quad- \frac{i}{2} \, (\sAt{a} + \partial_a \phizT) \brk{\prn{\tilde{X}^\mu + (\sBdel, X^\mu)_\Kref } \prn{ \tilde{X}_\mu - \mathfrak{J}_\mu + 2 (\sBdel, X_\mu)_\Kref }  + 2 \, (\source{h_{\theta\thb}} + \tilde{\gref}_{\theta\thb} ) - 2\, \phizT \, \tilde{X}^d \Kref_d} \\
 &\quad - \sAt{c}  \bigg[ 
- \frac{1}{2} \, \gref^{cd} \prn{ \source{h_{da}} + \tilde{\gref}_{da} - (\phizT-\sBdel, \gref_{da})_\Kref} + \gref^{cd} \, \partial_{[a} \brk{\partial_{d]}X^\mu \prn{\mathfrak{J}_{\mu} - (\sBdel, X_\mu)_\Kref}  }\\
&\qquad\qquad\quad + \frac{i}{2} \, e_\mu^c \prn{\mathfrak{J}^\mu - (\sBdel,X^\mu)_\Kref }\; \partial_a X^\nu (\tilde{X}_\nu + \mathfrak{J}_\nu) + \brk{\source{\mathcal{F}^c}+\partial^c \phizT - i\, \mathfrak{J}^c \, \phizT} \Kref_a 
 \bigg]
  - \phizT \, \Cwvt^\theta{}_{\theta a}
 \bigg\}
 \end{split}
 \end{equation}
 \begin{equation}
 \begin{split}
  \Dwv_\thb \SF{\mathscr{F}}_{\theta \thb} &= \theta \,\bigg\{ (\sBdel,\phizT)_\Kref - \frac{1}{2} \,  \sAt{a} \, e^a_\mu \prn{ \tilde{X}^\mu - \mathfrak{J}^\mu + 2 (\sBdel, X^\mu)_\Kref }\\
  &\qquad  + \frac{i}{2} \,  \phizT \brk{ \prn{\tilde{X}^\mu + (\sBdel, X^\mu)_\Kref } \prn{ \tilde{X}_\mu - \mathfrak{J}_\mu + 2 (\sBdel, X_\mu)_\Kref }  + 2 \, (\source{h_{\theta\thb}} + \tilde{\gref}_{\theta\thb} )- 2 \, \phizT \, \tx^d \Kref_d }
  \bigg\}
  \\
  \Dwv_\theta \SF{\mathscr{F}}_{\theta \thb} &=  \thb \,\bigg\{ (\phizT-\sBdel,\phizT)_\Kref - \frac{1}{2} \, ( \sAt{a} + \partial_a \phizT) \, e^a_\mu \prn{ \tilde{X}^\mu - \mathfrak{J}^\mu + 2 (\sBdel, X^\mu)_\Kref }\\
  &\qquad  - \frac{i}{2} \,  \phizT \brk{ -\prn{\mathfrak{J}^\mu + (\sBdel, X^\mu)_\Kref } \prn{ \tilde{X}_\mu - \mathfrak{J}_\mu + 2 (\sBdel, X_\mu)_\Kref }  + 2 \, (\source{h_{\theta\thb}} + \tilde{\gref}_{\theta\thb} ) -2 \, \phizT \, \mathfrak{J}^d \Kref_d }
  \bigg\}
\end{split}
\end{equation}
}\normalsize

\paragraph{Shear tensor:}

Now, we wish to compute the shear tensor and its square as it appears in \eqref{eq:bvterm3}. We start by computing the components of
\begin{equation}
\begin{split}
 \Dwv_I \SF{u}_J \equiv 
 (-)^{K} \, \prn{\Dwv_I \, \SF{u}^K  }\SF{\gref}_{KJ}
\,.
\end{split}
\label{eq:defDulow}
\end{equation}
We find:
{\small
\begin{subequations}
\begin{align}
\Dwv_a \SF{u}_b 
&= \nabla_a u_b + \thb \theta \bigg\{ \nabla_a \tilde{u}_b+\gref_{bc}\, (\sAt{a} , u^c)_\Kref + (\source{h_{bc}}+\tilde{\gref}_{bc}) \, \nabla_a u^c + \prn{ \Kref_{(a} \nabla_b \sAt{c)} - \Kref_{(a} \nabla_{c)} \sAt{b} } \, u^c   \nonumber \\
&\qquad + \frac{1}{2} \, u^c \, \brk{2 \nabla_{(a} (\source{h_{c)b}} + \tilde{\gref}_{c)b} ) + 2 (\source{\mathcal{F}_{(a}},\gref_{c)b})_\Kref -  \nabla_b (\source{h_{ac}} + \tilde{\gref}_{ac} ) - (\sAt{b}, \gref_{ac})_\Kref } \bigg\} \nonumber \\
\Dwv_a \SF{u}_\theta 
&= \thb \; \bigg\{\prn{ \mathfrak{J}_\mu - (\sBdel, X_\mu)_\Kref   } \partial_{c} X^\mu \, \nabla_a u^c + \nabla_{(a} \brk{\prn{ \mathfrak{J}_\mu - (\sBdel, X_\mu)_\Kref   } \partial_{c)} X^\mu \; }\, u^c   \nonumber \\
&\qquad + \frac{1}{2} \, u^d \prn{\source{h_{da}} +\tilde{\gref}_{da} - (\phizT-\sBdel, \gref_{da})_\Kref }
- u^c \prn{ \sAt{c} + \partial_c \phizT} \Kref_a +\frac{1}{T} \, (\sAt{a} + \partial_a \phizT)  \bigg\}\nonumber \\
\Dwv_a \SF{u}_\thb 
&= \theta \; \bigg\{\prn{ \tilde{X}_\mu + (\sBdel, X_\mu)_\Kref   } \partial_{c} X^\mu \, \nabla_a u^c + \nabla_{(a} \brk{\prn{ \tilde{X}_\mu + (\sBdel, X_\mu)_\Kref   } \partial_{c)} X^\mu \; }\, u^c   \nonumber \\
&\qquad - \frac{1}{2} \, u^d \prn{\source{h_{da}} +\tilde{\gref}_{da} + (\sBdel, \gref_{da})_\Kref }
+ u^c  \sAt{c} \, \Kref_a -\frac{1}{T} \, \sAt{a}   \bigg\}\nonumber \\
\Dwv_\theta \SF{u}_a 
&= \thb \; \bigg\{- \tilde{T}\, \Kref_a + \gref_{ac}\, (\phizT - \sBdel, u^c)_\Kref- \frac{1}{2} \, u^d \prn{\source{h_{da}} +\tilde{\gref}_{da} - (\phizT-\sBdel, \gref_{da})_\Kref } \nonumber \\
&\qquad  + \frac{i}{2} \,  \partial_a X^\mu \, \prn{\mathfrak{J}_\mu - (\sBdel,X_\mu)_\Kref }\; u^c\,\partial_c X^\nu (\tilde{X}_\nu + \mathfrak{J}_\nu) + u^c \, \partial_{[c} \brk{ \partial_{a]} X^\mu \; \prn{ \mathfrak{J}_\mu - (\sBdel,X_\mu)_\Kref} }\nonumber \\
&\qquad - \frac{1}{T} \,\brk{\sAt{c} + \partial_c\phizT -\frac{i\phizT}{2} \, \prn{\mathfrak{J}_\mu - (\sBdel,X_\mu)_\Kref} \, \partial_c X^\mu }
 \bigg\}\nonumber \\
\Dwv_\thb \SF{u}_a 
&= \theta \; \bigg\{ \tilde{T}\, \Kref_a + \gref_{ac}\, (\sBdel, u^c)_\Kref+ \frac{1}{2} \, u^d \prn{\source{h_{da}} +\tilde{\gref}_{da} + (\sBdel, \gref_{da})_\Kref } \nonumber \\
&\qquad  - \frac{i}{2} \,  \partial_a X^\mu \, \prn{\tilde{X}_\mu + (\sBdel,X_\mu)_\Kref }\; u^c\,\partial_c X^\nu (\tilde{X}_\nu + \mathfrak{J}_\nu) + u^c \, \partial_{[c} \brk{ \partial_{a]} X^\mu \; \prn{ \tilde{X}_\mu + (\sBdel,X_\mu)_\Kref} }\nonumber \\
&\qquad + \frac{1}{T} \,\brk{\sAt{a}  -\frac{ i\phizT}{2} \, \prn{\tilde{X}_\mu + (\sBdel,X_\mu)_\Kref} \, \partial_a X^\mu }
 \bigg\}\nonumber \\
\end{align}
\begin{align}
\Dwv_\theta \SF{u}_\thb 
& =   \frac{\phizT}{2T} +\thb\theta \, \bigg\{ - \Kref^a e_a^\mu \prn{\tilde{X}_\mu + (\sBdel,X_\mu)_\Kref} \, \tilde{T} - \frac{\phizT}{2T^2} \, \tilde{T} + \prn{\phizT-\sBdel, \, \prn{ \tilde{X}_\mu +  (\sBdel, X_\mu)_\Kref  }  \partial_c X^\mu \; u^c }_\Kref \nonumber \\
&\qquad - \frac{1}{2} \, u^a \brk{ \prn{ \phizT- \sBdel, (\tilde{X}_\mu + (\sBdel,X_\mu)_\Kref)\partial_a X^\mu ) }_\Kref  + \prn{  \sBdel, (\mathfrak{J}_\mu - (\sBdel,X_\mu)_\Kref)\partial_a X^\mu ) }_\Kref } \nonumber \\
&\qquad  + \frac{1}{2} \, (u.\partial)(\source{h_{\theta\thb}}+\tilde{\gref}_{\theta\thb})
 - \frac{\phizT}{T} (\Kref.\source{\mathcal{F}}) + (\Kref.\source{\mathcal{F}}) \, u^a \, e_a^\mu \prn{ \mathfrak{J}_\mu - (\sBdel,X_\mu)_\Kref } 
 \bigg\} \nonumber \\
\Dwv_\thb \SF{u}_\theta 
& =   \frac{\phizT}{2T} +\thb\theta \, \bigg\{- \Kref^a e_a^\mu \prn{\mathfrak{J}_\mu - (\sBdel,X_\mu)_\Kref} \, \tilde{T} - \frac{\phizT}{2T^2} \, \tilde{T} - \prn{\sBdel, \, \prn{ \mathfrak{J}_\mu -  (\sBdel, X_\mu)_\Kref  }  \partial_c X^\mu \; u^c }_\Kref \nonumber \\
&\qquad + \frac{1}{2} \, u^a \brk{ \prn{ \phizT- \sBdel, (\tilde{X}_\mu + (\sBdel,X_\mu)_\Kref)\partial_a X^\mu ) }_\Kref  + \prn{  \sBdel, (\mathfrak{J}_\mu - (\sBdel,X_\mu)_\Kref)\partial_a X^\mu ) }_\Kref } \nonumber \\
&\qquad   -\frac{1}{2} \, (u.\partial)(\source{h_{\theta\thb}}+\tilde{\gref}_{\theta\thb}) - \frac{\phizT}{T} (\Kref. \source{\mathcal{F}})
+ \Kref^c (\source{\mathcal{F}}_c + \partial_c \phizT) \, u^a \, e_a^\mu \prn{ \tilde{X}_\mu + (\sBdel,X_\mu)_\Kref } 
 \bigg\} \nonumber \\
\Dwv_\theta \SF{u}_\theta& = \Dwv_\thb \SF{u}_\thb = 0	 
\end{align}
\end{subequations}
}\normalsize
This implies for $\SF{\Sigma}_{IJ}$ defined as 
\begin{equation}
\SF{\Sigma}_{IJ} = \Dwv_I \SF{u}_J + (-)^{IJ} \, \Dwv_J \SF{u}_I
\end{equation}	
one finds the explicit expressions
\small{
\begin{equation}
\begin{split}
\SF{\Sigma}_{ab} 
&=2\,\nabla_{(a} u_{b)} + \thb \theta \bigg\{2\, \nabla_{(a} \tilde{u}_{b)}+2\,\gref_{(ac}\, (\sAt{b)} , u^c)_\Kref +2\, (\source{h_{(ac}}+\tilde{\gref}_{(ac}) \, \nabla_{b)} u^c + \prn{ \Kref_{a} \nabla_{[b} \sAt{c]} + \Kref_{b} \nabla_{[a} \sAt{c]} } \, u^c   \\
&\qquad + u^c \, \brk{\nabla_{c} (\source{h_{ab}} + \tilde{\gref}_{ab} ) + (\source{\mathcal{F}_{c}},\gref_{ab})_\Kref  } \bigg\}  \\
\SF{\Sigma}_{a\theta} &= 
\thb \; \bigg\{- \tilde{T}\, \Kref_a + \prn{ \mathfrak{J}_\mu - (\sBdel, X_\mu)_\Kref   } \partial_{c} X^\mu \, \nabla_a u^c   + u^c \, \nabla_{c} \brk{ \partial_{a} X^\mu \; \prn{ \mathfrak{J}_\mu - (\sBdel,X_\mu)_\Kref} } \\
&\qquad - u^c \prn{ \sAt{c} + \partial_c \phizT} \Kref_a  + \gref_{ac}\, (\phizT - \sBdel, u^c)_\Kref   \bigg\}
 \\
\SF{\Sigma}_{a\thb} &= 
\theta \; \bigg\{\tilde{T}\, \Kref_a + \prn{ \tilde{X}_\mu + (\sBdel, X_\mu)_\Kref   } \partial_{c} X^\mu \, \nabla_a u^c   + u^c \, \nabla_{c} \brk{ \partial_{a} X^\mu \; \prn{ \tilde{X}_\mu + (\sBdel,X_\mu)_\Kref} } \\
&\qquad + u^c \, \sAt{c} \, \Kref_a  + \gref_{ac}\, (\sBdel, u^c)_\Kref   \bigg\}
\\
 \SF{\Sigma}_{\theta\thb} 
& =  \thb\theta \, \bigg\{
 (u.\partial)(\source{h_{\theta\thb}}+\tilde{\gref}_{\theta\thb})
 + (\Kref.\source{\mathcal{F}}) \, u^a \, e_a^\mu \prn{ \mathfrak{J}_\mu - (\sBdel,X_\mu)_\Kref } 
 - \Kref^c (\source{\mathcal{F}}_c + \partial_c \phizT) \, u^a \, e_a^\mu \prn{ \tilde{X}_\mu + (\sBdel,X_\mu)_\Kref } 
 \\
&\qquad  + \Kref^a e_a^\mu \prn{\mathfrak{J}_\mu-\tilde{X}_\mu +2\, (\sBdel,X_\mu)_\Kref} \, \tilde{T}  \\
&\qquad 
 + \prn{\phizT-\sBdel, u^c }_\Kref \, \prn{ \tilde{X}_\mu +  (\sBdel, X_\mu)_\Kref  }  \partial_c X^\mu + \prn{\sBdel, u^c }_\Kref \, \prn{ \mathfrak{J}_\mu -  (\sBdel, X_\mu)_\Kref  }  \partial_c X^\mu  
\bigg\} \\
\end{split}
\label{eq:}
\end{equation}
}\normalsize
Note that the projected trace of $\SF{\Sigma}_{IJ}$ is the same as the usual (metric) trace in superspace: 
\begin{equation}
\label{eq:strace}
\begin{split}
(-)^{P+Q+PQ} \, \SF{P}^{PQ} \, \SF{\Sigma}_{PQ} 
  &=2 \vartheta + \thb\theta \, \Big\{  (u.\nabla) \left[ P^{ab}(\source{h_{ab}} + \tilde{\gref}_{ab})\right]   + \frac{2}{T}\,(\tilde{T} + u^c \sAt{c}) \, \vartheta  + 2 i \, (u.\partial) \source{h_{\theta\thb}} \Big\} \\
  &=  (-)^{P+Q+PQ} \, \SF{g}^{PQ} \, \SF{\Sigma}_{PQ} \\
  &  = 2\,\Dwv_I \SF{u}^I \equiv 2 \SF{\vartheta}  \,.
\end{split}
\end{equation}
Let us also write down the projector with indices up and down:
\begin{equation}
\begin{split}
\SF{P}_I{}^J &\equiv \SF{\gref}_{IK} \, \SF{P}^{KJ} \\
\SF{P}_a{}^b 
&= P_a^b + \thb \theta \, P_a^c  \,u^d (\source{h_{cd}} + \tilde{\gref}_{cd}) \, u^b\\
\SF{P}_\theta{}^a 
&= \thb \, \prn{ \mathfrak{J}^\mu - (\sBdel, X^\mu)_\Kref } e_\mu^b \, u_b \, u^a  \\
\SF{P}_\thb{}^a 
&= \theta \, \prn{ \tilde{X}^\mu + (\sBdel, X^\mu)_\Kref } e_\mu^b \, u_b \, u^a \\
\SF{P}_a{}^\theta = \SF{P}_a{}^\thb
&= 0 \\
\SF{P}_\theta{}^\theta = \SF{P}_\thb{}^\thb 
&= -1 
 \end{split}
\label{eq:projSFud}
\end{equation}
Note that the Grassmann-odd directions count in the super-trace, for example: the metric trace $\SF{\gref}_I{}^I = d-2$, whereas $\SF{P}_I{}^I = d-3$, as one can readily verify ($d$ denotes the number of spacetime dimensions of the physical fluid).
This finally let's us compute the superfield version of the graded symmetric, transverse shear tensor:
\small{
\begin{equation}
\begin{split}
\SF{\sigma}_{IJ} &\equiv \frac{1}{2}\,(-)^{M(1+J+N) + N} \, \SF{P}_I{}^M\, \SF{P}_J{}^N \prn{ \SF{\Sigma}_{MN} -  \frac{2}{d-1} \, \SF{\gref}_{MN} \, \SF{\vartheta}} \\
\SF{\sigma}_{ab} 
&= \sigma_{ab} + \thb\theta \, P_{(a}^eP_{b)}^f \Big\{\prn{ \source{h_{c(e}} + \tilde{\gref}_{c(e}} \nabla_{f)} u^c +\frac{1}{2} \, u^m u^n (\source{h_{mn}}+\tilde{\gref}_{mn}) \nabla_{e}u_f \\
&\qquad\qquad\qquad\qquad + \frac{1}{2} \, (u.\nabla) (\source{h_{ef}} + \tilde{\gref}_{ef})  
 + u^g (\source{h_{fg}}+\tilde{\gref}_{fg}) u^{d} \, \nabla_d u_e+ \gref_{(ec} \, (\source{\mathcal{F}_{f)}},u^c)_\Kref\\
&\qquad\qquad\qquad\qquad  + \frac{1}{2} \, u^c \, (\sAt{c}, \gref_{ef})_\Kref 
 - \frac{\vartheta}{d-1} \, (\source{h_{ef}}+\tilde{\gref}_{ef}) - \frac{1}{d-1} \, P_{ef} \, \tilde{\vartheta} \Big\}
\\
\SF{\sigma}_{a\theta} &= 
  \frac{1}{2} \,\thb \,P_a^b\, \Big\{   (u.\nabla) \brk{ \prn{\mathfrak{J}_\mu - (\sBdel,X_\mu)_\Kref} \partial_b X^\mu  } +  \prn{\mathfrak{J}_\mu - (\sBdel,X_\mu)_\Kref} \partial_c X^\mu \, \prn{ \nabla_b u^c   + u^c \, \acc_b } \\
&\qquad\qquad    - \frac{2\,\vartheta}{d-1} \, \partial_b X^\mu \prn{ \mathfrak{J}_\mu - (\sBdel,X_\mu)_\Kref } \Big\} \\
\SF{\sigma}_{a\thb} 
	&= 
   \frac{1}{2} \,\theta \,P_a^b\, \Big\{   (u.\nabla) \brk{ \prn{\tilde{X}_\mu + (\sBdel,X_\mu)_\Kref} \partial_b X^\mu  } +  \prn{\tilde{X}_\mu + (\sBdel,X_\mu)_\Kref} \partial_c X^\mu \, \prn{ \nabla_b u^c   + u^c \, \acc_b } \\
&\qquad\qquad    - \frac{2\,\vartheta}{d-1} \, \partial_b X^\mu \prn{ \tilde{X}_\mu + (\sBdel,X_\mu)_\Kref } \Big\} \\
 \SF{\sigma}_{\theta\thb} &= - \SF{\sigma}_{\thb\theta} 
	= -\frac{i}{d-1} \, \vartheta  + \thb \theta\,\frac{1}{2} \,  \bigg\{ (u.\nabla) \brk{ \source{h_{\theta\thb}} + \tilde{\gref}_{\theta\thb}+(u.\tilde{X})(u.\mathfrak{J}) } \\
	&\qquad\qquad\qquad\qquad\qquad\qquad - \frac{1}{d-1} \brk{ 2\vartheta \prn{\source{h_{\theta\thb}} + \tilde{\gref}_{\theta\thb} +(u.\tilde{X})(u.\mathfrak{J})} +2 i \, \tilde{\vartheta}} \bigg\}
\end{split}
\label{eq:supersigma}
\end{equation}
}\normalsize
where $2\tilde{\vartheta}$ can be read off as the top component of \eqref{eq:strace}. 
Note that this definition of shear tensor is not super-traceless. However, its bottom component $\SF{\sigma}_{ab}|$ reduces to the usual shear tensor, which is traceless in ordinary spacetime.


\providecommand{\href}[2]{#2}\begingroup\raggedright\endgroup

\end{document}